%% file: ESL.tex
\documentclass[a4paper,UKenglish,cleveref,autoref,thm-restate]{lipics-v2021}

\usepackage{cite}

\title{Exact Separation Logic (Extended Version)}
\subtitle{Towards bridging the gap between verification and bug-finding}

\titlerunning{Exact Separation Logic}

\author
	{Petar Maksimovi\'c}
	{Imperial College London, United Kingdom \and Runtime Verification Inc., United States of America}{}{}{}

\author
	{Caroline Cronj\"ager}
	{Ruhr-Universit\"at Bochum, Germany}{}{}{}

\author
	{Andreas L\"o\"ow}
	{Imperial College London, United Kingdom}{}{}{}

\author
	{Julian Sutherland}
	{Nethermind, United Kingdom}{}{}{}

\author
	{Philippa Gardner}
	{Imperial College London, United Kingdom}{}{}{}

\authorrunning{P.~Maksimovi\'c, C.~Cronj\"ager, A.~L\"o\"ow, J.~Sutherland, and P.~Gardner} 

\Copyright{Jane Open Access and Joan R. Public} 

\ccsdesc[500]{Theory of computation~Logic and verification, Separation logic, Hoare logic, Abstraction}

\keywords{Separation logic, program correctness, program incorrectness, abstraction}

\category{} 

\relatedversion{\href{https://arxiv.org/abs/2208.07200}{https://arxiv.org/abs/2208.07200} (extended version, i.e., this version)} 

\acknowledgements{
We thank Sacha-\'Elie Ayoun and Daniele Nantes Sobrinho for the many discussions that have improved the quality of the paper. We also thank the anonymous reviewers for their comments. Maksimović, L\"o\"ow, and Gardner were supported by the EPSRC Fellowship ‘VetSpec: Verified Trustworthy Software Specification’
(EP/R034567/1). Cronj\"ager was partially supported by the Erasmus Plus Student Mobility for Traineeships scheme.}

\nolinenumbers 

\EventEditors{John Q. Open and Joan R. Access}
\EventNoEds{2}
\EventLongTitle{37th European Conference on Object-Oriented Programming (ECOOP 2023)}
\EventShortTitle{ECOOP 2023}
\EventAcronym{ECOOP}
\EventYear{2023}
\EventDate{July 17--21, 2023}
\EventLocation{Seattle, United States of America}
\EventLogo{}
\SeriesVolume{42}
\ArticleNo{19}

\usepackage{sl}
\usepackage{mathtools}
\usepackage[utf8]{inputenc}

\usepackage{listings}
\lstdefinestyle{code}{
    commentstyle=\color{ForestGreen},
    keywordstyle=\color{Blue},
    stringstyle=\color{OliveGreen},
    basicstyle=\ttfamily\footnotesize,
    showstringspaces=false,
    upquote=true,
    numbers=left,
    numberstyle=\scriptsize\color{Gray},
    frame=lines
}

\lstdefinelanguage{wisl}{
  morekeywords={predicate, +, lemma, proof, apply, statement, 
                forall, if, assert, bind, unfold, function, return},
  keywordstyle=\color{Blue}\bfseries,
  sensitive=false,
  comment=[l]{//}
}

\lstdefinelanguage{ocaml}{
 language=caml,
 columns=[c]fixed,
 keywordstyle=\bfseries,
 upquote=true,
 commentstyle=,
 breaklines=true,
 showstringspaces=false,
 stringstyle=\color{blue},
 literate={'"'}{\textquotesingle "\textquotesingle}3
}

\usepackage{macros}
\usepackage{c_macros}
\usepackage{js_macros}

\usepackage{xspace}

\begin{document}

\maketitle

\begin{abstract}
Over-approximating (OX) program logics, such as separation logic (SL), are used for {\it verifying} properties of heap-manipulating programs: all terminating behaviour is characterised, but established results and errors need not be~reachable. OX function specifications are thus incompatible with true bug-finding supported by symbolic execution tools such as Pulse and Pulse-X. In contrast, under-approximating (UX) program logics, such as incorrectness separation logic, are used to {\it find} true results and bugs: established results and errors are reachable, but there is no mechanism for understanding if all terminating behaviour has been characterised.

We introduce exact separation logic (ESL), which provides fully-verified function specifications compatible with both OX verification and UX true bug-funding: all terminating behaviour is characterised and all established results and errors are reachable. We prove soundness for ESL with mutually recursive functions, demonstrating, for the first time, function compositionality for a UX logic. We show that UX program logics require subtle definitions of internal and external function specifications compared with the familiar definitions of OX logics. We investigate the expressivity of ESL and, for the first time, explore the role of abstraction in UX reasoning by verifying abstract ESL specifications of various data-structure algorithms. In doing so, we highlight the difference between {\it abstraction} (hiding information) and {\it over-approximation} (losing information). Our findings demonstrate that abstraction cannot be used as freely in UX logics as in OX logics, but also that it should be feasible to use ESL to provide tractable function specifications for self-contained, critical code, which would then be used for both verification and true bug-finding.
\end{abstract}

\input{sections/introduction}
\input{sections/language}

\input{sections/proofsystem}
\input{sections/examples}
\input{sections/related-work}

\input{sections/conclusion}
\bibliography{ESL.bib}

\newpage
\appendix
%
\input{sections/app-semantics}
\input{sections/app-proofrules}
\input{sections/app-soundness}

\input{sections/app-scott}

\input{sections/app-envsoundness}
\input{sections/app-examples}
\input{sections/app-sl-list-length}

\end{document}

%% file: sections/introduction.tex

\newcommand{\even}{\predd{even}}
\newcommand{\odd}{\predd{odd}}
\newcommand{\llmeas}{\mathbin{3|\lstvs|+2}}
\newcommand{\fmeas}{\mathbin{3|v:\lstvs'|+(v \bmod 2)}}
\newcommand{\gmeas}{\mathbin{3|v:\lstvs'|+1-(v \bmod 2)}}

\newcommand{\strlit}[1]{\mathtt{``#1"}}

\newcommand{\nlist}{\predd{list_{\Nat}}}

\section{Introduction}

Over-approximating (OX) program logics were introduced to reason about program correctness, starting with Hoare logic~\cite{hoare} and evolving to separation logic (SL) \cite{seplogic,reyseplogic}. 
SL is used for {\it verification} and features function specifications of the form $\triple P {f( \pvvar x)} Q$, the meaning of which is that all terminating executions of the function $f$ that start from a state in the pre-condition $P$ end in a state covered by the post-condition $Q$. 
SL has the standard  rule of  {\it forward consequence}, which allows one to {\it lose information}  (for example, if we had a post-condition with $x = 42$, we could soundly weaken this precise information to the less precise $x > 0$ or even to the non-informative $\true$).
In essence, the philosophy underlying the OX approach in general can be stated as: 
\begin{center}
\it no paths can be cut, but information can be lost.
\end{center} 

A key property of SL is that function specifications are {\it compositional}, enabling {\it scalable} reasoning about the heap. This is due to their {\it locality}, which allows the 
pre-condition to describe only the partial state sufficient for the function to execute,
and the {\it frame} property, which allows the function to be called in any larger state. 
SL function specifications have been used for verification of complex, real-world code in tools such as VeriFast~\cite{verifast}, Iris~\cite{jung:popl:2015}, and Gillian~\cite{gillianpldi,gilliancav}. However, given that their post-conditions may describe states that are not reachable from their pre-conditions, such OX specifications are not compatible with true bug-finding, as found, for example, in Meta's Pulse~\cite{isl} and Pulse-X~\cite{Le22} tools.

Under-approximating (UX) program logics were recently introduced, originating from reverse Hoare logic~(RHL)~\cite{reverselogic} 
for reasoning about correctness of probabilistic programs, and coming to prominence with incorrectness logic~\cite{il} and incorrectness separation logic (ISL)~\cite{isl}, which identified
their bug-finding potential. ISL function specifications are of the form $\isltripleok P {f(\pvvar x)} \Qok$ and $\isltripleerr P {f(\pvvar x)} \Qerr$, the meaning of which is that any state in the success post-condition $\Qok$ or the error post-condition $\Qerr$ is reachable from some state in the pre-condition $P$ by executing the function $f$; this guarantees that all results and bugs reported in the post-conditions will be true. In contrast to SL, ISL uses the rule of {\it backward consequence}, which allows one to {\it cut paths} (for example, if we had a post-condition with $x > 0$, we could soundly strengthen this information to consider only the path in which $x = 42$). Therefore, the philosophy underlying UX logics in general can be summarised as:
\begin{center}
\it paths can be cut, but no information can be lost.
\end{center}

When it comes to the use of ISL function specifications, whilst this has been implemented in Pulse-X, as far as we are aware, ISL does not feature function-call rules, 
and function compositionality for ISL and  UX logics has not been proven. Moreover, as it is not possible to determine if UX specifications cover all terminating behaviour, they remain incompatible with verification and cannot therefore be used in tools such as VeriFast, Iris, and~Gillian.

Our challenge is to develop a program logic in which we can state and prove function specifications that are compatible with {\it both} verification and true bug-finding.
Our motivation comes from the unique flexibility and expressivity that such specifications would provide, as they could be used by verification and bug-finding tools alike, closing the gap between these two contrasting paradigms. From our experience in program logics and associated tool-building, we believe that the main use case for exact specification should be self-contained, critical code, such as widely-used data-structure libraries. 

We introduce exact separation logic (ESL), with {\it exact} (EX)  function specifications of the form $\uquadruple P {f(\pvvar x)} \Qok \Qerr$, whose meaning combines that of SL and ISL specifications:
all terminating executions of the function that start from a state in the pre-condition $P$ end in a state covered by the post-conditions; {\it and} all states in two post-conditions are reachable from a state in the pre-condition by executing the function. The exactness of ESL can be captured by the~slogan:
\begin{center}
\it no paths can be cut and no information can be lost.
\end{center}
The slogan is supported by the rule of {\it equivalence}, which combines the forward consequence of SL and the backward consequence of ISL. In fact, ESL proof rules form a common core of SL and ISL, and ESL should therefore be a familiar setting to those acquainted with either.

We prove soundness for ESL with mutually recursive functions, which we believe is the first proof of  function compositionality for a  UX logic, and which transfers immediately to ISL. In doing so, drawing inspiration from InsecSL~\cite{insecsl}, we provide formal definitions of {\it external} and {\it internal} function specifications, which describe, respectively, the interface a function exposes towards its clients and towards its implementation, and highlight the difference in complexity between these two types of specifications in OX and UX reasoning.

Using numerous examples, we demonstrate (in the main text and in Appendix~\ref{apdx:examples}) how ESL can be used to reason about data-structure libraries, language errors, mutual recursion, and non-termination. In doing so, we introduce, for the first time, abstract predicates to UX reasoning and provide abstract function specifications for a number of data-structure algorithms, focussing on singly-linked lists and binary trees. In doing so, we highlight an important difference between the concepts of abstraction and over-approximation: in particular, abstraction corresponds to {\it hiding} information whereas over-approximation corresponds to {\it losing} it. Our findings demonstrate that, while abstraction cannot be used as freely in UX logics as in OX logics, sometimes resulting in less abstract specifications and more complex proofs, it should be feasible to  use ESL to provide  tractable function specifications  for  self-contained, critical code that can then be used for both verification and true bug-finding. 

\section{Exact Separation Logic by Example}
\label{sec:overview}

%
We guide the reader through what it means to write ESL specifications and proofs by intuition and example, contrasting our findings with those known from SL and ISL. 

%

\subparagraph*{Illustrative Example.} 
Consider the command $C \defeq \pifelses{ \pvar{x} > 0 }  \passign{\pvar{y}}{42}  \pifelsem \passign{\pvar{y}}{21} \}$, which can be specified, starting from the pre-condition $\pvar x \in \mathbb{Z}$, in ESL, SL, and ISL as follows: \\[-2.5mm]
\begin{minipage}{0.33\textwidth}
\[
\begin{array}{l}
\especlines{\pvar x \in \mathbb{Z} } \\
\pifelses{ \pvar{x} > 0 } \\
\tab\especlines{\pvar x > 0} \\ 
\tab\passign{\pvar{y}}{42} \\
\tab\especlines{Q_1 : \pvar x > 0 \land \pvar y \doteq 42} \\
 \pifelsem \\
\tab\especlines{\pvar x \leq 0 } \\ 
\tab\passign{\pvar{y}}{21} \\
\tab\especline{Q_2: \pvar x \leq 0 \land \pvar y \doteq 21} \\
\pifelsee\\
\especlines{Q_1 \lor Q_2} 
\end{array}
\]
\end{minipage}
\begin{minipage}{0.33\textwidth}
\[
\begin{array}{l}
\specline{\pvar x \in \mathbb{Z} } \\
\pifelses{ \pvar{x} > 0 } \\
\tab\ldots \\
\tab\polish{\text{// Same as ESL}} \\
\tab\ldots \\
\pifelsee\\
\specline{Q_1 \lor Q_2} \\
\polish{\text{// Losing information }} \\
\specline{\pvar x \in \mathbb{Z} \land \pvar y > 0}
\end{array}
\]
\end{minipage}
\begin{minipage}{0.33\textwidth}
\[
\begin{array}{l}
\islspeclineok[blue]{\pvar x \in \mathbb{Z} } \\
\pifelses{ \pvar{x} > 0 } \\
\tab\islspeclineok[blue]{\pvar x > 0} \\ 
\tab\passign{\pvar{y}}{42} \\
\tab\islspeclineok[blue]{Q_1: \pvar x > 0 \land \pvar y \doteq 42} \\
 \pifelsem~\passign{\pvar{y}}{21}~\pifelsee\\
 \polish{\text{// Path cutting}} \\
\islspeclineok[blue]{\pvar x > 0 \land \pvar y \doteq 42} 
\end{array}
\]
\end{minipage} \\

As ESL specifications  must neither cut paths nor lose information (in this example, about the values of $\pvar x$ and $\pvar y$), the ESL post-condition of $C$ must be equivalent to  $(\pvar x > 0 \land \pvar y \doteq 42) \lor ( \pvar x \leq 0 \land \pvar y \doteq 21 )$. In SL, it is possible to use forward consequence to weaken this information and obtain, for example,  $\pvar x \in \mathbb{Z} \land \pvar y > 0$, or just $\pvar x \in \mathbb{Z}$, or even just $\true$. 
In ISL, it is possible to cut, for example, the $\mathtt{else}$ branch of the $\mathtt{if}$ statement, but the values of $\pvar x$ and $\pvar y$ must be maintained  in the post-condition of the $\mathtt{then}$ branch, $\pvar x > 0 \land \pvar y \doteq 42$. 

One question that we have been often asked is whether it is simpler to prove an exact specification $\uquadruple {P} {\cmd}  {\Qok} {\Qerr}$ in ESL,  or to prove it separately in SL and ISL. The answer is that it is simpler to prove the specification in ESL.  If a specification is exact, then  it does not cut paths and it does not lose information. Therefore, the tools that make SL and ISL proofs simpler than ESL proofs, namely forward consequence and backward consequence, can only be used in very limited ways, if at all.
From our experience, the ISL proof of an exact specification will turn out to be almost identical to the ESL one, and an SL proof on top of that would duplicate a large part of the work. In fact, if one were to try to prove the exact specification $\utripleq {\pvar x \in \mathbb{Z}} C {(\pvar x > 0 \land \pvar y \doteq 42) \lor ( \pvar x \leq 0 \land \pvar y \doteq 21 )}$ from the above example in either SL or ISL, they would obtain exactly the same proof as in ESL. 

We also emphasise that ESL is not meant to replace either SL or ISL. If one is interested in only verification or only bug-finding, then one should use a formalism tailored to that type of analysis to  exploit the available shortcuts. However, if one wanted to use the same codebase for both  verification and bug-finding, then ESL offers a way of providing specifications useful for both. One example of such a codebase would be a widely-used data-structure library, where some of the users use it for verification and others for bug-finding.

\subparagraph*{List-length in ESL: Intuition.}
We consider a list-length function, $\mathtt{LLen}(\pvar x)$, which takes a list at~$\pvar x$, does not modify it, and returns its length, and the following
ESL specification:
\[
\utripleok{\pvar x\doteq x \lstar \llist{x,n} }{~\mathtt{LLen}(\pvar x)~}{\llist{x,n} \lstar\pvar{ret}\doteq n}
\]
This specification uses a standard list-length predicate, $\llist{x, n}$, which states that the length of the list at $x$ equals $n$ and is defined as follows:
\[
	\llist{x, n} \defeq (x\doteq\nil \lstar  n\doteq0) \lor (\exsts{v,x'} x\mapsto v,x' \lstar\llist{x',n - 1} ),
\]
      hiding the information about the values and internal node addresses of the list. Before proving this specification, we establish some intuition about why it holds. Let us assume that it does not hold and try to find a counter-example: by the meaning of ESL specifications, it is either not OX-valid or it is not UX-valid. The former, however, is not possible, as the analogous SL specification holds. The latter means that it is possible to 
find a state in the post-condition not reachable by the execution of $f$ from any state in the pre-condition, and may be unfamiliar to the reader as UX program logics have been introduced only recently.

We start looking for such a state in the post-condition (post-model) by choosing some values for $x$ and $n$: say, $x = 0$ and $n = 2$. This also fixes $\pvar{ret}$ to 2.
Then, we fully unfold $\llist{0,2}$ to obtain $\exists v_1, x_1, v_2.~0 \mapsto v_1, x_1 \lstar x_1 \mapsto v_2, \nil$, and instantiate the existentials $v_1$, $x_1$, and $v_2$: say, with $1$, $4$, and $9$, respectively. In this way, we obtain the state described by the assertion $0 \mapsto 1, 4 \lstar 4 \mapsto 9, \nil$.
When it comes to the pre-condition, $x$ and $n$ (and also $\pvar x$) are fixed by the post-model choices, and when we unfold the list, the pre-condition becomes $\pvar x = 0 \lstar \exists v_1, x_1, v_2.~0 \mapsto v_1, x_1 \lstar x_1 \mapsto v_2, \nil$. As the algorithm does not modify the list, it becomes clear that if we choose $v_1$, $x_1$, and $v_2$ as for the post-model (that is, 1, 4, and 9, respectively), the algorithm will reach our post-model. Given that the same reasoning would apply for any choice of $x$ and $n$, we realise that the given specification is, in fact, also UX-valid and hence exact. This reveals an important observation, which is that
\begin{center}
\it abstraction does not always equate to over-approximation, that is,\\hiding information does not always mean losing information.
\end{center}

For those used to OX reasoning, it might appear that the post-condition $\llist{x,n} \lstar\pvar{ret}\doteq n$ loses information about the structure of the list, but the insight here is that this information was never known in the pre-condition in the first place, as we also only had $\llist{x,n}$ there.



\subparagraph*{List-length in ESL: Proof Sketch.}
Reasoning about function specifications in the UX/EX setting has not been studied previously and requires subtle definitions of {\it external function specifications}, which provide the interface that the function exposes to the client, and {\it internal function specifications}, which provide the interface to the function’s implementation. With OX logics, these are well-understood and the gap between them is small. For UX/EX logics, this gap is larger. We illustrate these concepts informally using the list-length example, and give the corresponding formal definitions in~\S\ref{ch:proofsys}.

The proof sketch of the ESL external specification of the list-length algorithm is given in Figure~\ref{fig:esl-list-length}.
It is more complex than its SL counterpart (see Appendix~\ref{apdx:sl-list-length}), but is manageable and comes with the benefit that this ESL specification can be used for both verification and~bug-finding.

First, as the function is recursive, we have to provide a measure and prove the specification extended with this  measure: in this case, the measure is $\alpha = n$, given by the length of the list. This measure is necessary to ensure the finite reachability property for mutually recursive functions in UX logics, and is a known technique from the work on total correctness specifications for OX logics~\cite{Floyd1967Flowcharts,total-tada}. Recursive function calls are then allowed only if they use specifications of a strictly smaller~measure, represented in the proof sketch by the function specification context $\Gamma (\alpha)$, which contains the specification of $\mathtt{LLen}(\pvar x)$ for all $\beta < \alpha$. 

The move from the external to the internal pre-condition initialises the local function variables to $\nil$. The ESL rule for the $\mathtt{if}$ statement, just like in~SL, adds the condition to the then-branch, its negation to the else-branch, and collects the branch post-conditions using disjunction. The rules for the basic commands (here, the assignments $\passign{\pvar r}0$ and $\passign{\pvar r}{\pvar r + 1}$ and the lookup $\passign{\pvar{x}}{[\pvar x + 1]}$) are also the same as in SL, as these are already exact. 
The unfolding of the list is also done in the same way, as unfolding always preserves equivalence; note how the condition of the $\mathtt{if}$ statement determines the appropriate disjunct for the list predicate. The recursive function call is allowed to go through as it is used with measure $n - 1$ (with the parts of the assertion representing the pre- and the post-condition highlighted).

The major difference between ISL/ESL and SL proofs is that we cannot lose
information about the function parameters and local variables in the middle of the former. Therefore, we cannot simplify the assertions $Q'_1$ and $Q'_2$ further and cannot fold back the list predicate within the internal specification, as we would do in SL (cf. corresponding proof in Appendix~\ref{apdx:sl-list-length}).

\begin{figure}[!t]
\[
\small
\begin{array}{r@{~}l}
& \hspace*{-1.15cm} \polish{\text{// Function is recursive and requires a measure: $\alpha \doteq n$}} \\
\Gamma (\alpha) \vdash &
\especlines{\pvar x \doteq x \lstar \llist{x, n} \lstar \alpha \doteq n} \\
& \mathtt{LLen}(\pvar x)~\{ \\
& \tab\polish{\text{// Transition from external to internal pre-condition: initialise locals  to $\nil$}} \\
& \tab\especlines{\pvar x \doteq x \lstar \llist{x, n} \lstar \alpha \doteq n \lstar \pvar r \doteq \nil } \\
& \tab\pifelses{ \pvar{x} = \nil } \\
& \tab\tab\especlines{\pvar x \doteq x \lstar \llist{x, n} \lstar \alpha \doteq n \lstar \pvar r\doteq\nil \lstar {\pvar x \doteq \nil}} \\ 
& \tab\tab\passign{\pvar{r}}{0} \\
& \tab\tab\especlines{Q'_1 : \pvar x \doteq x \lstar \llist{x, n} \lstar \alpha \doteq n \lstar {\pvar r \doteq 0}\lstar \pvar x \doteq \nil } \\
& \tab \pifelsem \\
& \tab\tab\especlines{\pvar x \doteq x \lstar \shade{\llist{x,n}} \lstar \alpha \doteq n \lstar \pvar r\doteq\nil \lstar \shade{\pvar x \dotneq \nil}} \\ 
& \tab\tab\text{\polish{// Unfold $\llist{x, n}$ using the equivalence }} \\
& \tab\tab\text{\polish{// $\models \llist{\pvar x, n} \lstar \pvar x \dotneq \nil \Leftrightarrow \exists v, x'.~\pvar x \mapsto v, x' \lstar \llist{x', n-1}$}} \\
& \tab\tab\especlines{\shade{\exists v, x'}.~\pvar x \doteq x \lstar  \shade{x \mapsto v, x' \lstar \llist{x', n-1}} \lstar \alpha \doteq n \lstar \pvar r \doteq \nil} \\ 
& \tab\tab\passign{\pvar{x}}{[\pvar x + 1]}; \\
& \tab\tab\especlines{\exists v, x'.~\pvar x \doteq x' \lstar  x \mapsto v, x' \lstar \llist{x', n-1} \lstar \alpha \doteq n \lstar \pvar r \doteq \nil} \\ 
& \tab\tab\text{\polish{// As $\alpha - 1 < \alpha$, we can use the specification of $\mathsf{LLen}(\pvar x)$ with measure $\alpha - 1$}} \\
& \tab\tab\especlines{\exists v, x'.~\shade{\pvar x \doteq x'} \lstar  x \mapsto v, x' \lstar \shade{\llist{x', n-1} \lstar \alpha - 1 \doteq n - 1} \lstar \pvar r \doteq \nil} \\ 
& \tab\tab\passign{\pvar{r}}{\mathtt{LLen}(\pvar x)}; \\
& \tab\tab\especlines{\exists v, x'.~\pvar x \doteq x' \lstar  x \mapsto v, x' \lstar \llist{x', n-1} \lstar  \alpha - 1 \doteq n - 1 \lstar \shade{\pvar r \doteq n -1}} \\ 
& \tab\tab\passign{\pvar{r}}{\pvar r + 1} \\
& \tab\tab\especlines{Q_2' : \exists v, x'.~\pvar x \doteq x' \lstar  x \mapsto v, x' \lstar \llist{x', n-1} \lstar  \alpha - 1 \doteq n - 1 \lstar \pvar r \doteq n} \\ 
& \tab\pifelsee;\\
& \tab\especlines{Q': Q_1' \lor Q_2'} \\
& \tab\preturn{\pvar r} \\
& \tab\especlines{ Q' \lstar \pvar{ret} = \pvar r} \\
& \tab\polish{\text{// Transition from internal to external post-condition given in text}} \\
& \} \\
& \especlines{\llist{x, n} \lstar \pvar{ret} \doteq n \lstar \alpha = n}
\end{array}
\]
\caption{ESL proof sketch: $\utripleok{\pvar x\doteq x \lstar \llist{x,n} }{~\mathtt{LLen}(\pvar x)~}{\llist{x,n} \lstar\pvar{ret}\doteq n}$.}%
\label{fig:esl-list-length}
\end{figure}

The most complex part of the proof sketch is the transition from the internal to the external post-condition, in which we have to somehow forget the local variables of the function, given that they must not spill out into the calling context. This is done by replacing them with fresh, existentially quantified logical variables, which in this case also allows us to use equivalence to fold back the list predicate and reach the target post-condition. The details of this transition, in which we denote $\pvar{ret} \doteq n \lstar \alpha \doteq n$ by $R$, are as follows:
\[
\begin{array}{@{}l}
\exsts{x_q, r_q} Q'[x_q/ \pvar x][r_q / \pvar r]\lstar\pvar{ret}\doteq \pvar r[x_q/ \pvar x][r_q / \pvar r] \\[0.3mm]
 \quad \Leftrightarrow  ((x \doteq \nil \lstar n \doteq 0) \lor 
		   (\exists x_q, r_q, v, x'.~x_q \doteq x' \lstar  x \mapsto v, x' \lstar \llist{x', n-1} \lstar  r_q \doteq n)) \lstar R
	   \\
 \quad \Leftrightarrow  ((x \doteq \nil \lstar n \doteq 0) \lor 
		   (\exists v, x'.~x \mapsto v, x' \lstar \llist{x', n-1})) \lstar R ~\text{\polish{// can fold now}} \\
 \quad \Leftrightarrow  \llist{x, n} \lstar (n \doteq 0 \lor n \dotgt 0) \lstar R \\
 \quad \Leftrightarrow \llist{x, n} \lstar \pvar{ret} \doteq n \lstar \alpha \doteq n
\end{array}
\]

Observe that, since we are proving an EX specification, we are not allowed to cut paths. This means that the ISL proof of the analogous ISL specification of $\texttt{LLen} (\pvar x)$ would be identical, noting that the use of equivalence would technically be replaced by backward consequence.

\subparagraph*{List-insert in ESL: Intuition.} The list-length function, $\texttt{LLen} (\pvar x)$, is an example of an algorithm where the EX specification is analogous to the  traditional OX specification. At times, however, ESL specifications have to be more complex.
Consider, for example, the list-insert algorithm $\texttt{LInsertFirst}(\pvar x, \pvar v)$, which inserts the element~$\pvar v$ at the beginning of the list $\pvar x$. Its traditional OX specification~is:
\[
\tripleq{\pvar x\doteq x \lstar \pvar v \doteq v \lstar \llist{x, \lstvs} }{~\mathtt{LInsertFirst}(\pvar x, \pvar v)~}{\llist{\pvar{ret}, v \cons \lstvs}}
\]
where $\llist{x, \lstvs}$ is the standard list predicate that exposes the values of the list:
\[ \llist{x, \lstvs} \defeq (x\doteq\nil \lstar \lstvs \doteq \emplist) \lor (\exsts{v, x', \lstvs'} x \mapsto v, x' \lstar\llist{x',\lstvs'}\lstar \lstvs \doteq v \cons \lstvs') \]
Using the counter-example approach to check if this specification is EX-valid, we easily see that it loses information: in particular, no end-state where $x$ is not the \textit{second} pointer in the returned list $\pvar{ret}$ is reachable from the given pre-condition.
Consequently, for EX validity, we are required to use the following, less abstract, ESL specification for $\texttt{LInsertFirst}$:
\[
\utripleq{\pvar x\doteq x \lstar \pvar v \doteq v \lstar \llist{x, \lstxs, \lstvs} }{~\mathtt{LInsertFirst}(\pvar x, \pvar v)~}{\llist{\pvar{ret}, \pvar{ret} \cons \lstxs, v \cons \lstvs} \lstar \listptr{x, \lstxs}} \\
\]
where $\llist{x, \lstxs, \lstvs}$ is a predicate that exposes the internal pointers of a given list in addition to the values, and $\listptr{x, \lstxs}$ states that the list $\lstxs$ starts with $x$. 

\subparagraph*{Further Examples.}
In~\S\ref{ch:examples} and Appendix~\ref{apdx:examples}, we give many additional examples of ESL specifications and proofs to illustrate reasoning about  
list algorithms and binary trees, as well as language errors, mutual recursion, non-termination, and client programs.

%% file: sections/language.tex
\newcommand{\serr}{\sto[\pvar{err} \rightarrow \verr]}

\section{The Programming Language}
\label{ch:wisl}

We introduce ESL using a simple programming language, the syntax of which is given below.

\smallskip
\begin{display}{Language Syntax}
    $\begin{array}{r@{~~}c@{~~}l}
      \gv \in \vals & ::= & \nv \in \nats \mid \bv \in \bools \mid \sv \in \strings \mid \nil \mid \lst{\gv} \qquad\quad~ \pvar x \in \pvars \\
      \pexp \in \pexps & ::= & \gv \mid \pvar x \mid \pexp + \pexp \mid \pexp - \pexp \mid ... \mid \pexp = \pexp \mid \pexp < \pexp 
      \mid \neg \ \pexp \mid  \pexp \wedge \pexp \mid ... \mid \pexp \cons \pexp \mid \pexp \lcat \pexp \mid ... \\[2mm]
      \cmd  \in \cmds & ::= &\pskip \mid \passign{\pvar{x}}{\expr{\pexp}} \mid \passign{\pvar{x}}{\prandom}               \mid \perror(\pexp) \mid \pifelse{\pexp}{\cmd}{\cmd} \mid \pwhile{\pexp}{\cmd} \mid \cmd; \cmd \mid\\
                            && \pfuncall{\pvar{y}}{\fid}{\vec{\expr{\pexp}}} \mid \pderef{\pvar{x}}{\expr{\pexp}} \mid \pmutate{\expr{\pexp}}{\expr{\pexp}}   \mid \palloc{\pvar{x}}{\pexp}   \mid  \pdealloc{\expr{\pexp}}                           
\end{array}  
$
\end{display}

Values, $\gv \in \vals$, include: natural numbers, $\nv \in
\nats$; Booleans, $\bv \in \bools \defeq$ $\{ \true,
\false \}$; strings, $\sv \in \strings$; a dedicated value
$\nil$; and lists of values, $\lst{\gv} \in \mathsf{List}$. 
Expressions, $\pexp \in \Exp$, comprise values, program variables, $\pvar x \in \pvars$, and various unary and binary operators (e.g., addition, equality, negation, conjunction, list prepending, and list concatenation).
Commands comprise:
the variable assignment; non-deterministic number generation; error raising;
the $\mathtt{if}$ statement; the $\mathtt{while}$ loop; command sequencing;
function call; and memory management commands, that is, lookup,
mutation, allocation, and deallocation.  The sets of program variables
for  expressions and commands, denoted by $\pv{\pexp}$ and $\pv{\cmd}$
respectively, and the 
sets of modified variables for commands, denoted by $\updt(\cmd)$, are
defined in the standard~way.

\begin{definition}[Functions]
A function, denoted by
$\pfunction{\fid}{\pvvar x}{\cmd; \preturn{\pexp}}$, comprises:
a function identifier, $\fid \in \strings$; the function
parameters, $\pvvar x$,  given by 
a list of distinct program variables; a  function
body, $\cmd \in \Cmd$; and a {return expression}, $\pexp \in \PExp$, with  $\pv{\pexp} \subseteq
\{\vec{\pvar x}\} \cup \pv{\cmd}$. 
\end{definition}

Program variables in function bodies that are not the function 
parameters are treated as local variables initialised to $\nil$, with their scope not 
extending beyond the function. 

\begin{definition}[Function Implementation Contexts]
A function implementation context, $\fictx : \strings \rightharpoonup_{\mathit{fin}} \pvars~\mathsf{List} \times \cmds \times \pexps$, is a finite partial function from function identifiers to their implementations. For $\fictx(\fid)=(\vec{\pvar{x}}, \cmd, {\pexp})$,  we also write
$\fid(\vec{\pvar{x}})\{\cmd; \preturn{\pexp}\} \in \fictx$. \end{definition}


We next define an operational semantics that gives a complete account of
the behaviour of commands and does not get stuck on any input, as
we explicitly account for language errors and missing resource errors.

\begin{definition}[Stores, Heaps, States]
\label{def:stohsta}
Variable stores, $\sto : \pvars \rightharpoonup_{\mathtt{fin}} \vals$, are partial finite functions from program variables to values. Heaps, $\hp : \nats \rightharpoonup_{\mathtt{fin}} (\vals
\uplus \cfreed)$, are partial finite
functions from natural numbers to values extended with a dedicated
symbol $\cfreed \notin \vals$. Program states, $\cst = (\sto, \hp)$, consist of a store and a heap. 
\end{definition}

Heaps are used to model the memory, and the dedicated
symbol $\cfreed \notin \vals$ is required for UX frame preservation\footnote{UX frame preservation  means that if a program runs with a non-missing outcome to a given final state, then it also runs with the same outcome to an extended final state, with the extension (the \emph{frame}) unaffected by the execution. From ISL~\cite{isl}, it is known that losing deallocation information breaks UX frame preservation; the solution is to keep track of deallocated cells, which we achieve by using $\cfreed$.} to hold
(cf.~Definition~\ref{def:gvalspec}). In particular, $\hp(\nv) = \gv$ means that an
allocated heap cell with address~$\nv$  contains the value $\gv$;
and $\hp(\nv) = \cfreed$ means that a heap cell with address $\nv$ has been
deallocated~\cite{cosette,javert,javert2,jslogic,isl}. This linear memory model is used in much of the
SL literature, including
ISL~\cite{isl}. 
Onward, $\emptyset$ denotes the empty heap,  $\hp_1 \uplus \hp_2$ denotes heap disjoint union, and $\hp_1 \hpdisj \hp_2$ denotes that $\hp_1$ and $\hp_2$ are disjoint. 

\begin{definition}[Expression Evaluation]
    \label{def:evalfunc}
    The evaluation of an expression $\pexp$ with respect to a store $\sto$,
    denoted $\esem{\pexp}{\sto}$, 
    results in either a value or a dedicated symbol denoting an
    evaluation error, $\undefd \notin \vals$.
    Some illustrative
    cases are:
\[
        \esem{\gv}{\sto} = \gv \quad
      \esem{\pvar{x}}{\sto} = \begin{cases} \sto(\pvar{x}), & \mbox{$\pvar x \in \dom(\sto)$} \\ {\undefd}, & \mbox{otherwise} \end{cases}
      \quad \esem{\pexp_1 ~{+}~ \pexp_2}{\sto} = \begin{cases} \esem{\pexp_1}{\sto} ~{+}~ \esem{\pexp_2}{\sto}, & \mbox{$\esem{\pexp_1}{\sto}, \esem{\pexp_2}{\sto} \in \nats$} \\ {\undefd}, & \mbox{otherwise} 
            \end{cases}\]

\end{definition}

The big-step operational semantics uses judgements of the form 
$\strans \cst \cmd {\cst'} \outcome$, read: given implementation context
$\fictx$ and starting from state $\cst$, the execution of command $\cmd$ 
results in outcome $\outcome \in \outcomes = {\{\oxok, \oxerr, \oxm\}}$ and
 state  $\cst'$. The outcome can
either equal: $\oxok$ (elided where possible), denoting a 
 successful execution;  $\oxerr$, denoting an execution
faulting with a language error, or $\oxm$, denoting an
execution faulting with a missing resource error.

\begin{definition}[Operational Semantics]
\label{def:wislopsem}
The representative cases of the big-step operational semantics are given in
Figure~\ref{fig:whilesem}. The complete semantics is given in Appendix~\ref{apdx:cos}.
\end{definition}

The successful transitions are straightforward: for
example, the $\mathsf{nondet}$ command generates an arbitrary natural number; the
function call executes the function body in a 
store where the function parameters are given the values of the function arguments
and the function locals are initialised to $\nil$; 
and the control flow statements behave as expected. Allocation
requires the specified amount of contiguous cells (always available
as heaps are finite), and lookup, mutation, and deallocation
require the targeted cell not to have been freed.

The semantics stores error information in a dedicated program variable $\pvar{err}$, not available to the programmer. For simplicity of error messages, we assume to have a function $\mathsf{str}: \pexps \rightarrow \strings$, which serialises program expressions into strings. The faulting semantic transitions are split into {\it language errors}, which can be captured by program-logic reasoning, and
{\it missing resource errors}, which cannot, as such errors break the frame property. Language errors arise due to, for example, expressions being incorrectly typed
(e.g. ~$\nil + 1$) or an attempt to access deallocated cells (that is, the use-after-free error). 
On the other hand, missing resource errors arise from accessing cells that are not present in memory.

\begin{figure*}[!t]
\small
    \begin{mathpar}
       \infer{
          \sthreadp{ \sto }{ \hp }, \passign{\pvar{x}}{\pexp} \baction_{\fictx} \sthreadp{\sto'}{ \hp }
       }{
       \esem{\pexp}{\sto} = \gv \quad \sto' = \sto [\pvar{x} \storearrow \gv]
       } \qquad
       \infer{
          \sthreadp{ \sto }{ \hp }, \passign{\pvar{x}}{\prandom} \baction_{\fictx} \sthreadp{\sto'}{ \hp }
       }{
       \nv \in \Nat \quad \sto' = \sto [\pvar{x} \storearrow \nv]
       } \qquad
       \infer{
          \sthreadp{ \sto }{ \hp }, \pwhile{\pexp}{\cmd} \baction_{\fictx} \sthreadp{ \sto }{ \hp }
       }{
          \esem{ \pexp }{ \sto } = \false
       } \\
       \infer{
          \sthreadp{ \sto }{ \hp }, \pwhile{\pexp}{\cmd} \baction_{\fictx} \outcome : \cst'
       }{
          \begin{array}{c}
          \esem{ \pexp }{ \sto } = \true \quad \sthreadp{ \sto }{ \hp }, \cmd \baction_{\fictx} \cst'' \\
          \cst'', \pwhile{\pexp}{\cmd} \baction_{\fictx} \outcome : \cst'
          \end{array}
       }
       \quad
       \infer{
     \sthreadp{ \sto }{ \hp },
     \pfuncall{\pvar{y}}{\fid}{\vec{\pexp}} \baction_{\fictx}
     \sthreadp{\sto [\pvar{y} \storearrow \gv' ] }{ \hp' }
  }{
    \begin{array}{c}
    \pfunction{\procname}{\vec{\pvar{x}}}{\cmd; \preturn{\pexp'}} \in \scontext
    \quad
    \esem{\vec{\pexp}}{\sto} = \vec{\gv} 
     \quad \pv{\cmd} \setminus
      \{\pvvar{x}\} = \{\pvvar{z} \}
    \\
 \sto_p  =  \emptyset [ \vec{\pvar{x}} \storearrow \vec{v}] [ \vec{\pvar{z}} \storearrow \nil]
     \quad (\sto_p, \hp), \cmd
      \baction_{\fictx} \sthreadp{ \sto_q }{ \hp' }  \quad  \esem{\pexp'}{\sto_q} =v' 
    \end{array}
  }
  \\
      \infer{
         \sthreadp{ \sto }{ \hp }, \pderef{\pvar{x}}{\pexp} \baction_{\fictx} \sthreadp{\sto [\pvar{x} \storearrow v] }{ \hp }
      }{
      \esem{\pexp}{\sto} = \nv \quad \hp(\nv) = \gv
      }
      \qquad
       \infer{
          \sthreadp{ \sto }{ \hp }, \pmutate{\pexp_1}{\pexp_2}
          \baction_{\fictx} \sthreadp{ \sto }{ \hp' }
       }{
       \esem{\pexp_1}{\sto} = \nv \quad \hp(\nv) \in \vals \quad \esem{\pexp_2}{\sto} = \gv \quad \hp' = \hp [\nv \mapsto \gv] 
       } \\
       \infer{
          \sthreadp{ \sto }{ \hp }, \palloc{\pvar{x}}{\pexp} \baction_{\fictx} \sthreadp{\sto [\pvar{x} \storearrow \nv'] }{ \hp' }
       }{
       \begin{array}{c}
       		 (\nv' + i \notin \dom (\hp))|_{i = 1}^{\esem{\pexp}{\sto} - 1} \quad
       		\hp' = \hp \uplus \{ (\nv' + i \mapsto \nil)|_{i = 1}^{\esem{\pexp}{\sto} - 1} \} 
       \end{array}
       } \quad
       \infer{
          \sthreadp{ \sto }{ \hp }, \pdealloc{\pexp} \baction_{\fictx} \sthreadp{ \sto }{ \hp[\nv \mapsto \cfreed]}
       }{
       \esem{\pexp}{\sto} = \nv \quad  \hp(\nv) \in \vals 
       } \\
       \infer{
          \sthreadp{ \sto }{ \hp },  \pderef{\pvar{x}}{\pexp} \baction_{\fictx} {\oxerr} : \sthreadp{ \sto_{\oxerr} }{ \hp }
       }{
       \begin{array}{c}
       \esem{\pexp}{\sto} = \undefd \quad \verr = [``\mathsf{ExprEval}", \stringify {\pexp}]
       \end{array}
       }
       \qquad
        \infer{
          \sthreadp{ \sto }{ \hp },  \pderef{\pvar{x}}{\pexp}
          \baction_{\fictx} {\oxm} : \sthreadp{ \sto_{\oxerr} }{ \hp }
       }{
       \begin{array}{c}\esem{\pexp}{\sto} = \nv  \notin \dom (\hp) \quad \verr = [``\mathsf{MissingCell}", \stringify{\pexp}, \nv]\end{array}
       } 
       \\
       \infer{
          \sthreadp{ \sto }{ \hp },  \pderef{\pvar{x}}{\pexp}
          \baction_{\fictx} {\oxerr} : \sthreadp{ \sto_{\oxerr} }{ \hp }
       }{
       \hp(\esem{\pexp}{\sto}) = \cfreed \quad \verr = [``\mathsf{UseAfterFree}", \stringify{\pexp_1}, \esem{\pexp}{\sto}]
       } \qquad
       \infer{
       \sthreadp{ \sto }{ \hp }, \perror(\pexp) \baction_\fictx \oxerr :  \sthreadp{ \sto_{\oxerr} }{ \hp }
       }{\esem{\pexp}{\sto} = \gv \quad \verr = [``\mathsf{Error}", \gv]}
    \end{mathpar}
    \caption{Operational semantics (excerpt),
      with $\sto_{\oxerr} \defeq \serr$ and 
    $\mathsf{str}: \PExp \rightarrow \strings$.}
    \label{fig:whilesem}
    \end{figure*}

%% file: sections/proofsystem.tex

\newcommand{\admiss}[1]{\mathcal{A}\mkern-1mu\mathit{dm}(#1)}

\section{Exact Separation Logic}
\label{ch:proofsys}

We introduce an  exact separation logic for our programming language,
giving the assertion language in \S\ref{sec:asrtlang}, specifications in \S\ref{sec:specesl}, and the program logic rules in \S\ref{sec:proglog}.

\subsection{Assertion Language}
\label{sec:asrtlang}

To define assertions and their meaning, we introduce \emph{logical variables}, $x, y, z, \in \LVar$, distinct from program variables, and define the set of \emph{logical expressions} as follows:
\[
    \begin{array}{l@{~}l@{~}lr}
    \lexp \in \LExp & \defeq & \gv \mid x \mid \pvar{x}  \mid \lexp + \lexp \mid \lexp - \lexp \mid ... \mid \lexp = \lexp \mid \neg \ \lexp \mid  \lexp \land \lexp \mid ... \mid \lexp \cdot \lexp \mid \lexp \cons \lexp \mid... \\
    %
    \end{array}
  \]

Note that we can use program expressions in assertions (for example, $\pexp \in \Val$), as they form a proper subset of logical expressions.

\begin{definition}[Assertion Language]
\label{def:logasrt}
The assertion language is defined as follows:
\[
	\begin{array}{r@{~}c@{~}l}
   \pc \in \PAssert & \defeq  & \lexp_1 = \lexp_2 \mid
   \lexp_1 < \lexp_2 \mid
   \lexp \in X \mid \ldots \mid \lnot \pc \mid \pc_1 \Rightarrow \pc_2

   \\
   P \in \Assert & \defeq &

   \pc 
     \mid \AssFalse \mid P_1 \Rightarrow P_2   
  \mid \exists x \ldotp P 
    \mid \emp \mid  
   \lexp_1 \mapsto \lexp_2 \mid \lexp \mapsto \cfreed \mid  P_1 \lstar P_2 \mid \bigoast_{\lexp_1 \leq x < \lexp_2} P 
 \end{array}
 \]
 where $\lexp, \lexp_1, \lexp_2  \in \LExp$,  $X \subseteq \Val$, and $\lvar{x}
 \in \LVar$.
\end{definition}

\medskip
Boolean  assertions, $\pc \in \PAssert$, lift Boolean logical
expressions to assertions.
Assertions, $P \in \Assert$, contain Boolean assertions,
standard first-order connectives and quantifiers, and spatial
assertions.
Spatial assertions include: the empty memory assertion, $\emp$; the
positive cell assertion, $\lexp_1 \mapsto \lexp_2$; the negative cell
assertion, $\lexp\mapsto\cfreed$ (as in~\cite{javert,cosette,javert2,jslogic} and denoted in ISL by
$\lexp \negheap$~\cite{isl}), the separating conjunction (star); and its iteration (iterated~star).

\medskip
To define assertion satisfiability, we introduce \emph{substitutions},
$\subst : \LVar \rightharpoonup_{\mathtt{fin}} \Val$, which are partial
finite mappings from logical variables to values,
 extending expression evaluation of Definition \ref{def:evalfunc} to
$\esem{\lexp}{\subst,\sto}$ straightforwardly, with a new base case for logical variables:
\[
\esem{x}{\subst,\sto} = \subst(x), \mbox{if $x \in \domain(\subst)$} \qquad \esem{x}{\subst, \sto} = \undefd, \mbox{if $x \notin \domain(\subst)$}
\]

\begin{definition}[Satisfiability]%
\label{def:lsat}

The assertion satisfiability relation, denoted by $\subst, \stt \models P$, is defined as follows:
\[
\begin{array}{@{}l@{~}c@{~\ }l}
\subst, (\sto, \hp) \models \pc & \Leftrightarrow & \esem{\pc}{\subst, \sto} = \true \land \hp = \emptyset \\
\subst, (\sto, \hp) \models \AssFalse &\Leftrightarrow& \text{never} \\
\subst, (\sto, \hp) \models P_1 \Rightarrow P_2 &\Leftrightarrow& \subst, (\sto, \hp) \models  P_1 \Rightarrow \subst, (\sto, \hp) \models  P_2  \\
\subst, (\sto, \hp) \models \exists \lvar{x} \ldotp P &\Leftrightarrow& \exists v \in \Val \ldotp \subst[\lvar{x} \mapsto v], (\sto, \hp) \models P \\
\subst, (\sto, \hp) \models \emp &\Leftrightarrow& {\hp = \emptyset } \\
\subst, (\sto, \hp) \models {\lexp_1 \mapsto \lexp_2} &\Leftrightarrow& \hp = \{ \esem{\lexp_1}{\subst, \sto} \mapsto \esem{\lexp_2}{\subst, \sto}\} \\
\subst, (\sto, \hp) \models {\lexp_1 \mapsto \cfreed} &\Leftrightarrow& \hp = \{ \esem{\lexp_1}{\subst, \sto} \mapsto \cfreed\} \\
\subst, (\sto, \hp) \models P_1 \lstar P_2 &\Leftrightarrow& \exists \hp_1, \hp_2 \ldotp \hp = \hp_1 \uplus \hp_2 \land \subst, (\sto, \hp_1) \models P_1 \land \subst, (\sto, \hp_2) \models P_2 \\
\subst, (\sto, \hp) \models \bigoast_{\lexp_1 \le x < \lexp_2} P &\Leftrightarrow& \exists \hp_i, \ldots, h_{k-1}.\, \hp = \uplus_{j = i}^{k-1} \hp_j \land \forall j.\, i \le j < k  \Rightarrow \subst, (\sto, \hp_j) \models P[j/x] \\
                            && \text{where $i = \esem{\lexp_1}{\subst, \sto}$, $k = \esem{\lexp_2}{\subst, \sto}$, and $x$ is not free in $\lexp_1$ or $\lexp_2$}.
\end{array}
\]
\end{definition}

Assertion satisfiability is defined in the standard way. For convenience, we choose Boolean assertions to be satisfiable only in the empty heap.

\begin{definition}[Validity] 
  An assertion $P$ is \emph{valid}, denoted by $\models P$, iff $\forall \subst, \stt.~\subst, \stt \models P$.
\end{definition}

\subsection{Specifications}
\label{sec:specesl}

We define specifications for commands and functions, focussing in particular on external and internal function specifications and the relationship between them.





\begin{definition}
  {\em Specifications}, $\sspec = \uspec{P}{\Qok}{\Qerr} \in \specs$, comprise a \emph{pre-condition}, $P$, a \emph{success post-condition}, $\Qok$, and a \emph{faulting post-condition},~$\Qerr$.
\end{definition}

We denote that command $\cmd$ has specification~$t$ by $\cmd : t$, or by $\uquadruple{P}{\cmd}{\Qok}{\Qerr}$ in quadruple form. Additionally, we use the following shorthand:
\[
\begin{array}{r@{~\defeq~}l}
\utripleq{P}{\cmd}{Q} & \uquadruple{P}{\cmd}{Q}{\AssFalse} \\
\utripleerr{P}{\cmd}{Q} & \uquadruple{P}{\cmd}{\AssFalse}{Q} \\
\utripleq P \cmd \bigq & \uquadruple P \cmd - -
\end{array}
\]
noting the use of  $\bigq$ for cases in which the
post-condition details are not relevant.
We use quadruples rather than triples since, even though the post-condition could be expressed as a disjunction of $\oxok$- and $\oxerr$-labelled assertions, we find the quadruple distinction  helpful as compound commands (e.g. sequence) treat the two differently (cf.~Figure~\ref{fig:unify}).

The EX-validity of a specification $t$ for a command $\cmd$ in an implementation context $\fictx$ requires both OX and UX frame-preserving validity.

\begin{definition}[$\fictx$-Valid Specifications]
\label{def:gvalspec}
Given implementation context $\fictx$, command $\cmd$, and
specification $t = \uspec{P}{\Qok}{\Qerr}$,
$t$ is {\em $\fictx$-valid} for $\cmd$, denoted by $\fictx\models \cmd:\sspec$ or $\fictx \models
\uquadruple{P}{\cmd}{\Qok}{\Qerr}$,  if and only~if:
\[
    \begin{array}{@{}l}
     \text{\polish{// Frame-preserving over-approximating validity }} \\
 (\forall \subst, \sto, \hp, \hp_f, \outcome, \sto', \hp''.~\subst, (\sto, \hp) \models P \implies \\
\quad  (\sto, \hp \uplus \hp_f), \cmd \baction_{\fictx} \outcome: (\sto', \hp'') \implies
(\outcome \neq \oxm \land
    \exists \hp'.~\hp'' = \hp' \uplus \hp_f \land \subst, (\sto', \hp') \models Q_\outcome))~\land \\
     \text{\polish{// Frame-preserving under-approximating validity}} \\
     (\forall \subst, \sto', \hp', \hp_f, \outcome.~\subst, (\sto', \hp') \models Q_\outcome \implies \hp_f~\sharp~\hp' \implies \\
     \quad  (\exsts{\sto,\hp}~\subst, (\sto, \hp) \models P~\land~(\sto, \hp \uplus \hp_f), \cmd \baction_\fictx \outcome: (\sto', \hp' \uplus \hp_f)))
    \end{array}
  \]
\end{definition}




Observe that the outcome $\outcome$ can either be success or a language error; it cannot be a missing resource error as this would break UX frame preservation.
As our operational semantics is complete, we can also use ESL
to characterise non-termination. In particular, if a
command satisfies a specification with both post-conditions
$\AssFalse$, then the execution of the command is guaranteed to not terminate if executed from a
state satisfying the pre-condition. Were the semantics incomplete (for example, if it did not reason about errors), then such a
specification might  also indicate the absence of a semantic
transition.

Compared to traditional OX reasoning, UX reasoning brings additional complexity to proofs of function specifications. To handle this complexity, we introduce two types of function specifications: {\em external} specifications, which provide the {interface} the function exposes to the client, and the related {\em internal} specifications, which provide the interface to the function implementation. This terminology is also used informally in InsecSL~\cite{insecsl}. We use these in subsequent sections to show that ESL exhibits function compositionality.

\begin{definition}[External Specifications]
\label{def:espec}
A specification $\uspec{P}{\Qok}{\Qerr}$ is
an {\em external function specification} if and only if:
\begin{itemize}
\item $P = (\vec {\pvar x} \doteq \vec x \lstar P')$,
  for some distinct program variables $\pvvar x$, distinct logical variables $\vec x$, and assertion~$P'$, with $\pv{P'} = \emptyset$; and
%
%
\item either $\pv{\Qok} = \{ \pvar{ret} \}$ or $\Qok = \AssFalse$, and 
either $\pv{\Qerr} = \{ \pvar{err} \}$ or $\Qerr = \AssFalse$.
\end{itemize}

The set of external specifications is denoted by $\especs$.

\end{definition}

\begin{definition}[Function Specification Contexts]
\label{def:fspeccontext}
A function specification context, \linebreak $\fsctx : \fids \rightharpoonup_{\mathit{fin}}  \mathcal{P}(\especs)$, is a finite partial function from function
identifiers to a set of external
 specifications, with the more familiar notation $\uquadruple{\pvvar x = \vec x \lstar P}{f(\pvvar x)}{\Qok}{\Qerr} \in \fsctx$ at times used in place of $\uspec{\pvvar x = \vec x \lstar P}{\Qok}{\Qerr} \in \fsctx(f)$.
\end{definition}


%
The constraints on external
specifications are well-known from OX logics and follow the
usual scoping of function parameters and local variables, which are limited to the function body:
the pre-conditions
only contain the function
parameters, $\pvvar x$; and the post-conditions may only have the
(dedicated) program variables $\pvar{ret}$ or $\pvar{err}$, which hold,
respectively, the return value  on successful termination or the
error value on faulting
termination.

Internal function specifications are more  interesting for exact and UX than for OX
reasoning. The internal pre-condition is straightforward, extending the
external pre-condition by instantiating the local variables to $\nil$. The
internal post-condition must therefore include information about the parameters
and local variables, as the internal specification cannot
lose information. This means that the connection between
internal and external post-conditions is subtle, given the constraints on the latter.
To address this, we define an internalisation function, relating
an external function specification with a set of possible internal
specifications. In particular, the external post-condition has to
be equivalent to an internal one in which the parameters and local
variables of the internal post-condition have been replaced by fresh
existentially quantified logical variables.


\begin{definition}[Internalisation]
Given implementation context $\fictx$ and function $\fid \in
\dom(\fictx)$, a  function specification internalisation, $\fext_{\fictx, \fid} : \especs \tmap \mathcal{P}(\specs)$, is defined as follows: 
\[
\begin{array}{l}
    \fext_{\fictx, \fid} (\uspec{P}{\Qok}{\Qerr}) = \\
        \quad \begin{array}{l@{\,\,}l}
	\{ \uspec{P \lstar \vec {\pvar z}\doteq\nil}{\Qok'}{\Qerr'} \mid & \models \Qok' \Rightarrow \pexp \in \Val \lstar \AssTrue~\text{and} \\
        & \models \Qok \Leftrightarrow\exsts{\vec{p}} \Qok'[\vec{p} / \vec{\pvar p}]\lstar\pvar{ret}\doteq \pexp[\vec{p} / \vec{\pvar p}]~\text{and} \\
        & \models \Qerr \Leftrightarrow\exsts{\vec{p}} \Qerr'[\vec{p} / \vec{\pvar p}] \},
        \end{array}
\end{array}
\]
where $f(\vec{\pvar{x}})\{\cmd; \preturn{\pexp}\} \in \fictx$, $\pvvar z=\pv{\cmd} \setminus \pv{P}$, $\vec{\pvar p} = \pv{P} \uplus \{\pvvar z\}$, and the logical variables $\vec p$ are fresh with respect to $\Qok$ and $\Qerr$.
\end{definition}



This approach also works for SL and ISL as well (with $\Leftarrow$ instead of $\Leftrightarrow$ for the post-conditions for SL, and $\Rightarrow$ instead of $\Leftrightarrow$ for ISL). It is not strictly necessary for SL, however, as information about program variables can
be forgotten in the internal post-conditions before the transition to the external post-condition.


\newcommand{\lenvx}{\omega}
\newcommand{\Lenvx}{\Omega}


\begin{definition}[Environments]
	An environment, $(\fictx,\fsctx)$, is a pair consisting of an
        implementation context $\fictx$ and a specification context
        $\fsctx$.
\end{definition}

An environment $(\fictx,\fsctx)$ is {\em valid}
if and only if every function specified
in $\fsctx$ has an implementation in $\fictx$ and every specification in $\fsctx$ has a $\fictx$-valid
internal specification.

\begin{definition}[Valid Environments]
\label{def:validenv}
Given an implementation context $\fictx$ and a
specification context~$\fsctx$, the environment $(\fictx,\fsctx)$
is {\em valid}, written $\models (\fictx,\fsctx)$,   if and only if
\[
\begin{array}{l}
 	\dom(\fsctx)\subseteq\dom(\fictx)~\land \\ \quad
		 (\forall \fid, \pvvar x, \cmd, \pexp.~f(\vec{\pvar{x}})\{\cmd; \preturn{\pexp}\} \in \fictx \implies
		 (\forall \sspec.~\sspec \in \fsctx(\fid) \implies
		  \exists \sspec' \in \fext_{\fictx, \fid}(\sspec).~\fictx\models \cmd:\sspec'))
\end{array}
\]

\end{definition}

Finally, a specification $t $ is valid for a command~$\cmd$ in a
specification context $\fsctx$
if and only if $t$ is $\fictx$-valid for  all implementation contexts
$\fictx$ that validate $\fsctx$.

\begin{definition}[$\fsctx$-Valid Specifications]
\label{def:Gammavalid} Given a
  specification context $\fsctx$, a command $\cmd$, and a specification
  $t = \uspeconecaseq{P}{\bigq}$, the specification $t$ is {\em $\fsctx$-valid} for
  command $\cmd$, written $\fsctx \models \cmd : t$ or $\fsctx \models
  \utripleq{P}{\cmd}{\bigq}$,  if and only if
$
\frall{\fictx}~\models(\fictx,\fsctx) \implies
\fictx\models \utripleq{P}{\cmd}{\bigq}
$.
\end{definition}

\subsection{Program Logic}
\label{sec:proglog}

We give the representative ESL proof rules in Figure~\ref{fig:unify} and all in Appendix~\ref{apdx:esl}. We introduce and discuss in detail the function-related rules, given for the first time in a UX setting. We denote the repetition of the pre-condition in the post-condition by $\pvar{pre}$.
When reading the rules, it  is important to remember that we {\it must not drop paths and must not lose information}.
The judgement
$\fsctx \vdash \uquadruple{P}{\cmd}{\Qok}{\Qerr}$ means that the specification $t$ is {\it derivable} for a command $\cmd$ given the specifications recorded in $\fsctx$.

The basic command rules are fairly straightforward.  The
\prooflab{nondet} rule
existentially quantifies the generated value via
$\pvar{x} \dotint \Nat$ to capture all paths, in contrast with the
 RHL~\cite{reverselogic} and ISL~\cite{isl} rules,
which explicitly choose one value to describe one path. The
$\lexp' \dotint \Val$ in the post-condition is necessary as we know that
$\lexp'$ evaluates to a value from the
pre-condition and cannot lose information;
the same principle applies to many other rules. The
\prooflab{assign} rule requires that the evaluation of $\pexp$
does not fault in the pre-condition via $\dotin{\pexp}{\Val}$.
Strictly speaking, we should
have an additional case in which the assigned variable is not in the store. To avoid
this clutter, we instead assume that program variables are always in
the store as we are analysing function bodies and, in our programming language, all local
variables are initialised on function entry.  The error-related rules
capture cases in which expression evaluation faults
(e.g. \prooflab{lookup-err-val} rule, using $\pexp \notin \Val$), expressions are of the
incorrect type, or memory is accessed after it has been freed (e.g. \prooflab{lookup-err-use-after-free} rule, using
$\pexp \mapsto \cfreed$). Note that missing
resource errors cannot be captured without breaking frame preservation, as the added-on
frame could~contain~the~missing~resource.


\begin{figure*}[!t]
	\small
  \begin{minipage}{\textwidth}
\begin{mathpar}
  \inferrule[\mbox{skip}]
  {}
    { \fsctx \vdash \utripleq{\emp}{\texttt{skip}}{\emp} }
  \qquad
    \inferrule[\mbox{nondet}]
    {\pvar x \notin \pv{\lexp'} \\\\ Q \defeq \lexp' \dotint \Val \lstar \pvar x \dotint {\Nat}}
    { \fsctx \vdash \utripleq{\pvar x \doteq \lexp'}{\passign{\pvar{x}}{\prandom}}{Q} }
  \qquad
  \inferrule[\mbox{assign}]
  {\pvar x \notin \pv{\lexp'} \\\\
  Q \defeq \dotin{\lexp'}{\Val} \lstar  \pvar{x} \doteq \pexp[\lexp'/\pvar x]}
    { \fsctx \vdash \utripleq{{\pvar{x} \doteq \lexp' \lstar \dotin{\pexp}{\Val}}}{\passign{\pvar{x}}{\pexp}}{Q} }
  \\
  \inferrule[\mbox{lookup}]
      { \pvar x \notin \pv{\lexp'} \\\\
      Q \defeq \dotin{\lexp'}{\Val} \lstar \pvar{x} \doteq \lexp_1[\lexp'/ \pvar{x}] \lstar \pexp[\lexp'/ \pvar{x}] \mapsto \lexp_1 [\lexp'/ \pvar{x}]}
    { \fsctx \vdash \utripleq{\pvar{x} \doteq \lexp'  \lstar \pexp \mapsto \lexp_1}{\pderef{\pvar{x}}{\pexp}}{Q} }
  \qquad
  \inferrule[\mbox{mutate}]
      { Q \defeq \expr{\pexp_1} \mapsto  \expr{\pexp_2} \lstar \lexp \dotint \Val}
    { \fsctx \vdash \utripleq{\expr{\pexp_1} \mapsto \lexp \lstar \dotin{\pexp_2}{\Val}}{\pmutate{\expr{\pexp_1}}{\expr{\pexp_2}}}{Q} }
  \\
  \inferrule[\mbox{new}]
    { \pvar{x}  \notin \pv{\lexp'}  \\\\
     Q \defeq \lexp' \dotint \Val \lstar \bigoast{_{0 \le i < \pexp[\lexp'/\pvar x]}} (( \pvar x + i) \mapsto \nil)}
          { \fsctx \vdash \utripleok{\pvar x \doteq \lexp' \lstar \pexp \dotint \Nat}{\palloc{\pvar{x}}{\pexp}}{Q} }
  \qquad
\inferrule[\mbox{error}]
{ \eerr \defeq [``\mathsf{Error}", \pexp] }
{ \fsctx \vdash \utripleerr{\dotin{\pexp}{\Val}}{\perror(\pexp)}{\pvar{err} \doteq \eerr} }
 \\
\inferrule[\mbox{free}]
{Q \defeq \dotin{\lexp'}{\Val}\lstar \pexp \mapsto \cfreed}
{ \fsctx \vdash \utripleok{\pexp \mapsto \lexp'}{\pdealloc{\pexp}}{Q} }
\qquad
  \inferrule[\mbox{lookup-err-val}]
  { P \defeq \pvar x \doteq \lexp' \lstar E \notdotint \Val \\\\ \eerr \defeq [``\mathsf{ExprEval}", \stringify {\pexp}]}
  { \fsctx \vdash \utripleerr{P}{\pderef{\pvar{x}}{\pexp}}{\Qerr^*} }
\qquad
\inferrule[\mbox{lookup-err-use-after-free}]
{ P \defeq \pvar x \doteq \lexp' \lstar \pexp \mapsto \cfreed \\\\ \eerr \defeq [``\mathsf{UseAfterFree}", \stringify{\pexp}, \pexp]}
{ \fsctx \vdash \utripleerr{P }{\pderef{\pvar{x}}{\pexp}}{\Qerr^*} }
     \\
  \inferrule[\mbox{if-then}]
  { \cmd \defeq \pifelse{\pexp}{\cmd_1}{\cmd_2} \\\\ \fsctx \vdash \utripleq{P \lstar \pexp}{\cmd_1}{\bigq} }
  { \fsctx \vdash \utripleq{P \lstar \pexp }{\cmd}{\bigq} }
\qquad
  \inferrule[\mbox{if-else}]
  { \cmd \defeq \pifelse{\pexp}{\cmd_1}{\cmd_2} \\\\\fsctx \vdash \utripleq{P \lstar \lnot \pexp}{\cmd_2}{\bigq} }
  { \fsctx \vdash \utripleq{P \lstar \lnot \pexp }{\cmd}{\bigq} }
\qquad
  \inferrule[\mbox{if-err-val}]
  { \cmd \defeq \pifelse{\pexp}{\cmd_1}{\cmd_2} \\\\ \eerr \defeq ["\mathsf{ExprEval}", \stringify{\pexp}] }
  { \fsctx \vdash \utripleerr{P \lstar \pexp\notdotint\Val}{\cmd}{\Qerr^*} }
\\
 \inferrule[\mbox{seq}]
 { {\begin{array}{c} \fsctx \vdash \uquadruplex{P}{\cmd_1}{R}{\Qerr^1} \\ \fsctx \vdash \uquadruplex{R}{\cmd_2}{\Qok}{\Qerr^2} \end{array}}}
 { \fsctx \vdash \uquadruplex{P}{\psequence{\cmd_1}{\cmd_2}}{\Qok}{\Qerr^1 \lor \Qerr^2} }
\quad
%
	\inferrule[\mbox{while}]
{
\frall{i\in\Nat}~\models P_i \Rightarrow \pexp \in \Bool \lstar \AssTrue \\\\
\frall{i\in\Nat}~\fsctx \vdash \uquadruplex{P_i \lstar \pexp}{\cmd}{P_{i+1} }{Q_i}
  }
  { \fsctx \vdash \uquadruplex{P_0}{\pwhile{\pexp}{\cmd}}{\neg\pexp\lstar\exsts{i} P_i}{\exsts{i}  Q_i}
  }
  \quad
%
\\
 \inferrule[\mbox{equiv}]
  { {\begin{array}{c} \fsctx \vdash \uquadruplex{P'}{\cmd}{\Qok'}{\Qerr'} \\ \models P', \Qok', \Qerr' \Leftrightarrow P, \Qok, \Qerr
  \end{array}}}
 { \fsctx \vdash \uquadruplex{P }{\cmd}{\Qok}{\Qerr}}
 \quad
\inferrule[\mbox{frame}]
{ {\begin{array}{c} \updt(\cmd) \cap \fv{R} = \emptyset \\ \fsctx \vdash \uquadruplex{P}{\cmd}{\Qok}{\Qerr}\end{array}} }
{ \fsctx \vdash \uquadruplex{P \lstar R}{\cmd}{\Qok \lstar R}{\Qerr \lstar R}}
\\
 \inferrule[\mbox{exists}]
{ \fsctx \vdash \uquadruplex{ P}{\cmd}{\Qok}{\Qerr} }
{ \fsctx \vdash \uquadruplex{\exists x.\, P}{\cmd}{\exsts x \Qok} {\exsts x \Qerr}} \quad
  %
 \inferrule[\mbox{disj}]
  { \fsctx \vdash \uquadruplex{P_1}{\cmd}{\Qok^1}{\Qerr^1} \\\\
   \fsctx \vdash \uquadruplex{P_2}{\cmd}{\Qok^2}{\Qerr^2}}
 { \fsctx \vdash \uquadruplex{P_1 \lor P_2}{\cmd}{\Qok^1 \lor \Qok^2}{\Qerr^1 \lor \Qerr^2} }
\end{mathpar}
\end{minipage}
\caption{ESL proof rules (excerpt), with $\Qerr^* = (\pvar{pre} \lstar \pvar{err} \doteq \eerr)$}
\label{fig:unify}
\end{figure*}

When it comes to composite commands, we opt for two
$\mathtt{if}$-rules, covering the branches separately.
The sequencing rule shows how exact quadruples of successive
commands can be joined together, highlighting, in particular, how
errors are collected using disjunction. One interesting aspect of this rule is
what happens when $C_1$ only throws an error or does not terminate, meaning that $R = \AssFalse$.
In both those cases, given the exactness of the rules, it has to be that $\Qok = \Qerr^2 = \AssFalse$,
and the post-condition of the sequence becomes $(\mathit{ok}: \AssFalse)\ (\mathit{err}: \Qerr^1 \lor~\AssFalse)$,
meaning that, if $C_1$ only throws an error (that is, $\Qerr^1 \neq \AssFalse$) then that is
the only error that can come out of the sequence, and if $C_1$ does not terminate (that is, $\Qerr^1 = \AssFalse$)
then the sequence does not terminate either.

The while rule is an adaptation of the RHL while rule~\cite{reverselogic}, generalising the invariant of the SL while rule with two natural-number-indexed families of variants, $P_i$ and $Q_i$, which explicitly maintain the iteration index. Note how the $i$ in the premise is a meta-variable representing a natural number, which in the conclusion gets substituted for an existentially quantified logical variable; a similar principle will be applied later when dealing with environment extension. Interestingly, this rule does not require adjustment to reason about non-termination.

The structural rules are not surprising, with equivalence
replacing the forward/backward consequence of OX/UX reasoning
and with frame, existential
introduction, and disjunction affecting both post-conditions.
Disjunction  allows us to derive the standard SL $\mathtt{if}$ rule, which captures both branches at the same time.
Note that there, however, is no sound conjunction rule,
as the conjunction rules of SL and ISL cannot be combined in
ESL, since conjunction does not distribute over the star in both
directions, breaking frame preservation.




\myparagraph{Function Call}
We discuss the ESL function-call rule in detail, creating it starting from the standard OX-sound SL rule, adapted for quadruples:
\begin{mathpar}
  \small
  \inferrule
  {\quadruple {\pvvar x = \vec x \lstar P}{f(\pvvar x)} \Qok \Qerr \in \fsctx \qquad \pvar{y} \not\in \pv{\lexp_y}}
  {
    \fsctx \vdash
    \quadruple{\pvar{y} \doteq \lexp_y \lstar \vec \pexp = \vec x \lstar P}
                {\pfuncall{\pvar y}{f}{\vec \pexp}}
                {\Qok[\pvar{y} / \mathtt{ret}]}{\pvar y = \lexp_y \lstar \Qerr}
  }
\end{mathpar}
In order to make this rule UX-sound, we only have to ensure that no information from the pre-condition is lost in the post-conditions. In particular, we have to remember:
\begin{itemize}
\item that the evaluation of $\lexp_y$ does not fault, captured by $\lexp_y \in \vals$ and needed in the success post-condition only, as it is already implied by the error post-condition; and
\item that $\vec \pexp = \vec x$ holds (with the substitution $[\lexp_y/\pvar y]$ needed in the success post-condition as the value of $\pvar y$ may change if the function call succeeds);
\end{itemize}
bringing us to the ESL function call rule:
\begin{mathpar}
	\small
  \inferrule
  {\uquadruple {\pvvar x = \vec x \lstar P}{f(\pvvar x)} \Qok \Qerr \in \fsctx \qquad \pvar{y} \not\in \pv{\lexp_y} \\ \Qok' \defeq \lexp_y \dotint \Val \lstar \vec{\pexp}[\lexp_y/\pvar{y}] \doteq \vec x \lstar \Qok[\pvar{y} / \mathtt{ret}] \qquad \Qerr' \defeq \pvar y \doteq \lexp_y \lstar \vec{\pexp} \doteq \vec x \lstar \Qerr}
  {
    \fsctx \vdash
    \uquadruple{\pvar{y} \doteq \lexp_y \lstar \vec \pexp \doteq \vec x \lstar P}
                {\pfuncall{\pvar y}{f}{\vec \pexp}}
                {\Qok'}
                {\Qerr'}
  }
\end{mathpar}

\myparagraph{Environment Formation}
Whereas the ESL function-call rule does not deviate substantially from its OX counterpart, the environment formation rules illustrate the difference in complexity between OX and UX function compositionality. These rules use the judgement $\vdash (\fictx, \fsctx)$ to state that the environment $(\fictx, \fsctx)$ is {\it well-formed}. The base case, $\vdash (\emptyset, \emptyset)$, is trivial and the same as for SL, stating that the environment consisting of an empty implementation context and an empty specification context is well-formed. For illustrative purposes, we give a simplified version of the extension rule, extending the environment with a single, possibly recursive, function. The full rule, which extends the environment with a group of mutually recursive functions, is given in Appendix~\ref{apdx:esl}.
We start from the corresponding OX-sound rule from SL:
\[
	\small
   	\inferrule[env-extend-sl]
{
  		{\begin{array}{c}
  		\vdash (\fictx,\fsctx)
  		\quad
  		f \not\in \dom(\fictx)
  		\quad
		\fictx' = \fictx[f \mapsto (\vec{\pvar x}, \cmd, \pexp)] 
		\quad
  		\fsctx' = \fsctx[f \mapsto \{ t \}]
		\quad
		\exists t' \in \fext_{\fictx', f}(t).~\fsctx' \vdash \cmd : \sspec'
		\end{array}}
	}{
  		\vdash (\fictx', \fsctx')}
\]
which states that a well-formed environment $(\fictx, \fsctx)$ can be extended with a given function $f$ and its external specification $t$ to $(\fictx', \fsctx')$ if some corresponding internal specification of $f$ can be proven for the body of $f$ under the extended specification context $\fsctx'$.\footnote{Internalisation is normally omitted in SL as forward consequence allows information about program variables to be lost, making internal and external post-conditions of SL specifications almost the same.} Note that using $\fsctx'$ means that a specification can be used to prove itself, which is sound in SL but unsound in ISL: specifically, it would allow us to prove UX-invalid specifications of non-terminating functions. For example, we
would be able to prove that the function
$
\mathtt{f}()~\{~\pvar r := \mathtt{f}();~\preturn \pvar r~\}
$
satisfies the EX-valid specification
$\utripleq{\emp}{\mathtt{f}()}{\AssFalse}$, but also the
EX-invalid specification $\utripleq{\emp}{\mathtt{f}()}{\pvar{ret} \doteq
  42}$. The latter is vacuously OX-valid as there are
  no terminating executions, but when considered from the UX viewpoint,
it implies the existence of an execution path from the pre- to the
post-condition, contradicting the non-termination of $\mathtt{f}$.

Therefore, to be soundly usable in UX reasoning, speci\-fications with satisfiable post-conditions (onward: terminating specifications) must come with a mechanism that disallows the above counter-example. We achieve this by following a standard approach for reasoning about termination~\cite{Floyd1967Flowcharts,total-tada}, based on decreasing measures on well-ordered sets.
In particular, we require the terminating specification to be proven, $t \defeq \uspec{P}{\Qok}{\Qerr}$ to be extended with a {\it measure} $\alpha \in \mathbb{N}$, denoting this extension by $\sspec(\alpha)$:\footnote{We can extend the measure beyond natural numbers to computable ordinals, $\ord\defeq\ckord$, allowing us to reason about a broader set of functions, such as those with non-deterministic nested recursion (cf. Appendix~\ref{apdx:esl}).}
\[
	\begin{array}{r@{~\defeq~}l}
	\sspec(\alpha) & \uspec{P \lstar \alpha = E_\mu}{\Qok \lstar \alpha = E_\mu}{\Qerr \lstar \alpha = E_\mu}
	\end{array}
\]
where $E_\mu$ is a logical expression describing how the measure is computed.
Then, we prove that $\sspec(\alpha)$ holds for every specific $\alpha$, assuming that recursive calls to $f$ can only use the terminating specifications $\sspec(\beta)$ of a measure $\beta$ {\it strictly smaller} than $\alpha$.
This restriction is standard and, if the proof succeeds, ensures that $f$ has at least one terminating~execution. Also, it disallows the above-mentioned counter-example, as no measure given in the pre-condition would be able to decrease before the recursive call. Importantly, the measure is only a tool required for proving specification validity and once this proof has been completed, the specification without the measure is added to $\fsctx$ and can be used in proofs of client code.

In addition, we incorporate reasoning about non-terminating specifications (NT-specifica\-tions). This is relevant in situations in which the operational semantics of the analysed language is complete, which allows non-termination to be captured using the post-condition $\AssFalse$. As NT-specifications are vacuously UX-sound, our focus is on ensuring their OX-soundness, which we do by again imposing a measure $\alpha$, but allowing recursive calls for a measure $\beta$ {\it smaller or equal} than $\alpha$, that is, for an NT-specification to be used to prove itself. This, for example, allows for a proof of the specification $\utripleq{\emp}{\mathtt{f}()~\{~\pvar r := \mathtt{f}();~\preturn \pvar r~\}}{\AssFalse}$ by choosing a constant measure $\alpha$. The ESL environment extension rule, therefore, is as follows:
\begin{mathpar}
	\small
	\inferrule[\mbox{env-extend}]{
		{\begin{array}{l}
		\text{\polish{// Assume valid environment, extend implementation context with new function $f$ }} \\
		\vdash (\fictx,\fsctx)
		\qquad
		\fid \not\in \dom(\fictx)
		\qquad \fictx' = \fictx[\fid \mapsto (\vec{\pvar x}, \cmd, \pexp)]
		\\[1mm]
		\text{\polish{// Extend the specifications of $f$ with a measure $\alpha$}} \\
		\sspec \defeq \uspec{P}{\Qok}{\Qerr} \quad
	\sspec_\infty \defeq \uspeconecase{P_\infty}{\AssFalse} \quad \sspec_\infty(\alpha) \defeq \uspeconecase{P_\infty \lstar \alpha = E_\mu}{\AssFalse} \\
		\sspec(\alpha) \defeq \uquadruple {P \lstar \alpha = E_\mu} {\!} {\Qok \lstar \alpha = E_\mu} {\Qerr \lstar \alpha = E_\mu} \quad
		 \\ ~ \\[-3mm]
		\text{\polish{// Construct $\fsctx(\alpha)$: assume $\sspec$ for measure $\beta < \alpha$ {and $\sspec_\infty$ for measure $\beta \leq \alpha$}}} \\
		\fsctx(\alpha) = \fsctx[\fid\mapsto
				\{ \sspec(\beta) \mid {\beta < \alpha}\} \mathbin\cup \{ t_\infty \mid \beta \leq \alpha \}]
		\\[1mm]
		\text{\polish{// For every $\alpha$, prove internal specifications of $f$ corresponding to $\sspec$ {and $\sspec_\infty$}}} \\
		\frall{\alpha} \exists t' \in \fext_{\fictx', \fid}\big(\sspec(\alpha)\big).~\fsctx(\alpha)\vdash \cmd : \sspec' \qquad
		{\frall{\alpha} \exists t' \in \fext_{\fictx', \fid}(\sspec_\infty(\alpha)).~\fsctx(\alpha)\vdash \cmd : \sspec'}  \\[1mm]
		%
		%
		\text{\polish{// Extend $\fsctx$ with $\sspec$ {and $\sspec_\infty$}}} \\
		\fsctx':=\fsctx[\fid\mapsto\{\sspec, \sspec_\infty\}]
		\end{array}}
	}
	{
		\vdash (\fictx',\fsctx')
	}
\end{mathpar}

\subsection{Soundness}
\label{ss:soundness} We state the soundness results for ESL and give intuition about the proofs; the full proofs can be found in Appendix~\ref{apdx:soundesl}, \ref{apdx:scott}, and \ref{apdx:envsound}.

\begin{theorem}
\label{logicsoundness}
Any derivable specification is valid:
$
\fsctx \vdash \utripleq P \cmd \bigq  \implies   \fsctx \models \utripleq P \cmd \bigq
$.
\end{theorem}

\begin{proof}
By induction on $\fsctx \vdash \utripleq P \cmd \bigq$. Most cases are straightforward; the \prooflab{fun-call} rule obtains a valid specification for the function body from the validity of the environment.
\end{proof}

\begin{restatable}{theorem}{restateenvsound}\label{thm:envsound} Any well-formed environment is valid:
	$
		\vdash (\fictx,\fsctx)\implies \;\models (\fictx,\fsctx)
	$.
\end{restatable}

\begin{proof}
At the core of the proof is a lemma stating that $\models(\fictx,\fsctx) \implies (\frall{\alpha} \models(\fictx',\fsctx(\alpha)))$, where $\fictx'$ and $\fsctx(\alpha)$ have been obtained from $\fictx$ and $\fsctx$ as per \prooflab{env-extend}. Using this lemma, we derive the desired $\models(\fictx',\fsctx')$, where $\fsctx'$ is obtained from $\fsctx$ and $\fsctx(\alpha)$ as per \prooflab{env-extend}. The proof of this lemma is done by transfinite induction on $\alpha$ and has the standard zero, successor, and limit ordinal cases. We outline the proof for the case in which a single, possibly recursive, function $f$ with body $\cmd_\fid$ is added; the generalisation to $n$ mutually recursive functions is straightforward and can be seen in Appendix~\ref{apdx:envsound}.

In all three cases, the soundness of all specifications except the NT-specification with the highest considered ordinal follows from the inductive hypothesis. This remaining NT-specification is vacuously UX-valid, meaning that we only need to prove its OX-validity.
For this, we use a form of fixpoint induction called Scott induction (see, e.g., Winskel~\cite{scott}), required when specifications can be used to prove themselves (e.g. any SL specification).

We set up the Scott induction by extending the set of commands with two pseudo-commands, $\mathtt{scope}$ and $\mathtt{choice}$, 
with the former modelling the function call but allowing arbitrary commands to be executed in place of the function body, and the latter denoting non-deterministic choice.
We then construct the greatest-fixpoint closure of these extended commands, denoted by $\ccmdplain$, whose elements may contain infinite applications of the command constructors. We define a behavioural equivalence relation $\simeq_{\fictx'}$ on $\ccmdplain$ and denote by $\ccmdpr$ the obtained quotient space. This relation induces a partial order $\poccmd$, and a join operator that coincides with $\mathtt{choice}$, and we show that $(\ccmdpr,\poccmd)$ is a domain.

We next define $S^\alpha$ as the set of all equivalence classes that hold an element that, for every specification in $(\fsctx(\alpha))(f)$, OX-satisfies at least one of its internal specifications, and show that $S^\alpha$ is an admissible subset of $\ccmdpr$, that is, that it contains the least element of $\ccmdpr$ (represented, for example, by the infinite loop $\mathtt{while}~(\true)~\{~\pskip~\}$) and is chain-closed.

We then define the function $h(\cmd)\defeq\fcallsub{\cmd_f}{\cmd, \fictx', \fid}$, which replaces
all function calls to $\fid$ in $\cmd_\fid$ with $\cmd$ using the $\mathtt{scope}$ command, and the function $g$ as the lifting of $h$ to $\ccmdpr$:
$
g(\repr{\cmd}):= \left[h(\cmd)\right]
$.
We next prove that $g$ is continuous (that is, monotonic and supremum-preserving) and that $g(S^\alpha)\subseteq S^\alpha$, from which we can apply  the Scott induction principle, together with a well-known identity of the least-fixpoint, which implies that $\cmd_\fid\in\lfp{g}$, to obtain that $\repr{\cmd_\fid}\in S^\alpha$. From there, we are finally able to prove that $\models(\fictx',\fsctx(\alpha))$.
\end{proof}

These two theorems, to the best of our knowledge, are the first to demonstrate sound function compositionality for UX logics. Previous work on UX logics~\cite{isl,cisl,insecsl} used function specifications in examples but did not include rules in the logic for calling functions and managing a function specification environment. Program logics that do not include function rules and function specification environments in effect delegate soundness responsibilities to the meta-logic within which they are embedded. This might be appropriate in some contexts, such as in interactive theorem provers, whose meta-logic is reliable. Chargu\'{e}raud's clean-slate tutorial SL implementation in Coq~\cite{Chargueraud20}, for example, does not provide either function call rules or program-logic-level infrastructure for a function specification environment; instead, it relies on Coq's induction mechanism and definitional mechanism to use and store function specifications. However, when implementing an SL/ISL/ESL-based tool in a mainstream non-ITP language, such as C++ or OCaml, no reliable meta-logic that can act as a safety net is available. This is particularly concerning for UX logics, which require complex rules for handling functions, including forgetting information about function-local mutable variables. In ESL, the handling of functions is fully internalised into the logic, no meta-logic facilities are required to handle function calls, and the program-logic-level facilities of ESL for handling functions are validated by the soundness~proof.

The proof of Theorem~\ref{thm:envsound} can be adjusted for ISL: the function call rule would remain the same, and \prooflab{env-extend} would not include NT-specifications, removing the need for Scott induction. On the other hand, the Scott induction itself could be easily adapted for~SL.

%% file: sections/examples.tex
\section{Examples: ESL in Practice}
\label{ch:examples}

We demonstrate how to use ESL to specify and verify correctness and incorrectness properties of data-structure algorithms, focussing on singly-linked lists and binary search trees. 

We investigate, for the first time, the use of abstract predicates in a UX program logic, decoupling abstraction from over-approximation. Our findings show that UX/EX specifications can soundly incorporate abstraction, but also that it, ultimately, cannot be used as freely as in the OX setting. Firstly, since UX reasoning cannot lose information, not all algorithms can be UX-specified at all levels of abstraction, and hence sometimes specifications have to be less abstract than in OX reasoning. Secondly, because specifications are only composable when expressed at the same level of abstraction, specifications of a library client have to be written at the ``least common level of abstraction'' of the specifications of all of the library functions that the client calls. 

Building on \S\ref{sec:overview}, we give further intuition on how to think informally about UX/EX specifications using a number of list algorithms and predicates describing lists with various degrees of abstraction (\S\ref{ss:lpreds}, \S\ref{ss:uxexabs}), more detail on how to write formal ESL proofs (\S\ref{ss:eslproofs}), and examples of ESL reasoning for binary-search-tree algorithms~(\S\ref{ss:beyond}).


\subsection{List Predicates}
\label{ss:lpreds}

We implement singly-linked lists in the standard way: every list consists of two contiguous cells in the heap (denoted $x \mapsto a, b$, meaning $x \mapsto a \lstar x + 1 \mapsto b$), with the first holding the value of the node, the second holding a pointer to the next node in the list, and the list terminating with a $\nil$ pointer.
To capture lists in ESL, we use a number of list predicates:
\[
\begin{array}{r@{~}c@{~}l}
	\llist{x} & \defeq & (x\doteq\nil)~\lor (\exsts{v, x'} x\mapsto v,x' \lstar\llist{x'}) \\[1mm]
	\llist{x, n} & \defeq & (x\doteq\nil \lstar  n\doteq0) \lor (\exsts{v,x'} x\mapsto v,x' \lstar\llist{x',n - 1} ) \\[1mm]
\llist{x, \lstvs} & \defeq & (x\doteq\nil \lstar \lstvs \doteq \emplist) \lor (\exsts{v, x', \lstvs'} x \mapsto v, x' \lstar\llist{x',\lstvs'}\lstar \lstvs \doteq v \cons \lstvs') \\[1mm]
\llist{x, \lstxs}& \defeq & (x\doteq\nil \lstar \lstxs \doteq \emplist) \lor  (\exsts{v, x', \lstxs'} x \mapsto v, x' \lstar\llist{x',\lstxs'}\lstar \lstxs \doteq x \cons \lstxs') \\[1mm]
\llist{x, \lstxs, \lstvs} & \defeq & (x\doteq\nil \lstar \lstxs \doteq \emplist \lstar \lstvs \doteq \emplist)~\lor \\[0.25mm] && (\exsts{v, x', \lstxs', \lstvs'} x \mapsto v, x' \lstar\llist{x',\lstxs', \lstvs'} \lstar  \lstxs \doteq x \cons \lstxs' \lstar \lstvs \doteq v \cons \lstvs') \\[1mm]
\end{array}
\]
These predicates expose different parts of the list 
structure in their parameters, {\it hiding} the rest via
existential quantification: the 
$\llist{x}$ predicate hides all information about the represented  mathematical list,
just declaring that there is a singly-linked list at address~$x$; 
the
$\llist{x, n}$ predicate hides the internal node addresses and values,
exposing the list length via the parameter $n$;
the $\llist{x, \lstxs} $ predicate hides information about the values
of the mathematical list, exposing the internal addresses of the list
via the parameter $\lstxs$;
the $\llist{x, \lstvs}$ predicate hides   information about the internal
addresses, exposing  the list's values  via the parameter $\lstvs$; 
and the list predicate $\llist{x, \lstxs, \lstvs}$ hides nothing, exposing the entire node-value structure via the parameters $\lstxs$ and $\lstvs$.
These predicates are related to each other via logical equivalence; for example, it holds that:
\[
\begin{array}{c@{\qquad}c}
\models \llist{x} \Leftrightarrow \exists n.~\llist{x, n} & \models \llist{x} \Leftrightarrow \exists \lstvs.~\llist{x, \lstvs} \\[1mm]
\models \llist{x} \Leftrightarrow \exists \lstxs, \lstvs.~\llist{x, \lstxs, \lstvs} & 
\models \llist{x, n} \Leftrightarrow \exists \lstvs.~\llist{x, \lstvs} \lstar |\lstvs| \doteq n
\end{array}
\]

\subsection{Writing UX/EX Abstract Specifications}
\label{ss:uxexabs}
We consider a number of list algorithms, described in words, and guide the reader on how to write correct UX/EX specifications for these algorithms using the list abstractions given in the previous section (\S\ref{ss:lpreds}), comparing how the UX/EX approach and specifications differ from their OX counterparts. We provide detailed proofs and implementations for each type of algorithm (iterative/recursive, allocating/deallocating, pure/mutative, etc.) in Appendix~\ref{apdx:examples}.

An important point is to understand how to look for counter-examples to a given specification: from the definition of OX validity, it follows that breaking OX reasoning amounts to ``finding a state in the pre-condition (pre-model) for which the execution of~$f$ terminates and does not end in a state in the post-condition''; from the definition of UX validity, breaking UX reasoning means ``finding a model of the post-condition (post-model) not reachable by execution of $f$ from any state in the pre-condition''; and from the definition of EX validity, it follows that breaking EX reasoning means breaking either OX or UX reasoning. In addition, it is useful to remember that, for breaking UX validity, it is sufficient to find information known in the pre-condition but lost in the post-condition. 


\subparagraph*{Length.} We first revisit the list-length function $\mathtt{LLen}(\pvar x)$, which takes a list at address~$\pvar x$, does not modify it, and returns its length. In \S\ref{sec:overview}, we have shown that it satisfies the exact specification  
$
\utripleq{\pvar x\doteq x \lstar \llist{x,n} }{~\mathtt{LLen}(\pvar x)~}{\llist{x,n} \lstar\pvar{ret}\doteq n}
$
and observed that
\begin{center}
{\bfseries (O1)}~abstraction does not always equal over-approximation.
\end{center}


Using similar reasoning, we can come to the conclusion that the following, less abstract specifications for $\mathtt{LLen}(\pvar x)$ are also EX-valid:
\[
\begin{array}{c}
\utripleq{\pvar x\doteq x \lstar \llist{x,\lstvs} }{~\mathtt{LLen}(\pvar
  x)~}{\llist{x,\lstvs} \lstar\pvar{ret}\doteq |\lstvs|} \\
\utripleq{\pvar x\doteq x \lstar \llist{x,\lstxs,\lstvs} }{~\mathtt{LLen}(\pvar
  x)~}{\llist{x,\lstxs,\lstvs} \lstar\pvar{ret}\doteq |\lstxs|}
\end{array}
\]
On the other hand, if we consider the following OX-valid specification:
\[
\triple{\pvar x\doteq x \lstar \llist{x} }{~\mathtt{LLen}(\pvar
  x)~}{\exists n \in \Nat.~\llist{x} \lstar\pvar{ret}\doteq n}
\]
we see that it is not UX-valid as the post-condition does not connect the return value to the list.
In particular, if we choose a post-model for $\llist{x}$  that has length 2, 
but then choose, for example, $\pvar{ret} = 42$, we run into a problem: as the algorithm
does not modify the list, we have to choose the same model of $\llist{x}$ for the pre-model to have a chance of reaching the post-model, but then the algorithm will return 2, not 42, meaning that this specification is indeed not UX/EX-valid. From this discussion, we observe that:
\begin{center}
{\bfseries (O2)}~in valid UX/EX specifications, data-structure abstractions used in a post-condition must expose enough information to capture the behaviour of the function being specified with respect to the information given in the pre-condition.
\end{center}

Note that, given a specification less abstract than strictly needed, one can obtain more abstract ones by using validity-preserving transformations on specifications that correspond to the structural rules of the logic. We refer to these as {\it admissible} transformations, give the ones for existential introduction and equivalence below, and the rest in Appendix~\ref{apdx:esl}:
\begin{mathpar}
	\small
 \inferrule[\mbox{adm-exists}]
  { \Gamma \models  \uquadruple{\vec{\pvar x} = \vec x \lstar P} {f(\pvvar x)} \Qok \Qerr  \qquad y \notin \vec x}
  { \Gamma \models \uquadruple{\vec{\pvar x} = \vec x \lstar \exists y.~P}{f(\pvvar x)}{\exists y.~\Qok}{\oxerr : \exists y.~\Qerr} } 
\and
 \inferrule[\mbox{adm-equiv}]
  { {\begin{array}{c} \Gamma \models  \uquadruple{\vec{\pvar x} = \vec x \lstar P'} {f(\pvvar x)} {\Qok'}{\Qerr'}  \qquad \models P', \Qok', \Qerr' \Leftrightarrow P, \Qok, \Qerr \end{array}}}
  { \Gamma \models  \uquadruple{\vec{\pvar x} = \vec x \lstar P} {f(\pvvar x)} \Qok \Qerr}
\and

  \end{mathpar}

For list-length, starting from the least abstract specification using $\llist{x,\lstxs,\lstvs}$, we can derive, for example, the specification using $\llist{x,n}$, as follows:
\[
\begin{array}{@{}r@{\quad}l}
&  \utripleq{\pvar x\doteq x \lstar \llist{x,\lstxs,\lstvs} }{~\mathtt{LLen}(\pvar
  x)~}{\llist{x,\lstxs,\lstvs} \lstar\pvar{ret}\doteq |\lstxs|} \\
\text{\prooflab{adm-exists}} &  \utripleq{\exists \lstvs.~\pvar x\doteq x \lstar \llist{x,\lstxs,\lstvs} }{~\mathtt{LLen}(\pvar
  x)~}{\exists \lstvs.~\llist{x,\lstxs,\lstvs} \lstar\pvar{ret}\doteq |\lstxs|} \\
\text{\prooflab{adm-exists}} &  \utripleq{\exists \lstxs, \lstvs.~\pvar x\doteq x \lstar \llist{x,\lstxs,\lstvs} }{~\mathtt{LLen}(\pvar
  x)~}{\exists \lstxs, \lstvs.~\llist{x,\lstxs,\lstvs} \lstar\pvar{ret}\doteq |\lstxs|} \\
  \text{\prooflab{adm-equiv}} &  \utripleq{\pvar x\doteq x \lstar \llist{x, n}}{~\mathtt{LLen}(\pvar
  x)~}{\llist{x,n} \lstar\pvar{ret}\doteq n}
\end{array}
\]
Interestingly, from one more application of \prooflab{adm-exists} and \prooflab{adm-equiv}, we can derive
\[
\utripleq{\pvar x\doteq x \lstar \llist{x} }{~\mathtt{LLen}(\pvar x)~}{\exists n.~\llist{x,n} \lstar\pvar{ret}\doteq n}
\]
which further illustrates observation (O2), in that even though the pre-condition does not talk about the length of the list, the post-condition has to expose it because the function output depends on it, and hence the post-condition must connect up the return value to the length of the list, here by an existentially quantified variable.

This approach of deriving abstract specifications can be used in general for working with ESL: for a given algorithm, first prove the least abstract specification, which exposes all details, and then adjust the degree of abstraction to fit the needs of the client code. We discuss this further in the upcoming paragraph on reasoning about client programs.

\subparagraph*{Membership.}
\label{spec:mem} 
Next, we consider the list-membership function $\mathtt{LMem}(\pvar x, \pvar v)$, which takes a list at address~$\pvar x$, does not modify it, and returns $\true$ if $\pvar v$ is in the list, and false otherwise. Given (O2) and the fact that the function output depends on the values in the list, we understand that, for its UX/EX specification, we should be using a list abstraction that exposes at least the values, that is, $\llist{x, \lstvs}$ or $\llist{x, \lstxs, \lstvs}$. The corresponding specifications are:
\[
\begin{array}{c}
\utripleq{\pvar x\doteq x \lstar \pvar v \doteq v \lstar \llist{x, \lstvs} }{~\mathtt{LMem}(\pvar
  x, \pvar v)~}{\llist{x, \lstvs} \lstar \pvar{ret} \doteq (v \in \lstvs)} \\
\utripleq{\pvar x\doteq x \lstar \pvar v \doteq v \lstar \llist{x, \lstxs, \lstvs} }{~\mathtt{LMem}(\pvar
  x, \pvar v)~}{\llist{x, \lstxs, \lstvs} \lstar \pvar{ret} \doteq (v \in \lstvs)}
\end{array}
\]
and are proven similarly to list-length. We can check that a more abstract specification, say:
\[
\triple{\pvar x\doteq x \lstar \pvar v \doteq v \lstar \llist{x} }{~\mathtt{LMem}(\pvar x, \pvar v)~}{\exists b \in \Bool.~\llist{x} \lstar \pvar{ret} \doteq b}
\]
is not UX-valid, by choosing, as the post-model, $b$ to be false and the list at $x$ to contain~$v$. As for list-length, since list-membership does not modify the list, we have to choose the same model of $\llist{x}$ for the pre-model to have a chance of reaching the post-model, but then the algorithm will return $\true$, not $\false$, so this post-model is not reachable.

\subparagraph*{Swap-First-Two.} Next, we consider the list-swap-first-two function $\mathtt{LSwapFirstTwo}(\pvar x)$, which takes a list  at address~$\pvar x$, swaps its first two values if the list is of sufficient length, returning $\nil$, and throws an error otherwise without modifying the list. Given (O2), to specify this function we need an abstraction that captures list length and, apparently, also the list values; for example, $\llist{x, \lstvs}$. As this function can throw errors, its full EX specification has to use the ESL quadruple, in which the two post-conditions are constrained with the corresponding, shaded, branching conditions:
\[
\begin{array}{c}
 {\color{blue} (\pvar x\doteq x \lstar \llist{x, \lstvs})} \\
 {\mathtt{LSwapFirstTwo}(\pvar x, \pvar v)} \\
 {\color{blue} (\oxok: \exists v_1, v_2, \lstvs'.~\shade{\lstvs = v_1 : v_2 : \lstvs'} \lstar \llist{x, v_2 : v_1 : \lstvs'} \lstar \pvar{ret} \doteq \nil)} \\
 {\color{blue} (\oxerr: \llist{x, \lstvs} \lstar \shade{|\lstvs| < 2} \lstar \pvar{err} \doteq \strlit{List~too~short!} )}
\end{array}
\]
observing that the success post-condition, given the used abstraction, has to not only state that the length of the list is not less than two, but also how the values are manipulated, and also that the error message is chosen for illustrative purposes.

However, note that the swapped values {\it are not featured in the function output}, but instead remain contained within the predicate. This indicates that a more abstract specification: 
\[
\begin{array}{c}
 {\color{blue} (\pvar x\doteq x \lstar \llist{x, n})} \\
 {\mathtt{LSwapFirstTwo}(\pvar x, \pvar v)} \\
 {\color{blue} (\oxok: \llist{x, n} \lstar \shade{n \geq 2} \lstar \pvar{ret} \doteq \nil)}~
 {\color{blue} (\oxerr: \llist{x, n} \lstar \shade{n < 2} \lstar \pvar{err} \doteq \strlit{List~too~short!} )}
\end{array}
\]
which only reveals the list length, might be EX-valid, and indeed it is. Any list we choose in the error post-model will have length less than two, and can then be used in the pre-model to reach the post-model. On the other hand, whichever list we choose in the success post-model will have length at least two, that is, its values will be of the form $v_1 : v_2 : \lstvs$ and it will have some addresses, and then we can choose a list with the same addresses and values $v_2 : v_1 : \lstvs$ in the pre-model and we will reach the post-model by executing the function.

\subparagraph*{Pointer-Reverse.} 
\label{spec:rev} 
Let us now examine the list-pointer-reverse function, $\mathtt{LPRev}(\pvar x)$, which takes a list at address~$\pvar x$ and reverses it by reversing the direction of the next-pointers, returning the head of the reversed list. Given (O2) and the fact that the algorithm manipulates pointers and returns an address, but the actual values in the list are not exposed, we will try to use the address-only $\llist{x, \lstxs}$ predicate to specify this function as in the following OX triple, where $\lstxs^\dagger$ denotes the reverse of the mathematical list $\lstxs$:
\[
\triple{\pvar x\doteq x \lstar \llist{x, \lstxs} }{~\mathtt{LPRev}(\pvar x)~}{\llist{\pvar{ret}, \lstxs^\dagger}}
\]
which would seem to be UX-valid given our OX experience and previous examples, but is~not. In particular, it has no information about the logical variable $x$, which exists only in the pre-condition. This is not an issue in OX reasoning, but in UX reasoning it would mean that there exists a logical environment that interprets the post-condition but not the pre-condition, and such a specification, by the definition, could never be UX-valid.

To understand which specific information about $x$ is required, we first add the general $x \in \Val$, making the post-condition $\llist{\pvar{ret}, \lstxs^\dagger} \lstar x \in \Val$, and then try to choose a post-model by picking values for $\pvar{ret}$, $\lstxs$, and $x$. Note that, given the definition of $\llist{x, \lstxs}$, we cannot just pick any non-correlated values for $\pvar{ret}$ and $\lstxs$: in particular, either $\lstxs$ is an empty list and $\pvar{ret}$ is $\nil$, or $\lstxs$ is non-empty and $\pvar{ret}$ is its last element. This observation, in fact, reveals the information needed about $x$: either $x$ is $\nil$ and $\lstxs$ is empty, or $\lstxs$ is non-empty and $x$ is its first element. We capture this information using the $\listptr{x, \lstxs}$ predicate:
\[
\listptr{x, \lstxs} \defeq (\lstxs \doteq \emplist \lstar x \doteq \nil) \lor (\exsts{\lstxs'} \lstxs = x \cons \lstxs')
\]
and arrive at the desired EX specification of the list-pointer-reverse algorithm:
\[
\utripleq{\pvar x\doteq x \lstar \llist{x, \lstxs} }{~\mathtt{LPRev}(\pvar
  x)~}{\llist{\pvar{ret}, \lstxs^\dagger} \lstar \listptr{x, \lstxs}}
\]
Let us make sure that this specification is UX-valid. If we pick a post-model with $\lstxs = \emplist$, then $x = \pvar{ret} = \nil$ and the pre-model with the same $x$ and $\lstxs$ will work, as the list holds no values. For a post-model with non-empty $\lstxs$, $x$ must equal the head of $\lstxs$, $\pvar{ret}$ must equal the tail of $\lstxs$, and we also have to pick some arbitrary values $\lstvs$, with $|\lstvs| = |\lstxs|$. Then, given the described behaviour of the algorithm, we know that this post-model is reachable from a pre-model which has the list at $x$ with addresses $\lstxs$ and values $\lstvs^\dagger$.


\subparagraph*{Free.} 
\label{spec:free} 
Next, we take a look at the $\mathsf{LFree}(\pvar x)$ function, which frees a given list at address~$\pvar x$. Its OX specification is
$
\!\!
\begin{array}{l}
\triple{\pvar x \doteq x \lstar \llist{\pvar x} }{\mathtt{LFree}(\pvar x)}{\pvar{ret} \doteq \nil}\!\!\!\!
\end{array}
$, but it does not transfer to UX contexts because no resource from the pre-condition can be forgotten in the post-condition as that would break the UX frame property~\cite{isl}. Instead, we have to keep track of the addresses to be freed, which we can do using the $\llist{x, \lstxs}$ predicate (or $\llist{x, \lstxs, \lstvs}$), and we also have to explicitly state in the post-condition that these addresses have been freed:
\[
		\utripleq{\pvar x\doteq x \lstar \llist{x, \lstxs} }{~\mathtt{LFree}(\pvar x)~}{\pfreed{\lstxs} \lstar \listptr{x, \lstxs} \lstar \pvar{ret} \doteq \nil}
\]	
using the $\pfreed{\lstxs}$ predicate, which is defined as follows:
\[
	\pfreed{\lstxs} \defeq  (\lstxs\doteq \emplist) \lor (\exsts{x, \lstxs'} \lstxs \doteq x : \lstxs' \lstar  x \mapsto \cfreed, \cfreed \lstar \pfreed{xs'})
\]

\subparagraph*{Client Programs and Specification Composition.} 
\label{spec:client}
We discuss the usability of ESL specifications in general and abstraction in particular in the context of client programs that call multiple library functions. Consider the following (slightly contrived) client program, which takes a list and: pointer-reverses it if its length is between 5 and~10; frees it and then throws an error if its length is smaller than 5; and does not terminate otherwise:
\[
\begin{array}{l}
 \mathtt{LClient}(\pvar x)~\{ \\
 \tab\passign{\pvar{l}}{\mathtt{LLen(\pvar x)}}; \\
 \tab\texttt{if}~( \pvar{l} < 5 ) \\ 
 \tab\tab \{~\passign{\pvar{r}}{\mathtt{LFree(\pvar x)}};~\perror(\strlit{List~too~short!})~\}~\mathtt{else}~ \\ 
 \tab\tab \{~\pifelses{ \pvar{l} > 10}~\pwhiles{\true} ~\pskip~\}~\}~\mathtt{else}~\{~\passign{\pvar{r}}{\mathtt{LPRev(\pvar l)}}~\}~\}; \\
 \tab\preturn{\pvar{r}} \\
 \} \\
\end{array}
\]
Our first goal is to understand which is the most abstract list predicate that could be used for reasoning about this client, since we want to minimise the amount of details we need to carry along in the proof, noting that the least abstract one, $\llist{x, \lstxs, \lstvs}$, will always work. Observe that, importantly, only specifications expressed at the same abstraction level are composable with each other, because they must be composed using equivalence. We explore this in more detail in the subsequent formal discussion (see, in particular, observation (O5)).

When it comes to $\mathtt{LClient}(\pvar x)$, for list-length, we need information about the list length, meaning that we can use either $\llist{x, n}$, $\llist{x, \lstxs}$, or $\llist{x, \lstvs}$, but not $\llist{x}$. For list-free, we must have information about the addresses, meaning that $\llist{x, n}$ and $\llist{x, \lstvs}$ will not work, leaving us with $\llist{x, \lstxs}$, which is also usable for list-pointer-reverse. Therefore, we can write the specification of this client using the $\llist{x, \lstxs}$ predicate, as follows:
\[
\begin{array}{c}
 {\color{blue} (\pvar x\doteq x \lstar  \llist{x, \lstxs})} \\
 {\mathtt{LClient}(\pvar x)} \\
 {\color{blue} (\oxok: 5 \dotle |\lstxs| \dotle 10 \lstar \llist{\pvar{ret}, \lstxs^\dagger} \lstar \listptr{x, \lstxs} )} \\
 {\color{blue} \left(\oxerr: |\lstxs| \dotlt 5 \lstar \pfreed{\lstxs} \lstar \listptr{x, \lstxs} \lstar \pvar{err} \doteq \strlit{List~too~short!}\right)}
\end{array}
\]

In general, however, it is sufficient for a client to call one function that works with addresses and another that works with values for the only applicable predicate to be $\llist{x, \lstxs, \lstvs}$, which is still abstract in the sense that it allows for unbounded reasoning about lists, but does not hide any of its internal information. This leads us to the following observation:
\begin{center}
  {\bfseries(O3)}~specifications that use predicates which hide data-structure information, \\ albeit provable,  may have limited use in UX client reasoning.
\end{center}
As a final remark on abstraction, note that we have only considered predicates that expose the data-structure sub-parts (for lists, these sub-parts are values $\lstvs$ and addresses $\lstxs$) either entirely or not at all. It would be also possible to expose \emph{some} of this structure for some of the algorithms, but because of (O3), specifications using such abstractions are only composable with specifications exposing the same partial structure, and hence likely to be of limited use.

\subparagraph*{Non-termination.} We conclude our discussion on specifications with two remarks on EX reasoning about non-terminating behaviour. First,  consider the non-terminating branch of the \texttt{LClient} function, which is triggered when $|\lstxs| > 10$. Observe that this branch is implicit in the client specification, in that it is subsumed by the success post-condition (since $\models P \lor (|\lstxs| \dotgt 10 \lstar \AssFalse) \Leftrightarrow P$). However, to demonstrate that it exists, we can constrain the pre-condition appropriately to prove the specification
$
\utripleq{\pvar x\doteq x \lstar  \llist{x, \lstxs} \lstar |\lstxs| \dotgt 10 } {\mathtt{LClient}(\pvar x)} \AssFalse
$.
This implicit loss of non-terminating branches can be characterised informally as follows:
\begin{center}
{\bfseries (O4)} if the post-conditions do not cover all paths allowed by the pre-condition, \\ then the ``gap'' is non-terminating.
\end{center}
In this case, the pre-condition implies $|\lstxs| \in \Nat$ and the post-conditions cover the cases where $|\lstxs| \leq 10$, leaving the gap when $|\lstxs| > 10$, for which we provably have client non-termination. 

Second, we observe that, in contrast to terminating behaviour, for non-terminating behaviour EX is as expressive as OX; that is, the EX triple $\utripleq{P}{\cmd}{\AssFalse}$ is equivalent to the OX triple $\triple{P}{\cmd}{\AssFalse}$ as the UX triple $\isltripleonecase{P}{\cmd}{\AssFalse}$ is vacuously true. This is not to say that all non-terminating behaviour can be captured by ESL specifications. For example, as in OX, if the code branches on a value that does not come from the pre-condition, and if one of the resulting branches does not terminate, and if the code can also terminate successfully, then the non-terminating branch will be implicit in the pre-condition, but no gap in the sense of (O4) will be present. This is illustrated by the code and specification below, where the pre- and the post-condition are the same, but a non-terminating path still~exists:
\[
\especlines{\pvar x \doteq 0 }~
\passign{\pvar{x}}{\prandom};~
\pifelses{ \pvar{x} > 42 }~
 \pwhiles{\true} \pskip \}~
 \pifelsem~
\passign{\pvar{x}}{0}~
\pifelsee~
\especlines{\pvar x \doteq 0} 
\]

%
%
%

\subsection{More ESL Proofs: Iterative list-length}
\label{ss:eslproofs}

In \S\ref{sec:overview}, we have shown a proof sketch for a recursive implementation of the list-length algorithm, demonstrating how to handle the measure for recursive function calls; how the folding of predicates works in the presence of equivalence; and how to move between external and internal specifications. We highlight again the UX-specific issue that we raised and that is related to predicate folding, which can be formulated generally as follows:
\begin{center}
{\bfseries (O5)} if the code accesses data-structure information that the used predicate hides, then that predicate might not be foldable in a UX-proof in all of the places in which \\ it would be foldable in the corresponding OX-proof.
\end{center}
Here, we show how to write ESL proofs for looping code, using as example an iterative implementation of the list-length algorithm. Proofs for the majority of the other algorithms mentioned in~\S\ref{ss:uxexabs} can be found in Appendix~\ref{apdx:examples}; the rest are similar.

\begin{figure*}[!t]
\[
\small
\begin{array}{@{}l@{~}l}
\Gamma \vdash & \especlines{\pvar x \doteq x \lstar \llist{x, n}} \\
& \mathtt{LLen}(\pvar x)~\{ \\
& \tab\especlines{\pvar x \doteq x \lstar \llist{x, n}  \lstar \pvar r \doteq \nil } \\
&	\tab\passign{\pvar{r}}{0} \\
& \tab\especlines{\pvar x \doteq x \lstar \llist{x, n}  \lstar \pvar r \doteq 0 } \\
& \tab\especlines{{P_0}   } \\
	& \tab\pwhiles{ \pvar{x} \neq \nil } \\
&	\tab\tab\especlines{P_i \lstar \pvar x \dotneq \nil } \\
&	\tab\tab\especline{\exists j, v, x'.~\llseg{x, \pvar x, i} \lstar \pvar x \mapsto v, x' \lstar  \llist{x', j - 1} \lstar n \doteq i + j \lstar \pvar r \doteq i} \\
&	\tab\tab\passign{\pvar{x}}{[\pvar x + 1]}; \\
&	\tab\tab\especline{\exists j, v, x'.~\llseg{x, x', i} \lstar x' \mapsto v, \pvar x \lstar  \llist{\pvar x, j} \lstar n \doteq i + (j + 1) \lstar \pvar r \doteq i} \\
& \tab\tab\polish{\text{// equivalence: $\models \llseg{x, y, n + 1} \Leftrightarrow \exists x', v.~\llseg{x, x', n} \lstar x' \mapsto v, y$]]}} \\
&	\tab\tab\especline{\exists j.~\llseg{x, \pvar x, i + 1} \lstar \llist{\pvar x, j} \lstar n \doteq (i + 1) + j \lstar \pvar r \doteq i} \\
	& \tab\tab\passign{\pvar{r}}{\pvar r + 1} \\
	& \tab\tab\especlines{P_{i+1} } \\
	& \tab\} 
	\\
	&	\tab\especline{\pvar x \doteq \nil \lstar \exists i.~P_i} \\
	&	\tab\especline{\exists i, j.~\llseg{x, \pvar x, i} \lstar \llist{\pvar x, j} \lstar 
		                n \doteq i + j \lstar \pvar r \doteq i \lstar \pvar x \doteq \nil} \\
	&	\tab\especline{\llseg{x, \nil, n} \lstar \pvar r \doteq n \lstar \pvar x \doteq \nil}~\polish{\text{//~equivalence: $\models \llist{\nil, j} \Leftrightarrow j \doteq 0$}} \\
	&	\tab\especline{\llist{x, n} \lstar \pvar{r} \doteq n \lstar \pvar x \doteq \nil}~\polish{\text{//~equivalence: $\models \llseg{x, \nil, n} \Leftrightarrow \llist{x, n}$}} \\
& \tab\preturn{\pvar r} \\
& \tab\especlines{\llist{x, n} \lstar \pvar{r} \doteq n \lstar \pvar x \doteq \nil \lstar \pvar{ret} = \pvar r} \\
& \tab\especlines{\exists x_q, r_q.~\llist{x, n} \lstar r_q \doteq n \lstar x_q \doteq \nil \lstar \pvar{ret} = r_q} \\
& \tab \especlines{\llist{x, n} \lstar \pvar{ret} \doteq n} \\
& \} \\
& \especlines{\llist{x, n} \lstar \pvar{ret} \doteq n}

\end{array}
\]
\caption{ESL proof sketch: iterative list-length.}
\label{ps:llen}
\end{figure*}


\subparagraph*{Iterative list-length in ESL: Proof Sketch.}
In Figure~\ref{ps:llen}, we give an iterative implementation of the list-length algorithm and show that it satisfies the same ESL specification as its recursive counterpart,
$
\utripleok{\pvar x\doteq x \lstar \llist{x,n} }{~\mathtt{LLen}(\pvar x)~}{\llist{x,n} \lstar\pvar{ret}\doteq n}
$. Since there is no recursion, we elide the (trivial) measure. To state the loop variant, we use the list-segment predicate, defined as follows:
\[
	\llseg{x, y, n} \defeq (x\doteq y \lstar  n\doteq0) \lor (\exsts{v,x'} x\mapsto v,x' \lstar \llseg{x',y, n - 1} )
\]
and to apply the \prooflab{while} rule, we define:
\[
P_i \defeq \exists j.~\llseg{x, \pvar x, i} \lstar \llist{\pvar x, j} \lstar n \doteq i + j \lstar \pvar r = i
\]
Note that we could have chosen to elide $i$ from the body of $P_i$ in this simple example, but since this is not necessarily possible or evident in general as well as for instructive purposes, we chose to keep it in the proof. Note how, on exiting the loop, the negation of the loop condition collapses the existentials~$i$ and $j$. This allows us to obtain the given internal post-condition, from which we then easily move to the desired external post-condition. For this proof, we also use three equivalence lemmas, which state that a non-empty list segment can be separated into its last element and the rest, that the length of an empty list equals zero, and that a null-terminated list-segment is~a~list.

\subsection{Beyond List Examples: Binary Search Trees}
\label{ss:beyond}

\newcommand{\bstroot}{BSTRoot}

While list algorithms illustrate many aspects of exact reasoning, it is also important to understand how ESL specification and verification works with other data structures. For this reason, we discuss two algorithms operating over binary search trees (BSTs) that are intended to represent sets of natural numbers. We use two abstractions for BSTs, one in which only their values are considered as a mathematical set:
\[
\begin{array}{@{}r@{~}l}
\bst{x,K} \defeq~& (x\doteq\nil\lstar K \doteq\emptyset) \lor (\exsts{k,l,r,K_l,K_r} x\mapsto k,l,r \lstar \bst{r,K_r} \lstar \bst{l,K_l} \lstar \\ & \qquad K\doteq K_l\uplus\{k\}\uplus K_r \lstar K_l<k \lstar k<K_r)
\end{array}
\]
and another that fully exposes the BST structure:
\[
\begin{array}{@{}r@{~}l}
\bst{x, \tree} \defeq~& (x\doteq\nil\lstar \tree \doteq \emptytree) \lor (\exsts{k, l, r, \tree_l, \tree_r} E \mapsto k, l, r \lstar \bst{r, \tree_r} \lstar \bst{l, \tree_l} \lstar \\ & \qquad \tree \doteq ((x, k), \tree_l, \tree_r) \lstar \tree_l < k \lstar k < \tree_r)
\end{array}
\]
where $\tree$ is a mathematical tree, that is, an algebraic data type with two constructors representing, respectively, an empty tree and a root node with two child trees: $
\tree \in \mathsf{Tree} \defeq \emptytree~|~((x, n), \tree_l, \tree_r)
$,
where the notation $(x, n)$ represents a BST node with address $x$ and value~$n$. Note the overloaded $<$ notation, where one of the operands can be a set or a tree, which carry the intuitive meaning.

\myparagraph{BST algorithms} We first consider the BST-find-minimum algorithm, $\mathtt{BSTFindMin}(\pvar x)$, which takes a tree with root at $\pvar x$, does not modify it, and returns its minimum element or throws an empty-tree error. Since that algorithm operates only on the values in the tree, we are able to state its ESL specification using the $\bst{x,K}$ predicate as follows:
\[
\begin{array}{c}
 {\color{blue} (\pvar x \doteq x \lstar \bst{x,K})}~ 
 {\mathtt{BSTFindMin}(\pvar x)} 
 \begin{array}{l}
 {\color{blue} (\mathit{ok}{:}~x\neq\nil\lstar\bst{x,K}\lstar\pvar{ret}\doteq\min(K))} \\
 {\color{blue} (\mathit{err}{:}~x = \nil \lstar \bst{x,K} \lstar \pvar{err} \doteq \strlit{Empty~tree!})}
 \end{array}
\end{array}
\]
We have also considered the BST-insert algorithm, $\mathtt{BSTInsert}(\pvar x, \pvar v)$, which takes a tree with root at $\pvar x$ and inserts a new node with value $\pvar v$ into it as a leaf if $\pvar v$ is not already in the tree, or leaves the tree unmodified if it is. As this algorithm interacts both with values and addresses in the tree, the appropriate abstraction for it is $\bst{x, \tree}$, and its ESL specification is:
\[
\begin{array}{c}
 {\color{blue} (\pvar x \doteq x \lstar \pvar v \doteq v \lstar \bst{x, \tree})} \\
 {\mathtt{BSTInsert}(\pvar x, \pvar v)} \\
 {\color{blue} (\exists x'.~\bst{\pvar{ret}, \predd{BSTInsert}{\tree, (x', v)}} \lstar \predd{\bstroot}{x, \tree})}
\end{array}
\]
where $\predd{BSTInsert}{\tree, \nu}$ is the mathematical algorithm that inserts the node $\nu$ into the tree $\tree$:
\[
\begin{array}{l@{\qquad}l}
\predd{BSTInsert}{\emptytree, (x', v)} \defeq   
& \predd{BSTInsert}{((x, k), \tree_l, \tree_r), (x', v)} \defeq \\
\qquad ((x', v), \emptytree, \emptytree) & \qquad \mathsf{if}\ v < k\ \mathsf{then}~((x, k), \predd{BSTInsert}{\tree_l, (x', v)}, \tree_r) \\
& \qquad \mathsf{else\ if}\ k < v\ \mathsf{then}~((x, k), \tree_l, \predd{BSTInsert}{\tree_r, (x', v)}) \\
& \qquad \mathsf{else} ((x, k), \tree_l, \tree_r)
\end{array}
\]
and the predicate $\predd{\bstroot}{x, \tree}$ is defined analogously to $\listptr{x, \lstxs}$:
\[
\predd{\bstroot}{x, \tree} \defeq (\tree \doteq \emptytree \lstar x \doteq \nil) \lor (\exsts{k, \tree_l, \tree_r} \tree \doteq ((x, k), \tree_l, \tree_r))
\]

This example shows how in EX verification, just as in OX verification, we end up relating an imperative heap-manipulating algorithm to its mathematical/functional counterpart (cf.~Appel~\cite{Appel22} for a recent reiteration of this idea). The additional work required is that EX mathematical models must be more detailed: we are, yet again, not allowed to lose information. In particular, in OX verification we could relate $\mathtt{BSTInsert}(\pvar x, \pvar v)$ to mathematical sets, but in EX verification we must relate our imperative implementation to tree models, including both values and pointers. Moreover, our mathematical model of the algorithm, $\predd{BSTInsert}{\tree, (x', v)}$, must insert elements in the same way as the imperative implementation, that is, in this case at the leaves of the tree. The proofs for both algorithms are given in Appendix~\ref{apdx:examples}.

%

%% file: sections/related-work.tex
\section{Related Work}

In the previous sections, we have placed ESL in the context of related work on OX and UX logics and associated tools. Here, we discuss formalisms capable of reasoning both about program correctness and program incorrectness, as well as existing approaches to the use of function specifications (summaries) and abstraction in symbolic execution.

\subparagraph*{Program Logics for Both Correctness and Incorrectness.} Developed in parallel but independently of ESL, Outcome Logic (OL)~\cite{Zilberstein23}, much like ESL, brings together reasoning about correctness and incorrectness into one logic. Both OL and ESL rely on the traditional meaning of correctness, but OL introduces a new approach to incorrectness, based on reachability of sets of states. It has not yet been shown that this approach has the same bug-finding potential as that of ISL: in particular, bi-abduction has not yet been demonstrated to be compatible with~OL. In addition, the OL work, in contrast to ESL, does not discuss function compositionality or the interaction between abstraction, reachability, and incorrectness.


LCL$_A$~\cite{lics2021dist,Bruni23} is a non-function-composi\-tional, first-or\-der logic that combines UX and OX reasoning using abstract interpretation. It is parametric on an abstract domain~$A$, and proves UX triples of the form $\vdash_A [P]~\cmd~[Q]$ where, under certain conditions, the triple also guarantees verification. These conditions, however, normally mean that only a limited number of pre-conditions can be handled. The conditions also have to be checked per-command and if they fail to hold (due to, e.g., issues with Boolean guards, which are known to be a major source of incompleteness), then the abstract domain has to be incrementally adjusted; the complexity of this adjustment and the expressivity of the resulting formalism is unclear.

\subparagraph*{Compositional Symbolic Execution.} There exists a substantial body of work on symbolic execution with function summaries~(e.g.~\cite{godfsmart,demandsym,smash,cse1,cse3,cse4}), which is primarily based on first-order logic.
We highlight the work of Godefroid et al., which initially used exact summaries of bounded program behaviour to drive the compositional dynamic test generation of SMART~\cite{godfsmart}, and later distinguished between may (OX) and must (UX) summaries, leveraging the interaction between them to design the SMASH tool for compositional property checking and test generation~\cite{smash}. {SMASH, however, is limited in its ability to reason about heap-manipulating programs because, for example, it lacks support for pointer arithmetic. Nevertheless, it shows that interactions between OX and UX summaries can be exploited for automation, which is an important consideration for any automation of ESL. For example, SMASH is able to use not-may summaries (which amount to non-reachability) when constructing must-summaries (which amount to reachability), using the former to restrict the latter.}
When it comes to abstraction, for example, Anand et al.~\cite{cse2} implement linked-list and array abstractions for true bug-finding in non-compositional symbolic execution, in the context of the Java PathFinder, and use it to find bugs in list and array partitioning algorithms. True bug-finding is maintained by checking for state subsumption, which requires code modification rather than annotation and a record of all previously visited states.

%% file: sections/conclusion.tex
\section{Conclusions}

We have introduced ESL, a program logic for exact reasoning about
heap-manipulating programs.  ESL specifications provide a sweet spot
between verification and true bug-finding: as SL specifications, they capture all terminating behaviour, and, as
ISL specifications, they describe only results and errors that are
reachable. ESL specifications are therefore compatible with
tools that target OX
verification, such as VeriFast~\cite{verifast} and
Iris~\cite{jung:popl:2015},
tools that target UX true
bug-finding, such as Pulse~\cite{il,isl}, and tools capable of
targeting both, such as Gillian \cite{gillianpldi,gilliancav}. 
ESL supports reasoning about 
mutually recursive functions and comes with a soundness result that 
immediately transfers to SL and ISL, thus demonstrating, for the 
first time, scalable functional compositionality for UX~logics.

We have verified exact specifications for a number of illustrative examples, showing that ESL can  reason about data-structure libraries, language errors, mutual recursion, and non-termination. In doing so, we emphasise the distinction between the often-conflated concepts of abstraction and over-approximation. We have demonstrated that abstract predicates can be soundly used in EX and UX reasoning, albeit not as freely as in OX reasoning.



We believe that ESL reasoning, in its intended context of
semi-automatic verification of functional correctness properties, 
is  useful for the verification of self-contained, critical code that underpins a larger
codebase.
To demonstrate this, we will in the future incorporate UX
and EX verification inside
Gillian~\cite{gillianpldi,gilliancav},
which already has support for function compositionality and
semi-automatic predicate management as part of its OX
verification.

%% file: sections/app-semantics.tex

%
%
%
%

\section{Complete Operational Semantics}
\label{apdx:cos}
{\small
\begin{mathpar}
       \infer{
       \stt, \pskip \baction_{\fictx} \stt
       }{ \ } \and
       \infer{
          \sthreadp{ \sto }{ \hp }, \passign{\pvar{x}}{\pexp} \baction_{\fictx} \sthreadp{\sto [\pvar{x} \storearrow v] }{ \hp }
       }{
       \esem{\pexp}{\sto} = v
       } \and
       \infer{
          \sthreadp{ \sto }{ \hp }, \passign{\pvar x}{\pexp} \baction_{\fictx} {\oxerr} : \sthreadp{ \sto_{\oxerr} }{ \hp }
       }{
       \begin{array}{c}
       \esem{\pexp}{\sto} = \undefd \quad \verr = [``\mathsf{ExprEval}", \stringify {\pexp}]
       \end{array}
       } \and
       \infer{
          \sthreadp{ \sto }{ \hp }, \passign{\pvar{x}}{\prandom} \baction_{\fictx} \sthreadp{\sto [\pvar{x} \storearrow n] }{ \hp }
       }{
       n \in \Nat
       } \and
       \infer{
       \sthreadp{ \sto }{ \hp }, \perror(\pexp) \baction_\fictx \oxerr :  \sthreadp{ \sto_{\oxerr} }{ \hp }
       }{\esem{\pexp}{\sto} = \gv \quad \verr = [``\mathsf{Error}", \gv]}
       \and
       \infer{
          \sthreadp{ \sto }{ \hp }, \perror(\pexp) \baction_{\fictx} {\oxerr} : \sthreadp{ \sto_{\oxerr} }{ \hp }
       }{
       \begin{array}{c}
       \esem{\pexp}{\sto} = \undefd \quad \verr = [``\mathsf{ExprEval}", \stringify {\pexp}]
       \end{array}
       } \and
       \infer{
          \stt, \pifelse{\pexp}{\scmd_1}{\scmd_2} \baction_{\fictx} \outcome: \stt'
       }{
       \esem{ \pexp }{ \stt } = \true \quad  \stt, \scmd_1 \baction_{\fictx} \outcome: \stt'
       }
       \and
       \infer{
          \stt, \pifelse{\pexp}{\scmd_1}{\scmd_2} \baction_{\fictx} \outcome: \stt'
       }{
          \esem{ \pexp }{ \stt } = \false \quad  \stt, \scmd_2 \baction_{\fictx} \outcome: \stt'
       } \and
       \infer{
          \sthreadp{ \sto }{ \hp }, \pifelse{\pexp}{\scmd_1}{\scmd_2} \baction_{\fictx} {\oxerr} : \sthreadp{ \sto_{\oxerr} }{ \hp }
       }{
       \begin{array}{c}
       \esem{\pexp}{\sto} = \undefd \quad \verr = [``\mathsf{ExprEval}", \stringify {\pexp}]
       \end{array}
       }
       \and
       \infer{
          \sthreadp{ \sto }{ \hp },  \pifelse{\pexp}{\scmd_1}{\scmd_2}
          \baction_{\fictx} {\oxerr} : \sthreadp{ \sto_{\oxerr} }{ \hp }
       }{
              \begin{array}{c}
       \esem{\pexp}{\sto} = v \notin \Bool  \\ \verr = [``\mathsf{Type}", \stringify {\pexp}, v, ``\mathsf{Bool}"]
       \end{array}
       } \and
       \infer{
          \stt, \pwhile{\pexp}{\scmd} \baction_{\fictx} \stt
       }{
          \esem{ \pexp }{ \stt } = \false
       } \and
       \infer{
          \stt, \pwhile{\pexp}{\scmd} \baction_{\fictx} \outcome: \stt'
       }{
          \begin{array}{c}
          \esem{ \pexp }{ \stt } = \true \quad \stt, \scmd \baction_{\fictx} \stt'' \\
          \stt'', \pwhile{\pexp}{\scmd} \baction_{\fictx} \outcome: \stt'
          \end{array}
       } \and
       \infer{
          \sthreadp{ \sto }{ \hp }, \pwhile{\pexp}{\scmd} \baction_{\fictx} {\oxerr} : \sthreadp{ \sto_{\oxerr} }{ \hp }
       }{
       \begin{array}{c}
       \esem{\pexp}{\sto} = \undefd \\ \verr = [``\mathsf{ExprEval}", \stringify {\pexp}]
       \end{array}
       }
       \and
       \infer{
          \sthreadp{ \sto }{ \hp },  \pwhile{\pexp}{\scmd}
          \baction_{\fictx} {\oxerr} : \sthreadp{ \sto_{\oxerr} }{ \hp }
       }{
              \begin{array}{c}
       \esem{\pexp}{\sto} = v \notin \mathbb{N}  \\ \verr = [``\mathsf{Type}", \stringify {\pexp}, v, ``\mathsf{Bool}"]
       \end{array}
       } \and
\infer{
          \stt, \pwhile{\pexp}{\scmd} \baction_{\fictx} \outcome: \stt'
       }{
          \begin{array}{c}
          \esem{ \pexp }{ \sto } = \true \quad \stt, \scmd \baction_{\fictx} \outcome: \stt' \quad \outcome \neq \oxok
          \end{array}
       } \and
       \infer{
    \stt, \scmd_1 ; \scmd_2 \baction_{\fictx} \stt'
       }{
       \begin{array}{c}
       \stt, \scmd_1 \baction_{\fictx} \stt'' \quad  \stt'', \scmd_2 \baction_{\fictx} \outcome:  \stt' \end{array}
       } \and
       \infer{
    \stt, \scmd_1 ; \scmd_2 \baction_{\fictx} \outcome: \stt'
       }{
       \begin{array}{c}
       \stt, \scmd_1 \baction_{\fictx} \outcome:   \stt' \quad  \outcome \neq \oxok \end{array}
       } \and
       \infer{
     \sthreadp{ \sto }{ \hp },
     \pfuncall{\pvar{y}}{\procname}{\vec{\pexp}} \baction_{\fictx}
     \sthreadp{\sto [\pvar{y} \storearrow v' ] }{ \hp' }
  }{
    \begin{array}{c}
    \pfunction{\procname}{\vec{\pvar{x}}}{\scmd; \preturn{\pexp'}} \in \scontext
    \\
    \esem{\vec{\pexp}}{\sto} = \vec{v}
     \quad \pv{\scmd} \setminus
      \{\vec{\pvar{x}}\} = \{\vec{\pvar{z}} \}
    \\
 \sto_p  =  \emptyset [ \vec{\pvar{x}} \storearrow \vec{v}] [ \vec{\pvar{z}} \storearrow \nil]
     \\ (\sto_p, \hp), \scmd
      \baction_{\fictx} \sthreadp{ \sto_q }{ \hp' }  \quad  \esem{\pexp'}{\sto_q} =v'
    \end{array}
  }
  \and
  \infer{
     \sthreadp{ \sto }{ \hp },
     \pfuncall{\pvar{y}}{\procname}{\vec{\pexp}} \baction_{\fictx}
     \oxerr: \sthreadp{\sto_{\oxerr} }{ \hp' }
  }{
    \begin{array}{c}
    \pfunction{\procname}{\vec{\pvar{x}}}{\scmd; \preturn{\pexp'}} \in \scontext
    \\
    \esem{\vec{\pexp}}{\sto} = \vec{v}
     \quad \pv{\scmd} \setminus
      \{\vec{\pvar{x}}\} = \{\vec{\pvar{z}} \}
    \\
 \sto_p  =  \emptyset [ \vec{\pvar{x}} \storearrow \vec{v}] [ \vec{\pvar{z}} \storearrow \nil]
     \\ (\sto_p, \hp), \scmd
      \baction_{\fictx} \sthreadp{ \sto_q }{ \hp' }  \quad
      \esem{\pexp'}{\sto_q} = \undefd \\ \verr = [``\mathsf{ExprEval}", \stringify {\pexp'}]
    \end{array}
  } \and
  \infer{
     \sthreadp{ \sto }{ \hp },
     \pfuncall{\pvar{y}}{\procname}{\vec{\pexp}} \baction_{\fictx}
     \oxerr: \sthreadp{\sto_{\oxerr} }{ \hp }
  }{
    \begin{array}{c}
    \pfunction{\procname}{\vec{\pvar{x}}}{\scmd; \preturn{\pexp'}} \in \scontext
    \\
    |\vec{\pvar x}| \neq |\vec{\pexp}| \\
    \verr = [``\mathsf{ParamCount}", f]
    \end{array}
  } \and
  \infer{
     \sthreadp{ \sto }{ \hp },
     \pfuncall{\pvar{y}}{\procname}{\vec{\pexp}} \baction_{\fictx}
     \outcome: \sthreadp{\sto [\pvar{err} \storearrow \sto_q(\pvar{err}) ] }{ \hp }
  }{
    \begin{array}{c}
    \pfunction{\procname}{\vec{\pvar{x}}}{\scmd; \preturn{\pexp'}} \in \scontext
    \\
    \esem{\vec{\pexp}}{\sto} = \vec{v}
     \quad \pv{\scmd} \setminus
      \{\vec{\pvar{x}}\} = \{\vec{\pvar{z}} \}
    \\
 \sto_p  =  \emptyset [ \vec{\pvar{x}} \storearrow \vec{v}] [ \vec{\pvar{z}} \storearrow \nil]
     \\ (\sto_p, \hp), \scmd
      \baction_{\fictx} \outcome: \sthreadp{ \sto_q }{ \hp' }  \quad
      \outcome \neq \oxok
    \end{array}
  } \and
  \infer{
     \sthreadp{ \sto }{ \hp },
     \pfuncall{\pvar{y}}{\procname}{\pexp_1, \ldots \pexp_n} \baction_{\fictx}
     \oxerr: \sthreadp{\sto_{\oxerr} }{ \hp }
  }{
    \begin{array}{c}
    \pfunction{\procname}{\vec{\pvar{x}}}{\scmd; \preturn{\pexp'}} \in \scontext
    \\
    k \in \{ 1, \ldots n \} \quad (\esem{\pexp_i}{\sto} = v_i)|_{i=1}^{k-1}  \quad \esem{\pexp_k}{\sto} = \undefd  \\
    |\vec{\pvar{x}}| = n \quad \verr = [``\mathsf{ExprEval}", \stringify {\pexp_k}]
    \end{array}
  } \and
        \infer{
     \sthreadp{ \sto }{ \hp }, \pfuncall{\pvar{x}}{\procname}{\vec{\pexp}} \baction_{\fictx} {\oxerr} :  (\sto_{\oxerr}, \hp)
  }{
    \fid \notin \dom(\scontext) \quad
    \verr = [``\mathsf{NoFunc}", \fid]
  } \and
       \infer{
          \sthreadp{ \sto }{ \hp }, \pderef{\pvar{x}}{\pexp} \baction_{\fictx} \sthreadp{\sto [\pvar{x} \storearrow v] }{ \hp }
       }{
       \esem{\pexp}{\sto} = n \quad \hp( n) = v
       }  \and
       \infer{
          \sthreadp{ \sto }{ \hp }, \pderef{\pvar x}{\pexp} \baction_{\fictx} {\oxerr} : \sthreadp{ \sto_{\oxerr} }{ \hp }
       }{
       \begin{array}{c}
       \esem{\pexp}{\sto} = \undefd \\ \verr = [``\mathsf{ExprEval}", \stringify {\pexp}]
       \end{array}
       }
       \and
       \infer{
          \sthreadp{ \sto }{ \hp }, \pderef{\pvar x}{\pexp}
          \baction_{\fictx} {\oxerr} : \sthreadp{ \sto_{\oxerr} }{ \hp }
       }{
              \begin{array}{c}
       \esem{\pexp}{\sto} = v \notin \mathbb{N}  \\ \verr = [``\mathsf{Type}", \stringify {\pexp}, v, ``\mathsf{Nat}"]
       \end{array}
       }
\end{mathpar}
\begin{mathpar}
        \infer{
          \sthreadp{ \sto }{ \hp }, \pderef{\pvar x}{\pexp}
          \baction_{\fictx} {\oxm} : \sthreadp{ \sto_{\oxerr} }{ \hp }
       }{
       \begin{array}{c}\esem{\pexp}{\sto} = n  \notin \dom (\hp) \\ \verr = [``\mathsf{MissingCell}", \stringify{\pexp}, n]\end{array}
       }
       \and
       \infer{
          \sthreadp{ \sto }{ \hp }, \pderef{\pvar x}{\pexp}
          \baction_{\fictx} {\oxerr} : \sthreadp{ \sto_{\oxerr} }{ \hp }
       }{
       \begin{array}{c}\esem{\pexp}{\sto} = n  \quad \hp(n) = \cfreed \\ \verr = [``\mathsf{UseAfterFree}", \stringify{\pexp}, n]\end{array}
       }
       \and
       \infer{
          \sthreadp{ \sto }{ \hp }, \pmutate{\pexp_1}{\pexp_2}
          \baction_{\fictx} \sthreadp{ \sto }{ \hp [n \mapsto v] }
       }{
       \esem{\pexp_1}{\sto} = n \quad \hp(n) \in \Val \quad \esem{\pexp_2}{\sto} = v}
\and
\infer{
          \sthreadp{ \sto }{ \hp }, \pmutate{\pexp_1}{\pexp_2} \baction_{\fictx} {\oxerr} : \sthreadp{ \sto_{\oxerr} }{ \hp }
       }{
       \begin{array}{c}
       \esem{\pexp_1}{\sto} = \undefd \\ \verr = [``\mathsf{ExprEval}", \stringify {\pexp_1}]
       \end{array}
       }
       \and
       \infer{
          \sthreadp{ \sto }{ \hp }, \pmutate{\pexp_1}{\pexp_2}
          \baction_{\fictx} {\oxerr} : \sthreadp{ \sto_{\oxerr} }{ \hp }
       }{
              \begin{array}{c}
       \esem{\pexp_1}{\sto} = v \notin \mathbb{N}  \\ \verr = [``\mathsf{Type}", \stringify {\pexp_1}, v, ``\mathsf{Nat}"]
       \end{array}
       }
       \and
        \infer{
          \sthreadp{ \sto }{ \hp }, \pmutate{\pexp_1}{\pexp_2}
          \baction_{\fictx} {\oxm} : \sthreadp{ \sto_{\oxerr} }{ \hp }
       }{
       \begin{array}{c}\esem{\pexp_1}{\sto} = n  \notin \dom (\hp) \\ \verr = [``\mathsf{MissingCell}", \stringify{\pexp_1}, n]\end{array}
       }
       \and
       \infer{
          \sthreadp{ \sto }{ \hp }, \pmutate{\pexp_1}{\pexp_2}
          \baction_{\fictx} {\oxerr} : \sthreadp{ \sto_{\oxerr} }{ \hp }
       }{
       \begin{array}{c}\esem{\pexp_1}{\sto} = n  \quad \hp(n) = \cfreed \\ \verr = [``\mathsf{UseAfterFree}", \stringify{\pexp_1}, n]\end{array}
       } \and
       \infer{
          \sthreadp{ \sto }{ \hp }, \pmutate{\pexp_1}{\pexp_2}
          \baction_{\fictx} {\oxerr} : \sthreadp{ \sto_{\oxerr} }{ \hp }  }{
       \begin{array}{c}
       \esem{\pexp_1}{\sto} = n \quad \hp(n) \in \Val \quad \esem{\pexp_2}{\sto} = \undefd \\ \verr = [``\mathsf{ExprEval}", \stringify {\pexp_2}]
       \end{array}
       } \and
       \infer{
          \sthreadp{ \sto }{ \hp }, \palloc{\pvar{x}}{\pexp} \baction_{\fictx} \sthreadp{\sto [\pvar{x} \storearrow n'] }{ \hp' }
       }{
       \begin{array}{c}
       		\esem{\pexp}{\sto} = n \qquad (n' + i \notin \dom (\hp))|_{0 \le i < n} \\
       		\hp' = \hp[n' \mapsto \nil]\cdots[n' + n - 1 \mapsto \nil]
       \end{array}
       } \and
       \infer{
          \sthreadp{ \sto }{ \hp }, \palloc{\pvar x}{\pexp} \baction_{\fictx} {\oxerr} : \sthreadp{ \sto_{\oxerr} }{ \hp }
       }{
       \begin{array}{c}
       \esem{\pexp}{\sto} = \undefd \\ \verr = [``\mathsf{ExprEval}", \stringify {\pexp}]
       \end{array}
       }
       \and
       \infer{
          \sthreadp{ \sto }{ \hp }, \palloc{\pvar x}{\pexp}
          \baction_{\fictx} {\oxerr} : \sthreadp{ \sto_{\oxerr} }{ \hp }
       }{
              \begin{array}{c}
       \esem{\pexp}{\sto} = v \notin \mathbb{N}  \\ \verr = [``\mathsf{Type}", \stringify {\pexp}, v, ``\mathsf{Nat}"]
       \end{array}
       } \and
       \infer{
          \sthreadp{ \sto }{ \hp }, \pdealloc{\pexp} \baction_{\fictx} \sthreadp{ \sto }{ \hp[n \mapsto \cfreed]}
       }{
       \esem{\pexp}{\sto} = n \quad  \hp(n) \in \Val
       } \and
       \infer{
          \sthreadp{ \sto }{ \hp }, \pdealloc{\pexp} \baction_{\fictx} {\oxerr} : \sthreadp{ \sto_{\oxerr} }{ \hp }
       }{
       \begin{array}{c}
       \esem{\pexp}{\sto} = \undefd \\ \verr = [``\mathsf{ExprEval}", \stringify {\pexp}]
       \end{array}
       }
       \and
       \infer{
          \sthreadp{ \sto }{ \hp }, \pdealloc{\pexp}
          \baction_{\fictx} {\oxerr} : \sthreadp{ \sto_{\oxerr} }{ \hp }
       }{
              \begin{array}{c}
       \esem{\pexp}{\sto} = v \notin \mathbb{N}  \\ \verr = [``\mathsf{Type}", \stringify {\pexp}, v, ``\mathsf{Nat}"]
       \end{array}
       } \and
        \infer{
          \sthreadp{ \sto }{ \hp }, \pdealloc{\pexp}
          \baction_{\fictx} {\oxm} : \sthreadp{ \sto_{\oxerr} }{ \hp }
       }{
       \begin{array}{c}\esem{\pexp}{\sto} = n  \notin \dom (\hp) \\ \verr = [``\mathsf{MissingCell}", \stringify{\pexp}, n]\end{array}
       }
       \and
       \infer{
          \sthreadp{ \sto }{ \hp }, \pdealloc{\pexp}
          \baction_{\fictx} {\oxerr} : \sthreadp{ \sto_{\oxerr} }{ \hp }
       }{
       \begin{array}{c}\esem{\pexp}{\sto} = n  \quad \hp(n) = \cfreed \\ \verr = [``\mathsf{UseAfterFree}", \stringify{\pexp}, n]\end{array}
       }
    \end{mathpar}}

\vspace{1em}
\noindent
where $\sto_{\oxerr} \defeq \serr$.

%% file: sections/app-proofrules.tex

%
%
%
%

\newpage
\section{Exact Separation Logic}
\label{apdx:esl}

The complete rules of ESL are as follows (with $\Qerr^* \defeq \pvar{pre} \lstar \pvar{err} \doteq \eerr$):

{\small
\begin{mathpar}
  \inferrule[\mbox{skip}]
  {}
    { \fsctx \vdash \utripleq{\emp}{\texttt{skip}}{\emp} }
  \and
    \inferrule[\mbox{nondet}]
    {\pvar x \notin \pv{\pexp'}}
    { \fsctx \vdash \utripleq{\pvar x \doteq \pexp'}{\passign{\pvar{x}}{\prandom}}{\pexp' \dotint \Val \lstar \pvar x \dotint {\Nat}} }
  \and
  \inferrule[\mbox{assign}]
  {\pvar x \notin \pv{\pexp'} \qquad \subst \defeq [\pexp'/\pvar x]}
    { \fsctx \vdash \utripleq{{\pvar{x} \doteq \pexp' \lstar \dotin{\pexp}{\Val}}}{\passign{\pvar{x}}{\expr{\pexp}}}{{\dotin{\pexp'}{\Val} \lstar  \pvar{x} \doteq \pexp\subst}} }
  \and
    \inferrule[\mbox{assign-err}]{
  \eerr \defeq ["\mathsf{ExprEval}", \stringify{\pexp}]}
    { \fsctx \vdash \utripleerr{{\pvar{x} \doteq \pexp' \lstar \pexp\notdotint\Val}}{\passign{\pvar{x}}{\expr{\pexp}}}{{\Qerr^*}} }
  \and
  \inferrule[\mbox{lookup}]
      { \pvar x \notin \pv{\pexp'} \qquad \subst \defeq [\pexp'/ \pvar{x}]}
    { \fsctx \vdash \utripleq{\pvar{x} \doteq \pexp'  \lstar \pexp \mapsto \pexp_1}{\pderef{\pvar{x}}{\pexp}}{\dotin{\pexp'}{\Val} \lstar \pvar{x} \doteq \pexp_1\subst \lstar \pexp\subst \mapsto   \pexp_1 \subst} }
  \and
      \inferrule[\mbox{lookup-err-val}]
  {\eerr \defeq [``\mathsf{ExprEval}", \stringify {\pexp}]}
  { \fsctx \vdash \utripleerr{\pvar x \doteq \pexp' \lstar \pexp \notdotint \Val}{\pderef{\pvar{x}}{\pexp}}{\Qerr^*} }
  \and
    \inferrule[\mbox{lookup-err-type}]
      { \eerr\defeq["\mathsf{Type}",\stringify{\pexp},\pexp,"\mathsf{Nat}"]}
    { \fsctx \vdash \utripleerr{\pvar{x} \doteq \pexp'   \lstar \pexp \dotint \Val \lstar \pexp \notdotint \Nat}{\pderef{\pvar{x}}{\pexp}}{\Qerr^*} } \qquad
          \inferrule[\mbox{lookup-err-use-after-free}]
  {\eerr \defeq ["\mathsf{UseAfterFree}", \stringify {\pexp}, \pexp]}
  { \fsctx \vdash \utripleerr{\pvar{x} \doteq \expr{\pexp'} \lstar \expr{\pexp}\mapsto\cfreed}{\pderef{\pvar x}{\expr{\pexp}}}{\Qerr^*} }
  \and
  \inferrule[\mbox{mutate}]
      {}
    { \fsctx \vdash \utripleq{\expr{\pexp_1} \mapsto \pexp \lstar \dotin{ \pexp_2}{\Val}}{\pmutate{\expr{\pexp_1}}{\expr{\pexp_2}}}{\expr{\pexp_1} \mapsto  \expr{\pexp_2} \lstar \pexp\dotint\Val} }
  \and
      \inferrule[\mbox{mutate-err-val-1}]
  {\eerr \defeq ["\mathsf{ExprEval}", \stringify {\pexp_1}]}
  { \fsctx \vdash \utripleerr{\expr{\pexp_1} \notdotint\Val}{\pmutate{\expr{\pexp_1}}{\expr{\pexp_2}}}{\Qerr^*} }
  \and
    \inferrule[\mbox{mutate-err-type}]
  {\eerr \defeq ["\mathsf{Type}", \stringify {\pexp_1}, \expr{\pexp_1}, "\mathsf{Nat}"]}
  { \fsctx \vdash \utripleerr{\pexp_1 \dotint \Val \lstar \pexp_1 \notdotint \Nat}{\pmutate{\expr{\pexp_1}}{\expr{\pexp_2}}}{\Qerr^*} }
  \and
      \inferrule[\mbox{mutate-err-use-after-free}]
  {\eerr \defeq ["\mathsf{UseAfterFree}", \stringify {\pexp_1}, \expr{\pexp_1}]}
  { \fsctx \vdash \utripleerr{\expr{\pexp_1}\mapsto\cfreed}{\pmutate{\expr{\pexp_1}}{\expr{\pexp_2}}}{\Qerr^*} }
  \and
    \inferrule[\mbox{mutate-err-val-2}]
  {\eerr \defeq ["\mathsf{ExprEval}", \stringify {\pexp_2}]}
  { \fsctx \vdash \utripleerr{\expr{\pexp_1} \mapsto \pexp\lstar \pexp_2\notdotint\Val}{\pmutate{\expr{\pexp_1}}{\expr{\pexp_2}}}{\Qerr^*} }
  \and
  \inferrule[\mbox{new}]
    { \pvar{x}  \notin \pv{\pexp'} \quad \subst \defeq [\pexp'/\pvar x]}
          { \fsctx \vdash \utripleok{\pvar x \doteq \pexp' \lstar \pexp \dotint \Nat}{\palloc{\pvar{x}}{\pexp}}{\pexp' \dotint \Val \lstar \bigoast{_{0 \le i < \pexp\subst}} (( \pvar x + i) \mapsto \nil)} }

  \and
          \inferrule[\mbox{new-err-eval}]
    { \eerr\defeq["\mathsf{ExprEval}",\stringify{\pexp}]}
          { \fsctx \vdash \utripleerr{\pvar x \doteq \pexp' \lstar \pexp \notdotint \Val}{\palloc{\pvar{x}}{\pexp}}{\Qerr^*} }
          \qquad
	\inferrule[\mbox{new-err-type}]
    { \eerr\defeq["\mathsf{Type}",\stringify{\pexp},\expr{\pexp},"\mathsf{Nat}"]}
          { \fsctx \vdash \utripleerr{\pvar x \doteq \pexp' \lstar \pexp \dotint \Val \lstar \pexp \notdotint \Nat}{\palloc{\pvar{x}}{\pexp}}{\Qerr^*} }
  \and

\inferrule[\mbox{free}]
      {}
   { \fsctx \vdash \utripleok{\pexp \mapsto \pexp'}{\pdealloc{\pexp}}{\dotin{\pexp'}{\Val}\lstar \pexp \mapsto \cfreed} }
   \and
   \inferrule[\mbox{free-err-eval}]
      {\eerr\defeq["\mathsf{ExprEval}",\stringify{\pexp}]}
   { \fsctx \vdash \utripleerr{\pexp \notdotint\Val}{\pdealloc{\pexp}}{\Qerr^*} }
   \and
\inferrule[\mbox{free-err-type}]
      {\eerr\defeq["\mathsf{Type}",\stringify{\pexp},\pexp,"\mathsf{Nat}"]}
   { \fsctx \vdash \utripleerr{\pexp \dotint \Val \lstar \pexp \notdotint \Nat}{\pdealloc{\pexp}}{\Qerr^*} }
   \and
\inferrule[\mbox{free-err-use-after-free}]
      {\eerr\defeq["\mathsf{UseAfterFree}",\stringify{\pexp},\pexp]}
   { \fsctx \vdash \utripleerr{\pexp \mapsto\cfreed}{\pdealloc{\pexp}}{\Qerr^*} }
\and
\inferrule[\mbox{error}]
{ \eerr \defeq [``\mathsf{Error}", \pexp] }
{ \fsctx \vdash \utripleerr{\dotin{\pexp}{\Val}}{\perror(\pexp)}{\pvar{err} \doteq \eerr} }
\and
\inferrule[\mbox{error-err}]{
  \eerr \defeq ["\mathsf{ExprEval}", \stringify{\pexp}]}
    { \fsctx \vdash \utripleerr{{\pexp\notdotint\Val}}{\perror(\pexp)}{{\Qerr^*}} }
\end{mathpar}
}

{\small
\begin{mathpar}
  \inferrule[\mbox{if-then}]
  { \fsctx \vdash \utripleq{P \lstar \pexp}{C_1}{\bigq} }
  { \fsctx \vdash \utripleq{P \lstar \pexp }{\pifelse{\pexp}{C_1}{C_2}}{\bigq} }
\and
  \inferrule[\mbox{if-else}]
  { \fsctx \vdash \utripleq{P \lstar \lnot \pexp}{C_2}{\bigq} }
  { \fsctx \vdash \utripleq{P \lstar \lnot \pexp }{\pifelse{\pexp}{C_1}{C_2}}{\bigq} }
\and
  \inferrule[\mbox{if-err-val}]
  { \eerr \defeq ["\mathsf{ExprEval}", \stringify{\pexp}] }
  { \fsctx \vdash \utripleerr{P \lstar \pexp\notdotint\Val}{\pifelse{\pexp}{C_1}{C_2}}{\Qerr^*} }
\and
    \inferrule[\mbox{if-err-type}]
  { \eerr \defeq ["\mathsf{Type}", \stringify{\pexp}, \pexp, "\mathsf{Bool}"] }
  { \fsctx \vdash \utripleerr{P \lstar \pexp \dotint \Val \lstar \pexp \notdotint \Bool}{\pifelse{\pexp}{C_1}{C_2}}{\Qerr^*} }
\and
 \inferrule[\mbox{seq}]
 { {\begin{array}{c} \fsctx \vdash \uquadruplex{P}{C_1}{R}{\Qerr^1} \\ \fsctx \vdash \uquadruplex{R}{C_2}{\Qok}{\Qerr^2} \end{array}}}
 { \fsctx \vdash \uquadruplex{P}{\psequence{C_1}{C_2}}{\Qok}{\Qerr^1 \lor \Qerr^2} }
\and
  \inferrule[\mbox{while}]
  {
  \frall{i\in\Nat}~\models P_i \Rightarrow \dotin \pexp \Bool \lstar \AssTrue \\\\
  \frall{i\in\Nat}~\fsctx \vdash \uquadruplex{P_i \lstar \pexp}{\cmd}{P_{i+1} }{Q_i}
	}
	{ \fsctx \vdash \uquadruplex{P_0}{\pwhile{\pexp}{\cmd}}{\neg\pexp\lstar\exsts{i\in\Nat} P_i}{\exsts{i\in\Nat}  Q_i}
	}
\and
  	\inferrule[\mbox{while-err-val}]
{
\frall{i\in\Nat^{<j}}~\models P_i \Rightarrow \pexp \lstar \AssTrue \qquad
\models P_j \Rightarrow \pexp \notin \Val \lstar \AssTrue \\
\frall{i\in\Nat^{<j}}~\fsctx \vdash \uquadruple{P_i}{C}{P_{i+1} }{Q_i} \quad
\eerr\defeq["\mathsf{ExprEval}",\stringify{\pexp}]
  }
  { \fsctx \vdash \utripleerr{P_0}{\pwhile{\pexp}{C}}{(\exsts{i < j}  Q_i) \lor (P_{j}\lstar\pvar{err}\doteq\eerr) }
  }
\and
    	\inferrule[\mbox{while-err-type}]
{
\frall{i\in\Nat^{<j}}~\models P_i \Rightarrow \pexp \lstar \AssTrue \qquad
\models P_j \Rightarrow \dotin \pexp (\Val\setminus\Bool) \lstar \AssTrue \\\\
\frall{i\in\Nat^{<j}}~\fsctx \vdash \uquadruple{P_i}{C}{P_{i+1} }{Q_i} \quad
\eerr\defeq["\mathsf{Type}",\stringify{\pexp},\pexp,"\mathsf{Bool}"]
  }
  { \fsctx \vdash \utripleerr{P_0}{\pwhile{\pexp}{C}}{(\exsts{i < j}  Q_i) \lor (P_{j}\lstar\pvar{err}\doteq\eerr)}
  }
 \and
   \inferrule[\mbox{equiv}]
  { \fsctx \vdash \uquadruple{P'}{C}{\Qok'}{\Qerr'} \quad \models P', \Qok', \Qerr' \Leftrightarrow P, \Qok, \Qerr }
 { \fsctx \vdash \uquadruple{P }{C}{\Qok}{\Qerr}}
\and
\inferrule[\mbox{frame}]
{  \fsctx \vdash \uquadruple{P}{C}{\Qok}{\Qerr} \quad \updt(C) \cap \fv{R} = \emptyset }
{ \fsctx \vdash \uquadruple{P \lstar R}{C}{\Qok \lstar R}{\Qerr \lstar R}}
\and
 \inferrule[\mbox{exists}]
{ \fsctx \vdash \uquadruple{ P}{C}{\Qok}{\Qerr} }
{ \fsctx \vdash \uquadruple{\exists x.\, P}{C}{\exsts x \Qok} {\exsts x \Qerr}} \qquad
\and
 \inferrule[\mbox{disj}]
  { \fsctx \vdash \uquadruplex{P_1}{C}{\Qok^1}{\Qerr^1} \quad
   \fsctx \vdash \uquadruplex{P_2}{C}{\Qok^2}{\Qerr^2}}
 { \fsctx \vdash \uquadruplex{P_1 \lor P_2}{C}{\Qok^1 \lor \Qok^2}{\Qerr^1 \lor \Qerr^2} }
\and
   \inferrule[\mbox{fun-call}]{
      \uspecq{\pvvar x \doteq \vec x \lstar P}{\Qok}{\Qerr} \in \admiss{\fsctx(\fid)} \qquad
      \pvar{y} \not\in \pv{\lexp_y} \qquad \subst \eqdef [\lexp_y/\pvar{y}]
  }
  {
      \fsctx \vdash
      \uquadruplex{\pvar{y} \doteq \lexp_y \lstar \vec \pexp \doteq \vec x \lstar P}
      {\pfuncall{\pvar y}{\fid}{\vec \pexp}}
      {\lexp_y \dotint \Val \lstar \vec{\pexp}\subst \doteq \vec x \lstar \Qok[\pvar{y} / \mathtt{ret}]}
      {\pvar y \doteq \lexp_y \lstar \vec{\pexp} \doteq \vec x \lstar \Qerr}
  }
\end{mathpar}
\begin{mathpar}
  \inferrule[\mbox{fun-call-err-fct-nofunc}]{
    	{\fid\notin\dom(\fsctx) \qquad \eerr\defeq["\mathsf{NoFunc}",\fid]}
  }
  {
      \fsctx \vdash
      \utripleerr{\pvar{y} \doteq \lexp_y \lstar \vec \pexp \doteq \vec x}
      {\pfuncall{\pvar y}{\fid}{\vec \pexp}}
      {\Qerr^*}}
\and
\inferrule[\mbox{fun-call-err-param-count}]{
	 \uspecq{\pvvar x' \doteq \vec x' \lstar P}{\Qok}{\Qerr}\in \admiss{\fsctx(\fid)} \\
	|\vec \pexp|\neq|\pvvar x'| \\
\eerr\eqdef["\mathsf{ParamCount}"]
}
{
  \fsctx \vdash
  \utripleerr{\pvar{y} \doteq \lexp_y}
  {\pfuncall{\pvar y}{\fid}{\vec \pexp}}
  {\Qerr^*}}
  \and
\inferrule[\mbox{fun-call-err-param-expr-eval}]{
	 \uspecq{\pvvar x \doteq \vec x \lstar P}{\Qok}{\Qerr}\in \admiss{\fsctx(\fid)} \\
	 n = |\pvvar x| \\
\eerr\eqdef["\mathsf{ParamExprEval}", \stringify{E_k}]
}
{
  \fsctx \vdash
  \utripleerr{\pvar{y} \doteq \lexp_y \lstar (\bigoast{_{0 \le i < k}} (\pexp_i \in \Val)) \lstar \pexp_k \notin \Val}
  {\pfuncall{\pvar y}{\fid}{\pexp_1, \ldots \pexp_n}}
  {\Qerr^*}}
\and
 		\inferrule[\mbox{env-empty}]{\vdash (\emptyset, \emptyset)}{}
		\and
		\inferrule[\mbox{env-extend}]
{
  		{\begin{array}{c}
  		\vdash (\fictx,\fsctx)
  		\qquad
  		I = \{1, \ldots ,n\}
  		\qquad
  		\frall{i \in I} \fid_i \not\in \dom(\fictx)
  		\qquad \fictx' = \fictx[\fid_i \mapsto (\vec{\pvar x}_i, C_i, \pexp_i)]_{i \in I}
		\\[1mm]
  		\fsctx(\alpha) = \fsctx[\fid_i\mapsto
				\{ \uspecq{\tprei(\beta)}{\Qoki(\beta)}{\Qerri(\beta)} \mid {\beta < \alpha}\} \cup \{ \uspeconecaseq{\ntprei(\beta)}{\AssFalse} \mid \beta \leq \alpha\}]_{i \in I}
		\\[1mm]
		\frall{i\in I, \alpha} \exists t \in \fext_{\fictx', \fid_i}\big(\uspecq{\tprei(\alpha)}{\Qoki(\alpha)}{\Qerri(\alpha)}\big).~\fsctx(\alpha)\vdash C_i : \sspec \\
		\frall{i\in I, \alpha} \exists t \in \fext_{\fictx', \fid_i}(\uspeconecaseq{\ntprei(\alpha)}{\AssFalse}).~\fsctx(\alpha)\vdash C_i : \sspec  \\[1mm]
		\tprei \defeq \exsts{\alpha}\tprei(\alpha)\lstar\dotin{\alpha}{\ord}
		\qquad
		\ntprei \defeq \exsts{\alpha}\ntprei(\alpha)\lstar\dotin{\alpha}{\ord}
		\\
		\Qoki \defeq \exsts{\alpha}\Qoki(\alpha)\lstar\dotin{\alpha}{\ord}
		\qquad
		\Qerri \defeq \exsts{\alpha}\Qerri(\alpha)\lstar\dotin{\alpha}{\ord}
		\\[1mm]
		\fsctx'':=\fsctx[\fid_i\mapsto\{\uspecq{\tprei}{\Qoki}{\Qerri}, \uspeconecaseq{\ntprei}{\AssFalse}\}]_{i\in I}
		\end{array}}
	}{
  		\vdash (\fictx',\fsctx'')}
\end{mathpar}}

\vspace{1cm}

The rule for function calls differs slightly from the one presented in the main body. In particular, it has the condition $\uspecq{\pvvar x \doteq \vec x \lstar P}{\Qok}{\Qerr}\in \admiss{\fsctx(\fid)}$ instead of $\uspecq{\pvvar x \doteq \vec x \lstar P}{\Qok}{\Qerr}\in \fsctx(\fid)$, with $\admiss{T}$ defined as follows:
{\small
\begin{mathpar}
  \inferrule[\mbox{adm-in}]
   {T \subseteq \admiss T}
   {}
\and
 \inferrule[\mbox{adm-exists}]
  { \uspecq{\vec{\pvar x} = \vec x \lstar P}{\Qok}{\Qerr} \in \admiss T \qquad y \notin \vec x}
  { \uspecq{\vec{\pvar x} = \vec x \lstar \exists y.~P}{\exists y.~\Qok}{\exists y.~\Qerr} \in \admiss T } \qquad
\and
 \inferrule[\mbox{adm-frame}]
  { \uspecq{P}{\Qok}{\Qerr} \in \admiss T \quad \pv{R} = \emptyset}
  { \uspecq{P \lstar R}{\Qok \lstar R}{\Qerr \lstar R} \in \admiss T }
\and
 \inferrule[\mbox{adm-equiv}]
  { {\begin{array}{c} \uspecq{\vec{\pvar x} = \vec x \lstar P'}{\Qok'}{\Qerr'} \in \admiss T \\ \models P', \Qok', \Qerr' \Leftrightarrow P, \Qok, \Qerr \end{array}}}
  { \uspecq{\vec{\pvar x} = \vec x \lstar P}{\Qok}{\Qerr} \in \admiss T}
\and
 \inferrule[\mbox{adm-disj}]
  { {\begin{array}{c} \uspecq{\vec{\pvar x} = \vec x \lstar P_1}{{\Qok}_1}{{\Qerr}_1} \in \admiss T \quad \uspecq{\vec{\pvar x} = \vec x \lstar P_2}{{\Qok}_2}{{\Qerr}_2} \in \admiss T \end{array}}}
  { \uspec{\vec{\pvar x} = \vec x \lstar (P_1 \lor P_2}{{\Qok}_1 \lor {\Qok}_2}{\oxerr : {\Qerr}_1 \lor {\Qerr}_2} \in \admiss T }
\and
 \inferrule[\mbox{adm-alpha-equiv}]
  { {\begin{array}{c} \uspecq{P'}{\Qok'}{\Qerr'} \in \admiss T \quad \vec y \defeq \fv{P'} \cup \fv\Qok \cup \fv\Qerr \\ \vec z~\text{distinct} \quad
  P \defeq P' [\vec z/\vec y] \quad  \Qok \defeq \Qok' [\vec z/\vec y] \quad \Qerr \defeq \Qerr' [\vec z/\vec y] \end{array}}}
  { \uspecq{P}{\oxok : \Qok}{\oxerr : \Qerr} \in \admiss T}
  \end{mathpar}
}
These admissibility rules are validity-preserving and they lift the logic rules $\mbox{frame}$, \mbox{exists}, \mbox{equiv} and \mbox{disj} to the level of function specifications, as well as allow for the renaming of logical variables. These are commonly implicitly used in OX logics.

The \mbox{env-extend} rule in the figure above differs in several ways from the rule given in the main body (\S\ref{sec:proglog}):

\begin{mathpar}
  \small
    \inferrule[\mbox{env-extend}]{
      {\begin{array}{l}
      \vdash (\fictx,\fsctx)
      \qquad
      \fid \not\in \dom(\fictx)
      \qquad \fictx' = \fictx[\fid \mapsto (\vec{\pvar x}, \cmd, \pexp)]
      \\[1mm]
      \begin{array}{@{}ll}
      \sspec \defeq \uspecq {P} {\Qok} {\Qerr} &
    \sspec_\infty \defeq \uspeconecaseq {P_\infty} {\AssFalse} \\
      \sspec(\alpha) \defeq \uspecq{P \lstar \alpha = E_\mu} {\Qok \lstar \alpha = E_\mu} {\Qerr \lstar \alpha = E_\mu} &
    \sspec_\infty(\alpha) \defeq \uspeconecaseq {P_\infty \lstar \alpha = E_\mu} {\AssFalse}
      \end{array} \\ ~ \\[-3mm]
      \fsctx(\alpha) = \fsctx[\fid\mapsto
          \{ \sspec(\beta) \mid {\beta < \alpha}\} \mathbin\cup \{ t_\infty \mid \beta \leq \alpha \}]
      \\[1mm]
      \frall{\alpha} \exists t' \in \fext_{\fictx', \fid}\big(\sspec(\alpha)\big).~\fsctx(\alpha)\vdash \cmd : \sspec' \qquad
      {\frall{\alpha} \exists t' \in \fext_{\fictx', \fid}(\sspec_\infty(\alpha)).~\fsctx(\alpha)\vdash \cmd : \sspec'}  \\[1mm]
      \fsctx':=\fsctx[\fid\mapsto\{\sspec, \sspec_\infty\}]
      \end{array}}
    }
    {
      \vdash (\fictx',\fsctx')
    }
  \end{mathpar}

The more complex rule allows us to reason about a broader set of functions, and we chose to present a simpler version in the body of the paper for didactical reasons.
The proof of soundness conducted in Appendix \ref{apdx:envsound} proves soundness of the more complex rule,
which has the following strengths over the simplified version:
\begin{itemize}
  \item Instead of adding only one function $f$, we add the finite set $\{f_i ~|~ i=1,..,n \}$, allowing us to reason about \textit{mutally recursive} functions.
  \item We extend the measure beyond natural numbers to computable ordinals, $\ord\defeq\ckord$, allowing us to reason about a broader set of functions, such as those with non-deterministic nested recursion.
  \item We use the more general notation $P(\alpha)$ instead of $P\lstar\alpha\doteq\lexp_\mu$, and similarly for post-conditions. For this reason, we require an explicit transition from $\alpha$ being a meta-variable to it being a logical variable, as in the while rule. It can then be removed from the specification using admissible equivalence.
  \end{itemize}

\newpage
\newcommand{\casex}[1]{\medskip\noindent{\emph{#1}}}

%% file: sections/app-soundness.tex
\section{Soundness: ESL}
\label{apdx:soundesl}

In order to prove the soundness result of theorem \ref{logicsoundness}, we require three auxiliary lemmas regarding environment validity. Their proofs are straightforward and will therefore be omitted. For legacy reasons, this appendix and the following appendices may use an alternative notation for satisfiability: $\subst, \sto, \hp \models P$ instead of $\subst, (\sto, \hp) \models P$, where $\subst$ (referred to as the \emph{substitution}) is of type: $\LVar \rightharpoonup_{\mathtt{fin}} \Val$.

\begin{lemma}[Auxiliary properties] The following properties hold:
\label{lem:auxprop}
\begin{enumerate}
\item \label{LemPv}
  $\frall{\subst,\sto,\sto',\hp,P} \subst,\sto,\hp\models P \wedge \sto|_{\pv{P}}=\sto'|_{\pv{P}} \Rightarrow \subst,\sto',\hp\models P$

\item \label{LemMods}
  $\frall{\subst,\sto,\sto',\hp,\hp',\outcome,\scmd,\fictx}
   (\sto,\hp), \scmd \baction_{\fictx} \outcome : (\sto',\hp') \implies \frall{x \in \dom(\sto) \setminus \updt(C)} \sto(x) = \sto'(x)$

\item\label{LemSubst}
  $\frall{\pexp,\subst, \sto,\sto',\pvar y} \sto|_{\pv{\pexp}\backslash\{\pvar y\}}=\sto'|_{\pv{\pexp}\backslash\{\pvar y\}} \Rightarrow \esem{\pexp[\sto'(\pvar y)/\pvar y]}{\subst,\sto}=\esem{\pexp}{\subst,\sto'}$
\end{enumerate}
\end{lemma}

We also require a lemma that states that environment validity is closed under admissibility:
\begin{lemma}[Admissibility and Environment Validity]
\label{lem:admiss}
If $\models(\fictx,\fsctx)$, then $\models(\fictx,\admiss \fsctx)$, where $\admiss \fsctx = \{ f \mapsto \admiss{\fsctx(f)} \mid f \in \dom(\fsctx) \}$.
\end{lemma}

\begin{proof}
Assume that $\models(\fictx,\fsctx)$. It is sufficient to prove the claim for an isolated $f \in \dom(\fsctx)$, by induction on the formation of $\admiss{\fsctx(f)}$. The proof of the four non-trivial cases is analogous to the proof of the respective structural rules in the proof of Theorem \ref{logicsoundness}, given below.
\end{proof}

\begin{proof}[Proof of Theorem \ref{logicsoundness}]
	By induction on the derivation $\fsctx \vdash \uquadruple P C {\Qok}{\Qerr}$. We prove a representative selection of rules; the proofs for the remaining ones are analogous.

	\casex{Function Call.} We first prove the successfully terminating case of under-approximation. Our hypotheses are:

	\begin{description}
	\item[(H1)] $\models(\fictx,\fsctx)$,
	\item[(H2)] $\subst,\sto',\hp' \models \lexp_y \dotint \Val \lstar \vec{\pexp}[\lexp_y/\pvar{y}] \doteq \vec x \lstar \Qok[\pvar{y} / \mathtt{ret}]$,
	\item[(H3)] $\dish{\hp'}{\hp_\fid}$,
	\item[(H5)] $\pvar y\not\in\fv{\lexp_y}$,
	\item[(H6)] $\uspec{\pvvar x \doteq \vec x \lstar P}{\Qok}{\Qerr}\in \admiss{\fsctx(\fid)}$ 
	\end{description}

	Our goal is to show that: $$\exsts{\sto,\hp} \subst,\sto,\hp \models \pvar{y} \doteq \lexp_y \lstar \vec \pexp \doteq \vec x \lstar P \wedge(\sto, \hp \uplus \hp_\fid), \passign{\pvar y}{\fid(\vec{\pvar x})} \baction_\fictx ok: (\sto', \hp' \uplus \hp_\fid)$$

\begin{itemize}
\item
	From (H1) and Lemma~\ref{lem:admiss}, we obtain \textbf{(H1a)}~$\models(\fictx,\admiss \fsctx)$
\item
	From (H6) and (H1a), we obtain $C$ and $\pexp$, such that \textbf{(H7)}~$\fid(\vec{\pvar x})\{C;\preturn{\pexp}\}\in\fictx$.
\item
	From (H2), we obtain that \textbf{(H2a)}~$\subst, \sto'\models \vec{\pexp}[\lexp_y/\pvar{y}] = \vec x$,
	and \textbf{(H2b)}~$\subst, \sto', \hp'\models \Qok[\pvar y/\pvar{ret}]$.
\item 
	From (H1a), (H6) and (H7), we obtain that there exists a specification
$$\uspec{\vec {\pvar x} \doteq \vec x  \lstar P \lstar \vec {\pvar z}\doteq\nil}{\Qok'}{\Qerr'}\in\fext_{\fictx,\fid}(\uspec{\pvvar x\doteq\vec x\lstar P}{\Qok}{\Qerr})$$
such that \textbf{(H8)}~$\fictx \models \uquadruple{\vec {\pvar x} \doteq \vec x \lstar P \lstar \vec {\pvar z}\doteq\nil }{C}{\Qok'}{\Qerr'}$, where, from the definition of the internalisation function, we know that \textbf{(H9a)}~$\Qok\Leftrightarrow\exsts{\vec p} \Qok'[\vec p/\vec{\pvar p}]\lstar \pvar{ret}\doteq \pexp[\vec p/\vec{\pvar p}]$, where $\pvvar z=\pv{C}\backslash\{\pvvar x\}$ and $\vec{\pvar p} = \{ \vec{\pvar x} \} \uplus \{ \pvvar z \} = \pv{C}$.

\item
Given (H2b) and (H9a), we derive the following:
\[
\begin{array}{l r}
	\subst,\sto',\hp'\models (\exsts{\vec p} \Qok'[\vec p/\vec{\pvar p}]\lstar \pvar{ret}\doteq \pexp[\vec p/\vec{\pvar p}])[\pvar y/\pvar{ret}] & \\
	\Rightarrow \subst,\sto',\hp'\models \exsts{\vec p} \Qok'[\vec p/\vec{\pvar p}]\lstar \pvar{y}\doteq \pexp[\vec p/\vec{\pvar p}] & \\
	\end{array}
	\]
	from which we obtain that there exist values $\vec w$, such that:
	\[
	\begin{array}{l r}
	\Rightarrow  \subst[\vec p\rightarrow\vec w],\sto',\hp'\models\Qok'[\vec p/\vec{\pvar p}]\lstar \pvar{y}\doteq \pexp[\vec p/\vec{\pvar p}]  & \\
		\Rightarrow  \subst[\vec p\rightarrow\vec w],\sto',\hp'\models\Qok'[\vec w/\pvvar p]\lstar \pvar{y}\doteq \pexp[\vec w/\vec{\pvar p}]  & \\
		\Rightarrow  \subst[\vec p\rightarrow\vec w],\sto'[\pvvar p\rightarrow\vec w],\hp'\models\Qok'\lstar \pvar{y}\doteq \pexp & \\
		\Rightarrow  \subst,\sto'[\pvvar p\rightarrow\vec w],\hp'\models\Qok'\lstar \pvar{y}\doteq \pexp & \textbf{(H10a)} \\
		\Rightarrow \subst,\sto'[\pvvar p\rightarrow\vec w],\hp'\models\Qok' & \textbf{(H10b)} \\
\end{array}
\]

\item
Instantiating (H8) with (H10b) and (H3), we obtain that there exist $\tilde \sto$ and $\hp$, such that
\textbf{(H11)}~$\subst,\tilde \sto,\hp \models \vec {\pvar x} \doteq \vec x \lstar P \lstar \vec {\pvar z}\doteq\nil$ and
\textbf{(H12)}~$(\tilde \sto, \hp \uplus \hp_\fid), C \baction_\fictx ok: (\sto'[\pvvar p\rightarrow \vec w], \hp' \uplus \hp_\fid)$.
\item
Let $\vec v = \subst(\vec x)$. Then, since $\pv{P} = \emptyset$ and given Lemma~\ref{lem:auxprop}(\ref{LemPv}), taking $\sto'':=\emptyset[\pvvar{x}\rightarrow \vec v][\pvvar z\rightarrow\nil]$, we obtain that \textbf{(H13)}~$\subst,\sto'',\hp\models \vec {\pvar x} \doteq \vec x \lstar P \lstar \vec {\pvar z}\doteq\nil$ and also that \textbf{(H14)}~$(\sto'', \hp \uplus \hp_\fid), C \baction_\fictx ok: (\sto'[\pvvar p\rightarrow \vec w], \hp' \uplus \hp_\fid)$.
\item
Let $v'= \esem{\pexp}{\sto'[\pvvar p\rightarrow \vec w]} = \esem{\pexp[\vec w / \pvvar p]}{\subst, \sto'}$, $v_y = \esem{\lexp_y}{\subst, \sto'}$ and \textbf{(H15)}~$\sto = \sto'[\pvar y\rightarrow v_y]$. Therefore, we also have that \textbf{(H16)}~$\sto'=\sto[\pvar y\rightarrow v']$.
\item We now need to prove that $(\sto, \hp \uplus \hp_\fid), \passign{\pvar y}{\fid(\vec \pexp)} \baction_\fictx ok: (\sto[\pvar y\rightarrow v'], \hp' \uplus \hp_\fid)$. For this, we already have:
$\pfunction{\procname}{\vec{\pvar{x}}}{\scmd; \preturn{\pexp}} \in \scontext$, $\pv{\scmd} \setminus \{\vec{\pvar{x}}\} = \{\vec{\pvar{z}} \}$, $\sto''  =  \emptyset [ \vec{\pvar{x}} \storearrow \vec{v}] [ \vec{\pvar{z}} \storearrow \nil]$, $(\sto'', \hp \uplus \hp_\fid), \scmd \baction_{\fictx} \sthreadp{ \sto' }{ \hp' \uplus \hp'_\fid }$ and $\esem{\pexp}{\sto'[\pvvar p\rightarrow \vec w]} =v'$,
and we still need $\esem{\vec{\vexp}}{\sto} = \vec{v}$.
Rewriting (H2a) given (H16), we get $\subst, \sto[\pvar y\rightarrow v'] \models \vec{\pexp}[\lexp_y/\pvar{y}] = \vec x$, that is, \textbf{(H17)}~$\esem{\vec{\pexp}[\lexp_y/\pvar{y}]}{\subst, \sto[\pvar y\rightarrow v']} =~\vec v$.
From (H17), the definition of $\sto$, and Lemma~\ref{lem:auxprop}(\ref{LemSubst}), we then obtain that $\esem{\vec{\vexp}}{\subst, \sto} = \vec{v}$, and from there, as $\vec \pexp$ are program expressions, we obtain the desired $\esem{\vec{\vexp}}{\sto} = \vec{v}$.
\item
Finally, we need to prove that $\subst, \sto, \hp \models \pvar{y} \doteq \lexp_y \lstar \vec \pexp \doteq \vec x \lstar P$. For the first ``starjunct'', we need to show $\esem{\pvar y}{\subst,\sto} = \esem{\lexp_y}{\subst, \sto}$:
$$
	\esem{\pvar y}{\subst,\sto} = \sto(\pvar y) \stackrel{\text{(H15)}}{=} (\sto'[\pvar y\rightarrow v_y])(\pvar y) = v_y = \esem{\lexp_y}{\subst,\sto'} \stackrel{\text{(H16)}}{=} \esem{\lexp_y}{\subst,\sto[\pvar y\rightarrow v']} \stackrel{\text{(H5)}}{=} \esem{\lexp_y}{\subst,\sto}
$$
For the second starjunct we do:
$$
\esem{\vec x}{\subst,\sto} = \subst(\vec x) = \esem{\vec x}{\subst,\sto'} \stackrel{\text{(H2a)}}{=} \esem{\vec \pexp[\lexp_y / \pvar y]}{\subst,\sto'} \stackrel{\text{(H15)}}{\ref{lem:auxprop}.\ref{LemSubst}}{=} \esem{\vec \pexp}{\subst,\sto}
$$
The third starjunct follows from (H13), given that $P$ has no logical variables by construction.
\end{itemize}

Next, we prove the erroneously terminating case of under-approximation. Our hypotheses are:
	\begin{description}
	\item[(H1)] $\models(\fictx,\fsctx)$,
	\item[(H2)] $\subst,\sto',\hp' \models \pvar y\doteq \lexp_y\lstar\vec{\pexp} \doteq \vec x \lstar \Qerr$,
	\item[(H3)] $\dish{\hp'}{\hp_\fid}$,
	\item[(H4)] $\pvar y\not\in\fv{\lexp_y}$,
	\item[(H5)] $\uspec{\pvvar x \doteq \vec x \lstar P}{\Qok}{\Qerr}\in \admiss {\fsctx(\fid)}$.
	\end{description}
	Our goal is to show that: $$\exsts{\sto,\hp} \subst,\sto,\hp \models \pvar{y} \doteq \lexp_y \lstar \vec \pexp \doteq \vec x \lstar P \wedge(\sto, \hp \uplus \hp_\fid), \passign{\pvar y}{\fid(\vec{\pvar x})} \baction_\fictx err: (\sto', \hp' \uplus \hp_\fid)$$

\begin{itemize}
\item
	From (H1) and Lemma~\ref{lem:admiss}, we obtain \textbf{(H1a)}~$\models(\fictx,\admiss \fsctx)$
\item
	From (H5) and (H1a), we obtain $C$ and $\pexp$, such that \textbf{(H7)}~$\fid(\vec{\pvar x})\{C;\preturn{\pexp}\}\in\fictx$.
\item
	From (H2), we obtain that \textbf{(H2a)} $\subst, \sto'\models \vec{\pexp} = \vec x$, \textbf{(H2b)}~$\subst, \sto', \hp'\models \Qerr$ and \textbf{(H2c)}~$\subst,\sto'\models \pvar y = \lexp_y$.
\item 
From (H1a), (H5) and (H7), we obtain that there exists a specification $$\uspec{\vec {\pvar x} \doteq \vec x  \lstar P \lstar \vec {\pvar z}\doteq\nil}{\Qok'}{\Qerr'}\in\fext_{\fictx,\fid}(\uspec{\pvvar x\doteq\vec x\lstar P}{\Qok}{\Qerr})$$ such that \textbf{(H8)}~$\fictx \models \uquadruple{\vec {\pvar x} \doteq \vec x \lstar P \lstar \vec {\pvar z}\doteq\nil }{C}{\Qok'}{\Qerr'}$, where, from the definition of the internalisation function, we know that \textbf{(H9a)}~$\Qerr\Leftrightarrow\exsts{\vec p} \Qerr'[\vec p/\vec{\pvar p}]$ with $\pvvar z=\pv{C}\backslash\{\pvvar x\}$ and $\vec{\pvar p} = \{ \vec{\pvar x} \} \uplus \{ \pvvar z \}$.
\item
Given (H2b) and (H9a), analogous to the previous case we obtain:
$$
	\subst,\sto',\hp'\models (\exsts{\vec p} \Qerr'[\vec p/\vec{\pvar p}]) \quad
	\Longrightarrow\quad\exsts{\vec w} \subst, \sto_{\pvar p}',\hp'\models \Qerr' \quad \textbf{(H10)}
$$
where $\sto_{\pvar p}' = \sto'[\vec{\pvar p}\mapsto \vec w]$
\item
Instantiating (H8) with (H1), (H10), and (H3), we obtain that there exist $\tilde \sto$ and $\hp$, such that
\textbf{(H11)}~$\subst,\tilde \sto,\hp \models \vec {\pvar x} \doteq \vec x \lstar P \lstar \vec {\pvar z}\doteq\nil$ and
\textbf{(H12)}~$(\tilde \sto, \hp \uplus \hp_\fid), C \baction_\fictx err: (\sto_{\pvar p}', \hp' \uplus \hp_\fid)$.
\item
Let $\vec v = \subst(\vec x)$. Then, since $\pv{P} = \emptyset$ and given Lemma~\ref{lem:auxprop}(\ref{LemPv}), taking $\sto'':=\emptyset[\pvvar{x}\rightarrow \vec v][\pvvar z\rightarrow\nil]$, we obtain that \textbf{(H13)}~$\subst,\sto'',\hp\models \vec {\pvar x} \doteq \vec x \lstar P \lstar \vec {\pvar z}\doteq\nil$ and also that \textbf{(H14)}~$(\sto'', \hp \uplus \hp_\fid), C \baction_\fictx err: (\sto_{\pvar p}', \hp' \uplus \hp_\fid)$.
\item
Let
$v_{err} = \esem{\pvar{err}}{\subst, \sto'}$ and \textbf{(H15)}~$\sto = \sto'\backslash\pvar{err}$. Therefore, we also have that \textbf{(H16)}~$\sto'=\sto[\pvar{err}\rightarrow v_{err}]$.
\item We now need to prove that $(\sto, \hp \uplus \hp_\fid), \passign{\pvar y}{\fid(\vec \pexp)} \baction_\fictx err: (\sto[\pvar {err}\rightarrow v_{err}], \hp' \uplus \hp_\fid)$. For this, we already have:
$\pfunction{\procname}{\vec{\pvar{x}}}{\scmd; \preturn{\pexp}} \in \scontext$, $\pv{\scmd} \setminus \{\vec{\pvar{x}}\} = \{\vec{\pvar{z}} \}$, $\sto''  =  \emptyset [ \vec{\pvar{x}} \storearrow \vec{v}] [ \vec{\pvar{z}} \storearrow \nil]$ and $(\sto'', \hp \uplus \hp_\fid), \scmd \baction_{\fictx}err: \sthreadp{ \sto' }{ \hp' \uplus \hp'_\fid }$,
and we still need $\esem{\vec{\vexp}}{\sto} = \vec{v}$.
Rewriting (H2a) given (H16), we get $\subst, \sto'\models \vec{\pexp} = \vec x$, that is, \textbf{(H17)}~$\esem{\vec{\pexp}}{\subst, \sto'} =~\vec v$.
From (H16) and (H17), we then obtain that $\esem{\vec{\vexp}}{\subst, \sto} = \vec{v}$, which yields $\esem{\vec{\vexp}}{\sto} = \vec{v}$ since $\pexp$ is a program expression.
\item
Finally, we need to prove that $\subst, \sto, \hp \models \pvar{y} \doteq \lexp_y \lstar \vec \pexp \doteq \vec x \lstar P$.
The first starjunct is proven as follows:
$$
	\esem{\pvar y}{\subst,\sto} = \sto(\pvar y) \stackrel{\text{(H16)}}{=} \sto'(\pvar y) = \esem{\lexp_y}{\subst,\sto'} \stackrel{\text{(H16)}}{=} \esem{\lexp_y}{\subst,\sto}
$$
For the second starjunct we do:
$$
\esem{\vec x}{\subst,\sto} = \subst(\vec x) = \esem{\vec x}{\subst,\sto'} \stackrel{\text{(H2a)}}{=} \esem{\vec \pexp}{\subst,\sto'} \stackrel{\text{(H15)}}{=} \esem{\vec \pexp}{\subst,\sto}
$$
The third starjunct follows from (H13), given that $P$ has no logical variables by construction.
\end{itemize}

We move on to proving the successfully terminating over-approximation soundness. Our hypotheses are:
\begin{description}
\item[(H1)] $\models(\fictx,\fsctx)$
\item[(H2)] $\subst,\sto,\hp\models \pvar y\doteq \lexp_y\lstar\vec \pexp\doteq\vec x\lstar P$
\item[(H3)] $\hp~\sharp~\hp_\fid$
\item[(H4)] $\pvar y\not\in\fv{\lexp_y}$
\item[(H5)] $\uspec{\pvvar x \doteq \vec x \lstar P}{\Qok}{\Qerr}\in \admiss{\fsctx(\fid)}$.
\end{description}
Our goal is to show that:
	$$
	\begin{array}{l}
	\frall{\sto',\hp''} (\sto, \hp \uplus \hp_\fid), \passign{\pvar y}{\fid(\vec{\pvar x})} \baction_\fictx ok: (\sto', \hp'')\\
	\quad\quad\Rightarrow
	(\outcome\neq\oxm\land\exsts{\hp'}\hp''=\hp'\uplus \hp_\fid\wedge \subst,\sto',\hp'\models \vec \pexp[\lexp_y/\pvar y]\doteq\vec x\lstar \Qok[\pvar y/\pvar{ret}])
	\end{array}
	$$

\begin{itemize}
\item
	From (H1) and Lemma~\ref{lem:admiss}, we obtain \textbf{(H1a)}~$\models(\fictx,\admiss \fsctx)$
\item
	From (H1a) and (H5), we obtain $C$ and $\pexp$ such that \textbf{(H6)}~$\fid(\pvvar x)\{C,\preturn{\pexp}\}\in\fictx$.
\item
	We define $\vec v:=\subst(\vec x)$ and obtain from (H2) that \textbf{(H2a)}~$\subst,\sto\models\vec \pexp=\vec x$ and $\subst,\sto,\hp\models P$, and hence \textbf{(H2b)}~$\subst,\sto[\pvvar x\rightarrow\vec v],\hp\models \pvvar x\doteq\vec x\lstar P$.
\item
	Since $\pv{P}=\emptyset$ we obtain from (H2b) with Lemma \ref{LemPv} that $\subst,\emptyset[\pvvar x\rightarrow\vec v],\hp\models \pvvar x\doteq\vec x\lstar P$ and hence \textbf{(H7)}~$\subst,\emptyset[\pvvar x\rightarrow\vec v][\pvvar z\rightarrow\nil],\hp\models \pvvar x\doteq\vec x\lstar P\lstar\pvvar z\doteq\nil$.
\item
	(H1a), (H5) and (H6) imply the existence of a specification\\ $\uspec{\pvvar x\doteq\vec x\lstar P\lstar\pvvar z\doteq\nil}{\Qok'}{\Qerr'}\in\fext_{\fictx,\fid}(\uspec{\pvvar x\doteq\vec x\lstar P}{\Qok}{\Qerr})$ such that \\
	$$\textbf{(H9)}\quad\fictx\models\uquadruple{\pvvar x\doteq\vec x\lstar P\lstar\pvvar z\doteq\nil}{C}{\Qok'}{\Qerr'}$$
\item
	Instantiating (H9) with (H1), (H7) and (H3) yields
$$
\textbf{(H11)}\quad
\begin{array}{l}
\frall{\sto',\hp''} (\sto'',\hp\uplus \hp_\fid),C\baction_{\fictx}ok:(\sto',\hp'') \\\quad\Rightarrow(\outcome\neq\oxm\land\exsts{\hp'}\hp''=\hp'\uplus \hp_\fid\wedge \subst,\sto',\hp'\models \Qok')
\end{array}
$$
\item
	Defining $v':=\esem{\pexp}{\subst,\sto'}$, we apply the operation semantics of the successfully termination function call, which yields
	$$ (\sto,\hp\uplus \hp_\fid),\passign{\pvar y}{\fid(\vec \pexp)}\baction_{\fictx}(\sto[\pvar y\rightarrow v'],\hp'')$$
\item
	To conclude the proof, it remains to show that $\subst,\sto[\pvar y\rightarrow v'],\hp'\models \vec \pexp[\lexp_y/\pvar y]\doteq\vec x\lstar \Qok[\pvar y/\pvar{ret}]$.
	(H11) implies $\subst,\sto',\hp'\models\Qok'$. Defining $\pvvar p:=\pv{\Qok'}$ and $\vec v:=\sto'(\pvvar p)$, we obtain $\subst[\vec p\rightarrow\vec v],-,\hp'\models\Qok'[\vec p/\pvvar p]$ where $-$ may denote any variable store, since the assertion does not hold any program variables. Therefore, $\subst[\vec p\rightarrow\vec v],\sto,\hp'\models\Qok'[\vec p/\pvvar p]$ and hence \textbf{(H12)}~$\subst,\sto[\pvar y\rightarrow v'],\hp'\models\exsts{\vec p}\Qok'[\vec p/\pvvar p]\lstar\pvar y\doteq v'$ hold.
\item
	From the definitions of $v'$, $\pvvar p$ and $\vec p$ we obtain $v':=\esem{\pexp}{\subst,\sto'}=\esem{\pexp[\vec p/\pvvar p]}{\subst,-}$ and therefore $\subst,\sto[\pvar y\rightarrow v'],\hp'\models\exsts{\vec p}\Qok'[\vec p/\pvvar p]\lstar\pvar y\doteq \pexp[\vec p/\pvvar p]$. Hence $\subst,\sto[\pvar y\rightarrow v'],\hp'\models\Qok[\pvar y/\pvar{ret}]$.
\item
	From (H2a) and Lemma \ref{LemSubst} , we obtain $\subst,\sto[\pvar y\rightarrow v']\models\vec \pexp[\lexp_y/\pvar y]=\vec x$ and therefore $\subst,\sto[\pvar y\rightarrow v'],\hp'\models\vec \pexp[\lexp_y/\pvar y]\doteq\vec x\lstar\Qok[\pvar y/\pvar{ret}]$, which concludes this case of the proof.
\end{itemize}

Finally, we prove the erroneously terminating over-approximation soundness. Our hypotheses are:
\begin{description}
\item[(H1)] $\models(\fictx,\fsctx)$
\item[(H2)] $\subst,\sto,\hp\models \pvar y\doteq \lexp_y\lstar\vec \pexp\doteq\vec x\lstar P$
\item[(H3)] $\hp~\sharp~\hp_\fid$
\item[(H4)] $\pvar y\not\in\fv{\lexp_y}$
\item[(H5)] $\uspec{\pvvar x \doteq \vec x \lstar P}{\Qok}{\Qerr}\in\admiss{\fsctx(\fid)}$.
\end{description}
Our goal is to show that:
	$$
	\begin{array}{l}
	\frall{\sto',\hp''} (\sto, \hp \uplus \hp_\fid), \passign{\pvar y}{\fid(\vec{\pvar x})} \baction_\fictx err: (\sto', \hp'')\\
	\quad\quad\Rightarrow(\outcome\neq\oxm \land \exsts{\hp'}\hp''=\hp'\uplus \hp_\fid\land \subst,\sto',\hp'\models \pvar y\doteq \lexp_y\lstar\vec \pexp\doteq\vec x\lstar \Qerr)
	\end{array} $$

\begin{itemize}
\item
	From (H1) and Lemma~\ref{lem:admiss}, we obtain \textbf{(H1a)}~$\models(\fictx,\admiss \fsctx)$
\item
	From (H1a) and (H5), we obtain $C$ and $\pexp$ such that \textbf{(H6)}~$\fid(\pvvar x)\{C,\preturn{\pexp}\}\in\fictx$.
\item
	We define $\vec v:=\subst(\vec x)$ and obtain from (H2), that \textbf{(H2a)}~$\subst,\sto\models\vec \pexp=\vec x$ and $\subst,\sto,\hp\models P$, and hence (H2b)~$\subst,\sto[\pvvar x\rightarrow\vec v],\hp\models \pvvar x\doteq\vec x\lstar P$.
\item
	Since $\pv{P}=\emptyset$ we obtain from (H2b) with Lemma \ref{LemPv} that $\subst,\emptyset[\pvvar x\rightarrow\vec v],\hp\models \pvvar x\doteq\vec x\lstar P$ and \textbf{(H7)}~$\subst,\emptyset[\pvvar x\rightarrow\vec v][\pvvar z\rightarrow\nil],\hp\models \pvvar x\doteq\vec x\lstar P\lstar\pvvar z\doteq\nil$.
\item
	(H1a), (H5) and (H6) imply the existence of a specification \\ $\uspec{\pvvar x\doteq\vec x\lstar P\lstar\pvvar z\doteq\nil}{\Qok'}{\Qerr'}\in\fext_{\fictx,\fid}(\uspec{\pvvar x\doteq\vec x\lstar P}{\Qok}{\Qerr})$ such that \\
	$$\textbf{(H9)}\quad\fictx\models\uquadruple{\pvvar x\doteq\vec x\lstar P\lstar\pvvar z\doteq\nil}{C}{\Qok'}{\Qerr'}$$
\item
	Instantiating (H9) with (H1), (H7) and (H3) yields
	$$\textbf{(H11)}\quad
	\begin{array}{l}
	\frall{\sto',\hp''} (\sto'',\hp\uplus \hp_\fid),C\baction_{\fictx}err:(\sto',\hp'') \\\quad\Rightarrow(\outcome\neq\oxm\land\exsts{\hp'}\hp''= \hp'\uplus \hp_\fid\wedge \subst,\sto',\hp'\models \Qerr')
	\end{array}
	$$
\item
	Defining $v_{err}:=\esem{\pvar{err}}{\sto'}$, we apply the operation semantics of the erroneously termination function call, which yields
	$$ (\sto,\hp\uplus \hp_\fid),\passign{\pvar y}{\fid(\vec \pexp)}\baction_{\fictx}(\sto[\pvar{err}\rightarrow v_{err}],\hp'')$$
\item
	To conclude the proof, it remains to show that $\subst,\sto[\pvar{err}\rightarrow v_{err}],\hp'\models \pvar y\doteq \lexp_y\lstar\vec \pexp\doteq\vec x\lstar \Qerr$.
	$\subst,\sto'\models \pvar y= \lexp_y$ holds trivially and (H11) implies $\subst,\sto',\hp'\models\Qerr'$. Defining $\pvvar p:=\pv{\Qerr'}$ and $\vec v:=\sto'(\pvvar p)$, we obtain $\subst[\vec p\rightarrow\vec v],-,\hp'\models\Qerr'[\vec p/\pvvar p]$ where $-$ may denote any variable store, since the assertion does not hold any program variables. Therefore, $\subst[\vec p\rightarrow\vec v],\sto[\pvar{err}\rightarrow v_{err}],\hp'\models\Qerr'[\vec p/\pvvar p]$  and hence $\subst,\sto[\pvar{err}\rightarrow v_{err}],\hp'\models\exsts{\vec p}\Qerr'[\vec p/\pvvar p]$ and \textbf{(H12)}~$\subst,\sto[\pvar{err}\rightarrow v_{err}],\hp'\models\Qerr$ holds.
\item
	(H2) implies $\subst(\vec x)=\esem{\vec \pexp}{\subst,\sto}=\esem{\vec \pexp}{\subst,\sto[\pvar{err}\rightarrow v_{err}]}$ and hence $\subst,\sto[\pvar{err}\mapsto v_{err}]\models \vec \pexp=\vec x$. Therefore, we obtain $\subst,\sto[\pvar{err}\rightarrow v_{err}],\hp'\models\pvar y\doteq \lexp_y\lstar\vec \pexp\doteq\vec x\lstar\Qerr$, which concludes the proof.
\end{itemize}

\casex{While.} The while rule is:

	$$
	\inferrule[\mbox{while}]
{
\frall{i\in\Nat}~\models P_i \Rightarrow \dotin \pexp \Bool \lstar \AssTrue \\\\
\frall{i\in\Nat}~\fsctx \vdash \uquadruplex{P_i \lstar \pexp}{\cmd}{P_{i+1} }{Q_i}
  }
  { \fsctx \vdash \uquadruplex{P_0}{\pwhile{\pexp}{\cmd}}{\neg\pexp\lstar\exsts{i\in\Nat} P_i}{\exsts{i\in\Nat}  Q_i}
  }
  $$

	We prove the under-approximation case for successful termination; the faulting case is proven analogously. Our hypotheses are as follows:

	\begin{description}
	\item[(H1)] $\models(\fictx,\fsctx)$;
	\item[(H2a)] $\frall{i\in\Nat}~\models P_i \Rightarrow \dotin \pexp \Bool \lstar \AssTrue$
	\item[(H2b)] $\frall{i\in\Nat} \fsctx\vdash\uquadruple{P_i \lstar \pexp}{C}{P_{i+1}}{Q_i}$;
	\item[(H3)] $\subst,\sto',\hp' \models \neg\pexp\lstar\exsts{i\in\Nat} P_i$;
	\item[(H4)] $\dish{\hp'}{\hp_\fid}$.
	\end{description}

	Our goal is to show that: $$\exsts{\sto,\hp} \subst,\sto,\hp \models P_0 \land (\sto, \hp \uplus \hp_\fid), \pwhile{\pexp}{C} \baction_\fictx ok: (\sto', \hp' \uplus \hp_\fid)$$
\begin{itemize}
\item (H3) implies $\subst,\sto',\hp' \models\exsts{i\in\Nat}\neg\pexp\lstar P_i$ and therefore there exists some $t\in\Nat$ such that \textbf{(H5)}~$\subst[i \mapsto t],\sto',\hp' \models\neg\pexp\lstar P_i$, which implies \textbf{(H6)}~$\subst[i \mapsto t],\sto',\hp' \models P_i$.

\item If $t=0$, (H5) implies $\subst,\sto',\hp' \models P_0$ and $\esem{\pexp}{\sto'}=\pfalse$. The operational semantics then yields $(\sto',\hp'\uplus\hp_f), \pwhile{\pexp}{\cmd}\baction_{\fictx}\oxok: (\sto',\hp'\uplus\hp_f)$ which implies the desired result.

\item Otherwise, we have $t>0$, and the iterated application of (H2a), (H2b), and the induction hypothesis yields the existence of some $\sto$ and $\hp$ such that \textbf{(H7)}~$\subst[i \mapsto t],\sto,\hp \models P_0\lstar\pexp$ and \textbf{(H8)}~$(\sto,\hp\uplus\hp_f), \cmd^t\baction_{\fictx}\oxok: (\sto',\hp'\uplus\hp_f)$. (H7) also implies $\subst,\sto,\hp \models P_0$. Finally, given (H2a), (H8), and the operational semantics of the while loop, we also obtain the desired $(\sto,\hp\uplus\hp_f), \pwhile{\pexp}{\cmd}\baction_{\fictx}\oxok: (\sto',\hp'\uplus\hp_f)$.

%
%
%
\end{itemize}

Second, we prove the over-approximating case. Our hypotheses are as follows:

\begin{description}
	\item[(H1)] $\models(\fictx,\fsctx)$;
	\item[(H2a)] $\frall{i\in\Nat}~\models P_i \Rightarrow \dotin \pexp \Bool \lstar \AssTrue$
	\item[(H2b)] $\frall{i\in\Nat} \fsctx\vdash\uquadruple{P_i \lstar \pexp}{C}{P_{i+1}}{Q_i}$;

	\item[(H3)] $\subst,\sto,\hp\models P_0$,
	\item[(H4)] $\dish{\hp}{\hp_f}$.
\end{description}

and our goal is to show that:
$$
\begin{array}{l}
\frall{\outcome, \sto',\hp''} (\sto,\hp\uplus \hp_\fid),\pwhile{\pexp}{C}\baction_{\fictx}\outcome: (\sto',\hp'')\implies \\ \qquad \exsts{\hp'}\hp''=\hp'\uplus \hp_\fid\land ((\outcome = \oxok \land \subst,\sto',\hp' \models \neg\pexp\land\exsts{i} P_i) \lor (\outcome = \oxerr \land \subst,\sto',\hp' \models \exists i.~Q_i))
\end{array}
$$

Taking \textbf{(H5)}~$(\sto,\hp\uplus \hp_\fid),\pwhile{\pexp}{C}\baction_{\fictx}\outcome: (\sto',\hp'')$, we have the following:

\begin{itemize}

\item If $\esem{\pexp}{\sto}=\pfalse$, the operational semantics yield that $\outcome=\oxok$, $\sto'=\sto$ and $\hp''=\hp_f\uplus\hp$, which, together with (H3) implies the desired result.

\item If $\esem{\pexp}{\sto}=\ptrue$, (H5) implies that there exists a $t>0$ such that $(\sto,\hp\uplus \hp_f),\cmd^t\baction_{\fictx}\outcome: (\sto',\hp'')$ and (H2a), (H2b), and the inductive hypothesis imply \textbf{(H6)}~$\exsts{\hp'}\hp''=\hp'\uplus \hp_f \land ((\outcome = \oxok \land \subst,\sto',\hp' \models P_t) \lor (\outcome = \oxerr \land \subst,\sto',\hp' \models \bigvee_{i=0}^{t-1} Q_{i}))$ and since the final states $(\sto',\hp'')$ coincide for (H5) and (H6) given the operational semantics of the while loop, we have the desired goal.
\end{itemize}

\casex{Frame.} The frame rule is:
\[
  \inferrule[\mbox{frame}]
  { \fsctx \vdash \uquadruple{P}{C}{\Qok}{\Qerr} \qquad \updt(C) \cap \fv{R} = \emptyset }
  { \fsctx \vdash \uquadruple{P \lstar R}{C}{\Qok \lstar R}{\Qerr \lstar R}}
\]
To prove the soundness of this rule, our hypotheses are that for arbitrary $\fictx$:
\begin{description}
\item[(H1)] $\models(\fictx,\fsctx)$
\item[(H2)] $\fsctx \vdash \uquadruple{P}{C}{\Qok}{\Qerr}$
\item[(H3)] $\updt(C) \cap \fv{R} = \emptyset$
\end{description}
From the inductive hypothesis and \textbf{(H2)}, it follows that $\fsctx \models \uquadruple{P}{C}{\Qok}{\Qerr}$ \textbf{(H4)}.
It then suffices to show that $\fsctx \models \uquadruple{P \lstar R}{C}{\Qok \lstar R}{\Qerr \lstar R}$ holds.
We start off by showing the over-approximating case. For this we assume that for some $\subst, \sto, \hp, \hp_\fid, \outcome, \sto', \hp''$:
\begin{description}
\item[(H5)] $\subst, \sto, \hp \models P \lstar R$
\item[(H6)] $(\sto, \hp \uplus \hp_\fid), \scmd \baction_{\fictx} \outcome: (\sto', \hp'')$
\end{description}
From the definition of the satisfiability relation and \textbf{(H5)}, it follows that there exists some heaps, $\bar{\hp}$ and $\bar{\hp}_r$, such that:
\begin{description}
\item[(H7)] $\hp = \bar{\hp} \dunion \bar{\hp}_r$
\item[(H8)] $\subst, \sto, \bar{\hp} \models P$
\item[(H9)] $\subst, \sto, \bar{\hp}_r \models R$
\end{description}
Letting $\bar{\hp}_\fid = \bar{\hp}_r \dunion \hp_\fid$, from \textbf{(H6)} and the associativity of $\dunion$, it follows that
$(\sto, \bar{\hp} \uplus \bar{\hp}_\fid), \scmd \baction_{\fictx} \outcome: (\sto', \hp'')$ \textbf{(H10)}.
From \textbf{(H4)}, \textbf{(H1)}, \textbf{(H8)} and \textbf{(H10)}, it follows that:
\[
  \outcome \neq \oxm \land \exists \hp'.~\hp'' = \hp' \uplus \bar{\hp}_\fid \land \subst, \sto', \hp' \models Q_\outcome
\]
By applying \cref{LemMods} and \cref{LemPv} to \textbf{(H9)}, we can infer that $\subst, \sto', \bar{\hp}_r \models R$.
Then letting $\bar{\hp}' = \hp' \dunion \bar{\hp}_r$, given the definition of the satisfiability relation,
it follows that $\subst, \sto', \bar{\hp}' \models Q_\outcome \lstar R$. We can then infer that:
\[
  \outcome \neq \oxm \land \exists \hp'.~\hp'' = \hp' \uplus \bar{\hp}_\fid \land \subst, \sto', \hp' \models Q_\outcome \lstar R
\]
as required. We now show the under-approximating case. For this we assume that for some $\subst, \sto', \hp', \hp_\fid, \outcome$:
\begin{description}
\item[(H11)] $\subst, \sto', \hp' \models Q_\outcome \lstar R$
\item[(H12)] $\hp_\fid~\sharp~\hp'$
\end{description}
From the definition of the satisfiability relation and \textbf{(H11)}, it follows that there exists some heaps, $\bar{\hp}'$ and $\bar{\hp}_r'$ such that:
\begin{description}
\item[(H13)] $\hp' = \bar{\hp}' \dunion \bar{\hp}_r'$
\item[(H14)] $\subst, \sto', \bar{\hp}' \models Q_\outcome$
\item[(H15)] $\subst, \sto', \bar{\hp}_r' \models R$
\end{description}
Letting $\bar{\hp}_\fid = \hp_\fid \dunion \bar{\hp}_r'$, from \textbf{(H12)} and \textbf{(H13)}, it  follows that $\bar{\hp}_\fid ~\sharp~\bar{\hp}'$ \textbf{(H16)}. From \textbf{(H4)}, applying \textbf{(H1)}, \textbf{(H14)} and \textbf{(H16)}, it follows that:
\[
  \exsts{\sto,\bar{\hp}}~\subst, \sto, \bar{\hp} \models P~\land~(\sto, \bar{\hp} \uplus \bar{\hp}_\fid), \scmd \baction_\fictx \outcome: (\sto', \bar{\hp}' \uplus \bar{\hp}_\fid)
\]
By applying \cref{LemMods} and \cref{LemPv} to \textbf{(H15)}, it follows that $\subst, \sto, \bar{\hp}_r' \models R$.
Letting $\hp = \bar{\hp} \dunion \bar{\hp}_r'$, by the definition of the satisfiability relation, it follows that:
\[
  \exsts{\sto,\bar{\hp}}~\subst, \sto, \hp \models P \lstar R ~\land~(\sto, \hp \uplus \hp_\fid), \scmd \baction_\fictx \outcome: (\sto', \hp' \uplus \hp_\fid)
\]
as required.


\casex{Equivalence.} The equivalence rule is
$$
\inferrule[\mbox{equiv}]
  { \fsctx \vdash \uquadruple{P'}{C}{\Qok'}{\Qerr'} \quad \models P', \Qok', \Qerr' \Leftrightarrow P, \Qok, \Qerr }
 { \fsctx \vdash \uquadruple{P }{C}{\Qok}{\Qerr}}
 $$
Our hypotheses for OX-soundness are
\begin{description}
	\item[(H1)] $\models(\fictx,\fsctx)$
	\item[(H2)] $\fsctx \models \uquadruple{P'}{C}{\Qok'}{\Qerr'}$
	\item[(H3)] $\models P', \Qok', \Qerr' \Leftrightarrow P, \Qok, \Qerr $
	\item[(H4)] $\subst,\sto,\hp\models P$
	\item[(H5)] $(\sto,\hp\uplus \hp_\fid),C\baction_{\fictx}(\sto',\hp'')$
\end{description}
and we aim to show that
$$
\outcome\neq\oxm\land
\exsts{\hp'} \hp=\hp'\uplus \hp_\fid \land \subst,\sto',\hp'\models Q_{\outcome}
$$
(H3) and (H4) implies $\subst,\sto,\hp\models P'$ and with (H2) and (H5) then implies \textbf{(H6)}~$\outcome\neq\oxm\land
\exsts{\hp'} \hp=\hp'\uplus \hp_\fid \land \subst,\sto',\hp'\models Q'_{\outcome}$. (H2) then implies the desired result.

For the UX-soundness, the hypotheses are
\begin{description}
	\item[(H1)] $\models(\fictx,\fsctx)$
	\item[(H2)] $\fsctx \models \uquadruple{P'}{C}{\Qok'}{\Qerr'}$
	\item[(H3)] $\models P', \Qok', \Qerr' \Leftrightarrow P, \Qok, \Qerr $
	\item[(H4)] $\subst,\sto',\hp'\models Q_{\outcome}$
	\item[(H5)] $\hp'~\sharp~\hp_\fid$
\end{description}
(H2) and (H4) implies $\subst,\sto',\hp'\models Q'_{\outcome}$. (H2) and (H5) implies
$\exsts{\sto,\hp} \subst,\sto,\hp\models P' \land (\sto,\hp\uplus \hp_\fid),C\baction_{\fictx}(\sto',\hp'\uplus \hp_\fid)$. (H3) then implies the desired result.

\casex{Existentials.} The existential rule is:
\[
  \inferrule[\mbox{exists}]
  { \fsctx \vdash \uquadruple{ P}{C}{\Qok}{\Qerr} }
  { \fsctx \vdash \uquadruple{\exists x.\, P}{C}{\exsts x \Qok} {\exsts x \Qerr}}
\]
To prove the soundness of this rule, our hypotheses are that for arbitrary $\fictx$:
\begin{description}
\item[(H1)] $\models(\fictx,\fsctx)$
\item[(H2)] $\fsctx \vdash \uquadruple{ P}{C}{\Qok}{\Qerr}$
\end{description}
Using the inductive hypothesis and \textbf{(H2)}, it follows that \textbf{(H3)}
$\fsctx \models \uquadruple{ P}{C}{\Qok}{\Qerr}$.
It then suffices to show that $\fsctx \models \uquadruple{\exists x.\, P}{C}{\exsts x \Qok} {\exsts x \Qerr}$.
We start off by showing the over-approximating case. To do so, we assume that for some $\subst, \sto, \hp, \hp_\fid, \outcome, \sto', \hp''$:
\begin{description}
\item[(H4)] $\subst, \sto, \hp \models \exists x.\, P$
\item[(H5)] $(\sto, \hp \uplus \hp_\fid), \scmd \baction_{\fictx} \outcome : (\sto', \hp' \uplus \hp_\fid)$
\end{description}
From \textbf{(H4)} and the definition of the satisfiability relation,
it follows that, for some $v$, \\ \textbf{(H6)} $\subst[x \mapsto v], \sto, \hp \models P$ holds.
From \textbf{(H3)}, \textbf{(H1)}, \textbf{(H6)} and \textbf{(H5)}, it follows that:
\[
  \outcome \neq \oxm \land \exists \hp'.~\hp'' = \hp' \uplus \hp_\fid \land \subst[x \mapsto v], \sto', \hp' \models Q_\outcome
\]
This trivially entails:
\[
  \outcome \neq \oxm \land \exists \hp'.~\hp'' = \hp' \uplus \hp_\fid \land \subst, \sto', \hp' \models \exists x. \, Q_\outcome
\]
as required. We now show the under-approximating case.
To do so, we assume that for some $\subst, \sto', \hp', \hp_\fid, \outcome$:
\begin{description}
\item[(H7)] $\subst, \sto', \hp' \models \exists x \, Q_\outcome$
\item[(H8)] $\hp_\fid~\sharp~\hp'$
\end{description}
From \textbf{(H7)} and the definition of the satisfiability relation, it follows that,
for some $v$, \\ \textbf{(H9)} $\subst[x \mapsto v], \sto', \hp' \models Q_\outcome$.
From \textbf{(H3)}, \textbf{(H1)}, \textbf{(H9)} and \textbf{(H8)},
it follows that:
\[
  \exsts{\sto,\hp}~\subst[x \mapsto v], \sto, \hp \models P~\land~(\sto, \hp \uplus \hp_\fid), \scmd \baction_\fictx \outcome: (\sto', \hp' \uplus \hp_\fid)
\]
and consequently:
\[
  \exsts{\sto,\hp}~\subst, \sto, \hp \models \exists x. \, P~\land~(\sto, \hp \uplus \hp_\fid), \scmd \baction_\fictx \outcome: (\sto', \hp' \uplus \hp_\fid)
\]
as required.

\casex{Disjunction.} The disjunction rule is
$$
\inferrule[\mbox{disj}]
  { \fsctx \vdash \uquadruplex{P_1}{C}{\Qok^1}{\Qerr^1} \quad
   \fsctx \vdash \uquadruplex{P_2}{C}{\Qok^2}{\Qerr^2}}
 { \fsctx \vdash \uquadruplex{P_1 \lor P_2}{C}{\Qok^1 \lor \Qok^2}{\Qerr^1 \lor \Qerr^2} }
$$
Our hypotheses for OX-soundness are
\begin{description}
	\item[(H1)] $\models(\fictx,\fsctx)$
	\item[(H2)] $\fsctx \models \uquadruplex{P_1}{C}{\Qok^1}{\Qerr^1}$
	\item[(H3)] $\fsctx \models \uquadruplex{P_2}{C}{\Qok^2}{\Qerr^2}$
	\item[(H4)] $\subst,\sto,\hp\models P_1\lor P_2$
	\item[(H5)] $(\sto,\hp\uplus \hp_\fid), C\baction_{\fictx}\outcome: (\sto',\hp'')$
\end{description}
(H4) implies that $(\subst,\sto,\hp\models P_1)\lor(\subst,\sto,\hp\models P_2)$. If the first case of the disjunct holds, (H2) implies \textbf{(H6a)}~$\outcome\neq\oxm\land
	\exsts{\hp'} \hp''=\hp'\uplus \hp_\fid \land \subst,\sto',\hp'\models Q_{\outcome}^1$. Otherwise, the second case holds and (H3) yields \textbf{(H6b)}~$\outcome\neq\oxm\land
	\exsts{\hp'} \hp''=\hp'\uplus \hp_\fid \land \subst,\sto',\hp'\models Q_{\outcome}^2$.
	The disjunction of (H6a) and (H6b) yields the desired result.

	For the UX-soundness, our hypotheses are
\begin{description}
	\item[(H1)] $\models(\fictx,\fsctx)$
	\item[(H2)] $\fsctx \models \uquadruplex{P_1}{C}{\Qok^1}{\Qerr^1}$
	\item[(H3)] $\fsctx \models \uquadruplex{P_2}{C}{\Qok^2}{\Qerr^2}$
	\item[(H4)] $\subst,\sto',\hp'\models Q^1_{\outcome}\lor Q^2_{\outcome}$
	\item[(H5)] $\hp'~\sharp~\hp_\fid$
\end{description}
(H4) implies $(\subst,\sto',\hp'\models Q^1_{\outcome})\lor(\subst,\sto',\hp'\models Q^2_{\outcome})$. If the first case of the disjunct holds, (H2) yields \textbf{(H6a)}~$\exists{\sto,\hp} (\sto,\hp\uplus \hp_\fid),C \baction_{\fictx}(\sto',\hp'\uplus \hp_\fid) \land \subst,\sto,\hp\models P_1$. Otherwise, the second case of the disjunction holds and (H3) implies \textbf{(H6b)}~$\exists{\sto,\hp} (\sto,\hp\uplus \hp_\fid),C \baction_{\fictx}(\sto',\hp'\uplus \hp_\fid) \land \subst,\sto,\hp\models P_2$. The disjunction of (H6a) and (H6b) yields the desired result.

\casex{Lookup.} The Lookup rule is
$$
\inferrule[\mbox{lookup}]
      { \pvar x \notin \pv{\pexp'} \qquad \subst \defeq [\pexp'/ \pvar{x}]}
    { \fsctx \vdash \utripleq{\pvar{x} \doteq \pexp'  \lstar \pexp \mapsto \pexp_1}{\pderef{\pvar{x}}{\pexp}}{\dotin{\pexp'}{\Val} \lstar \pvar{x} \doteq \pexp_1\subst \lstar \pexp\subst \mapsto   \pexp_1 \subst} }
$$
To prove its OX-soundness, the hypotheses are:
\begin{description}
	\item[(H1)] $\models(\fictx,\fsctx)$
	\item[(H2)] $\pvar x\notin\pv{\pexp'}$
	\item[(H4)] $\subst,\sto,\hp\models\pvar{x} \doteq \pexp'  \lstar \pexp \mapsto \pexp_1$
	\item[(H5)] $(\sto,\hp\uplus \hp_\fid),\passign{\pvar x}{[\pexp]}\baction_{\fictx}\outcome: (\sto',\hp'')$
\end{description}
and we aim to show
$$
	\outcome\neq\oxm\land
	\exsts{\hp'} \hp''=\hp'\uplus \hp_\fid \land \subst,\sto',\hp'\models \dotin{\pexp'}{\Val} \lstar \pvar{x} \doteq \pexp_1[\pexp'/ \pvar{x}] \lstar \pexp[\pexp'/ \pvar{x}] \mapsto   \pexp_1 [\pexp'/ \pvar{x}]
$$
(H5) yields \textbf{(H6)}~$\outcome\neq\oxm$, \textbf{(H7)}~$\esem{\pexp}{\sto}=v$ and \textbf{(H8)}~$\sto'=\sto[\pvar x\rightarrow v]$ and \textbf{(H9)}~$\hp''=\hp\uplus \hp_\fid$.
Chosing \textbf{(10)}$\hp'=\hp$, (H4) and (H8) yield \textbf{(H11)}~$\subst,\sto'\models \pexp'\in\Val$.
(H4), (H9), (H10) and (H11) imply the desired result.

For UX-soundness, our hypotheses are:
\begin{description}
	\item[(H1)] $\models(\fictx,\fsctx)$
	\item[(H2)] $\pvar x\notin\pv{\pexp'}$
	\item[(H3)] $\subst,\sto',\hp'\models\dotin{\pexp'}{\Val} \lstar \pvar{x} \doteq \pexp_1[\pexp'/ \pvar{x}] \lstar \pexp[\pexp'/ \pvar{x}] \mapsto   \pexp_1 [\pexp'/ \pvar{x}]$
	\item[(H4)] $\hp'~\sharp~\hp_\fid$
\end{description}
and we aim to show:
\[
\exists \sto,\hp. \, (\sto,\hp\uplus \hp_\fid),\passign{\pvar x}{[\pexp]}\baction_{\fictx}(\sto',\hp'\uplus \hp_\fid) \land \subst,\sto,\hp\models \pvar{x} \doteq \pexp'  \lstar \pexp \mapsto \pexp_1
\]
Letting $v = \esem{\pexp'}{\subst,\sto'}$, $\sto = \sto'[\pvar{x} \mapsto v]$ and $\hp = \hp'$, then:
\[
  (\sto,\hp\uplus \hp_\fid),\passign{\pvar x}{[\pexp]}\baction_{\fictx}(\sto',\hp'\uplus \hp_\fid)
\]
holds. Then given (H2), by applying \cref{LemSubst}, it is clear that $v = \esem{\pexp'}{\subst,\sto}$,
and therefore $\subst,\sto\models \pvar{x} \doteq \pexp'$. Finally, similarly,
$\subst,\sto,\hp\models \pexp \mapsto \pexp_1$, from which we can reach our goal by the definition
of the satisfiability relation.

\casex{Lookup-err-val.} The Lookup-err-val rule is
\[
  \inferrule[\mbox{lookup-err-val}]
  {\eerr \defeq [``\mathsf{ExprEval}", \stringify {\pexp}]}
  { \fsctx \vdash \utripleerr{\pvar x \doteq \pexp' \lstar \pexp \notdotint \Val}{\pderef{\pvar{x}}{\pexp}}{\Qerr} }
\]
where $\Qerr = \pvar x \doteq \pexp' \lstar \pexp \notdotint \Val \lstar \pvar{err} = \eerr$.
To prove OX-soundness, the hypothese are:
\begin{description}
\item[(H1)] $\models(\fictx,\fsctx)$
\item[(H2)] $\subst,\sto,\hp\models\pvar{x} \doteq \pexp'  \lstar \pexp \notdotint \Val$
\item[(H3)] $(\sto,\hp\uplus \hp_\fid),\passign{\pvar x}{[\pexp]}\baction_{\fictx}\outcome: (\sto',\hp'')$
\end{description}
It then suffices to show that $\subst,\sto',\hp''\models\pvar{x} \doteq \pexp'  \lstar \pexp \notdotint \Val \lstar \pvar{err} = \eerr$.
Given (H2) and the definition satisfiability relation, it follows that $\esem{\pexp}{\sto,\hp} \notin \Val$, and therefore $\esem{\pexp}{\sto,\hp} = \undefd$. From this we can infer that the only rule from the big-step operational semantics that
can apply is:
\[
  \infer{
    \sthreadp{ \sto }{ \hp }, \pderef{\pvar x}{\vexp} \baction_{\fictx} {\oxerr} : \sthreadp{ \sto_{\oxerr} }{ \hp }
  }{
    \begin{array}{c}
      \esem{\vexp}{\sto} = \undefd \\ \verr = [``\mathsf{ExprEval}", \stringify {\pexp}]
    \end{array}
  }
\]
From this, we can infer that $\hp'' = \hp$ and $\sto' = \sto[\pvar{err} \mapsto \verr]$.
From (H2) and the definition of the satisfiability relation, it then follows that
$\subst, \sto', \hp'' \models \Qerr$ as required.

For UX-soundness, our hypotheses are:
\begin{description}
\item[(H1)] $\models(\fictx,\fsctx)$
\item[(H2)] $\subst, \sto', \hp' \models \Qerr$
\item[(H3)] $\hp_\fid~\sharp~\hp'$
\end{description}
It then suffices to show that for some $\sto, \hp$:
\[
  \subst, \sto, \hp \models \pvar{x} \doteq \pexp'  \lstar \pexp \notdotint \Val~\land~(\sto, \hp \uplus \hp_\fid), \scmd \baction_\fictx \outcome: (\sto', \hp' \uplus \hp_\fid)))
\]
Letting $\sto = \sto' \setminus \set{\pvar{err}}$ and $\hp = \hp'$, from (H2) and the definition of the satisfiability relation, it follows that $\subst, \sto, \hp \models \pvar{x} \doteq \pexp'  \lstar \pexp \notdotint \Val$ and by applying the same big-step semantics rule as in the OX case, we derive the second starjunct of our goal as required.

\casex{Lookup-err-use-after-free.} The Lookup-err-use-after-free rule is
\[
  \inferrule[\mbox{lookup-err-use-after-free}]
{\eerr \defeq [``\mathsf{UseAfterFree}", \stringify{\pexp}, \pexp]}
{ \fsctx \vdash \utripleerr{\pvar x \doteq \pexp' \lstar \pexp \mapsto \cfreed }{\pderef{\pvar{x}}{\pexp}}{\Qerr} }
\]
where $\Qerr = \pvar x \doteq \pexp' \lstar \pexp \mapsto \cfreed \lstar \pvar{err} = \eerr$.
To prove OX-soundness, the hypothese are:
\begin{description}
\item[(H1)] $\models(\fictx,\fsctx)$
\item[(H2)] $\subst,\sto,\hp\models\pvar{x} \doteq \pexp'  \lstar \pexp \mapsto \cfreed$
\item[(H3)] $(\sto,\hp\uplus \hp_\fid),\passign{\pvar x}{[\pexp]}\baction_{\fictx}\outcome: (\sto',\hp'')$
\end{description}
It then suffices to show that $\subst,\sto',\hp''\models\pvar{x} \doteq \pexp'  \lstar \pexp \mapsto \cfreed \lstar \pvar{err} = \eerr$.
Given (H2), we can infer that $\hp(\esem{\pexp}{\sto,\hp}) = \cfreed$. From this we can infer that the only rule from the big-step operational semantics that can apply is:
\[
  \infer{
    \sthreadp{ \sto }{ \hp }, \pderef{\pvar x}{\vexp}
    \baction_{\fictx} {\oxerr} : \sthreadp{ \sto_{\oxerr} }{ \hp }
  }{
    \begin{array}{c}\esem{\vexp}{\sto} = n  \quad \hp(n) = \cfreed \\ \verr = [``\mathsf{UseAfterFree}", \stringify{\pexp}, n]\end{array}
  }
\]
From this, we can infer that $\hp'' = \hp$ and $\sto' = \sto[\pvar{err} \mapsto \verr]$.
From (H2) and the definition of the satisfiability relation, it then follows that
$\subst, \sto', \hp'' \models \Qerr$ as required. For UX-soundness, our hypotheses are:
\begin{description}
\item[(H1)] $\models(\fictx,\fsctx)$
\item[(H2)] $\subst, \sto', \hp' \models \Qerr$
\item[(H3)] $\hp_\fid~\sharp~\hp'$
\end{description}
It then suffices to show that for some $\sto, \hp$:
\[
  \subst, \sto, \hp \models \pvar{x} \doteq \pexp'  \lstar \pexp \mapsto \cfreed ~\land~(\sto, \hp \uplus \hp_\fid), \scmd \baction_\fictx \outcome: (\sto', \hp' \uplus \hp_\fid)))
\]
Letting $\sto = \sto' \setminus \set{\pvar{err}}$ and $\hp = \hp'$, from (H2) and the definition of the satisfiability relation, it follows that $\subst, \sto, \hp \models \pvar{x} \doteq \pexp'  \lstar \pexp \mapsto \cfreed$ and by applying the same big-step semantics rule as in the OX case, we derive the second conjunct of our goal as required.

\casex{New.} The New rule is
$$
\inferrule[\mbox{new}]
    { \pvar{x}  \notin \pv{\pexp'} \quad \subst \defeq [\pexp'/\pvar x]}
          { \fsctx \vdash \utripleok{\pvar x \doteq \pexp' \lstar \pexp \dotint \Nat}{\palloc{\pvar{x}}{\pexp}}{\pexp' \dotint \Val \lstar \bigoast{_{0 \le i < \pexp\subst}} (( \pvar x + i) \mapsto \nil)} }
$$
For the OX-soundness, our hypotheses are
\begin{description}
	\item[(H1)] $\models(\fictx,\fsctx)$
	\item[(H2)] $\pvar x\notin\pv{\pexp'}$
	\item[(H3)] $\fsctx \vdash \utripleok{\pvar x \doteq \pexp' \lstar \pexp \dotint \Nat}{\palloc{\pvar{x}}{\pexp}}{\pexp' \dotint \Val \lstar \bigoast{_{0 \le i < \pexp[\pexp'/\pvar x]}} (( \pvar x + i) \mapsto \nil)}$
	\item[(H4)] $\subst,\sto,\hp\models \pvar x \doteq \pexp' \lstar \pexp \dotint \Nat$
	\item[(H5)] $(\sto,\hp\uplus \hp_\fid), \palloc{\pvar x}{\pexp}\baction_{\fictx}\outcome: (\sto',\hp'')$
\end{description}
(H5) implies:
\begin{description}
	\item[(H6)] $\outcome\neq\oxm$
	\item[(H7)] $\esem{\pexp}{\sto}=n$
	\item[(H8)] $\frall{i\in\{0,...,n-1\}} n'+i\notin\dom(\hp\uplus \hp_\fid)$
	\item[(H9)] $\sto'=\sto[\pvar x\rightarrow n']$
	\item[(H10)] $\hp''=(\hp\uplus \hp_\fid)[n'\mapsto\nil]\hdots[n'+n-1\mapsto\nil]$
\end{description}
Defining \textbf{(H11)}~$\hp'=\hp[n'\mapsto\nil]\hdots[n'+n-1\mapsto\nil]$ yields with (H10) that \textbf{(H12)}~$\hp''=\hp'\uplus \hp_\fid$.
(H2), (H4) and (H9) implies \textbf{(H13)}~$\subst,\sto'\models \pexp'\in\Val\lstar\pvar x=n'$ and (H11) then implies $\subst,\sto',\hp'\models \pexp' \dotint \Val \lstar \bigoast{_{0 \le i < \pexp[\pexp'/\pvar x]}} ( \pvar x + i) \mapsto \nil$, which is the desired result.

For the UX direction, the hypotheses are
\begin{description}
	\item[(H1)] $\models(\fictx,\fsctx)$
	\item[(H2)] $\pvar x\notin\pv{\pexp'}$
	\item[(H3)] $\fsctx \vdash \utripleok{\pvar x \doteq \pexp' \lstar \pexp \dotint \Nat}{\palloc{\pvar{x}}{\pexp}}{\pexp' \dotint \Val \lstar \bigoast{_{0 \le i < \pexp[\pexp'/\pvar x]}} (( \pvar x + i) \mapsto \nil)}$
	\item[(H4)] $\subst,\sto',\hp'\models \pexp' \dotint \Val \lstar \bigoast{_{0 \le i < \pexp[\pexp'/\pvar x]}} (( \pvar x + i) \mapsto \nil)$
	\item[(H5)] $\hp'~\sharp~\hp_\fid$
\end{description}
(H4) implies that \textbf{(H5)}~$\pvar x\in\dom(\sto')$ and we define
\begin{description}
	\item[(H6)] $n'=\sto'(\pvar x)$
	\item[(H7)] $n=\esem{\pexp[\pexp'/\pvar x]}{\subst,\sto'}$
	\item[(H8)] $\hp=\hp'|_{d}$, where $d=\dom(\hp')\backslash\{\sto(\pvar x),\hdots,\sto(\pvar x+n)\}$
	\item[(H9)] $\sto=\sto'[\pvar x\rightarrow v]$ where $v=\esem{\pexp'}{\subst,\sto'}$
\end{description}
(H5) and (H8) impliy that \textbf{(H10)}~$n'+i\notin\dom(\hp\uplus \hp_\fid)\forall i\in\{0,...,n-1\}$ and \textbf{(H11)}~$\hp'\uplus \hp_\fid=(\hp\uplus \hp_\fid)[n'\mapsto\nil]\hdots[n'+n-1\mapsto\nil]$ and (H9) implies \textbf{(H12)}~$\sto'=\sto[\pvar x\rightarrow n']$.
 (H2), (H7) and (H9) imply \textbf{(H13)}~$\esem{\pexp}{\subst,\sto}=n$.
(H10),(H11),(H12) and (H13)imply
$$
(\sto,\hp\uplus \hp_\fid), \palloc{\pvar x}{\pexp}\baction_{\fictx}\outcome: (\sto',\hp'')
$$
(H2) and (H9) imply $\esem{\pexp}{\sto}=v$, which with (H7) and (H13) implies $\subst,\sto\models \pvar x\doteq \pexp'\lstar \pexp\dotint\Nat$.

(H9), (H10), (H11) and (H13) imply the desired result.

\casex{Free.} The free rule is
$$
\inferrule[\mbox{free}]
      {}
   { \fsctx \vdash \utripleok{\pexp \mapsto \pexp'}{\pdealloc{\pexp}}{\dotin{\pexp'}{\Val}\lstar \pexp \mapsto \cfreed} }
   $$
For OX-Soundness, the hypotheses are
\begin{description}
	\item[(H1)] $\models(\fictx,\fsctx)$
	\item[(H2)] $\fsctx \vdash \utripleok{\pexp \mapsto \pexp'}{\pdealloc{\pexp}}{\dotin{\pexp'}{\Val}\lstar \pexp \mapsto \cfreed}$
	\item[(H3)] $\subst,\sto,\hp\models \pexp \mapsto \pexp'$
	\item[(H4)] $(\sto,\hp\uplus \hp_\fid), \pdealloc{\pexp}\baction_{\fictx}(\sto',\hp'')$
\end{description}
(H4) implies
\begin{description}
	\item[(H5)] $\esem{\pexp}{\sto}=n$
	\item[(H6)] $(\hp\uplus \hp_\fid)(n)\in\Val$
	\item[(H7)] $\sto=\sto'$
	\item[(H8)]	$\hp'' = (\hp\uplus \hp_\fid)[n\mapsto\cfreed]$
\end{description}
(H4) and (H5) impliy that $n\in\dom(\hp)$, which with (H8) implies that $\hp''= \hp[n\mapsto\cfreed]\uplus \hp_\fid$. (H3) and (H7) imply $\subst,\sto'\models \pexp'\in\Val$.
(H5) and (H7) imply $\esem{\pexp}{\sto'}=n$ and
defining $\hp'=\hp[n\mapsto\cfreed]$, we obtain $\subst,\sto',\hp'\models \pexp'\in\Val \lstar \pexp\mapsto\cfreed$, which is the desired result.

For the UX direction, our hypotheses are
\begin{description}
	\item[(H1)] $\models(\fictx,\fsctx)$
	\item[(H2)] $\fsctx \vdash \utripleok{\pexp \mapsto \pexp'}{\pdealloc{\pexp}}{\dotin{\pexp'}{\Val}\lstar \pexp \mapsto \cfreed}$
	\item[(H3)] $\subst,\sto',\hp'\models \pexp'\dotint\Val\lstar \pexp\mapsto\cfreed$
	\item[(H4)] $\hp'~\sharp~\hp_\fid$
\end{description}
Defining $\sto=\sto'$, (H3) yields that $n=\esem{\pexp}{\subst,\sto'}=\esem{\pexp}{\subst,\sto}$ for some $n\in\Nat$. Defining $\hp=\hp'[n\mapsto v]$ for $v=\esem{\pexp'}{\subst,\sto'}=\esem{\pexp'}{\subst,\sto}$, we obtain $\hp'=\hp[n\mapsto\cfreed]$ and therefore $\hp'\uplus \hp_\fid=(\hp\uplus \hp_\fid)[n\mapsto\cfreed]$. The operational semantics of $\mathtt{free}$ then  yields
$$
	(\sto,\hp\uplus \hp_\fid), \pdealloc{\pexp}\baction_{\fictx}(\sto',\hp'\uplus \hp_\fid)
$$
and also obtain $\subst,\sto,\hp\models \pexp\mapsto \pexp'$, which is the desired result.

\end{proof}

%% file: sections/app-scott.tex
\newpage
\section{Basics of Scott Induction}\label{apdx:scott}

The second soundness statement that needs to be proven for ESL is that well-formed environments are valid. This requires reasoning about the use of function specifications in the context of the environment extension rule.

In particular, the use of specifications of non-recursive functions is trivially sound. For recursive functions that always terminate, soundness can be proven by transfinite induction, while establishing a measure on the function pre-conditions and allowing recursive use of specifications only if they have a strictly lower measure. Without this requirement, we could  prove an unsound specification $\utripleok{\emp}{\mathtt{\fid}()}{\pvar{ret}\doteq 42}$ for the function $\mathtt{\fid}()\{ \passign{\pvar x}{\mathtt{\fid}()}; \preturn \pvar x \}$, which does not hold since $\mathtt{\fid}$ never terminates and the (satisfiable) post-condition $\pvar {ret}\doteq 42$ implies the existence of at least one terminating execution. This soundness issue does not arise in over-approximating logics, since, due to the meaning of triples, a satisfiable post-condition does not imply the existence of terminating traces. In these logics, it is always sound to apply a specification to prove itself.

However,
for recursive functions with non-terminating branches due to infinite recursion, we also have to be able to allow recursive use of specifications whose measure does not decrease, and the tool to handle such use is a form of fixpoint induction called Scott induction (see, e.g., Winskel~\cite{scott}), which would normally be the tool for also proving soundness of well-formed environments in over-approximating logics.


%
%

In the following, we give the relevant Scott-induction-related definitions (from \cite{scott}), together with an instantiation that will be applied to prove soundness of well-formed environments in Appendix~\ref{apdx:envsound}.

%
%
%
%
%

\begin{definition}[Domain]
	 A partially ordered set $(D,\sqsubseteq)$ is a {\it domain}, iff
	 \begin{description}
	 \item[(D1)] $\exsts{\bot\in D}\frall{d\in D} \bot\sqsubseteq d $ (least element)
	 \item[(D2)] $\frall{(d_n)_{n\in\Nat}\subseteq D}(\frall{i\in\Nat} d_i\sqsubseteq d_{i+1})~\implies~\bigsqcup\limits_{n\in\Nat}d_n\in D$ (chain-closedness)
	 \end{description}
	 where $\sqcup_{n\in\Nat} d_n$ denotes the least upper bound or the supremum of the set $\{d_n~|~n\in\Nat\}$ with respect to $\sqsubseteq$.
\end{definition}

\begin{definition}[Admissible Subset]
	Given a domain $(D, \sqsubseteq)$ with least element $\bot$, a subset $\sto\subseteq D$ is called admissible, iff
	\begin{description}
	 \item[(S1)] $\bot\in \sto $ (least element)
	 \item[(S2)] $\frall{(\sto_n)_{n\in\Nat}\subseteq \sto}(\frall{i\in\Nat} \sto_i\sqsubseteq \sto_{i+1})~\implies~\bigsqcup\limits_{n\in\Nat}\sto_n\in \sto$  (chain-closedness)
	 \end{description}
\end{definition}

\begin{definition}[Continuity on Domains]
	Assuming two domains $(D,\sqsubseteq_D)$ and $(\pexp,\sqsubseteq_\pexp)$, a function $g:D\longrightarrow \pexp$ is {\it continuous}, iff
	\begin{description}
	\item[(C1)] $\frall{d,d'\in\D} d\po_D d' \Longrightarrow g(d)\po_\pexp g(d')$ (monotonicity)
	\item[(C2)] $\frall{(d_n)_{n\in\Nat}} (\frall{i\in\Nat} d_i\sqsubseteq d_{i+1})~\Rightarrow~ \bigsqcup\limits_{n\in\Nat} g(d_n) = g(\bigsqcup_{n\in\Nat} d_n)$ (supremum-preservation)
	\end{description}
\end{definition}

\begin{theorem}[Least fixpoint]\label{lem:lfpid}
	Given a domain $D$ and a continuous function $g:D\longrightarrow D$, the least fixpoint of $g$, denoted by $\lfp{g}$, has the identity
	$$ \lfp{g}=\bigsqcup_{n\in\Nat} g^n(\bot),$$
	where $\bot$ denotes the least element of $D$ and $g^n$ denotes the $n$-times application of $g$.
\end{theorem}

\begin{theorem}[Scott Induction Principle]\label{theorem:scottinduction}
	Given a domain $D$, an admissible subset $\sto\subseteq D$, and a continuous function $g:D\longrightarrow D$, it holds that
	$$ g(\sto)\subseteq \sto ~\Longrightarrow~ \lfp{g}\in \sto $$
\end{theorem}
Before presenting the instantiation of the Scott induction, we require a pseudo-command $\mathtt{scope}$ which models the function call, and pseudo-commands for non-deterministic choice.
\begin{definition}[The $\mathtt{scope}$ pseudo-command]
	We define a pseudo-command which closely models the behaviour of a function call:
	$$\pscope{\pvvar x, \vec \pexp}{C}{\pvar y, \pexp'}$$
	whose arguments are
	\begin{itemize}
	\item
	a pair $(\pvvar x,\vec \pexp)$ consisting of a list of distinct program variables and a list of expressions, such that both are of the same length,
	\item
	a command $C$ which is to be executed within the "scope",
	\item
	a tuple $(\pvar y, \pexp')$ of a program variable and an expression,
	\end{itemize}
	and whose semantics (eliding expression-evaluation fault cases) is given by
	\[
	\infer{
     \sthreadp{ \sto }{ \hp },
     \pscope{\pvvar x, \vec \pexp}{C}{\pvar y, \pexp'} \baction_{\fictx}
     \sthreadp{\sto [\pvar{y} \storearrow v' ] }{ \hp' }
  }{
    \begin{array}{c}
    \esem{\vec{\vexp}}{\sto} = \vec{v}
     \quad \pv{\scmd} \setminus
      \{\vec{\pvar{x}}\} = \{\vec{\pvar{z}} \}
    \\
 \sto_p  =  \emptyset [ \vec{\pvar{x}} \storearrow \vec{v}] [ \vec{\pvar{z}} \storearrow \nil]
     \quad (\sto_p, \hp), \scmd
      \baction_{\fictx} \sthreadp{ \sto_q }{ \hp' }
      \quad \esem{\pexp'}{\sto_q} =v'
	\end{array} }
	\]
		\[
	\infer{
     \sthreadp{ \sto }{ \hp },
     \pscope{\pvvar x, \vec \pexp}{C}{\pvar y, \pexp'} \baction_{\fictx} err:
     (\sto [\pvar{err} \storearrow v_{err} ],\hp')
  }{
    \begin{array}{c}
    \esem{\vec{\vexp}}{\sto} = \vec{v}
     \quad \pv{\scmd} \setminus
      \{\vec{\pvar{x}}\} = \{\vec{\pvar{z}} \}
    \\
 \sto_p  =  \emptyset [ \vec{\pvar{x}} \storearrow \vec{v}] [ \vec{\pvar{z}} \storearrow \nil]
     \quad (\sto_p, \hp), \scmd
      \baction_{\fictx} err:(\sto_q,\hp')
      \quad \esem{\pvar{err}}{\sto_q} =v_{err}
	\end{array} }
	\]
\end{definition}
\begin{definition}[Non-Deterministic Choice (pseudo-commands)]
	We furthermore add pseudo-commands that arbitrarily pick a command form a given set and executes it:
	\[
	\infer{
     \stt,
     \pchoice{C_1}{C_2} \baction_{\fictx} \outcome:
     \stt'
  }{
	  \stt,
     C_1 \baction_{\fictx} \outcome:
     \stt'
     \lor
     \stt,
     C_2 \baction_{\fictx} \outcome:
     \stt'
     }
     \quad
     \infer{
     \stt,
     \pnatchoice{C_n}{n} \baction_{\fictx} \outcome:
     \stt'
  }{
	 \exsts{m\in\Nat}
	 \stt,
     C_m \baction_{\fictx} \outcome:
     \stt'
     }
	\]
\end{definition}
We will now give some general definitions and lemmas, which will later on be instantiated to prove the soundness of the environment extension via Theorem \ref{theorem:scottinduction}.
\begin{definition}[Greatest-Fixpoint Closure of $\Cmd$]
	We define the greatest-fixpoint closure of the set of commands and pseudo-commands, $\Cmd\cup\{\mathtt{scope},\sqcup,\bigsqcup\}$, as the closure of that set under infinite applications of the command constructors, and denote that closure by $\ccmdplain$.
\end{definition}
\begin{definition}[Behavioural Equivalence on $\ccmdplain$]
Given an arbitrary function implementation context $\gammag$, we define the equivalence relation $\eqg$ on $\ccmdplain$ as
	$$ C_1\eqg C_2 \iff \{(\stt,\stt')\in\State^2 ~|~ \exsts{\outcome}~\stt,C_1\baction_{\gammag}\outcome : \stt'\}=\{(\stt,\stt')\in\State^2 ~|~ \exsts{\outcome}~\stt,C_2\baction_{\gammag} \outcome : \stt'\} $$
	where $C_1, C_2 \in \ccmdplain$, effectively meaning that $\eqg$ relates commands that exhibit the same set of behaviours. We denote the resulting quotient space as $\ccmd$ and the corresponding equivalence class of a command $C$ by $\repr{C}$.
 	This relation yields a partial order, denoted by $\pog$ and defined as:
	$$ C_1\pog C_2 \iff \{(\stt,\stt')\in\State^2 ~|~ \exsts{\outcome}~\stt,C_1\baction_{\gammag}\outcome : \stt'\}\subseteq\{(\stt,\stt')\in\State^2 ~|~ \exsts{\outcome}~\stt,C_2\baction_{\gammag}\outcome : \stt'\} $$
	Furthermore, we define the join operator on commands in $\ccmd$, $\joing$, as the non-deterministic choice,
	lift it to quotient space, overloading notation:
	$$ \repr{C_1}\sqcup\repr{C_2} = \repr{C_1\sqcup C_2}$$
	and generalise it to countably infinitely many commands/equivalence classes in the standard way.
	\end{definition}
The relation $\eqg$ is an equivalence relation as it inherits reflexivity, symmetry and transitivity from the equality relation on sets, and
$\pog$ is a partial order on $\ccmd$ as it inherits transitivity and reflexivity from set inclusion, {while $\eqg$ ensures anti-symmetry.}

Furthermore, note that, by design of the language, we do not have to bring the outcome $\outcome$ into the equivalence relation, as faulting states can be distinguished from successful ones by having the dedicated program variable $\pvar{err}$ in the store, and language errors can be distinguished from the missing resource errors by the value that $\pvar{err}$ holds.

\begin{lemma}[Domain Property]
	For any function implementation context $\gammag$, $(\ccmd,\pog)$ is a domain.
\end{lemma}
\begin{proof} Since we have already argued the partial order property, there are only remaining two properties to show:

\myparagraph{Chain-Closedness}
	For any chain $(\repr{C_n})_{n\in\Nat}\subseteq\ccmd$, its supremum is defined as $\repr{\sqcup(C_n |_{n \in \Nat})}$. Per definition of $\ccmdplain$ we have $\sqcup(C_n |_{n \in \Nat})\in\ccmdplain$ which implies $\repr{\sqcup(C_n |_{n \in \Nat})}\in\ccmd$.

	\myparagraph{Least Element}
	The least element of $\ccmd$, denoted by $\leg$, is the equivalence class of commands which, given the function implementation context $\fictx$, do not terminate on any state. One such representative is the command $C=\pwhile{\true}\pskip$. Since $\{(\stt,\stt')\in\State^2 ~|~ \stt, \pwhile{\true}\pskip\baction_{\gammag}\stt'\}=\emptyset$, we trivially obtain $\leg\pog \repr{C}$, for all $\repr{C}\in\ccmd$.
\end{proof}
%
%

%
\begin{lemma}[Scope and Function Call Equivalence]\label{lem:scopefcallequiv}
	Given a function implementation context $\gammag$ and a function $\fid$ such that $\gammag(\fid) = (\pvvar x, C_\fid, \pexp')$, it holds that
	$$ \pscope{\pvvar x,\vec \pexp}{C_\fid}{\pvar y, \pexp'} \eqg \passign{\pvar y}{\fid(\vec \pexp)}$$
\end{lemma}
\begin{proof}
    We show in detail the case of successful execution; the faulting cases are analogous.
	Let $\gammag(\fid)=(\pvvar x, C_\fid, \pexp')$ and $(\sto,\hp),(\sto',\hp')\in\State$, such that $$(\sto,\hp), \pscope{\pvvar x,\vec \pexp}{C_\fid}{\pvar y, \pexp'}\baction_{\gammag} (\sto',\hp').$$
	The operational semantics of $\mathtt{scope}$ implies:

	\begin{itemize}
	\item $\sto'=\sto[\pvar y\rightarrow v']$
	\item $\esem{\vec \pexp}{\sto}=\vec v$
	\item $\esem{\pexp'}{\sto_q}=v'$
	\item $\sto_p = \emptyset[\pvvar x\rightarrow\vec v][\pvvar z\rightarrow\nil]$
	\item $(\sto_p,\hp), C_\fid\baction_{\gammag}(\sto_q,\hp')$
	\end{itemize}

	Due to the assumption $\gammag(\fid)=(\pvvar x, C_\fid, \pexp')$, we fulfil all conditions in the antecedent of the operational semantics of the function call, and therefore obtain $(\sto,\hp), \passign{\pvar y}{\fid(\vec \pexp)}\baction_{\gammag} (\sto',\hp') $.\\
	Now, let $(\sto,\hp),(\sto',\hp')\in\State$ such that $(\sto,\hp), \passign{\pvar y}{\fid(\vec \pexp)}\baction_{\gammag} (\sto',\hp') $. The operational semantics of the function call implies:

	\begin{itemize}
		\item $\sto'=\sto[\pvar y\rightarrow v']$
		\item $\esem{\vec \pexp}{\sto}=\vec v$
		\item $\esem{\pexp'}{\sto_q}=v'$
		\item $\sto_p = \emptyset[\pvvar x\rightarrow\vec v][\pvvar z\rightarrow\nil]$
		\item $(\sto_p,\hp), C_\fid\baction_{\gammag}(\sto_q,\hp')$
		\item $\gammag(\fid)=(\pvvar x, C_\fid, \pexp')$
	\end{itemize}

	Therefore, we fulfil all conditions in the antecedent of the operational semantics of $\mathtt{scope}$, and therefore obtain
	\newline
	$(\sto,\hp), \pscope{\pvvar x,\vec \pexp}{C_\fid}{\pvar y, \pexp'}\baction_{\gammag}(\sto',\hp') $. 
\end{proof}
In the following, let $C_i$ denote the implementation of the function $\fid_i$, and $C^i$ denote the $i$-th component of a vector $C$.

\begin{definition}[Function Call Substitution]
	Given a command $\bar C\in\Cmd$, a vector of $n$ commands $C=(C^1,\hdots,C^n)\in\nccmd$, a vector of $n$ functions $F=(\fid_1,\hdots,\fid_n)$ and a function implementation context $\fictx$, such that $F\subseteq\dom(\fictx)$, we define a function call substitution $\fcallsub{\bar C}{C,\gammag,F}$ recursively on the structure of $\bar C$, with $\gammag(\fid_i) = ( \pvvar x_i, -, \pexp_i )$:
	\begin{itemize}
		\item $\fcallsub{(\pifelse{B}{C_1}{C_2})}{C,\gammag,F}:= \pifelse{B}{\{\fcallsub{C_1}{C,\gammag,F}\}}{\{\fcallsub{C_2}{C,\gammag,F}\}}$
		\item $ \fcallsub{(\pwhile{B}{\bar C})}{C,\gammag,F}:= \pwhile{B}{\{ \fcallsub{\bar C}{C,\gammag,F} \}} $
		\item $ \fcallsub{(C_1;C_2)}{C,\gammag,F}:= \fcallsub{C_1}{C,\gammag,F}; \fcallsub{C_2}{C,\gammag,F} $
		\item $\fcallsub{(\passign{\pvar y}{\mathtt{g}(\vec \pexp)})}{C,\gammag,F}:= \left\{
		\begin{array}{l l}
			\pscope{\pvvar x_i,\vec \pexp}{C^i}{\pvar y, \pexp_i} & \text{if } \fid_i=\mathtt{g}\\
			\passign{\pvar y}{\mathtt{g}(\vec \pexp)} & \text{otherwise},
		\end{array} \right.
		$
			\item $\fcallsub{\bar C}{C,\gammag,F}:= \bar C $, for all other $\bar C$.
	\end{itemize}
\end{definition}

\begin{lemma}[Substitution Preserves $\eqg$]
	Given $I=\{1,...,n\}$, $F=(\fid_1,...,\fid_n)$, and $\gammag$ such that $\frall{i\in I}\gammag(\fid_i)=(-,C_i,-)$, $C=(C_1,...,C_n)$, and $i\in I$, it holds that
	$$ C_i \eqg \fcallsub{C_i}{C,\gammag,F} $$
\end{lemma}
\begin{proof}
	We prove the statement by structural induction on $C_i$. Per definition of function implementation contexts, $C_i\in\Cmd$ and is therefore finite. The only non-trivial cases are the structural commands and the function call on $\fid_i$, as the substitution is the identity otherwise.

	\myparagraph{If-Else} $C_i=\pifelse{\pexp}{\ctrue}{\cfalse}$. Let $\stt,\stt'\in\State$ such that $\stt,\pifelse{\pexp}{\ctrue}{\cfalse}\baction_{\gammag}\stt'$, and let $\stt = (\sto, \hp)$.
	Then, the operational semantics implies
	\[
	\begin{array}{l l}
	& \big(\esem{\pexp}{\sto}=\ptrue\land(\sto,\hp),\ctrue\baction_{\gammag}\stt'\big)\lor\big( \esem{\pexp}{\sto} ={\pfalse}\land (\sto,\hp),\cfalse\baction_{\gammag}\stt'\big) \\
	\stackrel{\text{(IH)}}{\Leftrightarrow} & \big(\esem{\pexp}{\sto}=\ptrue\land(\sto,\hp),\fcallsub{\ctrue}{C,\gammag,F}\baction_{\gammag}\stt'\big)\lor\big( \esem{\pexp}{\sto} = {\pfalse}\land (\sto,\hp),\fcallsub{\cfalse}{C,\gammag,F}\baction_{\gammag}\stt'\big)
\end{array}
	\]
	which is equivalent to
	$ (\sto,\hp),\pifelse{\pexp}{\fcallsub{\ctrue}{C,\gammag,F}}{\fcallsub{\cfalse}{C,\gammag,F}}\baction_{\gammag}\stt'$. The faulting case  is proven analogously to the successful case, and the while loop and the sequencing are proven analogously to if-else.

	\myparagraph{Function Call} We have to show that
	$$ (\passign{\pvar y}{\fid_i(\vec \pexp)})\eqg \pscope{\pvvar x_i,\vec \pexp}{C_i}{\pvar y, \pexp_i} $$
	where $\gammag(fi)=(\pvvar x_i,C_i,\pexp_i),$ but this holds directly due to Lemma \ref{lem:scopefcallequiv}.
\end{proof}
Before moving on to the Scott instantiation, we define a notion of {(recursion)} depth, which keeps track of the maximum number of nested function calls 
during the execution of commands.
\begin{definition}[Depth]
	Given a command $C\in\Cmd$, a vector of functions $F=(\fid_1,...,\fid_n)$ and a derivation $\stt,C\baction_{\gammag}\outcome:\stt'$, we define $\depthF{\stt,C\baction_{\gammag}\outcome:\stt'}$ inductively on the structure of big-step derivations of $C$ as follows, noting that we extend the notion of a maximal element to the empty set by defining it to be zero:

	\myparagraph{If-Else}
		$
		\begin{array}{l}\depthF{(\sto,\hp),\pifelse{\pexp}{\ctrue}{\cfalse}\baction_{\gammag}\outcome:\stt'}~\defeq \\ \qquad
				 \max\left(
			\begin{array}{l}
				\max\{
			\depthF{(\sto,\hp),\ctrue\baction_{\gammag}\outcome:\stt'} \mid \esem{\pexp}{\sto}=\ptrue \}, \\
\max\{
			\depthF{(\sto,\hp),\cfalse\baction_{\gammag}\outcome:\stt'} \mid \esem{\pexp}{\sto}=\pfalse \} \\
			\end{array}
		\right)
		\end{array}$

	\myparagraph{Sequence}
		$
		\begin{array}{l}
		\depthF{\stt,C_1;C_2\baction_{\gammag}\outcome:\stt'} \defeq \\
		\qquad\max\left(
		\begin{array}{l}
		\max\{
			\depthF{\stt,C_1\baction_{\gammag}\bar\stt},~
			\depthF{\bar\stt,C_2\baction_{\gammag}\outcome:\stt'} \mid \\ \qquad\qquad\stt,C_1\baction_{\gammag}\bar\stt \land \bar\stt,C_2\baction_{\gammag}\outcome:\stt'
		\}, \\
		\max\{
			\depthF{\stt,C_1\baction_{\gammag}\outcome:\stt'} \mid \outcome = \oxerr/\oxm \land \stt,C_1\baction_{\gammag}\outcome:\stt'
		\}
		\end{array}\right)
		\end{array}
		$

	\myparagraph{While}
		$
		\begin{array}{l}
		\depthF{(\sto,\hp),\pwhile{\pexp}{C}\baction_{\gammag}\outcome:\stt'}~\defeq\\
		\qquad
		 \max\left(
			\begin{array}{l}
				\max\{
			\depthF{\stt,C\baction_{\gammag}\bar\stt},~
			\depthF{\bar\stt,\pwhile{\pexp}{C}\baction_{\gammag}\outcome:\stt'}~| \\
			\qquad\qquad\esem{\pexp}{\sto}=\ptrue \land \stt,C\baction_{\gammag}\bar\stt \land \bar\stt,\pwhile{\pexp}{C}\baction_{\gammag}\outcome:\stt'\}, \\
			\max\{
			\depthF{\stt,C\baction_{\gammag}\outcome:\stt'} \mid \outcome = \oxerr/\oxm \land  \esem{\pexp}{\sto}=\ptrue \land \stt,C\baction_{\gammag}\outcome:\stt' \}
			\end{array}
		\right)
		\end{array}
		$

	\myparagraph{Function Call}
		$\depthF{(\sto,\hp),\passign{\pvar y}{\mathtt{g}(\vec \pexp),\baction_{\gammag}\outcome:(\sto',\hp')}} \defeq 0, \quad \text{if $\frall{i\in I} \mathtt{g}\neq \fid_i$}$
		$$
		\begin{array}{l}
		\depthF{(\sto,\hp),\passign{\pvar y}{\fid_i(\vec \pexp),\baction_{\gammag}\outcome:(\sto',\hp')}} \defeq \\
		\qquad\max\left(\begin{array}{l}
			\max\{ 1+ \depthF{(\sto_p,\hp),C_i\baction_{\gammag} (\sto_q,\hp')} \mid \\
			\qquad \outcome = \oxok, \esem{\vec \pexp}{\sto} = \vec v, (\sto_p,\hp), C_\fid\baction_{\gammag} (\sto_q,\hp'), \esem{\pexp_i}{\sto_q} = v'\}, \\
			\max\{ 1+ \depthF{(\sto_p,\hp),C_i\baction_{\gammag} \outcome: (\sto_q,\hp')} \mid \\
			\qquad \outcome = \oxerr/\oxm, \esem{\vec \pexp}{\sto} = \vec v, (\sto_p,\hp), C_\fid\baction_{\gammag} \outcome: (\sto_q,\hp') \}, \\
			\max\{ 1+ \depthF{(\sto_p,\hp),C_i\baction_{\gammag} \outcome : (\sto_q,\hp')} \mid \\
			\qquad \outcome = \oxerr, \esem{\vec \pexp}{\sto} = \vec v, (\sto_p,\hp), C_\fid\baction_{\gammag} (\sto_q,\hp'), \esem{\pexp_i}{\sto_q} = \undefd \}, \\
		\end{array}\right)
		\end{array}
		$$
		where $\gammag(\fid_i) = (\pvvar x_i, C_i, \pexp_i)$, and $\sto_p$ and $\sto'$ are defined as in the operational semantics of the function call.

			\myparagraph{Remaining Commands}
		$\depthF{\stt,C\baction_{\gammag}\outcome:\stt'} \defeq 0$.

\end{definition}

\subsection{1-dimensional Scott Instantiation}

The env-extend rule allows us to add a set of $n$ functions to a given valid environment.
Soundness of this rule is proven in Appendix~\ref{apdx:envsound} through transfinite induction. In each of the cases (zero, successor ordinal, limit ordinal), a Scott induction is required.
In \S\ref{sec:ndimscott}, we present and prove the Scott induction required to show soundness of the env-extend rule.
Here, we present the Scott induction as required for the case where only {\it one} function is added to a given environment at a time.

The general proof, as presented in \S\ref{sec:ndimscott} evolves naturally from and relies heavily on this simpler case, while introducing heavier notation.
To minimize clutter in later definitions, we introduce the over-approximation quadruple $\quadruple{P}{C}{\Qok}{\Qerr}$ and define its notion of validity.

\begin{definition}[OX-Validity]
	Given an OX-quadruple $\quadruple{P}{C}{\Qok}{\Qerr}$ and a function we define for an arbitrary implementation context $\fictx$
\[
\begin{array}{l}
	\fictx\models \quadruple{P}{C}{\Qok}{\Qerr} \Longleftrightarrow\\
	\hspace*{1cm}
	\frall{\subst,\sto,\hp,\outcome,\sto',\hp'',\hp_\fid} \subst,\sto,\hp\models P \\
	\hspace*{1cm}\quad\quad~\Longrightarrow~ (\sto,\hp\uplus \hp_\fid),C\baction_{\fictx}\outcome: (\sto',\hp'') \\
	\hspace*{1cm}\quad\quad\quad~\Longrightarrow~ \outcome\neq\oxm \land \exsts{\hp'} \hp''=\hp'\uplus \hp_\fid \land \subst,\sto',\hp'\models Q_{\outcome}
\end{array}
\]
and for an arbitrary specification context $\fsctx$
\[
\begin{array}{l}
	\fsctx\models \quadruple{P}{C}{\Qok}{\Qerr} \Longleftrightarrow \\
	\hspace*{1cm}
	\frall{\fictx} \models(\fictx,\fsctx)\Longrightarrow \fictx\models \quadruple{P}{C}{\Qok}{\Qerr}
\end{array}
\]
\end{definition}

In the following, we will also write $\triple{P}{C}{\bigq}$ as a shorthand for $\quadruple{P}{C}{\Qok}{\Qerr}$. Onward, we assume the following:

\begin{description}
\item[(A1)] a valid environment, $\models(\fictx,\fsctx)$;
\item[(A2)] a function $\fid(\pvvar x) \{ C_\fid; \preturn{\pexp'}\}$ that is not in the domain of $\fictx$;
\item[(A3)] an arbitrary element $\alpha\in\ord$;
\item[(A4)]
a set of terminating (external) specifications for $\fid$, $\{ \uspeconecase{\tpre(\beta)}{\bigq(\beta)} ~|~ \beta<\alpha \}$;
\item[(A5)] a set of non-terminating (external) specifications for $\fid$, $\{ \uspeconecase{\ntpre(\beta)}{\AssFalse} ~|~ \beta\leq\alpha \}$;
\item[(A6)] an extension of $\fictx$ with $\fid$: $\fictx' \defeq \fictx[\fid\mapsto(\pvvar x, C_\fid, \pexp')]$; and
\item[(A7)] an extension of $\fsctx$ with the given specifications of $\fid$: $$\fsctx(\alpha) \defeq \fsctx[\fid\mapsto\{\uspeconecase{\tpre(\beta)}{\bigq(\beta)} ~|~ \beta<\alpha\}\cup\{(\uspeconecase{\ntpre(\beta)}{\AssFalse} ~|~ \beta\leq\alpha\})]$$
\end{description}
We next define the function $g : \ccmdpr\longrightarrow\ccmdpr$, to be used in the upcoming Scott induction, as follows:
\[
		g(\repr{C}) \eqdef \repr{\ghelp(C)}
\]
where $\ghelp:\ccmdplain\longrightarrow\ccmdplain$ is defined as $\ghelp(C):=\fcallsub{C_\fid}{C,\fictx',\fid}$.

Intuitively, $\hp$ takes an arbitrary command $C$ from $\ccmdplain$ as an argument and substitutes it for any function call on $\fid$ in function body $C_\fid$.
The function $g$ then lifts this operation to the quotient space $\ccmdpr$.
The definitions of $\hp$ and $g$ trivially yield the following identities for arbitrary $C\in\ccmdplain$ and $(C_n)_{n\in\Nat}\subseteq\ccmdplain$:
\begin{itemize}
	\item[(G1)]
	$\bigjoinpr{n\in\Nat}g(\repr{C_n})
	= \repr{\bigjoinpr{n\in\Nat} \ghelp(C_n)}$
	\item[(G2)]
	$ \bigjoinpr{n\in\Nat}g^n(\repr{C})
	= \repr{\bigjoinpr{n\in\Nat}\ghelp^n(C)}$
\end{itemize}

\begin{lemma}\label{lem:gcont}
	The function $g$ is continuous.
\end{lemma}
\begin{proof} We begin by proving monotonicity.

	\myparagraph{Monotonicity}
	We prove the monotonicity of $\ghelp$. We need to show that for all $C_1,C_2\in\ccmdplain$, it holds that
	$$ C_1\poccmd C_2 \Rightarrow \fcallsub{C_\fid}{C_1,\fictx',\fid}\poccmd \fcallsub{C_\fid}{C_2,\fictx',\fid} $$
	Let $\stt,\stt'\in\State$ such that $\stt,\fcallsub{C_\fid}{C_1,\fictx',\fid}\baction_{\fictx'}\outcome:\stt'$. We need to show that
	$$\stt,\fcallsub{C_\fid}{C_2,\fictx',\fid}\baction_{\fictx'}\outcome:\stt'$$
	and we do so by structural induction on $C_\fid$.
	The only non-trivial cases are the compound commands and the function call, as the substitution is the identity for all other cases.

	\casex{If-Else.} $C_\fid=\pifelse{\pexp}{\ctrue}{\cfalse}$. Let $(\sto,\hp),\stt'\in\State$ such that
	$$ (\sto,\hp),\pifelse{\pexp}{\fcallsub{\ctrue}{C_1,\fictx',\fid}}{\fcallsub{\cfalse}{C_1,\fictx',\fid}}\baction_{\fictx'}\outcome:\stt'$$
	The operational semantics yields (for an appropriate $\stt{\oxerr}$):
	\[
	\begin{array}{l l}
	& \big(\esem{\pexp}{\sto}=\ptrue\land(\sto,\hp),\fcallsub{\ctrue}{C_1,\fictx',\fid}\baction_{\fictx}\outcome:\stt'\big)~\lor \\ & \big( \esem{\pexp}{\sto} = {\pfalse}\land (\sto,\hp),\fcallsub{\cfalse}{C_1,\fictx',\fid}\baction_{\fictx}\outcome:\stt'\big)~\lor \\
	& \big( \esem{\pexp}{\sto} = \undefd \land \stt' = \stt{\oxerr} \big) \\
	\stackrel{\text{(IH)}}{\Rightarrow} & \big(\esem{\pexp}{\sto}=\ptrue\land(\sto,\hp),\fcallsub{\ctrue}{C_2,\fictx',\fid}\baction_{\fictx}\outcome:\stt'\big)~\lor \\
	& \big( \esem{\pexp}{\sto} = {\pfalse}\land (\sto,\hp),\fcallsub{\cfalse}{C_2,\fictx',\fid}\baction_{\fictx}\outcome:\stt'\big)~\lor \\
	& \big( \esem{\pexp}{\sto} = \undefd \land \stt' = \stt{\oxerr} \big) \\
	\end{array}
	\]
	which implies the desired
	$$ (\sto,\hp),\pifelse{\pexp}{\fcallsub{\ctrue}{C_2,\fictx',\fid}}{\fcallsub{\cfalse}{C_2,\fictx',\fid}}\baction_{\fictx'}\outcome:\stt'$$
	The while loop and the sequencing cases are proven analogously.

	\casex{Function Call.}  $C_\fid=\passign{\pvar y}{\fid(\vec \pexp)}$. Using Lemma \ref{lem:scopefcallequiv} and considering the successful case only as the faulting cases are proven analogously, let $\stt,\stt'\in\State$ such that
	$$ \stt, \pscope{\pvar x,\vec \pexp}{C_1}{\pvar y, \pexp'}\baction_{\fictx'}\stt' $$
	Letting $\stt = (\sto,\hp)$ and $\stt' = (\sto',\hp')$, the operational semantics for $\mathtt{scope}$ then implies
	\[
	\begin{array}{l l}
	& \sto'=\sto[\pvar y\rightarrow v']
	\land \esem{\vec \pexp}{\sto}=\vec v
	\land \esem{\pexp'}{\sto_q}=v'
	\land \pv{C}\setminus\{\pvvar x\} = \pvvar z \\
	& \land~\sto_p = \emptyset[\pvvar x\rightarrow\vec v][\pvvar z\rightarrow\nil]
	\land (\sto_p,\hp), C_1\baction_{\fictx}(\sto_q,\hp') \\
	\stackrel{C_1 \poccmd C_2}{\implies} & \sto'=\sto[\pvar y\rightarrow v']
	\land \esem{\vec \pexp}{\sto}=\vec v
	\land \esem{\pexp'}{\sto_q}=v'
	\land \pv{C}\setminus\{\pvvar x\} = \pvvar z \\
	& \land~\sto_p = \emptyset[\pvvar x\rightarrow\vec v][\pvvar z\rightarrow\nil]
	\land (\sto_p,\hp), C_2\baction_{\fictx}(\sto_q,\hp') \\
	\end{array}
	\]
	which implies the desired
	$$ \stt, \pscope{\pvar x,\vec \pexp}{C_2}{\pvar y, \pexp'}\baction_{\fictx'}\stt' $$

	The monotonicity of $g$ follows straightforwardly from the monotonicity of $\ghelp$.

\myparagraph{Supremum-Preservation}
Assume a chain $(C_n)_{n\in\Nat}$ in $\ccmdplain$.
First, we show that \\ 
$\bigsqcup\limits_{n\in\Nat} \ghelp(C_n) \poccmd \ghelp\big(\bigsqcup\limits_{n\in\Nat} C_n\big)$:

\[
\begin{array}{r l}
	& \stt, \bigsqcup\limits_{n\in\Nat} \ghelp(C_n) \baction_{\fictx'} \outcome: \stt' \\
	\Rightarrow & \exsts{m\in\Nat} \stt, \ghelp(C_m) \baction_{\fictx'}  \outcome:  \stt' \\
	\Rightarrow & \exsts{m\in\Nat} \stt, \fcallsub{C_\fid}{C_m,\fictx',\fid} \baction_{\fictx'}  \outcome: \stt' \\
	\Rightarrow & \stt, \fcallsub{C_\fid}{\bigsqcup\limits_{n\in\Nat} C_n,\fictx',\fid} \baction_{\fictx'}  \outcome: \stt' \\
	\Rightarrow & \stt, \ghelp\big(\bigsqcup\limits_{n\in\Nat} C_n\big) \baction_{\fictx'}  \outcome: \stt'
\end{array}
\]

%
%

Next, we show that $\ghelp\big(\bigsqcup\limits_{n\in\Nat} C_n\big) \poccmd \bigsqcup\limits_{n\in\Nat} \ghelp(C_n)$.
Let $\stt, \ghelp\big( \bigsqcup\limits_{n\in\Nat} C_n \big) \baction_{\fictx'} \stt'$, \\
i.e.
$\stt, \fcallsub{C_\fid}{\bigsqcup\limits_{n\in\Nat} C_n,\fictx',\fid} \baction_{\fictx'} \stt'$.
Then, since $C_\fid\in\Cmd$, it is a finite command and hence has a finite number $t$ of function calls on~$\fid$. At each function call substitution site, the execution will execute some command $C_n$. Assume $k_1,...,k_t\in\Nat$ such that at the $i$-th execution site, the command $C_{k_i}$ is executed, and let $k \defeq \max{(k_1,...,k_t)}$. Since $(C_n)_{n\in\Nat}$ is a chain, we have that $C_{k_i}\poccmd C_k$ for all $i\in\{1,...,t\}$ and therefore:
\[
\begin{array}{r l}
	& \stt, \fcallsub{C_\fid}{C_k,\fictx',\fid} \baction_{\fictx'} \stt' \\
	\Rightarrow & \stt, \bigsqcup\limits_{n\in\Nat} \fcallsub{C_\fid}{C_n,\fictx',\fid} \baction_{\fictx'} \stt' \\
	\Rightarrow & \stt, \bigsqcup\limits_{n\in\Nat} \ghelp(C_n) \baction_{\fictx'} \outcome: \stt'
\end{array}
\]

The supremum preservation of $g$ follows trivially from the supremum preservation of $\ghelp$.
\end{proof}
Next, we introduce the admissible set we will use in this instantiation of the Scott induction.

%
%

\begin{lemma}[Admissible Subset $\oneS$]\label{lem:oneS}
	The set $\oneS$, defined as
	$$
	\begin{array}{l}
	\oneS := \{ \repr{\bar C}\in\ccmdpr ~|~ \exsts{C\in\repr{\bar C}}
	\frall{t\in(\fsctx(\alpha))(\fid)}
\exsts{ \uspeconecase{P'}{\bigq'}\in\fext_{\fictx',\fid}(t)} \fsctx\models\triple{P'}{C}{\bigq'}
	\}
	\end{array}
	$$
	is an admissible subset of $(\ccmdpr,\poccmd)$.
\end{lemma}

\begin{proof}
 	\casex{Least Element.}
 	We know that $\lepr=\repr{\pwhile{\ptrue}{\pskip}}$ and that this commands trivially semantically satisfies any OX-quadruple. 
 	Therefore, $\lepr\in\oneS$.

 	\casex{Chain-Closure.}
 	We need to show that given an arbitrary chain $(\repr{C'_n})_{n\in\Nat}\subseteq\oneS$, it holds that $\pnatchoice{\repr{C'_n}}{n}\in\oneS$. Onwards, we will use the following notation:
 	\begin{itemize}
 	\item $\uspeconecase{P}{\bigq} \defeq \uspec{P}{\Qok}{\Qerr}$
 	\item $\uspeconecase{P_n}{\bigq_n} \defeq \uspec{P_n}{\Qok^n}{\Qerr^n}$
 	\end{itemize}
 	The definition of $\oneS$ yields the existence of a chain $(C_n)_{n\in\Nat}\subseteq\ccmdplain$ such that for all $n\in\Nat$, it holds that $C_n\in\repr{C'_n}$ and
 	$$
 	\frall{n\in\Nat,\uspeconecase{P}{\bigq}\in((\fsctx(\alpha))(\fid)}\exsts{\uspeconecase{P_n}{\bigq_n}\in\fext_{\fictx',\fid}(\uspeconecase{P}{\bigq})} \fsctx\models\triple{P_n}{C_n}{\bigq_n}
 	$$
	Per definition of $\fext$, we know that $P_n = P\lstar\pvvar z\doteq\nil$ for all $n\in\Nat$. Together with the definition of choice, we obtain that
 	 $$
	\fsctx\models\quadruple{P\lstar\pvvar z\doteq\nil}{\bigsqcup (C_n | n\in\Nat)}{ \bigvee_{n \in \Nat}\Qok^n}{ \bigvee_{n \in \Nat}\Qerr^n}
 	$$
 	It remains to show that $(P\lstar\pvvar z\doteq\nil)~(\bigvee_{n \in \Nat} \bigq_n)$ is an internalisation of $(P)~(\bigq)$. Since $(P\lstar\pvvar z\doteq\nil)~(\bigq_n)$ are internalisations of $(P)~(\bigq)$, we obtain
 	$$
 		\frall{n\in\Nat}
 		(\Qok
 		~\Leftrightarrow~
 		\exsts{\vec p} \Qok^n[\vec p / \pvvar p]\lstar \pvar{ret}\doteq \pexp'[\vec p / \pvvar p]) \land (\Qerr
 		~\Leftrightarrow~
 		\exsts{\vec p} \Qerr^n[\vec p / \pvvar p])
 	$$
 	This implies
 	$$
 	\begin{array}{r l}
 		\Qok
 		~\Leftrightarrow~
 		&\bigvee_{n \in \Nat}\big(\exsts{\vec p} \Qok^n[\vec p / \pvvar p]\lstar \pvar{ret}\doteq \pexp'[\vec p / \pvvar p]\big) \\
 		~\Leftrightarrow~
 		&\exsts{\vec p} \bigvee_{n \in \Nat}\big( \Qok^n[\vec p / \pvvar p]\lstar \pvar{ret}\doteq \pexp'[\vec p / \pvvar p]\big) \\
 		~\Leftrightarrow~
 		&\exsts{\vec p} \big(\bigvee_{n \in \Nat}\Qok^n[\vec p / \pvvar p]\big)\lstar \pvar{ret}\doteq \pexp'[\vec p / \pvvar p]\\
 		~\Leftrightarrow~
 		&\exsts{\vec p} \big(\bigvee_{n \in \Nat}\Qok^n\big)[\vec p / \pvvar p]\lstar \pvar{ret}\doteq \pexp'[\vec p / \pvvar p]
 	\end{array}
 	$$
 	and analogously
 	 	$$
 	\begin{array}{r l}
 		\Qerr
 		~\Leftrightarrow~
 		&\bigvee_{n \in \Nat}\big(\exsts{\vec p} \Qerr^n[\vec p / \pvvar p]\big) \\
 		~\Leftrightarrow~
 		&\exsts{\vec p} \bigvee_{n \in \Nat}\big(\Qerr^n[\vec p / \pvvar p]\big)\\
 		~\Leftrightarrow~
 		&\exsts{\vec p} \big(\bigvee_{n \in \Nat}\Qerr^n\big)[\vec p / \pvvar p]
 	\end{array}
 	$$
 	Therefore, $\repr{\sqcup(C_n|n\in\Nat)}\in\oneS$, which yields $\sqcup(\repr{C_n}|n\in\Nat)\in\oneS$ and finally $\sqcup(\repr{C'_n}|n\in\Nat)\in\oneS$.
\end{proof}

This concludes the set-up for Scott induction, which allows us to prove the inductive step.

\begin{lemma}[Scott Condition]\label{lem:scottcondition}
	Under the assumptions (A1)-(A7) and additionally assuming
 \begin{equation}\label{validforallx}
	\frall{t\in(\fsctx(\alpha))(\fid)}\exsts{t'\in \fext_{\fictx', \fid}(t)}
	\fsctx(\alpha) \vdash C_\fid:t'
 \end{equation}
 and
 \begin{equation}\label{validforallnt}
	\exsts{t' \in \fext_{\fictx', \fid}(\uspeconecase{\tpre(\alpha)}{\bigq(\alpha)})}
	\fsctx(\alpha) \vdash C_\fid:t'
 \end{equation}
	it holds that $g(\oneS)\subseteq\oneS$.
\end{lemma}
\begin{proof}
	Let $\repr{C}\in\oneS$. Therefore, there exists a $C'\in\repr{C}$ such that for all $t\in(\fsctx(\alpha))(\fid)$ exists a $(P)~(\bigq)\in\fext_{\fictx',\fid}(t)$ such that $\textbf{(H)}~\fsctx\models\triple{P}{C'}{\bigq}$. This implies that $C'$ does not call on $\fid$, because $\fsctx$ holds no specifications for $\fid$. Hence, $C_\fid[C',\fictx', \fid]$ does not call on $\fid$ either.
	We prove the statement by showing the more general claim
	 $$
	{\fsctx(\alpha)\vdash \utripleq{P}{C_\fid}{\bigq}}  \implies  \fsctx\models\triple{P}{\fcallsub{C_\fid}{C',\fictx',\fid}}{\bigq}
	 $$
	 for arbitrary precondition $P$ and postcondition $\bigq$ by induction over the structure of $C_\fid$.

	 \myparagraph{Base Commands and Function Calls on $g\not=\fid$} In this case, $\fcallsub{C_\fid}{C',\fictx',\fid}=C_\fid$ and $C_\fid$ does not call on $\fid$.
	 Therefore, the proof tree of $\fsctx(\alpha)\vdash \utripleq{P}{C_\fid}{\bigq}$ uses no specifications on $\fid$, which implies that $\fsctx\vdash \utripleq{P}{C_\fid}{\bigq}$. From $\text{(A1)}~\models(\fictx,\fsctx)$, we obtain $\fsctx\models \utripleq{P}{C_\fid}{\bigq}$, which implies $\fsctx\models\triple{P}{C_\fid}{\bigq}$.

	 \myparagraph{Compound Command Rules: if-then, if-else, sequence}
	 If-then, if-else and sequence are proven directly, using the IH.

	\myparagraph{Compound Command Rules: while} 
	The following proof is based on a while rule \mbox{while-iterate}, which is a while rule from an earlier version of our proof rules. In the current version of the proof rules, we have simplified the \mbox{while-iterate} rule to the \mbox{while} rule as introduced in the main body. As the earlier \mbox{while-iterate} rule encompasses the simplified \mbox{while} rule, we have keep the proof here as is.
	Given the rule
	$$
	\inferrule[\mbox{while-iterate}]
{
\frall{i\in\Nat}~\models P_i \Rightarrow \dotin \pexp \Bool \lstar \AssTrue \qquad P_\infty \defeq \AssFalse \\\\
\frall{i\in\Nat}~\fsctx \vdash \uquadruple{P_i \land \pexp}{C}{P_{i+1} }{Q_i} \\\\  m \defeq \min(\{i \in \mathbb{N} \uplus \{ \infty \} \mid~\models P_i \Rightarrow \lnot \pexp \lstar \AssTrue \})
  }
  { \fsctx \vdash \uquadruple{P_0}{\pwhile{\pexp}{C}}{P_m}{\exsts{n < m}  Q_n}
  }
	$$
	we assume
	\begin{description}
	\item[(W1)] $\subst,\sto,\hp\models P_0$
	\item[(W2)] $(\sto,\hp\uplus \hp_\fid),\pwhile{\pexp}{\fcallsub{C}{C',\fictx',\fid}}\baction_{\bar\fictx}(\sto',\hp'')$
	\item[(W3)] $\frall{i\in\Nat}~\models P_i \Rightarrow \dotin \pexp \Bool \lstar \AssTrue$
	\item[(W4)] $P_\infty \defeq \AssFalse$
	\item[(W5)] $\frall{i\in\Nat}~\fsctx(\alpha) \vdash \uquadruple{P_i \land \pexp}{C}{P_{i+1} }{Q_i}$
	\item[(W6)] $m \defeq \min(\{i \in \mathbb{N} \uplus \{ \infty \} \mid~\models P_i \Rightarrow \lnot \pexp \lstar \AssTrue \})$
	\end{description}
	Noting that $\fcallsub{(\pwhile{\pexp}{C})}{C',\fictx',\fid} = \pwhile{\pexp}{\fcallsub{C}{C',\fictx',\fid}}$,
	(W2) implies that
	\begin{description}
	\item[(W7)] $(\esem{\pexp}{\sto,\hp\uplus \hp_\fid}=\pfalse \land (\sto,\hp\uplus \hp_\fid)=(\sto',\hp'')) \lor \\
	(\esem{\pexp}{\sto,\hp\uplus \hp_\fid}=\ptrue \land
	(\sto,\hp\uplus \hp_\fid), \fcallsub{C}{C',\fictx',\fid}\baction_{\fictx'} \bar\stt \land
	\bar\stt, \pwhile{\pexp}{\fcallsub{C}{C',\fictx',\fid}} \baction_{\fictx'} (\sto',\hp''))$
	\end{description}

	\noindent
  	and (W5) and the inductive hypothesis imply
  	$$
  	\textbf{(W8)}\quad \frall{i\in\Nat} \fsctx\models\quadruple{P_i \land \pexp}{\fcallsub{C}{C',\fictx',\fid}}{P_{i+1} }{Q_i}
  	$$
  	 We know that $m \neq \infty$, as otherwise the loop would be non-terminating, contradicting (W2).
If $m = 0$, then $\esem{\pexp}{\subst,\sto}=\pfalse$, and (W7) trivially yields the desired result.

	Otherwise, we have that $m > 0$. Then,
  	(W8) yields
  	$$\fsctx \models \quadruple{P_{i-1} \land \pexp}{\fcallsub{C}{C',\fictx',\fid}}{P_i}{Q_{i-1}}$$
  	for all $0<i\leq m$ and (W6) yields $\fsctx \models \triple{P_{m} }{\pwhile{\pexp}{\fcallsub{C}{C',\fictx',\fid}}}{P_m}$.
  	From here, applying the operational semantics of while $m$ times, similarly to the proof in RHL~\cite{reverselogic}, we obtain the desired
	$$
	\fsctx\models\quadruple{P_0}{\pwhile{\pexp}{\fcallsub{C}{C',\fictx',\fid}}}{P_m}{\exsts{n<m} Q_n}
	$$

	 \myparagraph{Structural rules}
	 All four structural rules (equiv, exists, frame, and disj) are proven trivially, using the inductive hypothesis. We give the proof for the equivalence rule:
	$$
	 \inferrule[\mbox{equiv}]
  { \models P \Leftrightarrow P' \quad  \fsctx(\alpha) \vdash \uquadruple{P'}{C_\fid}{\Qok'}{\Qerr'} \quad  \models \Qok' \Leftrightarrow \Qok \quad \models \Qerr' \Leftrightarrow \Qerr }
 { \fsctx(\alpha) \vdash \uquadruple{P }{C_\fid}{\Qok}{\Qerr}}
 $$
 where the IH gives us $\fsctx \models \quadruple{P'}{\fcallsub{C_\fid}{C',\fictx',\fid}}{\Qok'}{\Qerr'}$, from which the desired claim is obtained trivially.

	 \myparagraph{Function call on $\fid$}
	 $\fcallsub{(\passign{\pvar y}{\mathtt{\fid}(\vec \pexp)})}{C',\fictx',\fid}=\pscope{\pvvar x,\vec \pexp}{C'}{\pvar y, \pexp'}$, where $\fictx'(\fid)=(\pvvar x,C_\fid,\pexp')$. Therefore, we need to show that
	 $$\fsctx\models\triple{P}{\pscope{\pvvar x,\vec \pexp}{C'}{\pvar y, \pexp'}}{\bigq}$$

	 The assumption $\fsctx(\alpha)\vdash\utripleq{P}{\passign{\pvar y}{\mathtt{\fid}(\vec \pexp)}}{\bigq}$ implies via the function call rule that $P = (\pvar y\doteq \lexp_y\lstar \vec \pexp\doteq\vec x\lstar P^*)$, where $P^*$ is the program-variable-free part either the pre-condition $\ntpre$ of the non-terminating or the pre-condition $\tpre$ of the (partially) terminating specification, and that $(\pvvar x \doteq \vec x \lstar P^*)~(\bigq)\in(\fsctx(\alpha))(\fid)$.
	 We assume the following:
	 \begin{description}
	 \item[(F0)] an arbitrary $\bar\fictx$ such that $\models(\bar\fictx,\fsctx(\alpha))$
	 \item[(F1)] $\subst,\sto,\hp\models\pvar y\doteq \lexp_y\lstar \vec \pexp\doteq\vec x\lstar P^*$
	 \item[(F2)] arbitrary $\hp_\fid$ and $\hp''$ such that $(\sto,\hp\uplus \hp_\fid),\pscope{\pvvar x,\vec \pexp}{C'}{\pvar y, \pexp'}\baction_{\bar\fictx}\outcome: (\sto',\hp'')$
	 \end{description}
	 We need to show that
	 $$
	 (\outcome\neq\oxm) \land
	 \big(
	 \exsts{\hp'} \hp''= \hp'\uplus \hp_\fid \land
	 	((
	 	\outcome=ok \land \subst,\sto',\hp'\models\Qok
	 	)
	 \lor
	 	(
	 	\outcome=err \land \subst,\sto',\hp'\models\Qerr
	 	))
	 \big)
	 $$
	 Defining
	 \begin{description}
	 \item[(F3)] $\vec v := \esem{\vec \pexp}{\sto}$
	 \item[(F4)] $\pvvar z := \pv{C'}\backslash\{\pvvar x\}$.
	 \item[(F5)] $\sto_p := \emptyset[\pvvar x\rightarrow \vec v][\pvvar z\rightarrow\nil]$,
	 \end{description}
	 the operational semantics of $\mathtt{scope}$ and (F2) imply
	 \begin{description}
	 \item[(F6)] $(\sto_p,\hp\uplus \hp_\fid),C'\baction_{\bar\fictx}\outcome:(\sto_q,\hp'')$
	 \item[(F7a)] $\outcome=ok\Rightarrow (\sto'=\sto[\pvar y\rightarrow \esem{\pexp'}{\sto_q}])$
	 \item[(F7b)] $\outcome=err\Rightarrow (\sto'=\sto[\pvar{err}\rightarrow \esem{\pvar{err}}{\sto_q}])$
	 \end{description}

	The definition of $\oneS$ implies existence of a $\uspec{P'}{\Qok'}{\Qerr'}\in\fext_{\fictx,\fid}(\uspeconecase{P^*}{\bigq})$ such that \textbf{(H8a)}~$\fsctx\models\triple{P'}{C'}{\bigq'}$. (A1) and (F0) imply $\models(\bar\fictx,\fsctx)$, which yields with (H8a) that  \textbf{(H8b)}~$\bar\fictx\models\triple{P'}{C'}{\bigq'}$

	Per definition of the internalisation, we know that $P'=\pvvar x \doteq \vec x \lstar P^*\lstar\pvvar z\doteq\nil$, which implies with (F3)-(F5) that \textbf{(H9)}~$\subst,\sto_p,\hp\models P'$.
	Then,
	(F6), (F8b) and (F9) imply
	$$
	\textbf{(H10)}~~
	(\outcome\neq\oxm)\land
	 \big(
	 \exsts{\hp'} \hp''= \hp'\uplus \hp_\fid \land
	 	((
	 	\outcome=ok \land \subst,\sto_q,\hp'\models\Qok'
	 	)
	 \lor
	 	(
	 	\outcome=err \land \subst,\sto_q,\hp'\models\Qerr'
	 	))
	 \big)
	$$
	This yields $\subst,\sto',\hp'\models Q_{\outcome}$ as $\Qok'$ and $\Qerr'$ are internalisations of $\Qok$ and $\Qerr$ respectively, which in turn implies the desired result, i.e.
	$$
	\fsctx\models\triple{P}{\pscope{(\pvvar x,\vec \pexp)}{C'}{\pvar y, \pexp'}}{\bigq}
	$$

%
%

	This concludes the proof of the general statement. Instantiating it with the existentially quantified $t'$ from (\ref{validforallx}) and (\ref{validforallnt}) then yields the desired result
	 $$\repr{C}\in\oneS
	 ~\Longrightarrow~
	 \repr{\hp(C)}\in\oneS
	 $$
	 i.e., the Scott condition $g(\oneS)\subseteq\oneS$.

\end{proof}
Finally, we need to show that the function body $C_\fid$ of $\fid$ is indeed equivalent to the least fixpoint of $g$, denoted by $\lfp{g}$.
\begin{lemma}[Scott's last step]\label{lem:lastscott}
	The function body $C_\fid$ is in the least fixpoint of $g$, i.e.
	$$
	C_\fid\in\lfp{g}
	$$
\end{lemma}
Again, we prove a slightly different statement first, and then apply it to prove the lemma. Onward, we write $\lepr$ to denote the least element of $\ccmdpr$ (and also of $\oneS$).
Onwards, we will use $\wts$ as a shorthand for the command $\pwhile{\ptrue}{\pskip}$. Keep in mind that $\lepr=\repr{\wts}$.
\begin{lemma}\label{lem:depth}
	For all $n\in\Nat$, it holds that
	$$ \frall{\stt,\stt'\in\State}~\stt,C_\fid\baction_{\fictx'}\stt'~\land~\depthF{\stt,C_\fid\baction_{\fictx'}\stt'}\leq n
	\iff
	\stt,\ghelp^{n+1}(\wts)\baction_{\fictx}\stt'	$$
where do not explicitly include the outcome statement $\outcome$ to avoid clutter.
\end{lemma}
\begin{proof} By induction on $n$.

	\casex{Base Case: $n=0$.}

	\medskip
	\noindent
	"$\Rightarrow$": Let $\stt,C_\fid\baction_{\fictx'}\stt'~\land~\depthF{\stt,C_\fid\baction_{\fictx'}\stt'}=0$. Since $C_\fid\in\Cmd$, it cannot include the $\mathtt{scope}$ command. Since the depth is zero, the execution path does not reach a function call on $\fid$. Hence, this call may be replaced by any other command, including $\wts$, i.e.
	\[
	\begin{array}{c l}
		& \stt,\fcallsub{C_\fid}{\wts,\fictx',\fid}\baction_{\fictx'}\stt' \\
		\Leftrightarrow & \stt,\ghelp(\wts)\baction_{\fictx'}\stt' \\
	\end{array}
	 \]

	 \noindent
	 "$\Leftarrow$": Assuming
	 $\stt,\ghelp(wts)\baction_{\fictx'}\stt'
	 $, we obtain $\stt,\fcallsub{C_\fid}{\wts,\fictx',\fid}\baction_{\fictx'}\stt'
	 $.
	 Since $\wts$ does not terminate on any state, this implies that the execution does not reach the command $\wts$, that is, it does not reach any function call site of $\fid$ and, therefore, we obtain $\stt,C_\fid\baction_{\fictx'}\stt'$ and  $\depthF{\stt,C_\fid\baction_{\fictx'}\stt'}=0$.

	 \casex{Inductive Step.} Assume that the equivalence holds for some $n\in\Nat$, and the goal is then to prove that it holds for $n+1$.

	 \medskip
	 \noindent
	 "$\Rightarrow$": Let $\stt,C_\fid\baction_{\fictx'}\stt'~\land~\depthF{\stt,C_\fid\baction_{\fictx'}\stt'}\leq n+1$.
	 If $\depthF{\stt,C_\fid\baction_{\fictx'}\stt'}< n+1$, the inductive hypothesis yields $\stt, \ghelp^{n+1}(\wts) \baction_{\fictx}\stt'$,
	Since $\ghelp^{n+1}(\wts)$ implies that the execution terminates and therefore does not reach any instance of the command $\wts$, we may substitute $\wts$ for any other command, including $\hp(\wts)$. This yields $\stt, \hp^{n+2}(\wts) \baction_{\fictx}\stt'$.


	 Finally, let $\depthF{\stt,C_\fid\baction_{\fictx'}\stt'} = n+1$ and let $\passign{\pvar y}{\fid(\pvvar x)}$ be an arbitrary function call on $\fid$ reached by the execution.
	  That means that there exist states $\stt1$ and $\stt2$, such that the execution of the initial part of $C_\fid$ from $\stt$ up to that function call yields $\stt1$, that $\stt1,\passign{\pvar y}{\fid(\pvvar x)}\baction_{\fictx'}\stt2$, and that the execution of the remaining part of $C_\fid$ on $\stt2$ yields $\stt'$.
	 By instantiating the operational semantics of the function call with $\stt1$ and $\stt2$, we obtain $(\sto_p,\hp),C_\fid\baction_{\fictx'}(\sto_q,\hp')$, where, per definition of the $\mathtt{depth}$ function, $\depthF{(\sto_p,\hp),C_\fid\baction_{\fictx'}(\sto_q,\hp')}=n$.
	 The inductive hypothesis then implies $(\sto_p,\hp),\ghelp^{n+1}(\wts)\baction_{\fictx}(\sto_q,\hp')$. Therefore, $\stt, \fcallsub{C_\fid}{\ghelp^{n+1}(\wts), \fictx',\\fid}\baction_{\fictx'}\stt'$ and since this command does not call on $\fid$ and given the definition of $\ghelp$, we obtain $\stt,\ghelp^{n+2}(\wts)\baction_{\fictx}\stt'$.

	 \medskip\noindent
	 "$\Leftarrow$": Assume $\stt,\ghelp^{n+2}(\wts)\baction_{\fictx}\stt'$: that is, $\stt,\fcallsub{C_\fid}{\ghelp^{n+1}(\wts),\fictx',\fid}\baction_{\fictx}\stt'$.
	If there are no substitution sites in $C_\fid$, the desired goal is obtained trivially, as the substitution is vacuous and the considered depth is zero by definition.
	 Otherwise, consider an arbitrary substitution site in $C_\fid$ reached by the execution starting from $\stt$, and let the state before the substitution site be some $\stt1$. By instantiating the operational semantics of $\mathtt{scope}$, we obtain $\sto_p$, $\sto_q$, $\hp'$ and $\stt2$, such that $(\sto_p,\hp),\ghelp^{n+1}(\wts)\baction_{\fictx}(\sto_q,\hp')$ and the remaining part of $\fcallsub{C_\fid}{\ghelp^{n+1}(\wts),\fictx',\fid}$ executed on $\stt2$ terminates in $\stt'$.
	 Per the inductive hypothesis, we have that $(\sto_p,\hp),C_\fid\baction_{\fictx'}(\sto_q,\hp')$ and $\depthF{(\sto_p,\hp),C_\fid\baction_{\fictx'}(\sto_q,\hp')}\leq n$. This implies, given the operational semantics of function call and the definition of depth, that $\stt,C_\fid\baction_{\fictx'}(\sto_q,\hp')$ and $\depthF{\stt,C_\fid\baction_{\fictx'}\stt'}\leq n+1$, which concludes the proof.
\end{proof}
With this in place, we can prove Lemma~\ref{lem:lastscott}, concluding the overall proof:
\begin{proof}[Proof of Lemma~\ref{lem:lastscott}]
	We know that the least fixpoint of $g$ obeys the following identity: 
	\begin{equation}\label{lfpid}
		\lfp{g} = \bigsqcup\limits_{n\in\Nat}g^n(\lepr) = \repr{ \bigsqcup\limits_{n\in\Nat} \ghelp^n(\wts) }
	\end{equation}

	To prove the statement of the lemma, we will show that
	$$
	\bigsqcup\limits_{n\in\Nat} \ghelp^n(\wts) \eqccmd C_\fid
	$$

	\noindent
	"$\Rightarrow$": Assume $\stt,\stt'\in\State$ such that $\stt,\bigsqcup\limits_{n\in\Nat} \ghelp^n(\wts)\baction_{\fictx'}\stt'$. Therefore, there exists an $n\in\Nat$ such that $\stt,\ghelp^{n}(\wts)\baction_{\fictx'}\stt'$.
	Since the command $\wts$ does not terminate on any state and since $\ghelp^0(\wts)=\wts$, $n$ must be strictly positive. Lemma \ref{lem:depth} then implies that $\stt,C_\fid\baction_{\fictx'}\stt'$, which concludes the proof.

	\medskip\noindent
	"$\Leftarrow$": Assuming $\stt,\stt'\in\State$ such that $\stt,C_\fid\baction_{\fictx'}\stt'$, we know that $C_\fid$ terminates when executed on $\stt$, and therefore has a finite execution depth: that is, $\depthF{\stt,C_\fid\baction_{\fictx'}\stt'}=n$ holds for some $n\in\Nat$.
	Lemma \ref{lem:depth} then implies $\stt,\ghelp^{n+1}(\wts)\baction_{\fictx}\stt'$, and therefore
	$\stt,\bigsqcup\limits_{n\in\Nat}\ghelp^{n}(\wts)\baction_{\fictx}\stt'$.

	This yields $C_\fid\in\lfp{g}$.
\end{proof}

Theorem \ref{theorem:scottinduction} and
Lemmas \ref{lem:scottcondition} and \ref{lem:lastscott} finally fnimply that $\repr{C_\fid}\in\oneS$.

\subsection{n-dimensional Scott Instantiation}%
\label{sec:ndimscott}

The instantiation presented so far suffices to prove soundness for recursive functions which potentially have both terminating and non-terminating specifications. However, we wish to allow clusters of {\it mutually} recursive functions, and hence require a more general instantiation of the Scott induction.
In particular, given an environment $(\fictx,\fsctx)$, we add on a mutually recursive cluster $F:=(\fid_1,...,\fid_n)$ of $n$ functions.

We assume the following:

\begin{description}
\item[(B1)] a valid environment, $\models(\fictx,\fsctx)$;
\item[(B2)] a set of $n$ functions $\fid_i(\pvvar x_i) \{ C_i; \preturn{\pexp_i}\}$ with $i\in I\defeq\{1,...,n\}$ that is not in the domain of $\fictx$;
\item[(B3)] an arbitrary element $\alpha\in\ord$;
\item[(B4)]
a set of terminating (external) specifications for each $\fid_i$, $\{ \uspeconecase{\tprei(\beta)}{\bigq^i(\beta)} ~|~ \beta<\alpha \}$;
\item[(B5)] a set of non-terminating (external) specifications for each  $\fid_i$, $\{ \uspeconecase{\ntprei(\beta)}{\AssFalse} ~|~ \beta\leq\alpha \}$;
\item[(B6)] an extension of $\fictx$: $\fictx' \defeq \fictx[\fid_i\mapsto(\pvvar x_i, C_i, \pexp_i)]_{i\in I}$; and
\item[(B7)] an extension of $\fsctx$ with the given specifications: $$\fsctx(\alpha) \defeq \fsctx[\fid_i\mapsto\{\uspeconecase{\tprei(\beta)}{\bigq^i(\beta)} ~|~ \beta<\alpha\}\cup\{(\uspeconecase{\ntprei(\beta)}{\AssFalse} ~|~ \beta\leq\alpha\}]_{i\in I}$$
\end{description}

\medskip
{\it Note on Notation.} As most elements we will be using in this section are $n$-tuples of commands, or chains of such $n$-tuples, we use the following notation:
\begin{itemize}
	\item $C_i$: a subscript $i$ denotes that the commands $C_i$ is the implementation of the function $\fid_i$, as recorded in the associated function implementation context $\fictx'$,
	\item $C^i$:
	a superscript $i$ denotes the $i$-th component of an $n$-tuple $C\in\nccmd$,
	\item $C(i)$: an index $i$ denotes the $i$-th element of a chain (monotonically increasing sequence) $\chain{C}{m}$.
\end{itemize}

\begin{lemma}[n-Dimensional Domain]
	Given the domain $(\ccmdpr,\poccmd)$, we lift the equivalence relation, partial order,  and the join operator of $\ccmdpr$ to elements of $\ccmdpr^n$ as follows:
	$$
	\begin{array}{l}
	C\neqv \bar C
	~\Longleftrightarrow~
	\frall{i\in I} C^i\eqccmd \bar C^i \\
	C\npo \bar C
	~\Longleftrightarrow~
	\frall{i\in I} C^i\poccmd \bar C^i \\
	C\njoin \bar C := \left( \begin{array}{c}
		C^1\joinccmd \bar C^1 \\
		\vdots \\
		C^n\joinccmd \bar C^n
	\end{array}\right)
	\end{array}
	$$
\end{lemma}

The proof that these satisfy the appropriate properties is trivial with the least element $nle$ being the equivalence class of the $n$-tuple which holds $\pwhile{\ptrue}{\pskip}$ in every component.

\begin{lemma}[n-Dimensional Continuous G]
	The function $\nggx:\nqccmd\longrightarrow\nqccmd$, defined as
	$$
	\nggx(\repr{C}):= \repr{H(C)}
	$$
	where
	$$
	H(C):=\left(\begin{array}{c}
		\hp_1(C) \\
		\vdots \\
		\hp_n(C)
	\end{array}\right)
	$$
	and
	$$
	\hp_i:\ccmdplain^n\longrightarrow\ccmdplain, ~ \hp_i(C):=\fcallsub{C_i}{C,\fictx',F},
	$$
	is continuous.
\end{lemma}


\begin{proof} Again, what needs to be proven is monotonicity and supremum-preservation.

	\myparagraph{Monotonicity} Analogously to the one-dimensional case, it is proven that $\hp_i$ is monotonic for all $i\in I$. This implies that $H$ and, therefore, $G$ is monotonic as well.

	\myparagraph{Supremum-Preservation} Given a chain $\chain{C}{m}$, we need to prove:
	$$
	\bigsqcup\limits_{m\in\Nat} \nggx(\repr{C(m)}) =
	\nggx\big( \bigsqcup\limits_{m\in\Nat} \repr{C(m)} \big)
	$$
	 Due to the definition of $\nggx$, we obtain for the left-hand side
	$$
	\bigsqcup\limits_{m\in\Nat}\nggx({\repr{C(m)}})
	= \repr{ \bigsqcup\limits_{m\in\Nat}H(C(m)) } $$
	and for the right-hand side
	$$
	\nggx\big( \bigsqcup\limits_{m\in\Nat} \repr{C(m)} \big)
	= \repr{H(\bigsqcup\limits_{m\in\Nat} C(m) )}
	$$
	It therefore suffices to show the equivalence of $\bigsqcup\limits_{m\in\Nat}H(C(m))$ and $H(\bigsqcup\limits_{m\in\Nat} C(m) )$, i.e. for all $i\in I$:
	$$
	\bigsqcup\limits_{m\in\Nat} \nfcallsub{C_i}{C(m),\fictx',F} \eqccmd \nfcallsub{C_i}{~\sqcup (C(m) | m\in\Nat), \fictx', F~}
	$$
	This is proven analogously to the continuity in lemma \ref{lem:gcont}.
\end{proof}

\begin{lemma}[Admissible Set]
	Defining $\nSi$ analogous to lemma \ref{lem:oneS} as
	\[
	\begin{array}{l}
	\nSi := \{ \repr{\bar C}\in\ccmdpr ~|~ \exsts{C\in\repr{\bar C}}
\frall{ t\in(\fsctx(\alpha))(\fid_i)}
\exsts{\uspeconecase{P}{\bigq}\in\fext_{\fictx',\fid_i}(t)} \fsctx\models\triple{P}{C}{\bigq}
	\}
	\end{array}
	\]
	then the set $\nS$, defined as
	$$
	\nS := \prod\limits_{i\in I} \nSi
	$$
	is an admissible subset of $(\nqccmd,\npo)$.
\end{lemma}

\begin{proof}

	\casex{Least Element.} The one-dimensional least element $\lepr$, which is equivalent to the command
	$\pwhile{\ptrue}{\pskip}$, trivially semantically satisfies any over-approximating quadruple. We therefore have $\nle\in\nS$.

	\casex{Chain-Closure.} Assume a chain $\chain{\repr{\bar C}}{m}\subseteq\nS$. Then, $\chain{\repr{\bar C^i}}{m}$ is a chain in $\nSi$ for all $i\in I$. Since $\nSi$ is chain-closed (Lemma \ref{lem:oneS}), we have $\bigsqcup (\chain{\repr{\bar C^i}}{m} | m\in\Nat)\in\nSi$. The desired result is then obtained by applying the definition of $\bigsqcup$.
%
\end{proof}

\begin{lemma}[n-Scott Condition]\label{lem:nscottcond}
Under the assumption that for all $i\in I$:
	%
	%
	$$\frall{t\in(\fsctx(\alpha))(\fid_i)}\exsts{t'\in\fext_{\fictx',\fid_i}(t)}
	\fsctx(\alpha)\vdash C_i:t'$$
and
	$$\exsts{t'\in\fext_{\fictx',\fid_i}(\uspeconecase{\tprei(\alpha)}{\bigq^i(\alpha)})}
	\fsctx(\alpha)\vdash C_i:t'$$
it holds that
	$$ \nggx(\nS)\subseteq\nS$$
\end{lemma}

\begin{proof}
	Assume $\repr{\bar C}\in\nS$, i.e. $\repr{\bar C^i}\in\nSi$, then for all $i\in I$ there exists a $C^i\in\repr{\bar C^i}$ such that
	$$
	\frall{ t\in(\fsctx(\alpha))(\fid_i)}
	\exsts{\uspeconecase{P}{\bigq}\in\fext_{\fictx',\fid_i}(t)}
	 \fsctx\models\triple{P}{C^i}{\bigq}
	$$
	We aim to show $\repr{H(C)}\in\nS$, i.e. for all $i\in I$
	$$
	\repr{\hp_i(C)}\in\nSi
	$$
	where $C:=(C^1,...,C^n)$. What we prove is the general claim
	$$
	\fsctx(\alpha) \vdash \utripleq{P}{C_i}{\bigq}
	~\Longrightarrow~
	\fsctx \models \triple{P}{\nfcallsub{C_i}{C,\fictx',F}}{\bigq}
	$$
	which is shown by induction over the structure of $C_i$ analogously to the proof of Lemma \ref{lem:scottcondition}. Then, instantiation with the existentially quantified $t'$ yields the desired result.
\end{proof}

\begin{lemma}[Auxiliary lemma]\label{lem:nlaststepaux}
	For all $m\in\Nat$, $i\in I$ and $\stt,\stt'\in\State$ it holds that
	$$ \depthF{\stt,C_i\baction_{\fictx'}\stt'}\leq m \land \stt,C_i\baction_{\fictx'}\stt'
	~\Longleftrightarrow~
	\stt,\nfcallsub{C_i}{H^m(C),\fictx',F}\baction_{\fictx'}\stt'
	$$
\end{lemma}

This lemma is proven analogously to lemma \ref{lem:depth}. We now proceed to the last step of the proof:

\begin{lemma}[n-Dimensional Last Step] \label{lem:nlaststep}
It holds that
$$
\left(\begin{array}{l}
	C_1 \\
	\vdots \\
	C_n
\end{array}\right)\in\lfp{G}
$$
\end{lemma}

\begin{proof}
	Given the known identities about least fixpoints and the definitions of $G$ and $H$, we obtain:
	$$
	\lfp{G}
	= \bigsqcup\limits_{m\in\Nat} G^m(\nle)
	= \bigsqcup\limits_{m\in\Nat} \repr{H^m(\wts^n)}
	= \repr{\bigsqcup\limits_{m\in\Nat} H^m(\wts^n)},
	$$
	where $\wts^n$ denotes the $n$-tuple whose every component is $\pwhile{\ptrue}{\pskip}$. For this proof, given a vector $v$, we write $(v)_i$ to denote its $i$-th component.
	%
	%
	%
%
 %
	 Assume $i\in I$ and $\stt,\stt'\in\State$ such that $\stt,C_i\baction_{\fictx'}\stt'$. This implies that the execution terminates and therefore there exists  $m > 0$ such that $\depthF{\stt,C_i\baction_{\fictx'}\stt'}\leq m$. Lemma \ref{lem:nlaststepaux} then yields
	 $$
	 \begin{array}{r l}
	 & \stt,\nfcallsub{C_i}{H^{m-1}(\wts^n)}\baction_{\fictx} \stt' \\
	 \Rightarrow & \stt,(H^m(\wts^n))_i\baction_{\fictx} \stt' \\
	 \Rightarrow & \stt,\bigsqcup\limits_{m\in\Nat}(H^m(\wts^n))_i\baction_{\fictx} \stt' \\
	 \Rightarrow & \stt,\big(\bigsqcup\limits_{m\in\Nat}H^m(\wts^n)\big)_i\baction_{\fictx} \stt' \\
	 \end{array}
	 $$
	 which concludes the proof.
\end{proof}

Theorem \ref{theorem:scottinduction} and Lemmas \ref{lem:nscottcond} and \ref{lem:nlaststep} then imply for all $i\in I$ that $\repr{C_i}\in\nSi$.

%% file: sections/app-envsoundness.tex
\newpage
\section{Soundness: Environment Formation}\label{apdx:envsound}

To prove Theorem \ref{thm:envsound}, we first prove the following, slightly weaker, lemma:

\begin{lemma}\label{fictxalpha}
	Assume an environment $(\fictx,\fsctx)$, a finite set of indices $I=\{1,...,n\}$ and
	let $\fictx' = \fictx[\fid_i\mapsto (\vec{\pvar x_i}, C_i, \pexp_i) | i \in I]$,
	with $\fid_i \not\in \dom(\fictx)$ for all $i \in I$.
	Then, for any ordinal $\alpha$, defining
	$$
		\fsctx(\alpha) = \fsctx[\fid_i\mapsto
				\{ \uspeconecase{\tprei(\beta)}{\mathcal{Q}_{i}(\beta)} \mid {\beta < \alpha}\}\cup\{\uspeconecase{\ntprei(\beta)}{\AssFalse} \mid {\beta \leq \alpha}\} \mid i \in I ],
	$$
	and assuming that
 \begin{equation}\label{validforall}
	\frall{i\in I, \alpha} \exsts{t'\in \fext_{\fictx', \fid_i}(\uspeconecase{\tprei(\alpha)}{\mathcal{Q}^{i}(\alpha)})}
	\fsctx(\alpha) \vdash C_i:t'
 \end{equation}
 and
 \begin{equation}\label{validforallntx}
	\frall{i\in I, \alpha} \exsts{t' \in \fext_{\fictx', \fid_i}(\uspeconecase{\ntprei(\alpha)}{\AssFalse})}
	\fsctx(\alpha) \vdash C_i:t'
 \end{equation}
it holds that
	\begin{equation*}
		\models(\fictx,\fsctx) \implies (\forall \alpha. \models(\fictx',\fsctx(\alpha))).
	\end{equation*}
\end{lemma}

\begin{proof}
    By transfinite induction on the ordinal $\alpha$, reasoning about the zero, successor and limit-ordinal cases.

    \myparagraph{Zero case}
    When $\alpha=0$, by definition we have $\fsctx(0)=\fsctx[\fid_i\mapsto\{\uspeconecase{\ntprei}{\AssFalse}\} | i\in I]$. Let $\fid\in\dom(\fictx')$: then,
    if $\fid\neq \fid_i, \forall i\in I$ and for any $t\in(\fsctx(0))(\fid)$, it holds that $t\in\fsctx(\fid)$ and from $\models(\fictx,\fsctx)$ we obtain the existence of a $t'\in\fext_{\fictx',\fid}(t)$ such that $\models C:t'$.

	Otherwise, $\fid = \fid_i$ for some $i\in I$, in which case $(\fsctx(0))(\fid)$ is a singleton set only holding the specification $t=\uspeconecase{\ntprei(0)}{\AssFalse}$.
	We instantiate the $n$-dimensional Scott induction of \S\ref{sec:ndimscott} with $\alpha=0$, which yields
	$$
	\frall{i\in I}
	\repr{C_i}\in\nSi
	$$
	that is, there exists a $C\in\repr{C_i}$ such that
	$$
	\fsctx\models \utripleq{\ntprei(0)}{C}{\AssFalse}
	$$
	The assumption $\models(\fictx,\fsctx)$ implies $\fictx\models\utripleq{\ntprei(0)}{C}{\AssFalse}$. Therefore, $C$ only calls on functions from $\dom(\fictx)$, allowing us to arbitrarily extend the domain of $\fictx$, yielding $\fictx'\models\utripleq{\ntprei(0)}{C}{\AssFalse}$, that is:
	\[
    \begin{array}{l}
\quad \forall \subst, \sto, \hp, \hp_\fid, \outcome, \sto', \hp''.~\subst, \sto, \hp \models \ntprei(0)\lstar\pvvar z\doteq\nil \implies \\
\qquad  (\sto, \hp \uplus \hp_\fid), \scmd \baction_{\fictx'} \outcome: (\sto', \hp'') \implies \\
    \qquad\quad
    \outcome \neq \oxm \land (\exists \hp'.~\hp'' = \hp' \uplus \hp_\fid \land \subst, \sto', \hp' \models \AssFalse)
    \end{array}
\]
	As $C$ and $C_i$ are behaviourally equivalent, we trivially obtain the same statement for $C_i$, i.e.
	$$
	\fictx'\models\utripleq{\ntprei(0)}{C_i}{\AssFalse}
	$$
	which concludes this case of the proof.

    \myparagraph{Successor case}
    In the successor case, we assume the inductive hypothesis (IH)~$\models(\fictx',\fsctx(\alpha))$ for an arbitrary ordinal $\alpha$, and we need to prove that $\models(\fictx',\fsctx(\alpha+1))$. By definition, (H1)~$\dom(\fsctx(\alpha)) = \dom(\fsctx(\alpha+1))$ and (H2)~$(\fsctx(\alpha))(\fid)\subset(\fsctx(\alpha+1))(\fid)$ holds for all $\fid$. Now, let $\fid$ be in $\dom(\fsctx(\alpha+1))$. From the definition, we obtain that
	\begin{equation*}
		(\fsctx(\alpha+1))(\fid) = \begin{cases}
		\fsctx(\fid), & \textrm{ for } \fid\neq \fid_i, \forall i \in I \\
			(\fsctx(\alpha))(\fid_i)\cup\{ \uspeconecase{P^{i}(\alpha)}{\bigq^{i}(\alpha)}\}, & \textrm{ for } \fid=\fid_i, \textrm{ for some } i\in I \\
			\hspace*{0.4cm}\cup~\{\uspeconecase{\ntprei(\alpha+1)}{\AssFalse}\} &
		\end{cases}
	\end{equation*}
	We prove the various cases separately.
	First, if $\fid\neq \fid_i\;\forall i \in I$, then the inductive hypothesis $\models(\fictx',\fsctx(\alpha))$ implies that
	\begin{equation*}
	\forall \fid \in \dom{(\fsctx(\alpha))}, \sspec \in (\fsctx(\alpha))(\fid).~\fid(\vec{\pvar{x}})\{C; \preturn{\pexp}\} \in \fictx' \Rightarrow \exists \sspec' \in \fext_{\fictx', \fid}(\sspec).~\fictx'\models C:\sspec'
	\end{equation*}
	Since in this case $(\fsctx(\alpha+1))(\fid) = \fsctx(\fid) = (\fsctx(\alpha))(\fid)$ and $\fictx(\fid)=\fictx'(\fid)$ since the domains of $\fictx$ and $\fsctx$ coincide, we immediately obtain the desired
	\begin{equation*}
		\frall{t \in (\fsctx(\alpha+1))(\fid)} \exsts{ t' \in \fext_{\fictx', \fid}(t)} \fictx'\models C:t'.
	\end{equation*}

	Next, let $\fid=\fid_i$ for some $i \in I$ and let $t\in(\fsctx(\alpha))(\fid_i)$. Then, the inductive hypothesis $\models(\fictx',\fsctx(\alpha))$ immediately implies the existence of a $t'\in \fext_{\fictx', \fid}(t)$ such that $\fictx'\models C_i:t'$.

	Next, let $t = \uspeconecase{P^{i}(\alpha)}{\bigq^{i}(\alpha)}$.
	From (\ref{validforall}), we know that there exists a $t'\in\fext_{\fictx',\fid_i}(t)$ such that
	$\fictx'\vdash C_i:t'$. From there,
	given Theorem~\ref{logicsoundness} and the inductive hypothesis $\models(\fictx',\fsctx(\alpha))$, we obtain that $\fictx'\models C_i:t'$.

	Finally, let $t = \uspeconecase{\ntprei(\alpha+1)}{\AssFalse}$.
	Per definition, we have \\
	$\fext_{\fictx',\fid_i}(t) = \{ \uspeconecase{\ntprei(\alpha+1)\lstar\pvvar z\doteq\nil}{\AssFalse} \}$, where $\pvvar z=\pv{C_i}\backslash\{\pvvar x_i\}$.
	We instantiate the $n$-dimensional Scott induction from \S\ref{sec:ndimscott} with $\alpha+1$ and obtain
	$$
	\repr{C_i}\in \sto_i^{\alpha+1}
	$$
	Analogously to the argumentation in the zero case, we obtain the desired
	$$
	\fictx'\models\utripleq{\ntprei(\alpha+1)}{C_i}{\AssFalse}
	$$
	concluding the successor case.
	\myparagraph{Limit Ordinal Case} Given an arbitrary limit ordinal~$\alpha$, we assume the inductive hypothesis (IH)~$\forall\lambda<\alpha.~\models(\fictx',\fsctx(\lambda))$, and our goal is to prove that $\models(\fictx',\fsctx(\alpha))$.

	Let $\fid$ be an arbitrary element from $\dom(\fsctx(\alpha))$. First, if $\fid\neq \fid_i, \forall i \in I$, then we have $(\fsctx(\alpha))(\fid)=\fsctx(\fid)$ and $\fictx'(\fid)=\fictx(\fid)$. The reasoning for this case is the same as for the first part of the successor case.
	Otherwise, we have that $\fid=\fid_i$ for some $i \in I$.
	For this case, we first prove that the following two sets are equal:
	\[
	\begin{array}{l}
		A =\{ \uspeconecase{P^{i}(\beta)}{\bigq^{i}(\beta)} \;|\; \beta<\alpha \} \cup \{ \uspeconecase{\ntprei(\beta)}{\AssFalse} \;|\; \beta\leq\alpha \}
		\text{~and~} \\
		B = \big(\bigcup\limits_{\beta<\alpha} (\fsctx(\beta))(\fid_i)\big) \cup \{\uspeconecase{\ntprei(\alpha)}{\AssFalse}\}.
	\end{array}
	\]
	For the left-to-right inclusion, let $t\in A$.

	\casex{Case 1.}
	$t=(P^{i}(\lambda))~(\bigq^{i}(\lambda))$ for some $\lambda<\alpha$. As $\alpha$ is a limit ordinal, there exists another ordinal $\beta$ such that $\lambda<\beta<\alpha$. By construction it holds that $t\in(\fsctx(\beta))(\fid_i)$ and hence $t\in B$.

	\casex{Case 2.}
	$t=\uspeconecase{\ntprei(\lambda)}{\AssFalse}$ for some $\lambda<\alpha$. Per construction, we have $t\in(\fsctx(\lambda))(\fid_i)$ and therefore $t\in B$.

	\casex{Case 3.}
	$t=\uspeconecase{\ntprei(\alpha)}{\AssFalse}$. Trivially, we have $t\in B$.

	\medskip\noindent
	For the right-to-left inclusion, let $t\in B$.

	\casex{Case 1.}
	$t\in(\fsctx(\lambda))(\fid_i)$ for some $\lambda<\alpha$. Then, per construction, $t\in A$.

	\casex{Case 2.}
	$t=\uspeconecase{\ntprei(\alpha)}{\AssFalse}$. Again, per construction, $t\in A$.

	\medskip
	This yields the equality of $A$ and $B$, which gives us that
	\begin{equation*}
		(\fsctx(\alpha))(\fid_i) = \big(\bigcup\limits_{\beta<\alpha} (\fsctx(\beta))(\fid_i)\big) \cup \{\uspeconecase{\ntprei(\alpha)}{\AssFalse}\}.
	\end{equation*}

	Now, let $t\in(\fsctx(\alpha))(\fid_i)$, and consider the following two cases:

	\casex{Case 1.} $t\in(\fsctx(\lambda))(\fid_i)$ for some $\lambda<\alpha$. The inductive hypothesis implies $\models(\fictx',\fsctx(\lambda))$ and therefore that $\fictx'\models C_i:t'$ for some $t'\in \fext_{\fictx', \fid_i}(t)$.

	\casex{Case 2.} $t=\uspeconecase{\ntprei(\alpha)}{\AssFalse}$. This is proven by instantiating the $n$-dimensional Scott induction from \S\ref{sec:ndimscott} with $\alpha$, which analogously to the zero case, yields $\fictx'\models\utripleq{\ntprei(\alpha)}{C_i}{\AssFalse}$, which conludes the proof.
\end{proof}

Before the final soundness proof, we require one further lemma:

\begin{lemma}[Existentialisation]\label{lem:existentialise}
Let $\models(\fictx,\fsctx)$, $\fictx(\fid)=(\pvvar x,C_\fid,\pexp')$, and let \\ $X =\{\xspec ~|~ x\in\ord\}$. Then, if $X \subseteq \fsctx(\fid)$, it holds that $\models(\fictx,\bar\fsctx)$, where
\[
\begin{array}{r l}
	\bar\fsctx = &\fsctx[\fid\mapsto (\fsctx(\fid)\\
	&\quad\quad\cup\{\exspec\} \setminus X]
\end{array}
\]
\end{lemma}

\begin{proof}
We need to prove that
$$
\exsts{t\in\fext_{\fictx,\fid}(\exspec)} \fictx\models C_\fid:t
$$
\myparagraph{Over-approximation}
Per construction, we know that the pre-condition of any internalisation of $\exspec$ equals $(\Pex)\lstar\pvvar z\doteq\nil$, where $\pvvar z = \pv{C}\backslash\pv{\Pex}$, which is equivalent to $\Pex\lstar\pvvar z\doteq\nil$. We assume $\subst,\sto,\hp,\sto',\hp'',\hp_\fid,\outcome$ such that
\begin{description}
\item[(O1)] $\subst,\sto,\hp\models \Pex\lstar\pvvar z\doteq\nil$
\item[(O2)] $(\sto,\hp\uplus \hp_\fid),C_\fid\baction_{\fictx}\outcome: (\sto',\hp'')$
\end{description}
and aim to show:
$$
\outcome\neq\oxm \land \exsts{\hp'} \hp''= \hp'\uplus \hp_\fid \land \subst,\sto',\hp'\models \Qeps
$$
where $\Qeps$ is an internal postcondition of \\ $\exspec$.
\begin{itemize}
	\item
(O1) implies that $\exsts{v\in\ord} \subst[x\rightarrow v],\sto,\hp\models\Px\lstar\pvvar z\doteq\nil$.
	\item
The assumptions of the lemma then imply that \textbf{(O3)}~$\outcome\neq\oxm \land \exsts{\hp'} \hp''= \hp'\uplus \hp_\fid \land \subst[x\rightarrow v],\sto',\hp'\models Q'_{\outcome}$ for some $Q'_{\outcome}$ that is an internal post-condition of \\ $\xspec$, where $x\in\ord$, i.e. we have either
$$
\begin{array}{l c}
\textbf{(O4a)}&\qquad (\outcome=ok) \land (\Qok(x)~\Leftrightarrow~\exsts{\vec p}\Qok'[\vec p / \pvvar p]\lstar\pvar{ret}\doteq \pexp'[\vec p / \pvvar p]) \\
\textbf{(O4b)}&\qquad (\outcome=err) \land (\Qerr(x)~\Leftrightarrow~\exsts{\vec p}\Qok'[\vec p / \pvvar p])
\end{array}
$$
where, without loss of generality, we may assume that $x\notin\vec p$.
	\item
(O3) implies that $\subst,\sto',\hp'\models \exsts{x}Q'_{\outcome}\lstar\dotin{x}{\ord}$.
	\item
To conclude the OX direction of the proof, we show that $\exsts{x}Q'_{\outcome}\lstar\dotin{x}{\ord}$ is an internal post-condition of $\exspec$
which is implied by (O4a) in case of successful, and by (O4b) in case of erroneous termination:
\begin{align*}
	\exsts{x}\Qok(x)\lstar\dotin{x}{\ord}
	~\Leftrightarrow~
	&\exsts{x}\exsts{\vec p}\Qok'[\vec p / \pvvar p]\lstar\pvar{ret}\doteq \pexp'[\vec p / \pvvar p]\lstar\dotin{x}{\ord} \\
	~\Leftrightarrow~
	&\exsts{\vec p}\exsts{x}\Qok'[\vec p / \pvvar p]\lstar\dotin{x}{\ord}\lstar\pvar{ret}\doteq \pexp'[\vec p / \pvvar p] \\
	~\Leftrightarrow~
	&\exsts{\vec p}(\exsts{x}\Qok'[\vec p / \pvvar p]\lstar\dotin{x}{\ord})\lstar\pvar{ret}\doteq \pexp'[\vec p / \pvvar p])\\
	~\Leftrightarrow~
	&\exsts{\vec p}(\exsts{x}\Qok'\lstar\dotin{x}{\ord})[\vec p / \pvvar p]\lstar\pvar{ret}\doteq \pexp'[\vec p / \pvvar p] \\
\end{align*}
\begin{align*}
	\exsts{x}\Qerr(x)\lstar\dotin{x}{\ord}
	~\Leftrightarrow~
	&\exsts{x}\exsts{\vec p}\Qerr'[\vec p / \pvvar p]\lstar\dotin{x}{\ord} \\
	~\Leftrightarrow~
	&\exsts{\vec p}\exsts{x}\Qerr'[\vec p / \pvvar p]\lstar\dotin{x}{\ord}\\
	~\Leftrightarrow~
	&\exsts{\vec p}(\exsts{x}\Qerr'\lstar\dotin{x}{\ord})[\vec p / \pvvar p]\\
\end{align*}
\end{itemize}
This concludes the OX direction.

\myparagraph{Under-approximation} We assume $\subst,\sto',\hp',\hp_\fid,\outcome$ such that $\hp'\sharp \hp_\fid$ and $\subst,\sto',\hp'\models\exsts{x} Q'_{\outcome}\lstar\dotin{x}{\ord}$, where $Q'_{\outcome}$ is obtained from the OX case. This implies the existence of a $v\in\ord$ such that
$$
	\subst[x\rightarrow v],\sto',\hp'\models Q'_{\outcome}
$$
From the assumptions, we obtain the existence of $\sto,\hp$ such that $\subst[x\rightarrow v],\sto,\hp\models P(x)\lstar\pvvar z\doteq\nil$, which implies
$$
	\subst,\sto,\hp\models \exsts{x} P(x)\lstar\pvvar z\doteq\nil\lstar\dotin{x}{\ord}
$$
This concludes the proof as the last assertion is the (only) internal pre-condition of of $\exspec$.

\end{proof}

With these lemmas, we can now easily prove~\cref{thm:envsound}:

\restateenvsound*

\begin{proof}[Proof of Theorem~\ref{thm:envsound}]
By induction on $\vdash (\fictx,\fsctx)$. When the last rule applied was the base rule for environments, we have that $(\fictx,\fsctx)=(\emptyset, \emptyset)$, meaning that $\dom(\fsctx)$ is empty, and the statement to prove is trivially true.
Otherwise,
we assume the following hypotheses and inductive hypothesis:
\begin{description}
\item[(H1)] $\vdash (\fictx, \fsctx)$
\item[(H2)] $\fictx' = \fictx[\fid_i\mapsto (\vec{\pvar x_i}, C_i, \pexp_i) | i \in I]$,
	with $\fid_i \not\in \dom(\fictx)$ for all $i \in I$
\item[(H3)] $\fsctx(\alpha) = \fsctx[\fid_i\mapsto
				\{ \uspec{\tprei(\beta)}{\Qoki(\beta)}{\Qerri(\beta)} \mid {\beta < \alpha}\} \cup \{\uspeconecase{\ntprei(\beta)}{\AssFalse} \mid \beta \leq \alpha\}]_{i \in I}$
\item[(H4)] $\frall{i\in I, \alpha\in\ord} \exists t \in \fext_{\fictx', \fid_i}\big(\uspec{\tprei(\alpha)}{\Qoki(\alpha)}{\Qerri(\alpha)}\big).~\fsctx(\alpha)\vdash C_i : \sspec$
\item[(H5)]
		$\frall{i\in I, \alpha\in\ord} \exists t \in \fext_{\fictx', \fid_i}(\uspeconecase{\ntprei(\alpha)}{\AssFalse}).~\fsctx(\alpha)\vdash C_i : \sspec$
\item[(IH)] $\models(\fictx, \fsctx)$
\end{description}
We will first prove that $\models(\fictx',\cup_{\alpha}\fsctx(\alpha))$, that is,
$$
\begin{array}{l}
\dom(\cup_{\alpha}\fsctx(\alpha))\subseteq\dom(\fictx') \land
\frall{\fid,\pvvar x, C, \pexp} \fid(\pvvar x)\{C;\preturn{\pexp}\}\in\fictx' \\
\qquad\Rightarrow
(\frall{t} t\in(\cup_{\alpha}\fsctx(\alpha))(\fid) \\
\qquad\qquad\Rightarrow
\exsts{t'\in\fext_{\fictx',\fid}(t)} \fictx'\models C:t')
\end{array}
$$
where the notation $\cup_{\alpha}\fsctx(\alpha)$ denotes the function which maps $\fid$ to the set $\cup_{\alpha}((\fsctx(\alpha))(\fid))$. The fist conjunct holds since Lemma \ref{fictxalpha} implies that $\dom(\fsctx(\alpha))\subseteq\fictx'$ for all $\alpha$. Now, assume $\fid,\pvvar x, C, \pexp, t$ such that
\begin{description}
	\item[(H6)] $\fid(\vec{\pvar{x}})\{C; \preturn{\pexp}\} \in \fictx'$
	\item[(H7)] $t\in(\cup_{\alpha}\fsctx(\alpha))(\fid)$
\end{description}
 By (H7), we have that there exists $\alpha'\in\ord$ such that \textbf{(H8)}~$t\in(\fsctx(\alpha'))(\fid)$.
Then, from Lemma~\ref{fictxalpha} applied to (IH), (H2), (H4), and (H5), we obtain
\textbf{(H9)}~$\models(\fictx',\fsctx(\alpha'))$. By construction, we know that $\dom(\cup_{\alpha}\fsctx(\alpha))=\dom(\fsctx(\alpha'))$, meaning that \textbf{(H10)}~$\fid \in \dom(\fsctx(\alpha'))$. Therefore, instantiating (H9) with (H10), (H8), and (H6), we obtain that there exists $t'\in \fext_{\fictx', \fid}(t)$ such that $\models C:t'$, which implies $\models(\fictx',\cup_{\alpha}\fsctx(\alpha))$. Defining
\begin{itemize}
\item $\tprei = \exsts{\alpha}\tprei(\alpha)\lstar\dotin{\alpha}{\ord}$
\item $\ntprei = \exsts{\alpha}\ntprei(\alpha)\lstar\dotin{\alpha}{\ord}$
\item $\Qoki = \exsts{\alpha}\Qoki(\alpha)\lstar\dotin{\alpha}{\ord}$
\item $\Qerri = \exsts{\alpha}\Qerri(\alpha)\lstar\dotin{\alpha}{\ord}$
\item $\fsctx'':=\fsctx[\fid_i\mapsto\{\uspec{\tprei}{\Qoki}{\Qerri}, \uspeconecase{\ntprei}{\AssFalse}\}]_{i\in I}$,
\end{itemize}
and applying Lemma \ref{lem:existentialise} twice to $\cup_{\alpha}\fsctx(\alpha)$ (once to the set of partially terminating specifications and once to the set of non-terminating specifications), we obtain $\models(\fictx',\fsctx'')$, concluding the soundness proof.
\end{proof}

%% file: sections/app-examples.tex
\newpage
\section{Further Examples}
\label{apdx:examples}

\subsection{Mutual Recursion: Even/Odd}

ESL allows us to reason about mutually recursive function. We illustrate this by using a simple example consisting of two functions which determine whether a natural number is even or odd, whose implementations are given below.
\[
\begin{array}{l}
	\pfunctions{isEven}{\pvar n} \\
	\tab \pifelses{\pvar n=0} \\
	\tab\tab \passign{\pvar b}{\ptrue} \\
	\tab \pifelsem \\
	\tab\tab \passign{\pvar n}{\pvar n -1}; \\
	\tab\tab \passign{\pvar b}{\mathtt{isOdd}(\pvar n)} \\
	\tab \pifelsee; \\
	\tab \preturn{\pvar b} \\
	\}
\end{array}
\quad\quad\quad
\begin{array}{l}
	\pfunctions{isOdd}{\pvar n} \\
	\tab \pifelses{\pvar n=0}; \\
	\tab\tab \passign{\pvar b}{\pfalse} \\
	\tab \pifelsem \\
	\tab\tab \passign{\pvar n}{\pvar n -1}; \\
	\tab\tab \passign{\pvar b}{\mathtt{isEven}(\pvar n)} \\
	\tab \pifelsee; \\
	\tab \preturn{\pvar b} \\
	\}
\end{array}
\]

To reason about these two functions, we introduce two (also mutually recursive) predicates:
\[
\begin{array}{l@{~\defeq~}l}
	\even{n} & n\doteq 0\vee\odd{n-1} \\
	\odd{n} & n\doteq 1\vee\even{n-1}. \\
\end{array}
\]
and give the pre-condition and the external ($Q_{even}(\alpha)$ and $Q_{odd}(\alpha)$) and internal ($Q'_{even}(\alpha)$ and $Q'_{odd}(\alpha)$) post-conditions for the two functions, noting that both share the same pre-condition $P(\alpha)$, and that we again use $\alpha=n$ for the decreasing measure, just as in the list length case:
\[
\begin{array}{r@{~}c@{~}l}
	P(\alpha) & \defeq & \pvar n\doteq n\lstar n\dotint\Nat\lstar n\doteq\alpha\\[2mm]
	Q_{even}(\alpha) & \defeq &  (\pvar{ret}\doteq\pfalse\lstar\odd{n}\lstar n\doteq\alpha)\vee(\pvar{ret}\doteq\ptrue\lstar\even{n}\lstar n\doteq\alpha) \\
	Q_{odd}(\alpha) & \defeq &  (\pvar{ret}\doteq\pfalse\lstar\even{n}\lstar n\doteq\alpha)\vee(\pvar{ret}\doteq\ptrue\lstar\odd{n}\lstar n\doteq\alpha) \\[2mm]
	Q'_{even}(\alpha) & \defeq &  (\pvar n\doteq n\lstar n\doteq 0\lstar\pvar b\doteq\ptrue\lstar n\doteq\alpha) \\
	&& \quad\lor (\pvar n\doteq n-1\lstar\pvar b\doteq\pfalse\lstar \odd{n}\lstar n\doteq\alpha) \\
	&& \quad\quad\lor(\pvar n\doteq n-1\lstar\pvar b\doteq\ptrue\lstar \even{n}\lstar n\doteq\alpha) \\
	Q'_{odd}(\alpha) & \defeq &  (\pvar n\doteq n\lstar n\doteq 0\lstar\pvar b\doteq\pfalse\lstar n\doteq\alpha) \\
	&& \quad\lor(\pvar n\doteq n-1\lstar\pvar b\doteq\ptrue\lstar \odd{n}\lstar n\doteq\alpha) \\
	&& \quad\quad\lor(\pvar n\doteq n-1\lstar\pvar b\doteq\pfalse\lstar \even{n}\lstar n\doteq\alpha) \\
\end{array}
\]

We further assume a well-formed environment $(\fictx,\fsctx)$ such that $\mathtt{isEven},\mathtt{isOdd}\not\in\dom(\fictx)$ and extend it as follows:
\[
\begin{array}{r c l}
\fictx' &\defeq& \fictx[\mathtt{isEven}\mapsto (\{ \pvar n \}, C_{even}, \pvar b), \mathtt{isOdd}\mapsto (\{ \pvar n \}, C_{odd}, \pvar b)] \\
\fsctx(\alpha) &\defeq& \fsctx[\mathtt{isEven}\mapsto
				\{ (P(\beta))~(\oxok: Q_{even}(\beta)) \mid {\beta < \alpha}\}, \\
				&&\quad\mathtt{isOdd}\mapsto
				\{ (P(\beta))~(\oxok: Q_{odd}(\beta))\mid {\beta < \alpha}\}]
\end{array}
\]
where $C_{even}$ and $C_{odd}$ denote the appropriate function bodies. Our goal is to prove
\[
\begin{array}{c}
\fsctx(\alpha)\vdash\utripleok{P(\alpha)}{C_{even}}{Q'_{even}(\alpha)} \\
\fsctx(\alpha)\vdash\utripleok{P(\alpha)}{C_{odd}}{Q'_{odd}(\alpha)}
\end{array}
\]
which we do in the proof sketches given below, first for $\mathtt{isEven}$ and then for $\mathtt{isOdd}$:

\[
\small
\begin{array}{r@{~}l}
	\fsctx(\alpha)~\vdash & \smallespecline{P(\alpha)\lstar\pvar b\doteq\nil}\\
	& \smallespecline{\pvar n\doteq n\lstar \dotin{n}{\Nat}\lstar n\doteq\alpha\lstar\pvar b\doteq\nil}\\
	& \pifelses{\pvar n=0} \\
	& \tab\smallespecline{\pvar n\doteq n\lstar n\doteq 0\lstar n\doteq\alpha \lstar\pvar b\doteq\nil} \\
	& \tab \passign{\pvar b}{\ptrue} \\
	& \tab\smallespecline{\pvar n\doteq n\lstar n\doteq 0\lstar n\doteq\alpha\lstar\pvar b\doteq\ptrue} \\
	& \pifelsem \\
	& \tab\smallespecline{\pvar n\doteq n\lstar n\dotgr0\lstar n\doteq\alpha\lstar\pvar b\doteq\nil} \\
	& \tab \passign{\pvar n}{\pvar n -1} \\
	& \tab\smallespecline{\pvar n\doteq n-1\lstar n\dotgr0\lstar n\doteq\alpha\lstar\pvar b\doteq\nil} \\
	& \tab\smallespecline{\pvar n\doteq n-1\lstar \dotin{n-1}{\Nat}\lstar n-1\doteq\alpha-1\lstar\pvar b\doteq\nil} \\
	& \tab \text{\polish{// as $\alpha-1<\alpha$, we can apply the specification for $\alpha-1$}} \\
	& \tab \passign{\pvar b}{\mathtt{isOdd}(\pvar n)} \\
	& \tab \smallespecline{\pvar n\doteq n-1\lstar\exsts{b} \big( (b\doteq\pfalse\lstar\even{n-1})\vee(b\doteq\ptrue\lstar\odd{n-1})\big)\lstar n\doteq\alpha\lstar \pvar b\doteq b}\\
	& \tab \smallespecline{\pvar n\doteq n-1\lstar\big( (\pvar b\doteq\pfalse\lstar\odd{n}\lstar n\doteq\alpha)\vee(\pvar b\doteq\ptrue\lstar\even{n}\lstar n\doteq\alpha)\big)}\\
	& \tab \smallespecline{(\pvar n\doteq n-1\lstar\pvar b\doteq\pfalse\lstar\odd{n}\lstar n\doteq\alpha)\vee(\pvar n\doteq n-1\lstar\pvar b\doteq\ptrue\lstar\even{n}\lstar n\doteq\alpha)}\\
	& \pifelsee \\
	& \especline{(\pvar n\doteq n\lstar n\doteq 0\lstar\pvar b\doteq\ptrue\lstar n\doteq\alpha)
	\vee (\pvar n\doteq n-1\lstar\pvar b\doteq\pfalse\lstar\odd{n}\lstar n\doteq\alpha)~\vee \\
	(\pvar n\doteq n-1\lstar\pvar b\doteq\ptrue\lstar\even{n}\lstar n\doteq\alpha)}
\end{array}
\]

\[
\small
\begin{array}{r@{~}l}
	\fsctx(\alpha)\vdash & \smallespecline{P(\alpha)\lstar\pvar b\doteq\nil}\\
	& \smallespecline{\pvar n\doteq n\lstar \dotin{n}{\Nat}\lstar n\doteq\alpha\lstar\pvar b\doteq\nil}\\
	& \pifelses{\pvar n=0} \\
	& \tab\smallespecline{\pvar n\doteq n\lstar n\doteq 0\lstar n\doteq\alpha\lstar\pvar b\doteq\nil} \\
	& \tab \passign{\pvar b}{\pfalse} \\
	& \tab\smallespecline{\pvar n\doteq n\lstar n\doteq 0\lstar n\doteq\alpha\lstar\pvar b\doteq\pfalse} \\
	& \pifelsem \\

	& \tab\smallespecline{\pvar n\doteq n\lstar n\dotgr0\lstar n\doteq\alpha\lstar\pvar b\doteq\nil} \\
	& \tab \passign{\pvar n}{\pvar n -1} \\
	& \tab\smallespecline{\pvar n\doteq n-1\lstar n\dotgr0\lstar n\doteq\alpha\lstar\pvar b\doteq\nil} \\
	& \tab\smallespecline{\pvar n\doteq n-1\lstar \dotin{n-1}{\Nat}\lstar n-1\doteq\alpha-1\lstar\pvar b\doteq\nil} \\
	& \tab \text{\polish{// since $\alpha-1<\alpha$, we can apply the specification for $\alpha-1$}} \\
	& \tab \passign{\pvar b}{\mathtt{isEven}(\pvar n)} \\
	& \tab \smallespecline{\pvar n\doteq n-1\lstar\exsts{b} \big( (b\doteq\pfalse\lstar\odd{n-1})\vee(b\doteq\ptrue\lstar\even{n-1})\big)\lstar n\doteq\alpha\lstar \pvar b\doteq b}\\
	& \tab \smallespecline{\pvar n\doteq n-1\lstar\big( (\pvar b\doteq\pfalse\lstar\even{n}\lstar n\doteq\alpha)\vee(\pvar b\doteq\ptrue\lstar\odd{n}\lstar n\doteq\alpha)\big)}\\
	& \tab \smallespecline{(\pvar n\doteq n-1\lstar\pvar b\doteq\pfalse\lstar\even{n}\lstar n\doteq\alpha)\vee(\pvar n\doteq n-1\lstar\pvar b\doteq\ptrue\lstar\odd{n}\lstar n\doteq\alpha)}\\
	& \pifelsee \\
	& \especline{(\pvar n\doteq n\lstar n\doteq 0 \lstar n\doteq\alpha\pvar b\doteq\pfalse) \vee
	 (\pvar n\doteq n-1\lstar\pvar b\doteq\pfalse\lstar\even{n}\lstar n\doteq\alpha)~\vee \\
	(\pvar n\doteq n-1\lstar\pvar b\doteq\ptrue\lstar\odd{n}\lstar n\doteq\alpha)}
\end{array}
\]

To complete the proof, we need to show that $Q_{i}\Leftrightarrow\exsts{\vec p} Q'_{i}[\vec p/\pvvar p]\lstar\pvar{ret}\doteq \pvar b[\vec p/\pvvar p]$ for $i\in\{\text{even, odd}\}$. As the two cases are analogous, we only show the $\mathtt{even}$ case in detail:
\[
\begin{array}{l l}
	& \exsts{\vec p} Q'_{even}(\alpha)[\vec p/\pvvar p]\lstar\pvar{ret}\doteq \pexp[\vec p/\pvvar p] \\
	\Leftrightarrow & \exists p_n, p_b. \begin{array}{l} ((p_n\doteq n\lstar n\doteq 0\lstar p_b\doteq\ptrue\lstar n\doteq\alpha)~\vee \\
                                                              (p_n\doteq n-1\lstar p_b\doteq\pfalse\lstar \odd{n}\lstar n\doteq\alpha)~\vee \\
                                                              (p_n\doteq n-1\lstar p_b\doteq\ptrue\lstar \even{n}\lstar n\doteq\alpha)) \lstar \pvar{ret}\doteq p_b \\
	\end{array}\\
	\Leftrightarrow & (n\doteq 0\lstar\pvar{ret}\doteq\ptrue\lstar\even{n}\lstar n\doteq\alpha) \vee (n~\dotgt~0\lstar\pvar{ret}\doteq\pfalse\lstar \odd{n}\lstar n\doteq\alpha)~\vee \\
	& (n~\dotgt~0\lstar\pvar{ret}\doteq\ptrue\lstar \even{n}\lstar n\doteq\alpha) \\
	\Leftrightarrow &  (n~\dotgt~0\lstar \pvar{ret}\doteq\pfalse\lstar \odd{n}\lstar n\doteq\alpha) \vee (n \in \Nat\lstar\pvar{ret}\doteq\ptrue\lstar \even{n}\lstar n\doteq\alpha) \\
	\Leftrightarrow & (\pvar{ret}\doteq\pfalse\lstar \odd{n}\lstar n\doteq\alpha)\vee(\pvar{ret}\doteq\ptrue\lstar \even{n}\lstar n\doteq\alpha) \\
	\Leftrightarrow & Q_{even}(\alpha)
\end{array}
\]

\subsection{More Complex Mutual Recursion: Even/Odd/List Length}
In the examples in the main body and the mutually recursive example in the previous section, the choice of the measure is straightforward. For, e.g., list length and list free, we traverse a non-cyclic list, therefore decreasing the distance to the end of the list in every step. For even/odd, each function decreases the function argument before passing it on to the other function, therefore also creating a natural measure.

In the real world, however, we might come across clusters of mutually recursive functions where not every function reduces the obvious measure (e.g., wrapper functions). As long as any function call terminates, we can still reason about such clusters by defining an appropriate measure. To illustrate this, we will look at a collection of three functions, which compute the length of a list in a convoluted, mutually-recursive way:
\[
\begin{array}{l}
	\pfunctions{LL}{\pvar x} \\
	\tab \pifelses{\pvar x=\nil} \\
	\tab\tab \passign{\pvar r}{0}; \\
	\tab \pifelsem \\
	\tab\tab \pderef{\pvar v}{\pvar x}; \\
	\tab\tab \pifelses{\even{v}} \\
	\tab\tab\tab \passign{\pvar r}{\pvar g(\pvar x)}; \\
	\tab\tab \pifelsem \\
	\tab\tab\tab \passign{\pvar r}{\pvar \fid(\pvar x)}; \\
	\tab\tab \pifelsee; \\
	\tab \pifelsee; \\
	\tab \preturn{\pvar r} \\
	\}\\
\end{array}
\hspace{1.5cm}
\begin{array}{l}
	\pfunctions{\fid}{\pvar x}\\
	\tab \pderef{\pvar v}{\pvar x}; \\
	\tab \pifelses{\even{v}} \\
	\tab\tab \pderef{\pvar x}{\pvar x+1}; \\
	\tab\tab \passign{\pvar r}{\mathtt{LL}(\pvar x)}; \\
	\tab\tab \passign{\pvar r}{\pvar r+1}; \\
	\tab \pifelsem \\
	\tab\tab \passign{\pvar r}{\pvar g(\pvar x)}; \\
	\tab \pifelsee; \\
	\tab \preturn{\pvar r} \\
	\}\\
	\quad \\
	\quad \\
\end{array}
\hspace{1.5cm}
\begin{array}{l}
	\pfunctions{g}{\pvar x}\\
	\tab \pderef{\pvar v}{\pvar x}; \\
	\tab \pifelses{\odd{v}} \\
	\tab\tab \pderef{\pvar x}{\pvar x+1}; \\
	\tab\tab \passign{\pvar r}{\mathtt{LL}(\pvar x)}; \\
	\tab\tab \passign{\pvar r}{\pvar r+1}; \\
	\tab \pifelsem \\
	\tab\tab \passign{\pvar r}{\pvar \fid(\pvar x)}; \\
	\tab \pifelsee; \\
	\tab \preturn{\pvar r} \\
	\} \\
	\quad \\
	\quad \\
\end{array}
\]

Intuitively, whenever either of the functions is called with an argument $\pvar x$, which is the head of a (non-cyclic) list, it computes the length of the list. The $\mathtt{LL}$ function calls $g$ if the first value of the list is even and $\fid$ otherwise. The function $g$, however, does the same test and calls $\fid$ if the input was even. Otherwise, it moves one element down the list and calls $\mathtt{LL}$ on the now shortened list. The function $\fid$ moves down the list by one element and calls $\mathtt{LL}$ on the shortened list if the first value is even, and calls $g$ on the initial list, if not.

As the functions branch on whether or not values of the list are divisible by 2, we adjust the $\llist{x,\lstvs}$ predicate slightly to include the condition that $\lstvs$ is a list of natural numbers:
$$\nlist{x,\lstvs}\defeq (x\doteq\nil\lstar\lstvs\doteq\epsilon) \vee
(\exsts{v,x,\lstvs'} x\mapsto v,x'\lstar v \dotint \Nat \lstar \nlist{x',\lstvs'}\lstar \lstvs\doteq v:\lstvs')$$
and furthermore require a trivial property of the previously introduced $\even{n}$ and $\odd{n}$ predicates, stating that
$\even{v}\vee\odd{v}\Leftrightarrow v\in\Nat$. Assuming a valid environment $\vdash(\fictx,\fsctx)$, we extend it as follows
\[
\begin{array}{r c l}
	\fictx' &\defeq& \fictx[\mathtt{LL}\mapsto(\{\pvar x\},C_{LL},\pvar r),
	\quad \mathtt{\fid}\mapsto(\{\pvar x\},C_{\fid},\pvar r),
	\quad \mathtt{g}\mapsto(\{\pvar x\},C_{g},\pvar r)] \\
	\fsctx(\alpha) &\defeq& \fsctx[
		\mathtt{LL}\mapsto\{(P_{LL}(\beta))(ok:Q_{LL})| \beta<\alpha \}, \\
		& & \quad\quad\mathtt{\fid}\mapsto\{(P_{\fid}(\beta))(ok:Q_{\fid})| \beta<\alpha \}, \mathtt{g}\mapsto\{(P_{g}(\beta))(ok:Q_{g}) | \beta<\alpha \}
	]
\end{array}
\]
Furthermore, we define
$$
\begin{array}{l}
\fsctx''\defeq \fsctx[
		\mathtt{LL}\mapsto\{(\nlist{x,\lstvs}\lstar\pvar x\doteq x \lstar \dotin{(3|\lstvs|+2)}{\ord}) \\
		\hspace*{3cm}(ok: \nlist{x,\lstvs}\lstar\pvar{ret}\doteq vs \lstar \dotin{(3|\lstvs|+2)}{\ord})\}, \\
		\hspace*{1.35cm}\mathtt{\fid}\mapsto\{(\nlist{x,v:\lstvs'}\lstar\pvar x\doteq x \lstar\dotin{3|v:\lstvs'|+(v \bmod 2)}{\ord})\\
		\hspace*{3cm}(ok:\nlist{x,v:\lstvs'}\lstar\pvar{ret}\doteq|v:\lstvs'|\lstar\dotin{3|v:\lstvs'|+(v \bmod 2))}{\ord} \}, \\
		\hspace*{1.35cm}\mathtt{g}\mapsto\{(\nlist{x,v:\lstvs'}\lstar\pvar x\doteq x \lstar\dotin{3|v:\lstvs'|+1-(v \bmod 2)}{\ord})\\
		\hspace*{3cm}(ok: \nlist{x,v:\lstvs'}\lstar\pvar{ret}\doteq|v:\lstvs'|\lstar\dotin{3|v:\lstvs'|+1-(v \bmod 2)}{\ord}) \}
	]
\end{array}
$$
and wish to prove $\vdash(\fictx',\fsctx'')$. To this end, we prove the following three specifications:
\[
\begin{array}{l}
	\fsctx(\alpha)\vdash\utripleok{P_{LL}(\alpha)\lstar\pvar r,\pvar v\doteq\nil}{C_{LL}}{Q'_{LL}(\alpha)}\\
	\fsctx(\alpha)\vdash\utripleok{P_{\fid}(\alpha)\lstar\pvar r,\pvar v\doteq\nil}{C_{\fid}}{Q'_{\fid}(\alpha)}\\
	\fsctx(\alpha)\vdash\utripleok{P_{g}(\alpha)\lstar\pvar r,\pvar v\doteq\nil}{C_{g}}{Q'_{g}(\alpha)}
\end{array}
\]
where $C_{LL}$, $C_{\fid}$, and $C_{g}$ denote the appropriate function bodies, and the function pre-conditions, capturing the function pre-conditions (with and without measure) and post-conditions (internal and external), are defined as follows:
\[
\begin{array}{r@{~}c@{~}l}
	P_{LL}(\alpha) &=& \nlist{x,\lstvs}\lstar\pvar x\doteq x \lstar\alpha\doteq 3|\lstvs|+2\\ 
	P_{\fid}(\alpha) &=& \nlist{x,v:\lstvs'}\lstar\pvar x\doteq x \lstar\alpha\doteq 3|v:\lstvs'|+(v \bmod 2)\\
	P_{g}(\alpha) &=& \nlist{x,v:\lstvs'}\lstar\pvar x\doteq x \lstar\alpha\doteq 3|v:\lstvs'|+1-(v \bmod 2) \\[2mm]
	Q'_{LL}(\alpha) &=& (x\doteq\nil\lstar\lstvs\doteq\epsilon\lstar\pvar x\doteq x\lstar\pvar r\doteq|\lstvs|\lstar\pvar v\doteq\nil\lstar\alpha\doteq 3|\lstvs|+2)~\vee \\
	& & (\exsts{v,\lstvs'}\nlist{x,v:\lstvs'}\lstar \lstvs\doteq v:\lstvs'\lstar\pvar x\doteq x\lstar \pvar v\doteq v\lstar\pvar r\doteq |\lstvs|\lstar\alpha\doteq 3|\lstvs|+2) \\
	Q'_{\fid}(\alpha) &=& \exsts{x'} x\mapsto v,x'\lstar\nlist{x',\lstvs'}\lstar\pvar v\doteq v\lstar\pvar r\doteq|v\cons\lstvs'|\\[-0.5mm] &&\lstar(\pvar x\doteq x'\lstar\even{v}\vee\pvar x\doteq x\lstar\odd{v})\lstar\alpha\doteq 3|v:\lstvs'|+(v \bmod 2) \\
	Q'_{g}(\alpha) &=& \exsts{x'} x\mapsto v,x'\lstar\nlist{x',\lstvs'}\lstar\pvar v\doteq v\lstar\pvar r\doteq|v\cons\lstvs'| \\ && \lstar(\pvar x\doteq x'\lstar\odd{v}\vee\pvar x\doteq x\lstar\even{v})\lstar\alpha\doteq 3|v:\lstvs'|+1-(v \bmod 2) \\[2mm]
	Q_{LL}(\alpha) &=& \nlist{x,\lstvs}\lstar\pvar{ret}\doteq|\lstvs|\lstar\alpha\doteq 3|\lstvs|+2\\
	Q_{\fid}(\alpha)  &=& \nlist{x,v:\lstvs'}\lstar\pvar{ret}\doteq|v:\lstvs'|\lstar\alpha\doteq 3|v:\lstvs'|+(v \bmod 2) \\
	Q_{g}(\alpha)  &=& \nlist{x,v:\lstvs'}\lstar\pvar{ret}\doteq|v:\lstvs'|\lstar\alpha\doteq 3|v:\lstvs'|+1-(v \bmod 2) \\
\end{array}
\]

We give the proof sketch for list length and $\mathtt{\fid}$ below. The proof sketch for $\mathtt{g}$ is analogous to that of $\mathtt{\fid}$.

\begingroup
\small
\allowdisplaybreaks
\setlength{\jot}{0pt}
\begin{align*}
	\fsctx(\alpha)\vdash
        &~\smallespecline{\nlist{x,\lstvs}\lstar\pvar x\doteq x\lstar\alpha\doteq 3|\lstvs|+2\lstar\pvar r,\pvar v\doteq\nil} \\
	&\tab\pifelses{\pvar x=\nil} \\
	&\tab\tab \smallespecline{\nlist{x,\lstvs}\lstar\pvar x\doteq x\lstar\alpha\doteq 3|\lstvs|+2\lstar\pvar r,\pvar v, \pvar x\doteq\nil} \\
	&\tab\tab \passign{\pvar r}{0}; \\
	&\tab\tab \smallespecline{\nlist{x,\lstvs}\lstar\pvar x\doteq x\lstar\alpha\doteq 3|\lstvs|+2\lstar\pvar r\doteq0\lstar\pvar x,\pvar v\doteq\nil} \\
	&\tab\tab \smallespecline{x\doteq\nil\lstar\lstvs\doteq\epsilon\lstar\pvar x\doteq x\lstar\pvar r\doteq|\lstvs|\lstar\alpha\doteq 3|\lstvs|+2\lstar\pvar v\doteq\nil} \\
	&\tab\pifelsem \\
	&\tab\tab \especline{\exsts{v,x',\lstvs'} x\mapsto v,x'\lstar v\dotint\Nat \lstar \nlist{x',\lstvs'}\lstar\lstvs\doteq v\cons\lstvs'\lstar\pvar x\doteq x\lstar \\ \alpha\doteq 3|\lstvs|+2\lstar\pvar r,\pvar v\doteq\nil} \\
	&\tab\tab \pderef{\pvar v}{\pvar x}; \\
	&\tab\tab \especline{\exsts{v, x',\lstvs'} x\mapsto v,x'\lstar v\dotint\Nat \lstar\nlist{x',\lstvs'}\lstar\lstvs\doteq v\cons\lstvs'\lstar\\
	                    \pvar x\doteq x\lstar\pvar v\doteq v\lstar\alpha\doteq 3|\lstvs|+2\lstar\pvar r\doteq\nil} \\
	&\tab\tab \smallespecline{\exsts{v, \lstvs'} \nlist{x,v\cons\lstvs'}\lstar\lstvs\doteq v\cons\lstvs'\lstar\pvar x\doteq x\lstar\pvar v\doteq v\lstar\alpha\doteq 3|\lstvs|+2\lstar\pvar r\doteq\nil} \\
	&\tab\tab\begin{leftvruled} { exists, equiv }
	\smallespecline{\nlist{x,v\cons\lstvs'}\lstar\lstvs\doteq v\cons\lstvs'\lstar\pvar x\doteq x\lstar\pvar v\doteq v\lstar \dotin{v}{\Nat}\lstar\alpha\doteq 3|\lstvs|+2\lstar\pvar r\doteq\nil} \\
	\pifelses{\even{\pvar v}} \\
	\tab \especline{\nlist{x,v\cons\lstvs'}\lstar\lstvs\doteq v\cons\lstvs'\lstar\pvar x\doteq x\lstar\pvar v\doteq v\lstar \even{v} \\
	\lstar\alpha-1\doteq 3|v:\lstvs'|+1-(v\bmod 2)\lstar\pvar r\doteq\nil} \\
	\tab\text{\polish{// as $\alpha-1<\alpha$, we can apply $g$'s specification for $\alpha-1$}} \\
	\begin{leftvruled} {fr, eq}
	\smallespecline{\pvar r\doteq\nil \lstar \pvar x\doteq x \lstar \nlist{x,v:\lstvs'}\lstar\alpha-1\doteq 3|v:\lstvs'|+1-(v\bmod 2) }\\
	\passign{\pvar r}{\mathtt{g}(\pvar x)}; \\
	\smallespecline{\pvar x\doteq x\lstar\nlist{x,v:\lstvs'}\lstar\pvar{r}\doteq|v:\lstvs'|\lstar\alpha-1\doteq 3|v:\lstvs'|+1-(v\bmod 2)}
	\end{leftvruled} \\
	\tab \smallespecline{\nlist{x,v:\lstvs'}\lstar \lstvs\doteq v:\lstvs'\lstar\pvar x\doteq x\lstar \pvar v\doteq v\lstar\even{v}\lstar\pvar r\doteq |\lstvs|\lstar\alpha\doteq 3|\lstvs|+2} \\
	\pifelsem \\
	\tab \especline{\nlist{x,v\cons\lstvs'}\lstar\lstvs\doteq v\cons\lstvs'\lstar\pvar x\doteq x\lstar\pvar v\doteq v\lstar \\
	\odd{v}\lstar\pvar r\doteq\nil\lstar\alpha-1\doteq 3|\lstvs|+(v\bmod 2)} \\
	\tab\text{\polish{// as $\alpha-1<\alpha$, we can apply $\fid$'s specification for $\alpha-1$ }} \\
	\begin{leftvruled} {fr, eq}
	\smallespecline{\nlist{x,v:\lstvs'}\lstar\pvar x\doteq x\lstar\pvar r\doteq\nil\lstar\alpha-1\doteq 3|\lstvs|+(v\bmod 2)}\\
	\passign{\pvar r}{\mathtt{\fid}(\pvar x)}; \\
	\smallespecline{\nlist{x,v:\lstvs'}\lstar\pvar x\doteq x\lstar\pvar{r}\doteq|v:\lstvs'|\lstar\alpha-1\doteq 3|\lstvs|+(v\bmod 2)}
	\end{leftvruled} \\
	\tab\smallespecline{\nlist{x,v:\lstvs'}\lstar \lstvs\doteq v:\lstvs'\lstar\pvar x\doteq x\lstar \pvar v\doteq v\lstar\odd{v}\lstar\pvar r\doteq |\lstvs|\lstar\alpha\doteq 3|\lstvs|+2} \\
	 \pifelsee \\
	\especline{\nlist{x,v:\lstvs'}\lstar \lstvs\doteq v:\lstvs'\lstar\pvar x\doteq x\lstar \pvar v\doteq v\lstar\pvar r\doteq |\lstvs|\lstar(\even{v}\vee\odd{v})\lstar \\ \alpha\doteq 3|\lstvs|+2}
	\tab\end{leftvruled} \\
	&\tab\tab\smallespecline{\exsts{v\in\Nat,\lstvs'}\nlist{x,v:\lstvs'}\lstar \lstvs\doteq v:\lstvs'\lstar\pvar x\doteq x\lstar \pvar v\doteq v\lstar\pvar r\doteq |\lstvs|\lstar\alpha\doteq 3|\lstvs|+2} \\
	&\tab\pifelsee \\
	&\tab\especline{(x\doteq\nil\lstar\lstvs\doteq\epsilon\lstar\pvar x\doteq x\lstar\pvar r\doteq|\lstvs|\lstar\pvar v\doteq\nil\lstar\alpha\doteq 3|\lstvs|+2)~\vee \\
	                (\exsts{v\in\Nat,\lstvs'}\nlist{x,v:\lstvs'}\lstar \lstvs\doteq v:\lstvs'\lstar\pvar x\doteq x\lstar \pvar v\doteq v\lstar\pvar r\doteq |\lstvs|\lstar\alpha\doteq 3|\lstvs|+2)}
\end{align*}

\begin{align*}
	\fsctx(\alpha)\vdash
	&~\smallespecline{\nlist{x,v:vs'}\lstar\pvar x\doteq x\lstar\alpha\doteq 3|v:vs'|+(v \bmod 2)\lstar\pvar v,\pvar r\doteq\nil} \\
	&\smallespecline{\exsts{x'} x\mapsto v,x'\lstar\dotin{v}{\Nat}\lstar \nlist{x',vs'}\lstar\pvar x\doteq x\lstar\alpha\doteq 3|v:vs'|+(v \bmod 2)\lstar\pvar v,\pvar r\doteq\nil} \\
	&\pderef{\pvar v}{\pvar x}; \\
	&\smallespecline{\exsts{x'} x\mapsto v,x'\lstar\dotin{v}{\Nat}\lstar \nlist{x',vs'}\lstar\pvar x\doteq x\lstar\pvar v\doteq v\lstar\alpha\doteq 3|v:vs'|+(v \bmod 2)\lstar\pvar r\doteq\nil} \\
	&\smallespecline{\exsts{x'} x\mapsto v,x'\lstar \nlist{x',vs'}\lstar\pvar x\doteq x\lstar\pvar v\doteq v\lstar\dotin{v}{\Nat}\lstar\alpha\doteq 3|v:vs'|+(v \bmod 2)\lstar\pvar r\doteq\nil} \\
	&\pifelses{\even{\pvar v}} \\
	&\tab \smallespecline{\exsts{x'} x\mapsto v,x'\lstar \nlist{x',vs'}\lstar\pvar x\doteq x\lstar\pvar v\doteq v\lstar\even{v}\lstar\alpha\doteq 3|v:vs'|\lstar\pvar r\doteq\nil} \\
	&\tab \pderef{\pvar x}{\pvar x+1}; \\
	&\tab \smallespecline{\exsts{x'} x\mapsto v,x'\lstar \nlist{x',vs'}\lstar\pvar x\doteq x'\lstar\pvar v\doteq v\lstar\even{v}\lstar\alpha-1\doteq 3|vs'|+2\lstar\pvar r\doteq\nil} \\
	&\tab \text{\polish{// as $3\alpha-1<\alpha$, we can apply LL's specifications for $\alpha-1$ }} \\
	&\tab\begin{leftvruled} {fr + ex}
	\tab \smallespecline{\nlist{x',vs'}\lstar\pvar x\doteq x'\lstar\pvar r\doteq\nil\lstar\alpha-1\doteq 3|vs'|+2} \\
	\tab \passign{\pvar r}{\mathtt{LL}(\pvar x)}; \\
	\tab \smallespecline{\nlist{x',vs'}\lstar\pvar x\doteq x'\lstar\pvar r\doteq|vs'|\lstar\alpha-1\doteq 3|vs'|+2} \\
	\tab \passign{\pvar r}{\pvar r+1}; \\
	\tab \smallespecline{\nlist{x',vs'}\lstar\pvar x\doteq x'\lstar\pvar r\doteq|vs'|+1\lstar\alpha-1\doteq 3|vs'|+2}
	\tab\end{leftvruled}\\
	&\tab \especline{\exsts{x'} x\mapsto v,x'\lstar\nlist{x',vs'}\lstar\pvar x\doteq x'\lstar\pvar v\doteq v\lstar\\
	                 \pvar r\doteq|v\cons\lstvs'|\lstar\even{v}\lstar\alpha\doteq 3|v:vs'|+(v\bmod 2)} \\
	&\pifelsem \\
	&\tab \especline{\exsts{x'} x\mapsto v,x'\lstar \nlist{x',vs'}\lstar\pvar x\doteq x\lstar\pvar v\doteq v\lstar\\
	                 \odd{\pvar v}\lstar\alpha-1\doteq 3|v:vs'|+1-(v\bmod 2)\lstar\pvar r\doteq\nil} \\
	&\tab \text{\polish{// as $\alpha-1<\alpha$, we can apply $g$'s specification for $\alpha-1$}} \\
	&\tab\begin{leftvruled} {fr + ex}
	\tab \smallespecline{\nlist{x,v:vs'}\lstar\pvar x\doteq x\lstar\pvar r\doteq\nil\lstar\alpha-1\doteq 3|v:vs'|+1-(v\bmod 2)} \\
	\tab \passign{\pvar r}{\mathtt{g}(\pvar x)}; \\
	\tab \smallespecline{\nlist{x,v:vs'}\lstar\pvar x\doteq x\lstar\pvar r\doteq|v:vs'|\lstar\alpha-1\doteq 3|v:vs'|+1-(v\bmod 2)}
	\tab\end{leftvruled} \\
	&\tab \especline{\exsts{x'} x\mapsto v,x'\lstar \nlist{x',vs'}\lstar\pvar x\doteq x\lstar\pvar v\doteq v\lstar\odd{\pvar v}\lstar\pvar r\doteq|v\cons\lstvs|\lstar \\ \alpha\doteq3|v:vs'|+(v\bmod 2)} \\
	&\pifelsee
	\especline{\exsts{x'} x\mapsto v,x'\lstar\nlist{x',vs'}\lstar\pvar v\doteq v\lstar\pvar r\doteq|v\cons\lstvs'|\lstar\\
                   (\pvar x\doteq x'\lstar\even{v}\vee\pvar x\doteq x\lstar\odd{v})\lstar\alpha\doteq3|v:vs'|+(v\bmod 2)}
\end{align*}
\endgroup

To conclude the proof, we need to show that $Q'_{LL}(\alpha)$, $Q'_{\fid}(\alpha)$ and $Q'_{g}(\alpha)$ are in the internalisations of their external counterparts, which is done as follows, again eliding the proof for $Q_g(\alpha)$ as it is analogous to that of $Q_\fid(\alpha)$.

\[
\begin{array}{l l}
	& \exsts{\vec p}Q'_{LL}(\alpha)[\vec p/\pvvar p]\lstar \pvar{ret}\doteq \pvar r[\vec p/\pvvar p] \\
	\Leftrightarrow & \exsts{p_x,p_r,p_v} \big((x\doteq\nil\lstar\lstvs\doteq\epsilon\lstar p_x\doteq x\lstar p_r\doteq|\lstvs|\lstar p_v\doteq\nil\lstar\alpha\doteq 3|\lstvs|+2)~\lor \\
	& \quad\quad(\exsts{v,\lstvs'}\nlist{x,v:\lstvs'}\lstar \lstvs\doteq v:\lstvs'\lstar p_x\doteq x\lstar p_v\doteq v\lstar p_r\doteq |\lstvs|\\
	& \quad\quad\quad\quad\lstar\alpha\doteq 3|\lstvs|+2)\big)\lstar\pvar{ret}\doteq p_r \\
	\Leftrightarrow & \nlist{x,\lstvs}\lstar\pvar{ret}\doteq|\lstvs|\lstar\alpha\doteq 3|\lstvs|+2 \\
\end{array}
\]

\[
\begin{array}{l l}
	& \exsts{\vec p}Q'_{\fid}(\alpha)[\vec p/\pvvar p]\lstar \pvar{ret}\doteq \pvar r[\vec p/\pvvar p] \\
	\Leftrightarrow & \exsts{p_x,p_v,p_r,x'}
	x\mapsto v,x'\lstar\nlist{x',\lstvs'}\lstar p_v\doteq v\lstar p_r\doteq|v\cons\lstvs'| \\
	& \quad\quad\lstar(p_x\doteq x'\lstar\even{v}\vee p_x\doteq x\lstar\odd{v})\lstar\alpha\doteq 3|v:\lstvs'|+(v \bmod 2) \\
	\Leftrightarrow & \exsts{x'} x\mapsto v,x'\lstar\nlist{x',\lstvs'}\lstar\pvar{ret}\doteq|v:\lstvs'|\\
	& \quad\quad\lstar(\even{v}\lstar\odd{v})\lstar\alpha\doteq 3|v:\lstvs'|+(v \bmod 2) \\
	\Leftrightarrow & \exsts{x'} x\mapsto v,x'\lstar\nlist{x',\lstvs'}\lstar\pvar{ret}\doteq|v:\lstvs'|\lstar\dotin{v}{\Nat}\lstar\alpha\doteq 3|v:\lstvs'|+(v \bmod 2)\\
	\Leftrightarrow &
	\nlist{x,v:\lstvs'}\lstar\pvar{ret}\doteq|v:\lstvs'|\lstar\alpha\doteq 3|v:\lstvs'|+(v \bmod 2)
\end{array}
\]



\subsection{List Swap}

We now consider a list-swap function that swaps the two first elements of a given list; if the given list is too short, an error is triggered.


\begingroup
\allowdisplaybreaks
\setlength{\jot}{0pt}
\begin{align*}
 &\especlines{\pvar x \doteq x \lstar \llist{x, n}} \\
 &\mathtt{LSwapFirstTwo}(\pvar x)~\{ \\
 &\tab\especlines{\oxok: \pvar x \doteq x \lstar \llist{x, n} \lstar \pvar a, \pvar b, \pvar y \doteq \nil} \\
 &\tab\pifelses{ \pvar{x} = \nil } \\
 &\tab\tab\especlines{\oxok: \pvar x \doteq x \lstar \llist{x, n} \lstar \pvar a, \pvar b, \pvar x, \pvar y \doteq \nil} \\
 &\tab\tab\perror{(\strlit{List~too~short!})} \\
 &\tab\tab\especlines{\oxerr: \pvar x \doteq x \lstar \llist{x, n} \lstar \pvar a, \pvar b, \pvar x, \pvar y \doteq \nil \lstar \pvar{err} \doteq \strlit{List~too~short!}} \\
 &\tab\pifelsem \\
 &\tab\tab\especlines{\oxok: \pvar x \doteq x \lstar \llist{x, n} \lstar \pvar a, \pvar b, \pvar y \doteq \nil \lstar \pvar x \dotneq \nil} \\
 &\tab\tab\especlines{\oxok: \exsts{v, x'} \pvar x \doteq x \lstar x \mapsto v, x' \lstar \llist{x', n - 1} \lstar \pvar a, \pvar b, \pvar y \doteq \nil} \\
 &\tab\tab\passign{\pvar y}{[\pvar x + 1]}; \\
 &\tab\tab\especlines{\oxok: \exsts{v} \pvar x \doteq x \lstar x \mapsto v, \pvar y \lstar \llist{\pvar y, n - 1} \lstar \pvar a, \pvar b \doteq \nil} \\
 &\tab\tab\pifelses{ \pvar{y} = \nil } \\
 &\tab\tab\tab\especlines{\oxok: \exsts{v} \pvar x \doteq x \lstar x \mapsto v, \pvar y \lstar \llist{\pvar y, n - 1} \lstar \pvar a, \pvar b \doteq \nil \lstar \pvar y \doteq \nil} \\
 &\tab\tab\tab\perror{(\strlit{List~too~short!})} \\
 &\tab\tab\tab\especline{\oxerr: \exsts{v} \pvar x \doteq x \lstar x \mapsto v, \pvar y \lstar \llist{\pvar y, n - 1} \lstar \pvar a, \pvar b \doteq \nil \lstar \pvar y \doteq \nil \\
                        \lstar \pvar{err} \doteq \strlit{List~too~short!}} \\
 &\tab\tab\pifelsem \\
 &\tab\tab\tab\especlines{\oxok: \exsts{v} \pvar x \doteq x \lstar x \mapsto v, \pvar y \lstar \llist{\pvar y, n - 1} \lstar \pvar a, \pvar b \doteq \nil \lstar \pvar y \dotneq \nil} \\
 &\tab\tab\tab\especline{\oxok: \exsts{v, v', x'} \pvar x \doteq x \lstar x \mapsto v, \pvar y \lstar \pvar y \mapsto v', x' \lstar \llist{x', n - 2} \lstar \pvar a, \pvar b \doteq \nil} \\
 &\tab\tab\tab\passign{\pvar a}{[\pvar x]}; \\
 &\tab\tab\tab\especline{\oxok: \exsts{v, v', x'} \pvar x \doteq x \lstar x \mapsto v, \pvar y \lstar \pvar y \mapsto v', x' \lstar \llist{x', n - 2} \lstar \pvar a \doteq v \lstar \pvar b \doteq \nil} \\
 &\tab\tab\tab\passign{\pvar b}{[\pvar y]}; \\
 &\tab\tab\tab\especline{\oxok: \exsts{v, v', x'} \pvar x \doteq x \lstar x \mapsto v, \pvar y \lstar \pvar y \mapsto v', x' \lstar \llist{x', n - 2} \lstar \pvar a \doteq v \lstar \pvar b \doteq v'} \\
 &\tab\tab\tab\passign{[\pvar y]}{\pvar a}; \\
 &\tab\tab\tab\especline{\oxok: \exsts{v, v', x'} \pvar x \doteq x \lstar x \mapsto v, \pvar y \lstar \pvar y \mapsto v, x' \lstar \llist{x', n - 2} \lstar \pvar a \doteq v \lstar \pvar b \doteq v'} \\
 &\tab\tab\tab\passign{[\pvar x]}{\pvar b} \\
 &\tab\tab\tab\especline{\oxok: \exsts{v, v', x'} \pvar x \doteq x \lstar x \mapsto v', \pvar y \lstar \pvar y \mapsto v, x' \lstar \llist{x', n - 2} \lstar \pvar a \doteq v \lstar \pvar b \doteq v'} \\
 &\tab\tab\tab\especline{\oxok: \exsts{x'} \pvar x \doteq x \lstar x \mapsto \pvar b, \pvar y \lstar \pvar y \mapsto \pvar a, x' \lstar \llist{x', n - 2}} \\
 &\tab\tab\} \\
 &\tab\tab\especline{\oxok: \exsts{x'} \pvar x \doteq x \lstar x \mapsto \pvar b, \pvar y \lstar \pvar y \mapsto \pvar a, x' \lstar \llist{x', n - 2} \\
                     \oxerr: \exsts{v} \pvar x \doteq x \lstar x \mapsto v, \pvar y \lstar \llist{\pvar y, n - 1} \lstar \pvar a, \pvar b \doteq \nil \lstar \pvar y \doteq \nil \lstar \\
                     \hspace*{1cm}\pvar{err} \doteq \strlit{List~too~short!}} \\
 &\tab\}; \\
 &\tab\especline{\oxok: \exsts{x'} \pvar x \doteq x \lstar x \mapsto \pvar b, \pvar y \lstar \pvar y \mapsto \pvar a, x' \lstar \llist{x', n - 2} \\
                 \oxerr: (\pvar x \doteq x \lstar \llist{x, n} \lstar \pvar a, \pvar b, \pvar x, \pvar y \doteq \nil \lstar \pvar{err} \doteq \strlit{List~too~short!})~\lor \\
                         \hspace*{1cm}(\exsts{v} \pvar x \doteq x \lstar x \mapsto v, \pvar y \lstar \llist{\pvar y, n - 1} \lstar \pvar a, \pvar b \doteq \nil \lstar \pvar y \doteq \nil \lstar \\
                         \hspace*{1cm}\pvar{err} \doteq \strlit{List~too~short!})} \\
 &\tab\preturn{\pvar \nil} \\
 &\} \\
 &\especline{\oxok: \llist{x, n} \lstar n \geq 2 \lstar \pvar{ret} \doteq \nil \\
            \oxerr: \llist{x, n} \lstar n < 2 \lstar \pvar{err} \doteq \strlit{List~too~short!}} \\
\end{align*}
\endgroup

\subsection{List Insert First}

We consider a list algorithm that inserts a new element $\pvar v$ at the beginning of a given list $\pvar x$:
\[
\begin{array}{l}
 \especlines{\pvar x\doteq x \lstar \pvar v \doteq v \lstar \llist{x, \lstxs, \lstvs}} \\
 \mathtt{LInsertFirst}(\pvar x, \pvar v)~\{ \\
 \tab\especlines{\pvar x\doteq x \lstar \pvar v \doteq v \lstar \llist{x, \lstxs, \lstvs} \lstar \pvar y \doteq \nil} \\
 \tab\palloc{\pvar y}{2}; \\
 \tab\especlines{\pvar x\doteq x \lstar \pvar v \doteq v \lstar \llist{x, \lstxs, \lstvs} \lstar \pvar y \mapsto \nil, \nil} \\
 \tab\pmutate{\pvar y}{v}; \\
 \tab\especlines{\pvar x\doteq x \lstar \pvar v \doteq v \lstar \llist{x, \lstxs, \lstvs} \lstar \pvar y \mapsto v, \nil} \\
 \tab\pmutate{\pvar y + 1}{x}; \\
 \tab\especlines{\pvar x\doteq x \lstar \pvar v \doteq v \lstar \llist{x, \lstxs, \lstvs} \lstar \pvar y \mapsto v, x} \\
 \tab\preturn{\pvar y} \\
 \} \\
 \especlines{\llist{x, \lstxs, \lstvs} \lstar \pvar{ret} \mapsto v, x} \\
 \especlines{\llist{\pvar{ret}, \pvar{ret} \cons \lstxs, v \cons \lstvs} \lstar \listptr{x, \lstxs}}
\end{array}
\]

\subsection{List Insert Last}

We consider an algorithm for inserting a new element $\pvar v$ at the end of a given list $\pvar x$:
\begingroup
\allowdisplaybreaks
\setlength{\jot}{0pt}
\begin{align*}
 &\especlines{\pvar x\doteq x \lstar \pvar v \doteq v \lstar \llist{x, \lstxs, \lstvs}} \\
 &\mathtt{LInsertLast}(\pvar x, \pvar v)~\{ \\
 &\tab\especlines{\pvar x\doteq x \lstar \pvar v \doteq v \lstar \llist{x, \lstxs, \lstvs} \lstar \pvar n, \pvar p, \pvar y \doteq \nil} \\
 &\tab\palloc{\pvar y}{2}; \\
 &\tab\especlines{\pvar x\doteq x \lstar \pvar v \doteq v \lstar \llist{x, \lstxs, \lstvs} \lstar \pvar y \mapsto \nil, \nil \lstar \pvar n, \pvar p \doteq \nil} \\
 &\tab\pmutate{\pvar y}{\pvar v}; \\
 &\tab\especlines{\pvar x\doteq x \lstar \pvar v \doteq v \lstar \llist{x, \lstxs, \lstvs} \lstar \pvar y \mapsto v, \nil \lstar \pvar n, \pvar p \doteq \nil} \\
 &\tab\pifelses{ \pvar{x} = \nil } \\
 &\tab\tab\especlines{\pvar x \doteq x \lstar \pvar v \doteq v \lstar \lstxs, \lstvs \doteq \emplist \lstar \pvar y \mapsto v, \nil \lstar \pvar n, \pvar p, \pvar x \doteq \nil} \\
 &\tab\tab\passign{\pvar x}{\pvar y} \\
 &\tab\tab\especlines{\pvar x \doteq \pvar y \lstar \pvar v \doteq v \lstar \llist{\pvar x, \lstxs \concat [\pvar x], \lstvs \concat [v]} \lstar \listptr{x, \lstxs} \lstar \pvar n, \pvar p, x \doteq \nil} \\
 &\tab\pifelsem \\
 &\tab\tab\especlines{\pvar x \doteq x \lstar \pvar v \doteq v \lstar \llist{x, \lstxs, \lstvs} \lstar \pvar y \mapsto v, \nil \lstar \pvar n, \pvar p \doteq \nil \lstar \pvar x \dotneq \nil} \\
 &\tab\tab\passign{\pvar p}{\pvar x}; \\
 &\tab\tab\especlines{\pvar x \doteq x \lstar \pvar v \doteq v \lstar \llist{x, \lstxs, \lstvs} \lstar \pvar y \mapsto v, \nil \lstar \pvar n \doteq \nil \lstar \pvar p \doteq \pvar x \lstar \pvar x \dotneq \nil} \\
 &\tab\tab\pderef{\pvar n}{\pvar x + 1}; \\
 &\tab\tab\especlines{P_0 \lstar \pvar x \doteq x \lstar \pvar v \doteq v \lstar \pvar y \mapsto v, \nil \lstar \pvar x \dotneq \nil \lstar \listptr{x, \lstxs}} \\
 &\tab\tab\pwhiles{\pvar n~\texttt{!=}~\nil} \\
 &\tab\tab\tab\especlines{P_i \lstar \pvar n \dotneq \nil} \\
 &\tab\tab\tab\passign{\pvar p}{\pvar n}; \\
 &\tab\tab\tab\especline{\exsts{\lstxs_\alpha, \lstxs_\beta, \lstvs_\alpha, v', \lstvs_\beta, p} \llseg{x, p, \lstxs_\alpha, \lstvs_\alpha} \lstar p \mapsto v', \pvar n \lstar \llist{\pvar n, \lstxs_\beta, \lstvs_\beta} \lstar \\
                        \lstxs \doteq \lstxs_\alpha \concat [p] \concat \lstxs_\beta \lstar \lstvs \doteq \lstvs_\alpha \concat [v'] \concat \lstvs_\beta \lstar |\lstxs_\alpha| = i \lstar \pvar n \dotneq \nil \lstar \pvar p \doteq \pvar n} \\
 &\tab\tab\tab\pderef{\pvar n}{\pvar n + 1} \\
 &\tab\tab\tab\especline{\exsts{\lstxs_\alpha, \lstxs_\beta, \lstvs_\alpha, v', v'', \lstvs_\beta, p} \llseg{x, p, \lstxs_\alpha, \lstvs_\alpha} \lstar p \mapsto v', \pvar p \lstar \pvar p \mapsto v'', \pvar n \lstar \llist{\pvar n, \lstxs_\beta, \lstvs_\beta} \lstar \\
                        \lstxs \doteq \lstxs_\alpha \concat [p, \pvar p] \concat \lstxs_\beta \lstar \lstvs \doteq \lstvs_\alpha \concat [v', v''] \concat \lstvs_\beta \lstar |\lstxs_\alpha| = i} \\
 &\tab\tab\tab\especline{\exsts{\lstxs_\alpha, \lstxs_\beta, \lstvs_\alpha, v'', \lstvs_\beta} \llseg{x, \pvar p, \lstxs_\alpha \concat [p], \lstvs_\alpha \concat [v']} \lstar \pvar p \mapsto v'', \pvar n \lstar \llist{\pvar n, \lstxs_\beta, \lstvs_\beta} \lstar \\
                        \lstxs \doteq (\lstxs_\alpha \concat [p]) \concat [\pvar p] \concat \lstxs_\beta \lstar \lstvs \doteq (\lstvs_\alpha \concat [v']) \concat [v''] \concat \lstvs_\beta \lstar |\lstxs_\alpha| = i} \\
 &\tab\tab\tab\especlines{P_{i+1}} \\
 &\tab\tab\pwhilee; \\
 &\tab\tab\especlines{\exsts{m} P_m \lstar \pvar n \doteq \nil \lstar \pvar x \doteq x \lstar \pvar v \doteq v \lstar \pvar y \mapsto v, \nil \lstar \pvar x \dotneq \nil \lstar \listptr{x, \lstxs}} \\
 &\tab\tab\especline{\exsts{v'} \llseg{x, \pvar p, \lstxs', \lstvs'} \lstar \lstxs \doteq \lstxs' \concat [\pvar p] \lstar \lstvs \doteq \lstvs' \concat [v'] \lstar \pvar p \mapsto v', \pvar n \lstar \\
                    \pvar n \doteq \nil \lstar \pvar x \doteq x \lstar \pvar v \doteq v \lstar \pvar y \mapsto v, \nil \lstar \pvar x \dotneq \nil \lstar \listptr{x, \lstxs}} \\
 &\tab\tab\pmutate{\pvar p + 1}{\pvar y} \\
 &\tab\tab\especline{\llist{\pvar x, \lstxs \concat [\pvar y], \lstvs \concat [v]} \lstar \pvar n \doteq \nil \lstar \pvar x \doteq x \lstar \pvar v \doteq v \lstar \pvar x \dotneq \nil \lstar \listptr{x, \lstxs}} \\
 &\tab\}; \\
 &\especline{(\pvar x \doteq \pvar y \lstar \pvar v \doteq v \lstar \llist{\pvar x, \lstxs \concat [\pvar x], \lstvs \concat [v]} \lstar \listptr{x, \lstxs} \lstar \pvar n, \pvar p, x \doteq \nil)~\lor \\
            (\llist{\pvar x, \lstxs \concat [\pvar y], \lstvs \concat [v]} \lstar \pvar n \doteq \nil \lstar \pvar x \doteq x \lstar \pvar v \doteq v \lstar \pvar x \dotneq \nil \lstar \listptr{x, \lstxs})} \\
 &\tab\preturn{\pvar x} \\
 &\} \\
 &\especlines{\exsts{x'} \llist{\pvar{ret}, \lstxs \concat [x'], \lstvs \concat [v]} \lstar \listptr{x, \lstxs}}
\end{align*}
\endgroup
where $P_i$ denotes the loop invariant $\exsts{\lstxs_\alpha, \lstxs_\beta, \lstvs_\alpha, v', \lstvs_\beta} \llseg{x, \pvar p, \lstxs_\alpha, \lstvs_\alpha} \lstar \pvar p \mapsto v', \pvar n \lstar \llist{\pvar n, \lstxs_\beta, \lstvs_\beta} \lstar \lstxs \doteq \lstxs_\alpha \concat [\pvar p] \concat \lstxs_\beta \lstar \lstvs \doteq \lstvs_\alpha \concat [v'] \concat \lstvs_\beta \lstar |\lstxs_\alpha| = i$.

\subsection{List Free}

We consider the list-free algorithm, $\mathtt{LFree}(\pvar x)$, which frees all the nodes of a given singly-linked list starting at $\pvar x$. The algorithm is implemented as follows:
\[
\begin{array}{l}
	\mathtt{LFree}(\pvar x)~\{ \\
	\tab\pifelses{ \pvar{x} = \nil } \\
	\tab\tab\pskip \\ \tab\pifelsem \\
	\tab\tab\passign{\pvar{y}}{\pvar x}; \passign{\pvar x}{[\pvar x+1]}; \\
	\tab\tab\pdealloc{\pvar y}; \pdealloc{\pvar y+1}; \\
	\tab\tab\passign{\pvar r}{\mathtt{LFree}(\pvar x)} \\
	\tab\}; \\
	\tab\preturn{\nil} \\
	\} \\
\end{array}
\]
The proof sketch of the body of the algorithm is given below. The measure we use to handle recursion is the length of the list, which corresponds to the number of pointers, $\alpha \defeq |\lstxs|$. 
\[
\small
\begin{array}{r@{~}l}
	\fsctx(\alpha)\vdash
	& \smallespecline{\pvar x\doteq x\lstar\llist{x,\lstxs}\lstar\alpha\doteq|\lstxs|\lstar\pvar r,\pvar y\doteq\nil} \\
	& \pifelses{ \pvar{x} = \nil }\\
	& \tab\smallespecline{\pvar x\doteq x\lstar x\doteq\nil\lstar\llist{x,\lstxs}\lstar\alpha\doteq|\lstxs|\lstar\pvar r,\pvar y\doteq\nil} \\
	& \tab\pskip \\
	& \tab\smallespecline{\pvar x\doteq x\lstar x\doteq\nil\lstar\llist{x,\lstxs}\lstar\alpha\doteq|\lstxs|\lstar\pvar r,\pvar y\doteq\nil} \\
	& \tab\smallespecline{\pvar x\doteq x\lstar x\doteq\nil\lstar\lstxs\doteq\epsilon \lstar \alpha\doteq|\lstxs|\lstar\pvar r,\pvar y\doteq\nil} \\
	& \pifelsem \\
	& \tab\smallespecline{\pvar x\doteq x\lstar x\dotneq\nil\lstar\llist{x,\lstxs}\lstar\alpha\doteq|\lstxs|\lstar\pvar r,\pvar y\doteq\nil} \\
	& \tab\smallespecline{\exsts{x',v,\lstxs'}\pvar x\doteq x\lstar x\mapsto v,x'\lstar\llist{x',\lstxs'}\lstar\lstxs\doteq x:\lstxs'\lstar\alpha\doteq|\lstxs|\lstar\pvar r,\pvar y\doteq\nil} \\
	& \tab\passign{\pvar{y}}{\pvar x}; \\
	& \tab\smallespecline{\exsts{x',v,\lstxs'}\pvar x\doteq x\lstar\pvar y\doteq x\lstar x\mapsto v,x'\lstar\llist{x',\lstxs'}\lstar\lstxs\doteq x:\lstxs'\lstar\alpha\doteq|\lstxs|\lstar\pvar r\doteq\nil} \\
	& \tab\passign{\pvar x}{[\pvar x+1]}; \\
	& \tab\smallespecline{\exsts{x',v,\lstxs'}\pvar x\doteq x'\lstar\pvar y\doteq x\lstar x\mapsto v,x'\lstar\llist{x',\lstxs'}\lstar\lstxs\doteq x:\lstxs'\lstar\alpha\doteq|\lstxs|\lstar\pvar r\doteq\nil} \\
	& \tab\smallespecline{\exsts{x',v,\lstxs'}\pvar x\doteq x'\lstar\pvar y\doteq x\lstar \pvar y\mapsto v,x'\lstar\llist{x',\lstxs'}\lstar\lstxs\doteq x:\lstxs'\lstar\alpha\doteq|\lstxs|\lstar\pvar r\doteq\nil} \\
	& \tab\pdealloc{\pvar y}; \\
	& \tab\smallespecline{\exsts{x',\lstxs'}\pvar x\doteq x'\lstar\pvar y\doteq x\lstar \pvar y\mapsto \cfreed,x'\lstar\llist{x',\lstxs'}\lstar\lstxs\doteq x:\lstxs'\lstar\alpha-1\doteq|\lstxs'|\lstar\pvar r\doteq\nil} \\
	& \tab\pdealloc{\pvar y+1}; \\
	& \tab\text{\polish{// as $\alpha-1<\alpha$, we can apply the specification for $\alpha-1$ }} \\
	& \tab\smallespecline{\exsts{x',\lstxs'}\pvar x\doteq x'\lstar\pvar y\doteq x\lstar \pvar y\mapsto \cfreed,\cfreed\lstar\llist{x',\lstxs'}\lstar\lstxs\doteq x:\lstxs'\lstar\alpha-1\doteq|\lstxs'|\lstar\pvar r\doteq\nil} \\
	& \tab\passign{\pvar r}{\mathtt{LFree}(\pvar x)} \\
	& \tab\especline{ \exsts{x',\lstxs'}\pvar x\doteq x'\lstar\pvar y\doteq x\lstar x \mapsto \cfreed,\cfreed\lstar\pfreed{\lstxs'} \lstar \listptr{x', \lstxs'} \lstar\lstxs\doteq x:\lstxs'\lstar \\
                          \alpha-1\doteq|\lstxs'|\lstar\pvar r\doteq\nil } \\
	& \}; \\
	& \especline{ (\pvar x\doteq x\lstar x\doteq\nil\lstar\lstxs\doteq\epsilon\lstar\pvar r,\pvar y\doteq\nil\lstar \alpha\doteq|\lstxs|)~\vee \\
	              (\exsts{x',\lstxs'}\pvar x\doteq x'\lstar\pvar y\doteq x\lstar x \mapsto \cfreed,\cfreed\lstar\pfreed{\lstxs'} \lstar \listptr{x', \lstxs'} \lstar\lstxs\doteq x:\lstxs'\lstar \\
                       \pvar r\doteq\nil\lstar\alpha\doteq|\lstxs|)}
\end{array}
\]

We conclude the proof by moving from the internal to the external specification:
$$
\begin{array}{r l}
	&\exists p_x, p_y, p_r.~\pvar{ret} \doteq \nil \lstar
	(p_x\doteq x\lstar x\doteq\nil\lstar\lstxs\doteq\epsilon\lstar p_r,p_y\doteq\nil\lstar\alpha\doteq|\lstxs|)~\vee \\
	 &\quad\quad(\exsts{x',\lstxs'}p_x\doteq x'\lstar p_y\doteq x\lstar x \mapsto \cfreed,\cfreed\lstar\pfreed{\lstxs'}\lstar\listptr{x', \lstxs'}\lstar\lstxs\doteq x:\lstxs'\lstar \\
	&\quad\quad p_r\doteq\nil\lstar\alpha\doteq|\lstxs|)\\
\Leftrightarrow &
\pvar{ret} \doteq \nil \lstar\alpha\doteq|\lstxs|\lstar
( (x\doteq\nil\lstar\lstxs\doteq\epsilon) \lor (\exsts{\lstxs'}x \mapsto \cfreed,\cfreed\lstar\pfreed{\lstxs'}\lstar\lstxs\doteq x:\lstxs')) \\
 \Leftrightarrow&\pfreed{\lstxs} \lstar \listptr{x, \lstxs} \lstar \pvar{ret} \doteq \nil \lstar\alpha\doteq|\lstxs|
\end{array}
$$

\subsection{List Algorithm Client}
We consider the following client of our proven-correct list algorithms
\[
\begin{array}{l}
 \mathtt{LClient}(\pvar x)~\{ \\
 \tab\passign{\pvar{l}}{\mathtt{LLen(\pvar x)}}; \\
 \tab\pifelses{ \pvar{l} < 5 } \\
 \tab\tab\passign{\pvar{r}}{\mathtt{LFree(\pvar x)}}; \\
 \tab\tab\perror(\strlit{List~too~short!}) \\
 \tab\pifelsem \\
 \tab\tab\pifelses{ \pvar{l} > 10} \\
 \tab\tab\tab\pwhiles{\true} ~\pskip~\} \\
 \tab\tab\pifelsem \\
 \tab\tab\tab\passign{\pvar{r}}{\mathtt{LPRev(\pvar l)}} \\
 \tab\tab\} \\
 \tab\}; \\
 \tab\preturn{\pvar{r}} \\
 \} \\
\end{array}
\]
and prove that it satisfies the following ESL specification:
\[
\begin{array}{c}
 \smallespecline{\pvar x\doteq x \lstar  \llist{x, \lstxs}} \\
 {\mathtt{LClient}(\pvar x)} \\
 \smallespecline{\oxok: 5 \dotle |\lstxs| \dotle 10 \lstar \llist{\pvar{ret}, \lstxs^\dagger} \lstar \listptr{x, \lstxs}} \\
 \smallespecline{\oxerr: |\lstxs| \dotlt 5 \lstar \pfreed{\lstxs} \lstar \listptr{x, \lstxs} \lstar \pvar{err} \doteq \strlit{List~too~short!}}
\end{array}
\]

We prove the three branches of the client separately, exposing the non-terminating case, and then join the obtained specifications through the admissible disjunction property, to obtain the above specification. We give two of the three proofs for the branches below; the third is analogous to the first. We denote the client's success post-condition by $\Qok$, the client's faulting post-condition by $\Qerr$, and assume a specification context $\fsctx$ that has the appropriate specifications of the called functions.
\begingroup
\small
\allowdisplaybreaks
\setlength{\jot}{0pt}
\begin{align*}
	\fsctx \vdash
	&~ \smallespecline{\pvar x \doteq x \lstar \llist{x, \lstxs} \lstar |\lstxs| \dotlt 5} \\
	&\mathtt{LClient}(\pvar x)~\{ \\
	&\tab\smallespecline{\pvar x \doteq x \lstar \llist{x, \lstxs} \lstar |\lstxs| \dotlt 5 \lstar \pvar l, \pvar r \doteq \nil} \\
	&\tab\passign{\pvar{l}}{\mathtt{LLen(\pvar x)}}; \\
	&\tab\especline{{\pvar x \doteq x \lstar \llist{x, \lstxs}} \lstar |\lstxs| \dotlt 5 \lstar \pvar l \doteq |\lstxs| \lstar \pvar r \doteq \nil} \\
	&\tab\pifelses{ \pvar{l} < 5 } \\
	&\tab\tab\especline{\pvar x \doteq x \lstar \llist{x, \lstxs} \lstar |\lstxs| \dotlt 5 \lstar \pvar l \doteq |\lstxs| \lstar \pvar r \doteq \nil} \\
	&\tab\tab\passign{\pvar{r}}{\mathtt{LFree(\pvar x)}}; \\
	&\tab\tab\especline{\pvar x \doteq x \lstar \pfreed{\lstxs} \lstar \listptr{x, \lstxs} \lstar |\lstxs| \dotlt 5 \lstar \pvar l \doteq |\lstxs| \lstar \pvar r \doteq \nil} \\
	&\tab\tab\perror(\strlit{LTS}) \\
	&\tab\tab\smallespecline{\oxerr:~\Qerr \lstar \pvar x \doteq x \lstar \pvar l \doteq |\lstxs| \lstar \pvar r \doteq \nil} \\
	&\tab \pifelsem \\
	&\tab\tab\smallespecline{\AssFalse}~\ldots~\smallespecline{\AssFalse} \\
	&\tab\}; \\
	&\tab\smallespecline{\oxerr:~\Qerr \lstar \pvar x \doteq x \lstar \pvar l \doteq |\lstxs| \lstar \pvar r \doteq \nil } \\
	&\tab\preturn{\pvar{r}} \\
	&\} \\
	& \especline{\oxerr: \Qerr \lstar \exists p_x, p_l, p_r.~p_x \doteq x \lstar p_l \doteq |\lstxs| \lstar p_r \doteq \nil} \\
	& \smallespecline{\oxerr: \Qerr }
\end{align*}
\begin{align*}
	\fsctx \vdash
	 &~ \smallespecline{\pvar x \doteq x \lstar \llist{x, \lstxs} \lstar |\lstxs| \dotgt 10} \\
	&\mathtt{LClient}(\pvar x)~\{ \\
	&\tab\especline{\pvar x \doteq x \lstar \llist{x, \lstxs} \lstar |\lstxs| \dotgt 10 \lstar \pvar l, \pvar r \doteq \nil} \\
	&\tab\passign{\pvar{l}}{\mathtt{LLen(\pvar x)}}; \\
	&\tab\especline{{\pvar x \doteq x \lstar \llist{x, \lstxs}} \lstar |\lstxs| \dotgt 10 \lstar \pvar l \doteq |\lstxs| \lstar \pvar r \doteq \nil} \\
	&\tab\pifelses{ \pvar{l} < 5 } \\
	&\tab\tab\smallespecline{\AssFalse}~\ldots~\smallespecline{\AssFalse} \\
	&\tab \pifelsem \\
	&\tab\tab\especline{{\pvar x \doteq x \lstar \llist{x, \lstxs}} \lstar |\lstxs| \dotgt 10 \lstar \pvar l \doteq |\lstxs| \lstar \pvar r \doteq \nil} \\
	&\tab\tab\pifelses{ \pvar{l} > 10} \\
	&\tab\tab\tab\especline{{\pvar x \doteq x \lstar \llist{x, \lstxs}} \lstar |\lstxs| \dotgt 10 \lstar \pvar l \doteq |\lstxs| \lstar \pvar r \doteq \nil} \\
	&\tab\tab\tab \pwhiles{\true} ~\pskip~\} \\
	&\tab\tab\tab\smallespecline{\AssFalse} \\
	&\tab\tab\pifelsem \\
	&\tab\tab\tab\smallespecline{\AssFalse}~\ldots~\smallespecline{\AssFalse} \\
	&\tab\tab\} \\
	&\tab\}; \\
        &\tab\smallespecline{\AssFalse} \\
	&\tab\preturn{\pvar{r}} \\
	&\} \\
        &\smallespecline{\AssFalse}
\end{align*}%
\endgroup

The three obtained specifications then yield via disjunction:
$$
\small
\uquadruple{(\pvar x\doteq x \lstar  \llist{x, \lstxs} \lstar |\lstxs| \dotlt 5)~\lor \\ (\pvar x\doteq x \lstar  \llist{x, \lstxs} \lstar 5 \dotle |\lstxs| \dotle 10)~\lor \\ (\pvar x\doteq x \lstar  \llist{x, \lstxs} \lstar |\lstxs| \dotgt 10) }{\mathtt{LClient}(\pvar x)}{\AssFalse \lor \Qok \lor \AssFalse}{\Qerr \lor \AssFalse \lor \AssFalse}
$$
and via equivalence the desired
$$
\small
\uquadruple{\pvar x\doteq x \lstar \llist{x, \lstxs}}{\mathtt{LClient}(\pvar x)}{\Qok}{\Qerr}
$$

\subsection{Binary Search Tree: BSTFindMin}

Now, we illustrate reasoning about binary search tree (BST) algorithms. We start with the following search algorithm, where we omit the case when the tree is empty to make the proof shorter and more focused:
\[
\begin{array}{l}
 	\mathtt{BSTFindMin}(\pvar x)~\{ \\
	\tab \pderef{\pvar l}{\pvar x +1};\\
 	\tab \pifelses{\pvar l = \nil} \\
 	\tab\tab \pderef{\pvar{min}}{\pvar x} \\
	\tab \pifelsem \\
 	\tab\tab \passign{\pvar{min}}{\mathtt{BSTFindMin}(\pvar l)} \\
	\tab \}; \\
	\tab \preturn{\pvar{min}} \\
	\pifelsee \\
\end{array}
\]
The specification we would like to prove is:
\[
\begin{array}{c}
 {\color{blue} (\pvar x \doteq x \lstar x \neq \nil \lstar \bst{x,K})} \\
 {\mathtt{BSTFindMin}(\pvar x)} \\
 {\color{blue} (x\neq\nil\lstar\bst{x,K}\lstar\pvar{ret}\doteq\min(K))}
\end{array}
\]
The proof sketch is as follows, using the measure that $\alpha$ is set-cardinality $\card{K}$:
\[
\small
\begin{array}{l}
	\fsctx(\alpha)\vdash \\
	\smallespecline{\pvar x \doteq x \lstar x \neq \nil \lstar \bst{x,K} \lstar \alpha\doteq\card{K} \lstar \pvar l,\pvar{min}\doteq\nil} \\
	\especline{\exsts{k,l,r,K_l,K_r} \pvar x\doteq x\lstar x\mapsto k,l,r \lstar \bst{r,K_r} \lstar \bst{l,K_l} \lstar K\doteq K_l\uplus\{k\}\uplus K_r
	\lstar K_l<k \lstar k<K_r\lstar \\
	\alpha\doteq\card{K}\lstar\pvar l,\pvar{min}\doteq\nil} \\
	\pderef{\pvar l}{\pvar x +1};\\
	\especline{\exsts{k,l,r,K_l,K_r} \pvar x\doteq x\lstar x\mapsto k,l,r \lstar \bst{r,K_r} \lstar \bst{l,K_l} \lstar K\doteq K_l\uplus\{k\}\uplus K_r
	\lstar K_l<k \lstar k<K_r\lstar \\
	\alpha\doteq\card{K}\lstar\pvar l\doteq l\lstar\pvar{min}\doteq\nil} \\
 	\pifelses{\pvar l = \nil} \\
	\tab \especline{\exsts{k,l,r,K_l,K_r} \pvar x\doteq x\lstar x\mapsto k,l,r \lstar \bst{r,K_r} \lstar \bst{l,K_l} \lstar K\doteq K_l\uplus\{k\}\uplus K_r
	\lstar K_l<k \lstar k<K_r\lstar \\
	\alpha\doteq\card{K}\lstar\pvar l\doteq l\lstar l\doteq\nil\lstar\pvar{min}\doteq\nil} \\
 	\tab \pderef{\pvar{min}}{\pvar x} \\
	\tab \especline{\exsts{k,l,r,K_l,K_r} \pvar x\doteq x\lstar x\mapsto k,l,r \lstar \bst{r,K_r} \lstar \bst{l,K_l} \lstar K\doteq K_l\uplus\{k\}\uplus K_r
	\lstar K_l<k \lstar k<K_r\lstar \\
	\alpha\doteq\card{K}\lstar\pvar l\doteq l\lstar l\doteq\nil\lstar\pvar{min}\doteq k} \\
	\tab \text{\polish{// $\bst{l,K_l}\lstar l\doteq\nil\Rightarrow K_l \doteq \emptyset \lstar \AssTrue$ which with $K\doteq K_l\uplus\{k\}\uplus K_r$ and $k<K_r$ implies $k=\min{K}$}} \\
	\tab \especline{\exsts{k,l,r,K_l,K_r} \pvar x\doteq x\lstar x\mapsto k,l,r \lstar \bst{r,K_r} \lstar \bst{l,K_l} \lstar K\doteq K_l\uplus\{k\}\uplus K_r
	\lstar K_l<k \lstar k<K_r\lstar \\
	\alpha\doteq\card{K}\lstar\pvar l\doteq l\lstar l\doteq\nil\lstar\pvar{min}\doteq \min(K)}\\
	\pifelsem \\
	\tab \especline{\exsts{k,l,r,K_l,K_r} \pvar x\doteq x\lstar x\mapsto k,l,r \lstar \bst{r,K_r} \lstar \bst{l,K_l} \lstar K\doteq K_l\uplus\{k\}\uplus K_r
	\lstar K_l<k \lstar k<K_r\lstar \\
	\alpha\doteq\card{K}\lstar\pvar l\doteq l\lstar l\neq\nil\lstar\pvar{min}\doteq\nil} \\
	\tab \text{\polish{// $K_l$ is a strict subset of $K$ and therefore of strictly smaller cardinality}} \\
	\tab\begin{leftvruled} {fr + ex}
		\tab \smallespecline{\pvar{min}\doteq\nil\lstar\pvar l\doteq l\lstar l\neq\nil\lstar\bst{l,K_l}\lstar\alpha>\card{K_l}} \\
		\tab \passign{\pvar{min}}{\mathtt{BSTFindMin}(\pvar l)} \\
		\tab \smallespecline{\pvar{min}\doteq\min(K_l)\lstar\pvar l\doteq l\lstar l\neq\nil\lstar\bst{l,K_l}\lstar\alpha>\card{K_l}}
	\tab\end{leftvruled}\\
	\tab \especline{\exsts{k,l,r,K_l,K_r} \pvar x\doteq x\lstar x\mapsto k,l,r \lstar \bst{r,K_r} \lstar \bst{l,K_l} \lstar K\doteq K_l\uplus\{k\}\uplus K_r
	\lstar K_l<k \lstar k<K_r\lstar \\
	\alpha\doteq\card{K}\lstar\pvar l\doteq l\lstar l\neq\nil\lstar\pvar{min}\doteq\min(K_l)} \\
	\tab \text{\polish{// $K\doteq K_l\uplus\{k\}\uplus K_r$, $K_l<k$, $k<K_r$ and $\min(K_l)=\min(K)$}} \\
	\tab \especline{\exsts{k,l,r,K_l,K_r} \pvar x\doteq x\lstar x\mapsto k,l,r \lstar \bst{r,K_r} \lstar \bst{l,K_l} \lstar K\doteq K_l\uplus\{k\}\uplus K_r
	\lstar K_l<k \lstar k<K_r\lstar \\
	\alpha\doteq\card{K}\lstar\pvar l\doteq l\lstar l\neq\nil\lstar\pvar{min}\doteq\min(K)} \\
	\}; \\
	\especline{(\exsts{k,l,r,K_l,K_r} \pvar x\doteq x\lstar x\mapsto k,l,r \lstar \bst{r,K_r} \lstar \bst{l,K_l} \lstar K\doteq K_l\uplus\{k\}\uplus K_r
	\lstar K_l<k \lstar k<K_r\lstar \\
	\alpha\doteq\card{K}\lstar\pvar l\doteq l\lstar l\doteq\nil\lstar\pvar{min}\doteq \min(K))~\vee \\
	(\exsts{k,l,r,K_l,K_r} \pvar x\doteq x\lstar x\mapsto k,l,r \lstar \bst{r,K_r} \lstar \bst{l,K_l} \lstar K\doteq K_l\uplus\{k\}\uplus K_r
	\lstar K_l<k \lstar k<K_r\lstar \\
	\alpha\doteq\card{K}\lstar\pvar l\doteq l\lstar l\neq\nil\lstar\pvar{min}\doteq\min(K))} \\
	\especline{\exsts{k,l,r,K_l,K_r} \pvar x\doteq x\lstar x\mapsto k,l,r \lstar \bst{r,K_r} \lstar \bst{l,K_l} \lstar K\doteq K_l\uplus\{k\}\uplus K_r
	\lstar K_l<k \lstar k<K_r\lstar \\
	\alpha\doteq\card{K}\lstar\pvar l\doteq l\lstar \pvar{min}\doteq \min(K)\lstar (l\doteq\nil\vee l \neq\nil)} \\
	\especline{\exsts{k,l,r,K_l,K_r} \pvar x\doteq x\lstar x\mapsto k,l,r \lstar \bst{r,K_r} \lstar \bst{l,K_l} \lstar K\doteq K_l\uplus\{k\}\uplus K_r
	\lstar K_l<k \lstar k<K_r\lstar \\
	\alpha\doteq\card{K}\lstar\pvar l\doteq l\lstar \pvar{min}\doteq \min(K)} \\
 	\preturn{\pvar{min}} \\
\end{array}
\]
Lastly, we need to prove the connection between the interior and exterior post-conditions:
\[
\begin{array}{l l}
	&\exsts{p_l,p_{min},p_x,k,l,r,K_l,K_r} p_x\doteq x\lstar x\mapsto k,l,r \lstar \bst{r,K_r} \lstar \bst{l,K_l} \lstar \\
	& K\doteq K_l\uplus\{k\}\uplus K_r \lstar K_l<k \lstar k<K_r\lstar\alpha\doteq\card{K}\lstar p_l\doteq l\lstar p_{min}\doteq \min(K) \lstar\pvar{ret}\doteq p_{min}\\
	\Leftrightarrow & \exsts{k,l,r,K_l,K_r} x\mapsto k,l,r \lstar \bst{r,K_r} \lstar \bst{l,K_l} \lstar \\
	& K\doteq K_l\uplus\{k\}\uplus K_r \lstar K_l<k \lstar k<K_r\lstar\alpha\doteq\card{K} \lstar\pvar{ret}\doteq \min(K)\\
	\Leftrightarrow & x\neq\nil \lstar \bst{x,K}\lstar\alpha\doteq\card{K} \lstar\pvar{ret}\doteq \min(K)\\
\end{array}
\]

\subsection{Binary Search Tree: BSTInsert}

Next, we consider a function that inserts a value into a BST:
\[
\begin{array}{l}
 	\mathtt{BSTInsert}(\pvar v,\pvar x)~\{ \\
 	\tab \pifelses{\pvar x = \nil} \\
 	\tab\tab \palloc{\pvar{x}}{3}; \pmutate{\pvar x}{\pvar v} \\
	\tab \pifelsem \\
	\tab\tab \pderef{\pvar k}{\pvar x}; \\
	\tab\tab \pifelses{\pvar k = \pvar v} \\
	\tab\tab\tab \pskip \\
	\tab\tab \pifelsem \\
	\tab\tab\tab \pifelses{\pvar v<\pvar k} \\
	\tab\tab\tab\tab \pderef{\pvar y}{\pvar x+1}; \passign{\pvar y}{\mathtt{BSTInsert}(\pvar v,\pvar y)}; \pmutate{\pvar x+1}{\pvar y} \\
	\tab\tab\tab \pifelsem \\
	\tab\tab\tab\tab \pderef{\pvar y}{\pvar x+2};  \passign{\pvar y}{\mathtt{BSTInsert}(\pvar v, \pvar y)};  \pmutate{\pvar x+2}{\pvar y} \\
	\tab\tab\tab \pifelsee \\
	\tab\tab \pifelsee \\
	\tab \pifelsee; \\
	\tab \preturn{\pvar x} \\
	\}
\end{array}
\]
For the specification, we require the $\predd{\bstroot}{x, \tree}$ stating that $x$ is the root of the tree $\tau$, for the same reason we need the $\listptr{x, \alpha}$ predicate for some list algorithms:
\[
\predd{\bstroot}{x, \tree} \defeq (\tree \doteq \emptytree \lstar x \doteq \nil) \lor (\exsts{k, \tree_l, \tree_r} \tree \doteq ((x, k), \tree_l, \tree_r))
\]
For completeness, we also give a functional description of our mathematical $\predd{BSTInsert}{\tree, (x', v)}$ function:
\[
\begin{array}{l}
\predd{BSTInsert}{\emptytree, (x', v)} \defeq ((x', v), \emptytree, \emptytree) \\[2mm]
\predd{BSTInsert}{((x, k), \tree_l, \tree_r), (x', v)} \defeq \\
\qquad \mathsf{if}\ v < k\ \mathsf{then} \\
\qquad\quad ((x, k), \predd{BSTInsert}{\tree_l, (x', v)}, \tree_r) \\
\qquad \mathsf{else\ if}\ k < v\ \mathsf{then} \\
\qquad\quad ((x, k), \tree_l, \predd{BSTInsert}{\tree_r, (x', v)}) \\
\qquad \mathsf{else} \\
\qquad\quad ((x, k), \tree_l, \tree_r)
\end{array}
\]
The specification we would like to prove is:
\[
\small
\begin{array}{c}
 {\color{blue} (\pvar x \doteq x \lstar \pvar v \doteq v \lstar \bst{x, \tree})} \\
 {\mathtt{BSTInsert}(\pvar x, \pvar v)} \\
 {\color{blue} (\exists x'.~ \bst{\pvar{ret}, \predd{BSTInsert}{\tree, (x', v)}} \lstar \predd{\bstroot}{x, \tree})}
\end{array}
\]
The proof sketch is as follows, using the measure that $\alpha$ equals $\predd{height}{\tree}$:
\[
\small
\begin{array}{l}
 	\mathtt{BSTInsert}(\pvar v,\pvar x)~\{ \\
	\hspace*{-0.63cm}\Gamma(\alpha) \vdash \specline{ \pvar x \doteq x \lstar \pvar v \doteq v \lstar \bst{x, \tree} \lstar \pvar k, \pvar y \doteq \nil \lstar \alpha \doteq \predd{height}{\tree}} \\
 	\tab \pifelses{\pvar x = \nil} \\
	\tab\tab \specline{ \pvar x \doteq x \lstar \pvar v \doteq v \lstar \bst{x, \tree} \lstar \pvar k, \pvar y \doteq \nil \lstar \shade{\pvar x \doteq \nil} \lstar \alpha \doteq \predd{height}{\tree}} \\
	\tab\tab \specline{ \pvar x \doteq x \lstar \pvar v \doteq v \lstar \shade{\predd{\bstroot}{x, \tree} \lstar \tree \doteq \emptytree} \lstar \pvar k, \pvar y \doteq \nil  \lstar \alpha \doteq \predd{height}{\tree}} \\
 	\tab\tab \palloc{\pvar{x}}{3}; \pmutate{\pvar x}{\pvar v} \\
	\tab\tab \specline{\shade{\pvar x \mapsto v, \nil, \nil} \lstar \pvar v \doteq v \lstar \predd{\bstroot}{x, \tree} \lstar \tree \doteq \emptytree \lstar \pvar k, \pvar y \doteq \nil  \lstar \alpha \doteq \predd{height}{\tree}} \\
	\tab\tab \text{\polish{// $\pvar x \mapsto v, \nil, \nil$ is equivalent to $\bst{\pvar x, ((\pvar x, v), \emptytree, \emptytree)}$}} \\
	\tab\tab \text{\polish{// When $\tree \doteq \emptytree$, $\predd{BSTInsert}{\tree, (\pvar x, v)} = ((\pvar x, v), \emptytree, \emptytree)$}} \\
	\tab\tab \specline{Q_1:~\shade{\bst{\pvar x, \predd{BSTInsert}{\tree, (\pvar x, v)}}} \lstar \predd{\bstroot}{x, \tree} \lstar \alpha \doteq \predd{height}{\tree} \lstar \tree \doteq \emptytree \lstar \pvar v \doteq v \lstar \pvar k, \pvar y \doteq \nil } \\
	\tab \pifelsem \\
	\tab\tab \specline{ \pvar x \doteq x \lstar \pvar v \doteq v \lstar \bst{x, \tree} \lstar \pvar k, \pvar y \doteq \nil \lstar \shade{\pvar x \dotneq \nil} \lstar \alpha \doteq \predd{height}{\tree}} \\
	\tab\tab \text{\polish{// Unfold BST, add on $\predd{\bstroot}{x, \tree}$ by equivalence }} \\
	\tab\tab \specline{\shade{\exsts {k, l, r, \tree_l, \tree_r}}~\pvar x \doteq x \lstar \pvar v \doteq v \lstar \shade{x \mapsto k, l, r \lstar \bst{r, \tree_r} \lstar \bst{l, \tree_l} \lstar \tree \doteq ((x, k), \tree_l, \tree_r) \lstar} \\ \shade{\tree_l < k \lstar k < \tree_r \star \predd{\bstroot}{x, \tree}} \lstar \pvar k, \pvar y \doteq \nil \lstar \alpha \doteq \predd{height}{\tree}} \\
	\tab\tab \pderef{\pvar k}{\pvar x}; \\
	\tab\tab \text{\polish{// $\pvar k$ in place of $k$, dropping the existential}} \\
	\tab\tab \specline{\exsts {l, r, \tree_l, \tree_r} \pvar x \doteq x \lstar \pvar v \doteq v \lstar x \mapsto \pvar k, l, r \lstar \bst{r, \tree_r} \lstar \bst{l, \tree_l} \lstar \tree \doteq ((x, \pvar k), \tree_l, \tree_r) \lstar \\ \tree_l < \pvar k \lstar \pvar k < \tree_r \star \predd{\bstroot}{x, \tree} \lstar \pvar y \doteq \nil \lstar \alpha \doteq \predd{height}{\tree}} \\
	\tab\tab \pifelses{\pvar k = \pvar v} \\
	\tab\tab\tab \specline{\exsts {l, r, \tree_l, \tree_r} \pvar x \doteq x \lstar \pvar v \doteq v \lstar x \mapsto \pvar k, l, r \lstar \bst{r, \tree_r} \lstar \bst{l, \tree_l} \lstar \tree \doteq ((x, \pvar k), \tree_l, \tree_r) \lstar \\ \tree_l < \pvar k \lstar \pvar k < \tree_r \star \predd{\bstroot}{x, \tree} \lstar \shade{\pvar k \doteq \pvar v} \lstar \pvar y \doteq \nil \lstar \alpha \doteq \predd{height}{\tree}} \\
	\tab\tab\tab \pskip \\
	\tab\tab\tab \specline{Q_2: \exsts {l, r, \tree_l, \tree_r} \pvar x \doteq x \lstar \pvar v \doteq v \lstar x \mapsto \pvar k, l, r \lstar \bst{r, \tree_r} \lstar \bst{l, \tree_l} \lstar \tree \doteq ((x, \pvar k), \tree_l, \tree_r) \lstar \\ \tree_l < \pvar k \lstar \pvar k < \tree_r \star \predd{\bstroot}{x, \tree} \lstar \pvar k \doteq \pvar v \lstar \pvar y \doteq \nil \lstar \alpha \doteq \predd{height}{\tree}} \\
	\tab\tab \pifelsem \\
	\tab\tab\tab \specline{\exsts {l, r, \tree_l, \tree_r} \pvar x \doteq x \lstar \pvar v \doteq v \lstar x \mapsto \pvar k, l, r \lstar \bst{r, \tree_r} \lstar \bst{l, \tree_l} \lstar \tree \doteq ((x, \pvar k), \tree_l, \tree_r) \lstar \\ \tree_l < \pvar k \lstar \pvar k < \tree_r \star \predd{\bstroot}{x, \tree}  \lstar \shade{\pvar k \dotneq \pvar v} \lstar \pvar y \doteq \nil \lstar \alpha \doteq \predd{height}{\tree}} \\
	\tab\tab\tab \pifelses{\pvar v<\pvar k} \\
	\tab\tab\tab\tab \specline{\exsts {l, r, \tree_l, \tree_r} \pvar x \doteq x \lstar \pvar v \doteq v \lstar x \mapsto \pvar k, l, r \lstar \bst{r, \tree_r} \lstar \bst{l, \tree_l} \lstar \tree \doteq ((x, \pvar k), \tree_l, \tree_r) \lstar \\ \tree_l < \pvar k \lstar \pvar k < \tree_r \star \predd{\bstroot}{x, \tree}  \lstar \shade{\pvar v \dotlt \pvar k} \lstar \pvar y \doteq \nil \lstar \alpha \doteq \predd{height}{\tree}} \\
	\tab\tab\tab\tab \pderef{\pvar y}{\pvar x+1}; \\
	\tab\tab\tab\tab \specline{\exsts {l, r, \tree_l, \tree_r} \pvar x \doteq x \lstar \pvar v \doteq v \lstar x \mapsto \pvar k, l, r \lstar \bst{r, \tree_r} \lstar \shade{\pvar y \doteq l \lstar \bst{l, \tree_l} \lstar \alpha - 1 \doteq \predd{height}{\tree_l}} \lstar \\ \tree \doteq ((x, \pvar k), \tree_l, \tree_r) \lstar  \tree_l < \pvar k \lstar \pvar k < \tree_r \star \predd{\bstroot}{x, \tree}  \lstar \pvar v \dotlt \pvar k \lstar \alpha \doteq \predd{height}{\tree}} \\
	\tab\tab\tab\tab \passign{\pvar y}{\mathtt{BSTInsert}(\pvar v,\pvar y)}; \\
	\tab\tab\tab\tab \specline{\exsts {l, r, \tree_l, \tree_r} \pvar x \doteq x \lstar \pvar v \doteq v \lstar x \mapsto \pvar k, l, r \lstar \bst{r, \tree_r} \lstar \\ \shade{(\exists x'.~ \bst{\pvar{y}, \predd{BSTInsert}{\tree_l, (x', v)}} \lstar \predd{\bstroot}{l, \tree_l} \lstar \alpha - 1 \doteq \predd{height}{\tree_l})} \lstar \\ \tree \doteq ((x, \pvar k), \tree_l, \tree_r) \lstar  \tree_l < \pvar k \lstar \pvar k < \tree_r \star \predd{\bstroot}{x, \tree}  \lstar \pvar v \dotlt \pvar k \lstar \alpha \doteq \predd{height}{\tree}} \\
	\tab\tab\tab\tab \pmutate{\pvar x+1}{\pvar y} \\
	\tab\tab\tab\tab \specline{\exsts {l, r, \tree_l, \tree_r} \pvar x \doteq x \lstar \pvar v \doteq v \lstar x \mapsto \pvar k, \shade{\pvar y}, r \lstar \bst{r, \tree_r} \lstar \\ {(\exists x'.~ \bst{\pvar{y}, \predd{BSTInsert}{\tree_l, (x', v)}} \lstar \predd{\bstroot}{l, \tree_l} \lstar \alpha - 1 \doteq \predd{height}{\tree_l})} \lstar \\ \tree \doteq ((x, \pvar k), \tree_l, \tree_r) \lstar  \tree_l < \pvar k \lstar \pvar k < \tree_r \star \predd{\bstroot}{x, \tree}  \lstar \pvar v \dotlt \pvar k \lstar \alpha \doteq \predd{height}{\tree}} \\
	\tab\tab\tab\tab \text{\polish{// Forget $\alpha - 1 = \predd{height}{\tree_l}$, as it can be entailed}} \\
	\tab\tab\tab\tab \text{\polish{// Forget $\exists l.~\predd{\bstroot}{l, \tree_l}$, as it is a tautology}} \\
	\tab\tab\tab\tab \text{\polish{// Re-instate the existential $l$ it as the value of $\pvar y$}} \\
	\tab\tab\tab\tab \specline{Q_3: \exsts {x', l, r, \tree_l, \tree_r} \pvar x \doteq x \lstar \pvar y \doteq l \lstar \pvar v \doteq v \lstar x \mapsto \pvar k, l, r \lstar  \bst{l, \predd{BSTInsert}{\tree_l, (x', v)}} \lstar \bst{r, \tree_r} \lstar \\  \tree \doteq ((x, \pvar k), \tree_l, \tree_r) \lstar  \tree_l < \pvar k \lstar \pvar k < \tree_r \star \predd{\bstroot}{x, \tree}  \lstar \pvar v \dotlt \pvar k \lstar \alpha \doteq \predd{height}{\tree}} \\
	\tab\tab\tab \pifelsem \\
	\tab\tab\tab\tab \specline{\exsts {l, r, \tree_l, \tree_r} \pvar x \doteq x \lstar \pvar v \doteq v \lstar x \mapsto \pvar k, l, r \lstar \bst{r, \tree_r} \lstar \bst{l, \tree_l} \lstar \tree \doteq ((x, \pvar k), \tree_l, \tree_r) \lstar \\ \tree_l < \pvar k \lstar \pvar k < \tree_r \star \predd{\bstroot}{x, \tree}  \lstar \shade{\pvar v \dotgt \pvar k} \lstar \pvar y \doteq \nil \lstar \alpha \doteq \predd{height}{\tree}} \\
	\tab\tab\tab\tab \pderef{\pvar y}{\pvar x+2};  \passign{\pvar y}{\mathtt{BSTInsert}(\pvar v, \pvar y)};  \pmutate{\pvar x+2}{\pvar y} \\
\end{array}
\]
\[
\small
\begin{array}{l}
	\tab\tab\tab\tab \text{\polish{// Analogously to previous branch}} \\
		\tab\tab\tab\tab \specline{Q_4: \exsts {x', l, r, \tree_l, \tree_r} \pvar x \doteq x \lstar \pvar y \doteq r \lstar \pvar v \doteq v \lstar x \mapsto \pvar k, l, r \lstar \bst{l, \tree_l} \lstar \\
                                           \tab\tab \bst{r, \predd{BSTInsert}{\tree_r, (x', v)}}  \lstar \predd{\bstroot}{x_r, \tree_r} \lstar \tree \doteq ((x, \pvar k), \tree_l, \tree_r) \lstar  \tree_l < \pvar k \lstar \pvar k < \tree_r \star \\
                                           \tab\tab \predd{\bstroot}{x, \tree} \lstar \pvar v \dotgt \pvar k \lstar \alpha \doteq \predd{height}{\tree}} \\
	\tab\tab\tab \pifelsee \\
	\tab\tab \pifelsee \\
	\tab \preturn{\pvar x} \\
	\tab \specline{\exists x', v', k', y'.~((Q_1 \lor Q_2 \lor Q_3 \lor Q_4) \lstar \pvar{ret} \doteq \pvar x)[x',v',k',y'/\pvar x, \pvar v, \pvar k, \pvar y]} \\
	\tab \specline{(\exists x', v', k', y'.~(Q_1 \lstar \pvar{ret} \doteq \pvar x)[x',v',k',y'/\pvar x, \pvar v, \pvar k, \pvar y])~\lor \\
	               (\exists x', v', k', y'.~(Q_2 \lstar \pvar{ret} \doteq \pvar x)[x',v',k',y'/\pvar x, \pvar v, \pvar k, \pvar y])~\lor \\
	               (\exists x', v', k', y'.~(Q_3 \lstar \pvar{ret} \doteq \pvar x)[x',v',k',y'/\pvar x, \pvar v, \pvar k, \pvar y])~\lor \\
	               (\exists x', v', k', y'.~(Q_4 \lstar \pvar{ret} \doteq \pvar x)[x',v',k',y'/\pvar x, \pvar v, \pvar k, \pvar y])} \\
	\tab \specline{(Q(\alpha) \lstar  \tree \doteq \emptytree)~\lor \\
	               (Q(\alpha) \lstar (\exists \tree_l, \tree_r.~\tree \doteq ((x, v), \tree_l, \tree_r)))~\lor \\
	               (Q(\alpha) \lstar (\exists k, \tree_l, \tree_r.~\tree \doteq ((x, k), \tree_l, \tree_r) \lstar v \dotlt k))~\lor \\
	               (Q(\alpha) \lstar (\exists k, \tree_l, \tree_r.~\tree \doteq ((x, k), \tree_l, \tree_r) \lstar v \dotgt k))}  \\
	\tab \text{\polish{// The four additional assertions cover the space of possibilities }} \\
	\tab \specline{Q(\alpha)}  \\
	\} \\
	\specline{Q(\alpha)}  \\
\end{array}
\]

\[
\begin{array}{r@{~}l}
& \exists x', v', k', y'.~(Q_1 \lstar \pvar{ret} \doteq \pvar x)[x',v',k',y'/\pvar x, \pvar v, \pvar k, \pvar y] \\
\Leftrightarrow & \exists ...~\bst{x', \predd{BSTInsert}{\tree, (x', v)}} \lstar \predd{\bstroot}{x, \tree} \lstar \alpha \doteq \predd{height}{\tree} \lstar \tree \doteq \emptytree \lstar v' \doteq v \lstar \\
& \tab\tab\tab\tab k', y' \doteq \nil \lstar \pvar{ret} \doteq x' \\
& \text{\polish{// Drop $v'$, $k'$, $y'$ }} \\
\Leftrightarrow & \exists x'.~\bst{x', \predd{BSTInsert}{\tree, (x', v)}} \lstar \predd{\bstroot}{x, \tree} \lstar \alpha \doteq \predd{height}{\tree} \lstar \pvar{ret} \doteq x' \lstar \tree \doteq \emptytree \\
& \text{\polish{// Substitute $\pvar{ret}$ for $x'$ in first parameter of BST}} \\
& \text{\polish{// Lose $\pvar{ret} \doteq x'$, it is entailed given definition of BSTInsert when $\tree \doteq \emptytree$ }} \\
\Leftrightarrow & \exists x'.~\bst{\pvar{ret}, \predd{BSTInsert}{\tree, (x', v)}} \lstar \predd{\bstroot}{x, \tree} \lstar \alpha \doteq \predd{height}{\tree} \lstar \tree \doteq \emptytree \\
\Leftrightarrow & Q(\alpha) \lstar \tree \doteq \emptytree
\end{array}
\]

\[
\begin{array}{r@{~}l}
& \exists x', v', k', y'.~(Q_2 \lstar \pvar{ret} \doteq \pvar x)[x',v',k',y'/\pvar x, \pvar v, \pvar k, \pvar y] \\
\Leftrightarrow & \exists x', v', k', y', l, r, \tree_l, \tree_r. \\
& x' \doteq x \lstar v' \doteq v \lstar x \mapsto k', l, r \lstar \bst{r, \tree_r} \lstar \bst{l, \tree_l} \lstar \tree \doteq ((x, k'), \tree_l, \tree_r) \lstar \\ & \tree_l < k' \lstar k' < \tree_r \star \predd{\bstroot}{x, \tree} \lstar k' \doteq v' \lstar y' \doteq \nil \lstar \alpha \doteq \predd{height}{\tree} \lstar \pvar{ret} \doteq x' \\
\Leftrightarrow & \exists x', k', l, r, \tree_l, \tree_r. \\
& \text{\polish{// Drop $v'$ (keeping $k' \doteq v$), $y'$ }} \\
& \text{\polish{// Add $Q_{\mathit{aux}}: (\exists \tree_l, \tree_r.~\tree \doteq ((x, v), \tree_l, \tree_r))$ (it can be entailed)}} \\
& \text{\polish{// Substitute $x$ for $x'$ in the unfolding of BST}} \\
& \text{\polish{// Forget $x' \doteq x$ and $k' \doteq v$, as they can be entailed given rest and $Q_{\mathit{aux}}$}} \\
& x' \mapsto k', l, r \lstar \bst{r, \tree_r} \lstar \bst{l, \tree_l} \lstar \tree \doteq ((x', k'), \tree_l, \tree_r) \lstar \\ & \tree_l < k' \lstar k' < \tree_r \star \predd{\bstroot}{x, \tree} \lstar \alpha \doteq \predd{height}{\tree} \lstar \pvar{ret} \doteq x' \lstar (\exists \tree_l, \tree_r.~\tree \doteq ((x, v), \tree_l, \tree_r))  \\
& \text{\polish{// Fold $\bst{x', \tree}$; by definition, $\predd{BSTInsert}{\tree, (x', v)} = \tree$ when $Q_\mathit{aux}$ holds}} \\
\Leftrightarrow & \exists x'.~\bst{x', \predd{BSTInsert}{\tree, (x', v)}} \lstar \predd{\bstroot}{x, \tree} \lstar \alpha \doteq \predd{height}{\tree} \lstar \pvar{ret} \doteq x' \lstar \\
& \tab\tab\tab\tab (\exists \tree_l, \tree_r.~\tree \doteq ((x, v), \tree_l, \tree_r))  \\
\Leftrightarrow & Q(\alpha) \lstar (\exists \tree_l, \tree_r.~\tree \doteq ((x, v), \tree_l, \tree_r))
\end{array}
\]

\[
\begin{array}{r@{~}l}
& \exists x', v', k', y'.~(Q_3 \lstar \pvar{ret} \doteq \pvar x)[x',v',k',y'/\pvar x, \pvar v, \pvar k, \pvar y] \\
\Leftrightarrow & \exists x', v', k', y', l, r, \tree_l, \tree_r, x''. \\
&  x' \doteq x \lstar y' \doteq l \lstar v' \doteq v \lstar x \mapsto k', l, r \lstar  \bst{l, \predd{BSTInsert}{\tree_l, (x'', v)}} \lstar \bst{r, \tree_r} \lstar \\
& \tree \doteq ((x, k'), \tree_l, \tree_r) \lstar  \tree_l < k' \lstar k' < \tree_r \star \predd{\bstroot}{x, \tree}  \lstar  v' \dotlt k' \lstar \alpha \doteq \predd{height}{\tree} \lstar \pvar{ret} \doteq x' \\
\Leftrightarrow & \exists x', x'', k', l, r, \tree_l, \tree_r. \\
& \text{\polish{// Drop $v'$, $y'$ }} \\
& \text{\polish{// Add $Q_{\mathit{aux}}: (\exists k, \tree_l, \tree_r.~\tree \doteq ((x, k), \tree_l, \tree_r) \lstar v \dotlt k)$ (it can be entailed)}} \\
& \text{\polish{// Substitute $x$ for $x'$ in the unfolding of BST}} \\
& \text{\polish{// Forget $x' \doteq x$ as it can be entailed given rest and $Q_{\mathit{aux}}$}} \\
&  x' \mapsto k', l, r \lstar  \bst{l, \predd{BSTInsert}{\tree_l, (x'', v)}} \lstar \bst{r, \tree_r} \lstar \tree \doteq ((x', k'), \tree_l, \tree_r) \lstar \\
&   \tree_l < k' \lstar k' < \tree_r \star \predd{\bstroot}{x, \tree}  \lstar  v \dotlt k' \lstar \alpha \doteq \predd{height}{\tree} \lstar \pvar{ret} \doteq x' \lstar \\
&   (\exists k, \tree_l, \tree_r.~\tree \doteq ((x, k), \tree_l, \tree_r) \lstar v \dotlt k)\\
\Leftrightarrow & \exists x', x'', k', l, r, \tree_l, \tree_r. \\
& \text{\polish{// Set up folding using the definition of BSTInsert and the present inequalities}} \\
& \text{\polish{// Drop $v < k'$ as it can be entailed given rest and $Q_{\mathit{aux}}$}} \\
&  x' \mapsto k', l, r \lstar  \bst{l, \predd{BSTInsert}{\tree_l, (x'', v)}} \lstar \bst{r, \tree_r} \lstar \\
&  \predd{BSTInsert}{\tree, (x'', v)} \doteq ((x, k), \predd{BSTInsert}{\tree_l, (x'', v)}, \tree_r) \lstar \\
&   \predd{BSTInsert}{\tree_l, (x'', v)} < k' \lstar k' < \tree_r \lstar \predd{\bstroot}{x, \tree}  \lstar \alpha \doteq \predd{height}{\tree} \lstar \\
& \pvar{ret} \doteq x' \lstar (\exists k, \tree_l, \tree_r.~\tree \doteq ((x, k), \tree_l, \tree_r) \lstar v \dotlt k)\\
& \text{\polish{// Fold }} \\
\Leftrightarrow & \exists x', x''.~\bst{x', \predd{BSTInsert}{\tree, (x'', v)}} \lstar \predd{\bstroot}{x, \tree}  \lstar \alpha \doteq \predd{height}{\tree} \lstar \pvar{ret} \doteq x' \lstar \\ & (\exists k, \tree_l, \tree_r.~\tree \doteq ((x, k), \tree_l, \tree_r) \lstar v \dotlt k) \\
& \text{\polish{// Substitute $\pvar{ret}$ for $x'$, drop $x'$, rename $x''$ as $x'$}} \\
\Leftrightarrow & \exists x'.~\bst{\pvar{ret}, \predd{BSTInsert}{\tree, (x', v)}} \lstar \predd{\bstroot}{x, \tree} \lstar \alpha \doteq \predd{height}{\tree} \lstar \\
& (\exists k, \tree_l, \tree_r.~\tree \doteq ((x, k), \tree_l, \tree_r) \lstar v \dotlt k)  \\
\Leftrightarrow & Q(\alpha) \lstar (\exists k, \tree_l, \tree_r.~\tree \doteq ((x, k), \tree_l, \tree_r) \lstar v \dotlt k)
\end{array}
\]

The reasoning for $Q_4$ is analogous to that for $Q_3$, obtaining
\[
\begin{array}{r@{~}l}
& \exists x', v', k', y'.~(Q_3 \lstar \pvar{ret} \doteq \pvar x)[x',v',k',y'/\pvar x, \pvar v, \pvar k, \pvar y] \\
\Leftrightarrow & \exists x', v', k', y', l, r, \tree_l, \tree_r, x''. \\
&  x' \doteq x \lstar y' \doteq l \lstar v' \doteq v \lstar x \mapsto k', l, r \lstar  \bst{l, \predd{BSTInsert}{\tree_l, (x'', v)}} \lstar \bst{r, \tree_r} \lstar \\
& \tree \doteq ((x, k'), \tree_l, \tree_r) \lstar  \tree_l < k' \lstar k' < \tree_r \star \predd{\bstroot}{x, \tree}  \lstar  v' \dotlt k' \lstar \alpha \doteq \predd{height}{\tree} \lstar \pvar{ret} \doteq x' \\
\Leftrightarrow & ... \\
\Leftrightarrow & Q(\alpha) \lstar (\exists k, \tree_l, \tree_r.~\tree \doteq ((x, k), \tree_l, \tree_r) \lstar v \dotgt k)
\end{array}
\]

\pagebreak

%% file: sections/app-sl-list-length.tex
\newpage
\section{SL Proof of List-length Algorithm}%
\label{apdx:sl-list-length}

\[
\begin{array}{l}
\hspace*{-0.46cm}
\text{\polish{// SL proofs do not require the measure $\alpha$}} \\
\hspace*{-0.46cm}\Gamma\vdash
\specline{\pvar x \doteq x \lstar \llist{x, n} \lstar \pvar r \doteq \nil } \\
\pifelses{ \pvar{x} = \nil } \\
\tab\specline{\pvar x \doteq x \lstar \llist{x, n} \lstar \pvar r\doteq\nil \lstar {\pvar x \doteq \nil}} \\ 
\tab\text{\polish{// Forward consequence: From $\pvar x \doteq x$, $\pvar x \doteq \nil$, and $\llist{x, n}$, learn that $n = 0$. Forget $\pvar x$.}} \\
\tab\specline{\llist{x, n} \lstar \pvar r\doteq\nil \lstar n \doteq 0} \\ 
\tab\passign{\pvar{r}}{0} \\
\tab\specline{\llist{x, n}  \lstar \pvar r \doteq 0 \lstar n \doteq 0} \\ 
\tab\text{\polish{// Forward consequence: From $\pvar r \doteq 0$ and $n \doteq 0$, learn that $\pvar r \doteq n$. Forget $\pvar r = 0$ and $\pvar n = 0$.}} \\
\tab\specline{\llist{x, n} \lstar \pvar r \doteq n } \\
\pifelsem \\
\tab\specline{\pvar x \doteq x \lstar {\llist{x,n}} \lstar \pvar r\doteq\nil \lstar {\pvar x \dotneq \nil}} \\ 
\tab\text{\polish{// Forward consequence: Unfold $\llist{x, n}$:}} \\
\tab\text{\polish{// $\models \llist{\pvar x, n} \lstar \pvar x \dotneq \nil \Rightarrow \exists v, x'.~\pvar x \mapsto v, x' \lstar \llist{x', n-1}$.}} \\
\tab\text{\polish{// Note, as with the ISL proof, the equivalence in the unfold step has been replaced by}} \\
\tab\text{\polish{// an implication, however, the step remains unchanged modulo this change.}} \\
\tab\specline{{\exists v, x'}.~\pvar x \doteq x \lstar  {x \mapsto v, x' \lstar \llist{x', n-1}}\lstar \pvar r \doteq \nil} \\ 
\tab\passign{\pvar{x}}{[\pvar x + 1]}; \\
\tab\specline{\exists v, x'.~{\pvar x \doteq x'} \lstar  x \mapsto v, x' \lstar {\llist{x', n-1} } \lstar \pvar r \doteq \nil} \\ 
\tab\passign{\pvar{r}}{\mathtt{LLen}(\pvar x)}; \\
\tab\specline{\shade{\exists v, x'}.~\pvar x \doteq x' \lstar \shade{x \mapsto v, x' \lstar \llist{x', n-1}} \lstar {\pvar r \doteq n -1}} \\ 
\tab\text{\polish{// Forward consequence: Forget $\pvar x$ and fold $\llist{x, n}$:}} \\
\tab\text{\polish{// $\models \exists v, x'.~\pvar x \mapsto v, x' \lstar \llist{x', n-1} \Rightarrow \llist{\pvar x, n}$.}} \\
\tab\specline{\shade{\llist{x, n}} \lstar{\pvar r \doteq n -1}} \\ 
\tab\passign{\pvar{r}}{\pvar r + 1} \\
\tab\specline{{\llist{x, n}} \lstar{\pvar r \doteq n}} \\ 
\pifelsee\\
\specline{\llist{x, n} \lstar{\pvar r \doteq n}} 
\end{array}
\]